\documentclass[12pt, oneside]{amsart}
\pdfoutput=1
\usepackage{geometry}
\usepackage{fancyhdr}
\usepackage{lmodern}
\usepackage{amsmath}
\usepackage{amssymb}
\usepackage{amsthm}
\usepackage{comment}
\usepackage{color}
\usepackage[usenames,dvipsnames]{xcolor}
\usepackage{graphicx}
\usepackage{wrapfig}
\usepackage{subfig}
\usepackage{soul}
\usepackage{tikz}
\usepackage[countmax]{subfloat}
\usepackage{float}
\usepackage{rotfloat}
\usepackage{setspace}
\usepackage{esint}
\usepackage[authoryear]{natbib}
\definecolor{MyDarkBlue}{rgb}{0,0.08,0.45}
\usepackage[unicode=true,pdfusetitle, bookmarks=true,bookmarksnumbered=false,bookmarksopen=false, breaklinks=false,pdfborder={0 0 1},backref=section,colorlinks=true, linkcolor = MyDarkBlue, citecolor = MyDarkBlue]{hyperref}
\usepackage{breakurl}
\usepackage[para]{threeparttable}
\usepackage{multirow}

\PassOptionsToPackage{normalem}{ulem}
\usepackage{ulem}

\makeatletter

\numberwithin{equation}{section}

\newtheorem{alg}{{\bf\sc Algorithm}}
\newtheorem{thm}{{\bf \sc Theorem}}[section]

\newtheorem{cor}{{\bf\sc Corollary}}[section]
\newtheorem{prop}{{\bf\sc Proposition}}[section]

\providecommand{\E}{\mathrm{E}}

\providecommand{\cov}{\mathrm{cov}}

\providecommand{\Prob}{\mathrm{P}}

\renewcommand{\Pr}{\Prob}

\providecommand{\1}{{\bf 1}}

\renewenvironment{proof}[1][Proof]{\noindent\text{#1.} }{\ \rule{0.5em}{0.5em}}

\makeatother

\title[Estimating networks with ARD]{Using Aggregated Relational Data to feasibly identify network structure without network data}

\author{Emily Breza$^{\dagger}$ }
\author{Arun G. Chandrasekhar$^{\ddagger}$}
\author{Tyler H. McCormick$^{\S}$ }
\author{Mengjie Pan$^{\star}$ }
\date{This version \today.}

\thanks{We thank Liran Einav,  Paul Goldsmith-Pinkham, Abhijit Banerjee, Esther Duflo, Axel Gandy, Ben Golub, Rema Hanna, Guy Harling, Jeff Eaton, Matthew Jackson, Michael Kremer, Rachael Meager, Betsy Ogburn, Elie Tamer, Tian Zheng and participants at various seminars/conferences
who provided helpful comments. We also thank Shobha Dundi, Varun Kapoor, Devika Lakhote, Ambika Sharma, Sneha Stephen, 
Tithee Mukhopadhyay, and Gowri Nagraj for outstanding research assistance.  This work is partially supported by grant SES-1559778 from the National Science foundation and grant number K01 HD078452 from the National Institute of Child Health and Human Development (NICHD).  This material is based upon work supported by, or in part by, the U.S. Army Research
Laboratory and the U.S. Army Research Office under grant number W911NF-12-1-0379.
}
\thanks{$^{\dagger}$Department of Economics, Harvard University; J-PAL; NBER. {ebreza@fas.harvard.edu}}
\thanks{$^{\ddagger}$Department of Economics, Stanford University; J-PAL; NBER. {arungc@stanford.edu}}
\thanks{$^{\S}$Departments of Statistics and Sociology, University of Washington. {tylermc@uw.edu}}
\thanks{$^{\star}$Department of Statistics, University of Washington. {mpan1@uw.edu}}

\begin{document}

\begin{titlepage}

\begin{abstract}

Social network data is often prohibitively expensive to collect, limiting empirical network research. We propose an inexpensive and feasible strategy for network elicitation using Aggregated Relational Data (ARD) -- responses to questions of the form ``how many of your links have trait $k$?'' Our method uses ARD to recover parameters of a network formation model, which permits the estimation of any arbitrary node- or graph-level statistic. We characterize both theoretically and empirically for which network features the procedure works.
In simulated and real-world graphs, the method performs well at matching a range of network characteristics. We replicate the results of two field experiments that used network data, and draw similar conclusions with ARD alone.


\end{abstract}

\maketitle
\thispagestyle{empty}
\end{titlepage}

\section{Introduction}

There has been a groundswell of empirical research on social and economic networks.\footnote{See, e.g.,  \citet*{karlanmrs2009,centola2010spread,tontarawongsa2011public,ligon2012motives,caijs2013,carrell2013natural,beaman2016can,blumenstock2016airtime,alatas2016network}. Also see \citet{chuang2015social,aralnetworked2016,boucher2016some,breza2016field} for overviews of empirical work using network data.} Nonetheless, a major barrier to entry into this space is access to network data, which is often extremely costly to collect.  A typical network elicitation exercise requires, (1) enumerating every member of the network in a census, (2)  asking each subject to name those individuals with whom they have a relationship and in what capacity, and (3) matching each individual's list of social connections back to the census. In field work, this can be difficult and expensive. Further, in other contexts, such as measuring networks of financial intermediaries or high-risk populations, proprietary data and privacy concerns may render steps (2) and (3) impossible. Moreover, this process needs to be repeated across many networks to conduct convincing inference. These barriers place significant limitations on conducting high-quality work in this space and discourage research, especially by those without access to considerable resources. 
 
The contribution of this paper is to present a technique that makes network research scalable and accessible on a budget. We propose that researchers collect aggregated relational data (ARD). ARD are responses to questions of the form \begin{quote}
\emph{``Think of all of the households in your village with whom you <<INSERT ACTIVITY>>. How many of these have trait $k$?''}
\end{quote}
ARD is considerably cheaper to obtain than full or even partial-network data. We show, using J-PAL South Asia cost estimates, that collecting ARD leads to a 70-80\% cost reduction.\footnote{While we present empirical evidence from village and neighborhood  networks in India, the method can also be extended to other settings.  See Section \ref{sec:conclusion} for a discussion of applications to firm and banking networks.}

Our proposed method is intuitive and comes down to the following three simple observations. First, ARD is considerably cheaper and easier to collect than network data. Second, ARD provides the researcher with enough information to identify parameters of an oft-used and standard network formation model in the statistics literature (see e.g. \cite{hoffrh2002}). The  argument builds on prior work by~\cite{mccormick2015latent}, which shows how the network formation model is related to a likelihood that depends only on ARD. We describe this and present an identification argument. Third, this parametric model of network formation is sufficiently rich to capture a number of features of real-world network structures, as we demonstrate through myriad simulations and empirical exercises. We characterize both theoretically and empirically for which network features the procedure works well.

We examine the performance of our method for estimating functions of the graph in several ways. First, develop a straightforward theoretical taxonomy, confirmed by empirical evaluation, that gives intuition about when the method will work under correct specification. Using a battery of simulations we show that we are able to guess what the underlying network structure looks like from the ARD, even as we vary the sparsity/density of the network, the size of the network, and the sampling share to a reasonable degree.   

Of course, real-world network data need not have been generated by the data generating process of our network-formation model. So we next consider an example where we have complete network data in nearly 16,500 households across 75 villages in Karnataka, India \citep*{banerjeegossip}. We show that had we collected ARD in these villages, even on a sample of 30\%, we would have been able to estimate reasonably-well a variety of features of the network that economists care about. 

We then provide two examples of recent research where either full or partial network data had been collected. \cite{SavingsMonitors} study how the observation of one's savings behavior by more central individuals in the network leads to greater savings in order to maintain a reputation for being responsible. We show with constructed ARD, we can replicate the paper's findings.  \cite*{banerjeebdk2016} use partial network data to study how exposure to microcredit erodes social capital by reducing support. The authors in part collected survey ARD in this sample, and we show we can replicate the findings. Further, the ARD enables conclusions about how microcredit exposure affected the neighborhood-level informal financial network structure. 
These examples show the effectiveness of our approach across different contexts and how ARD would have helped in policy-relevant empirical work. Researchers could have reached their conclusions without collecting full network data, which also means that the financial barrier to entry for such research would be considerably lower, thereby democratizing in part this research frontier.

We present a sample budget for survey data collection of full network data in 120 villages. Collecting ARD reduces the costs by approximately 70-80\%, depending on the sampling rate, using budgets prepared by J-PAL South Asia. While direct measurements of the network are always preferable to any estimation protocol, our calculations demonstrate that our proposed method can substantially expand the scope for and access to empirical networks research.

\subsection*{Overview of method}
For the bulk of the paper, we consider settings where we have ARD for a randomly-selected subset of nodes in the network and a basic vector of covariates for the full set of nodes.  ARD counts the number of links an agent has to members of different subgroups in the population. The core insight of our approach is that by combining ARD with a network formation model, we can derive the posterior distribution for the graph.  To do this, we assume a network formation model, which we refer to as the latent distance model, where the probability of a connection depends on individual heterogeneity and the positions of nodes in a latent social space~\citep{hoffrh2002}.  The distance between nodes in the space is a pair-specific latent variable that is inversely related to the probability of a tie: nodes that are closer together in the latent space are more likely to form ties.  The propensity to form ties across pairs is assumed conditionally independent given the latent variables.  ARD gives us information on where different subgroups lie relative to one another in this latent space. That is, ARD allows us to triangulate the relative locations of nodes. In prior work, \cite{mccormick2015latent} show how to relate the network formation model to a likelihood that depends only on ARD.  We extend that result and show how we can recover the parameters of the network formation model.  In our case, this consists of both individual-level effects for every node in the sample as well as the location of all nodes in the latent-space. Using a Bayesian framework for inference, we show that the choice of prior distribution has minimal impact on our ability to accurately recover moments for a variety of network configurations.  We note that, equipped with estimates of the degree distribution as well as the latent space locations in the ARD sample, we can use the demographic covariates for the entire sample to estimate the degree, fixed-effects, and latent locations for the entire population. We can then draw from the posterior distribution over graphs given the ARD response vector and compute network statistics based on these posterior samples.

\begin{figure}[!h]
\includegraphics[scale = 0.7, trim = 0in 0.67in 0in 0.82in, clip = true]{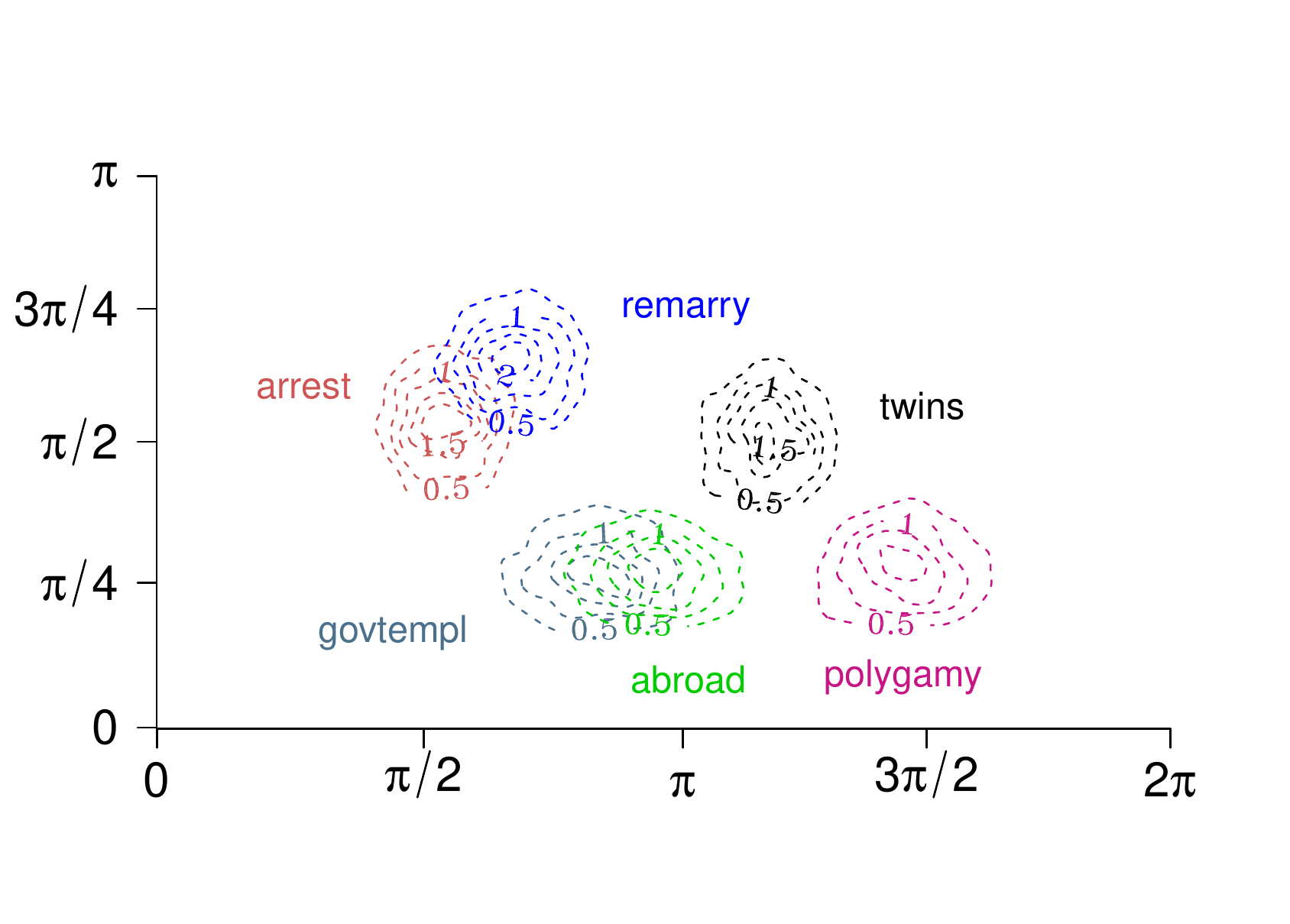}
\caption{Plot of the posterior densities for six ARD characteristic groups from Hyderabad.  The latent surface, a sphere, is represented by a cylindrical projection, with the vertical and horizontal axes representing latitude and longitude.  Positions of the groups indicate similarity in the networks of respondents that report connections with the group.  Concentration of the posterior density represents heterogeneity in the number known by respondents. \label{fig:intro_schematic}}
\end{figure}
Figure \ref{fig:intro_schematic} provides a simple illustration from one neighborhood in Hyderabad, India, where we collected ARD. The figure plots the positions on the latent surface, here a sphere, of six characteristic groups: households with histories of arrests, remarriages, members working abroad (likely in the Middle East), government employees, and twins. Several patterns emerge in this example. First, people tend to have joint knowledge of households with arrests and remarriages, consistent with both characteristics carrying negative social stigma.  Second, the arrested population is tightly correlated in space in comparison to other groups, indicating more extreme heterogeneity in the number of arrested individuals respondents know.  Third, people who know individuals with government employment also often know people who have household members abroad, again consistent with the local context where both government jobs and foreign migration require connections and lead to higher incomes.

The attractive features of our approach are not without costs. Our approach is  parametric, relying on guessing the network structure through the pseudo-true parameters of the latent distance formation model estimated from ARD. It can do no better than the best latent distance model at capturing the likely distribution that generated the network. It cannot, for example, represent clustering in a way that violates the triangle inequality.\footnote{For an example of a network formation model which can do this, see \cite{chandrasekharj2012}.}  To see this, consider a two-dimensional Euclidean space with four groups that have equal probability of cross group interaction. If the data generating process has this feature, we will not capture it well.

\subsection*{Relation to the literature}
Our work contributes to and builds on several literatures. First, there is a nascent literature that seeks to apply the lessons from the economics of networks without having access to network data (e.g.,  \cite{beaman2016can},  \cite{banerjeegossip}, and \cite{chassangetal2017}). These methods are limited because they only speak to identifying central individuals  or focus on proxies. Prior work shows that proxies such as geography or ethnic divisions do not capture the network well and augmenting sampled network data, which works, can still be expensive \citep{chandrasekharl2016}. Our approach does not restrict the researcher to inferences about  one specific aspect of the data, instead providing  a blueprint to recover a distribution over the entire graph at minimal cost.

 Second, our work builds on a sizable literature on ARD, but expands both the context and inferential quantities of interest.  In contrast to our work, most previous work on ARD focused on estimating the size of ``hard-to-reach'' populations (see e.g.  \citet{killworth1998estimation} or \citet{bernard2010counting}).  These groups consist of individuals who are outside the sampling frame of most surveys. Rather than needing to reach these individuals directly, using ARD allows researchers to study individuals through their interactions with others who are captured by more traditional sampling strategies. ~\cite{bernard2010counting} use ARD to estimate the number of individuals impacted by an earthquake whereas~\cite{kadushin2006scale} use ARD to estimate the number of individuals using heroine.\footnote{Perhaps the most common use of ARD is to estimate the number of individuals who are considered  high risk for HIV/AIDS (e.g., ~\cite{maghsoudi2014network},~\cite{guo2013estimating},~\cite{ezoe2012population},~\cite{salganik2011assessing}).}

The primary tool for estimating population size with ARD is the Network Scale-up Method (N-Sum) and variations thereof. Say the goal is to estimate the number of injection drug users in the population.  If a respondent reports knowing two injection drug users out of one-hundred total contacts, then approximately two percent of the respondent's network consists of individuals who are injection drug users.  If the respondent's network is characteristic, then in a population of 300,000,000 individuals, this would mean there are about 6,000,000 injection drug users.  Recent work has paid attention to estimating other features of the network\footnote{\cite{zheng2006many} estimate heterogeneity in the propensity to know members of groups, or overdispersion.}, but the majority of work on ARD still focuses on estimating population sizes.  As we do not focus on populations that are hard-to-reach, we can ask directly about whether a respondent is a member of a group  to estimate population sizes.  This distinction is essential for ``scaling'' a respondent's degree.  If the size of each ARD group and the total population are known, we can use the N-Sum logic to estimate individuals' degrees.

The closest related work from the ARD literature is~\cite{mccormick2015latent} -- here, we use the same network formation model and build on derivations that are the key contribution of that work.  Specifically,~\cite{mccormick2015latent} show that, for a specific formation model, it is possible to arrive at a likelihood that is informed by information in ARD.  That is,they interpret and do inference on a likelihood for ARD.  While we also have this likelihood, in our work it is merely an intermediate step.  In our paper, we perform inferences about the parameters of the formation model itself.  By explicitly making the link to the formation model, we can generate graphs and compute both graph and individual level statistics.  

Third, our latent surface model\footnote{In the context where the goal is inference about a regression coefficient that varies based on network connections, \cite{auerbach2016identification} presents a more general framework that links network formation to a function of distance between unobservable social characteristics that drive formation.} is closely related to the $\beta$-model~\citep{holland1981exponential,hunter2004mm,park2004statistical, blitzstein2011sequential} and the properties examined in \citet{chatterjeeds2010} and \citet{graham2014econometric}. Every node has a fixed-effect. Links form conditionally independently given the fixed effects of the nodes involved, modulated by a function of distance between the nodes in a latent space. Relative to the~\citet{graham2014econometric}  and ~\citet{chatterjeeds2010} models, our model places nodes in a latent space (as in \citet{hoffrh2002}), which we are trying to estimate, whereas the former only allows for observable covariates, and the latter has none.  Whereas previous approaches consider an asymptotic frame based on a growing graph, we consider an explicitly sampling-based framework. We empirically compare our proposed model to the beta model in Appendix~\ref{sec:betacompare}.

\subsection*{Organization}
 We begin with an overview of our method for an applied researcher in Section \ref{sec:overview}. Section \ref{sec:framework} presents the full framework, model, and estimation algorithm. In Section \ref{sec:sims} evaluates when, and how well, the method performs using simple theory and a variety of simulated graphs. Section \ref{sec:results} shows how our method works when we apply it to 75 village where we have complete network data. In Section \ref{sec:application}, we apply our results to two empirical examples. Section \ref{sec:cost} demonstrates the 70-80\% cost-savings of ARD versus full network elicitation. Section \ref{sec:conclusion} concludes.

\section{Overview of Method}\label{sec:overview}

We begin with a simple overview of the proposed method. Suppose that a researcher is interested in studying networks in a set of rural villages. A village network with $n$ households is given by  ${\bf g}$, which  is a collection of links $ij$ where $g_{ij}=1$ if and only if households $i$ and $j$ are linked and $g_{ij}=0$ otherwise.  To fix ideas, suppose that the researcher wants to learn how some outcome variable $W$ is related to a network statistic (or a vector of statistics) of interest $S({\bf g})$. Or, perhaps the researcher is interested in how a treatment (such as exposure to microcredit) affects features of network structure, $S({\bf g})$.

Our procedure takes five steps.
\begin{enumerate}
    \item[I.] {\bf Conduct ARD survey:} Sample a share $\psi$ (e.g., 30\%) of households. Have each enumerate a list of their network links.\footnote{\textcolor{purple}{}Note that this gives a direct estimate of the respondent's degree. The method laid out in Section \ref{sec:framework} does not require this and can also produce estimates for expected degree based on the ARD responses alone.} 
    Ask 5-8 ARD questions, such as
    \begin{quote}
         \emph{``How many households among your network list do you know where any adult has had typhoid, malaria, or cholera in the past six months?"}
    \end{quote}
    The ARD response for a household $i$ is
    \[
    y_{ik} = \sum_{j}g_{ij}\cdot {\bf 1}\{j \text{ has had one of those diseases in past 6 mo.}\}
    \]
    where trait $k$ denotes the disease question. This just adds up all friends that have had the diseases over the last six months. We include a sample ARD questionnaire in Online Appendix \ref{sec:ARDBihar}.
    \item[II.] {\bf Conduct census exercise:} Obtain basic information about the full set of households in the village in a very rapid survey (denoted $X_{i}$ for all $i=1,...,n$).
    \begin{itemize}
        \item Minimal demographics: e.g., GPS coordinates, caste/subcaste.
        \item ARD traits: e.g., whether the household has had typhoid, malaria, or cholera in the past six months.
    \end{itemize}
    A sample census questionnaire is in Online Appendix \ref{sec:CensusBihar}.
    \item[III.] {\bf Estimate network formation model with ARD:} Use the information from the ARD survey and the population counts from the census to estimate the parameters of a network formation model. In this model, the probability that two households $i$ and $j$ are linked depends on household fixed effects ($\nu_i$) and distance in some latent space (latent locations $z_i$) with 
    \[
    \Pr(g_{ij} = 1 \vert \nu_i, \nu_j, \zeta, z_i, z_j) \propto \exp(\nu_i + \nu_j + \zeta \cdot \text{distance}(z_i,z_j)).
    \]
    \begin{itemize}
        \item Fit a model to predict $\nu_i,z_i$ using $X_i$ in the ARD sample.
        \item Predict $\nu_i,z_i$ using $X_i$ for all households in the census but not in the ARD sample.
    \end{itemize}
    Equipped with estimated fixed effects and latent locations for all $n$ households in the network, the probability of any network ${\bf g}$ being drawn is fully computed. The code is freely available and discussed in Section \ref{sec:ARD_Codes}.
    \item[IV.] {\bf Compute network statistics of interest:} Use the estimated probability model (using $\zeta$, fixed effects $\nu_i$ and latent locations $z_i$) to compute $\E[S({\bf g})|{\bf Y}]$. The code is freely available and discussed in Section \ref{sec:ARD_Codes}.\footnote{Note that here, the method produces estimates of the latent locations of each node, which may themselves be useful for some research questions.}
    \item[V.] {\bf Estimate economic parameter of interest:} E.g., run regressions such as
    \[
    W_v = \alpha + \beta'\E[S({\bf g}_{v}) \vert {\bf Y}_v] + \epsilon_v
    \text{ or }
    \E[S({\bf g}_{v}) \vert {\bf Y}_v] = \alpha + \beta \text{Treatment}_v + \epsilon_v,
    \]
    though clearly one can do more complex exercises once one has estimated the above network formation model. 
    
\end{enumerate}

\section{Model and Estimation}\label{sec:framework}

In this section, we present formally the procedure outlined above. This includes defining ARD, introducing the network formation model, linking explicitly the formation model to the ARD, and finally, outlining how to generate graphs from that network formation model.   

\subsection{Setup}

We begin by describing the underlying graph and the ARD. Let ${\bf g} = (V,E)$ be an undirected, unweighted graph with vertex set $V$ and edge set $E$, with $|V|=n$ nodes. We let $g_{ij}=\1\{ij\in E\}$. We also assume that researchers have a vector of demographic characteristics, $X_i$ for every $i \in V$.

Finally, we assume that the researcher has an ARD sample of $m \leq n$ nodes which are selected uniformly at random (where we define $\psi=\frac{m}{n}$). These could be the whole sample, with $\psi=1$, or a smaller share, and will depend on the context. It is useful to define  $V_{ard}$ to be the ARD sample set and $V_{non} = V \setminus V_{ard}$.

Formally, an ARD response is a count $y_{ik}$ to a question ``How many households with trait $k$ do you know?'' which we can write as 
\[
y_{ik} = \sum_{j\in G_k} g_{ij}
\]
where $G_k \subset V$ is the set of nodes with trait $k$.   That is, $y_{ik}$ is a count of the number of households in group $k$ that person $i$ knows.  Note that throughout we assume that we observe $y_{ik}$ and, in some cases, additional information about the group of people with trait $k$ (e.g., the number of households with this trait in the population), but we do not observe any links in the network. 

It is easy to see how this could be applied to firm or banking network data. In the firm case, ${\bf g}$ is the directed, weighted supply-chain network, which is of course not observed by the researcher. $G_k$ would be set of firms in sector $k$ and $g_{ij}$ would be the volume of transactions between firms $i$ and $j$. Here $y^{out}_{ik}=\sum_{j\in G_k}g_{ij}$ and $y^{in}_{ik}=\sum_{j\in G_k}g_{ji}$ are  the total volume of directed transactions (inputs/outputs) between firm $i$ and firms in sector $k$. For the remainder of the paper, we proceed with the example of a social network survey, however.

\subsection{Latent surface model}

The setup and model we use is from \cite{mccormick2015latent}, motivated by, among others, \cite{hoffrh2002}. We model the underlying network as
\begin{align}
\Pr(g_{ij} = 1 \vert \nu_i, \nu_j, \zeta, z_i, z_j) \propto \exp(\nu_i + \nu_j + \zeta z_i'z_j), \label{eqn:network-model}
\end{align}
where $\nu_i$ are person-specific random effects that capture heterogeneity in linking propensity.\footnote{While we develop our methodology for this specific network formation model, we should note that it is likely possible to use ARD and other components of our method alongside a range of other formation models.  While generalizing the method is outside the scope of this paper, we do view it as an avenue for future work, especially in real-world settings where researchers have a strong preference for alternative models.}  The set $V$ of nodes occupy positions on the surface of a latent geometry.  As in previous latent geometry models in the statistics and machine learning literatures, the distance between nodes on the latent surface is inversely proportional to their propensity for interaction, parsimoniously encoding homophily.  Using a distance measure preserves the triangle inequality, thereby generating likely triadic closure.  That is, if the position of node $i$ is close to that of node $j$ and node $j$ is close to node $k$, then the triangle inequality limits the distance between $i$ and $k$.  As we show below, equipped with the latent space terms, the model has features akin to random geometric graphs where clusters of nodes that are nearby are more likely to link, capturing realistic clustering patterns \citep{penrose2003}.  For further discussion of the properties of this class of model see~\cite{hoff2008modeling}.  In our case, we use latent space positions on the surface of $p+1$ dimensional hypersphere, $\mathcal{Z} = \mathcal{S}^{p+1}$.  As described below, the hypersphere has both conceptual and computational advantages when working with ARD. Finally, $\zeta > 0$ modulates the intensity of the latent component.

We use a Bayesian framework and, therefore, complete the model by specifying priors on the model components.  We begin with the latent space. As in  \cite{mccormick2015latent}, we model priors for latent positions on $\mathcal{S}^{p+1}$ as
\begin{align*}
z_i |\mathbf{\upsilon}_z, \eta_z\sim\mathcal{M}(\mathbf{\upsilon}_z,0) 
\text{ and }
z_{j} | {j\in G_{k}}, \mathbf{\upsilon}_{k}, \eta_{k}\sim\mathcal{M}(\mathbf{\upsilon}_{k},\eta_{k}) 
\end{align*}
where $\mathcal{M}$ denotes the von Mises-Fisher distribution across $\mathcal{S}^{p+1}$.\footnote{Informally, the von Mises-Fisher distribution can be thought of as follows. If the concentration parameter is large.  It is similar to a normal distribution on the sphere in that it is unimodal and symmetrically dissipating in distance from the center (though it should not be confused with the wrapped normal distribution or other projection of the normal to a sphere). If the concentration parameter is small, it is essentially uniform over the sphere's surface.} 
Here $\upsilon_k$ denotes the location on the sphere and $\eta_k$ is the intensity: $\eta = 0$ means that the location is uniform at random, which makes sense since the ARD respondents are assumed to be drawn uniformly at random.  The $z_{j} | {j\in G_{k}}$ terms describe the latent positions of individuals who have a particular trait $k$.  For these groups, we estimate the center and spread of the distribution.  The positions of these groups then triangulate the positions of individuals who have ARD.  For individuals in the population without ARD data, we assign their positions based on the positions of individuals with ARD that have similar covariates.

Equipped with this,~\cite{mccormick2015latent} show that the expected ARD response by $i$ for category $k$ can be expressed as
\begin{align}\label{eq:lambda}
\lambda_{ik} = \E[y_{ik}] =  d_i b_k \left( \frac{C_{p+1}(\zeta) C_{p+1}(\eta_k)}{C_{p+1}(0)C_{p+1} \sqrt{\zeta^2 + \eta_k^2 + 2\zeta \eta_k \cos (\theta_{(z_i,\upsilon_k)})}} \right),
\end{align}
 where $d_i$ is the respondent degree and $b_k$ is the share of ties made with members of group $k$, $C_{p+1}(\cdot)$ is the normalizing constant of the von Mises-Fisher distribution (which is a ratio depending on modified Bessel functions that is easy to compute with standard statistical software), $\theta_{(z_i,\upsilon_i)}$ is the angle between the two vectors \citep{mccormick2015latent}. The expected number of nodes of type $k$ known by $i$ is roughly its expected degree scaled by the population share of the group, adjusted by a factor that captures the relative proximity of the node to the type in question in latent-space.  Note that, in the above expression, both the distance between an individual and the center of the latent trait distributions and the concentration of the latent trait distribution influence the (expected) number of individuals know.  Recall that our formation model only relies on the distance between individuals in the latent space.  The positions of individuals, however, are estimated using the likelihood above, meaning that both the position and concentration are relevant for our formation model.

A key assumption in our formation model is that the propensities for individuals to form ties are conditionally independent given the latent variables.  The likelihood for the formation model, conditional on the latent variables, is a Bernoulli trial for each pair.  ARD, then, is the sum of (conditionally) independent Bernoulli trials, which we can approximate with a Poisson distribution.  This allows us to compute the distribution of the ARD response, which will be distributed Poisson,
\[
y_{ik} \vert d_i,b_k,\zeta,\eta_k,\theta_{(z_i,\upsilon_k)} \sim \text{Poisson}\left(\lambda_{ik}\right).
\]

Though the likelihood above relies only on ARD, it does not uniquely identify the formation model since $\lambda_{ik}$ estimates on the degree, $d_i$, rather than the individual heterogeneity parameter $\nu_i$.  We can compute the expected degree as in \citep{mccormick2015latent},
\begin{equation}
d_i = n \exp(\nu_i) \E[\exp(\nu_j)] \left(\frac{C_{p+1}(0)}{C_{p+1}(\zeta)}\right).
\label{eq:deg}
\end{equation}
The virtue here is that this allows us to estimate $\nu_i$ for $i\in V_{ard}$.\footnote{Note that if in our ARD elicitation, we also collect information on each node's degree, which we recommend, then we can use that information here, without needing to first estimate $d_i$ above.} The logic is similar to that in \citet{chatterjeeds2010} or  \citet{graham2014econometric}: in a model like the $\beta$-model, having a vector of degrees essentially provides the researcher with enough information to recover the vector of fixed-effects.  If we take the above expression for each individual, then we have a system of $n$ equations with $n+1$ unknown terms ($n$ $\nu_i$ terms and one $\E [\exp(\nu_j)]$).  Assuming that $\E [\exp(\nu_j)]$ is well-approximated by the average of the $\exp(\nu_i)$'s, we have a system with $n$ equations and $n$ unknowns and can, therefore recover individual $\nu_i$ terms using degree and the latent scaling term, $\zeta$. 

To complete the model, we need priors for the remaining parameters. We propose Gamma priors for $\zeta$ and $\eta_k$ with conjugate priors on the hyperparameters.  Then if $\boldsymbol{\theta}$ is the shorthand for all parameters, the posterior is
\begin{align*}
\boldsymbol{\theta} \vert y_{ik} & \propto  \prod_{k=1}^K \prod_{i=1}^n \exp(-\lambda_{ik}) \lambda_{ik}^{y_{ik}}  \prod_{i=1}^n \text{Normal}(\log (d_i) \vert \mu_d, \sigma^2_d)\\ 
& \times \prod_{k=1}^K \text{Normal}(\log(b_k)\vert \mu_b, \sigma^2_b) \prod_{k=1}^K \text{Normal}(\log(\eta_k)\vert \mu_{\eta_k}, \sigma^2_{\eta_k}) \text{Gamma}(\zeta \vert \gamma_\zeta, \psi_\zeta).
\end{align*}

Given the data, we can compute posteriors over degrees of nodes, their unobserved heterogeneity, population shares of categories, intensity of the latent space component in the network formation model, relative locations of categories on the sphere, and how intensely they are concentrated at these locations. So with any draw of $(z_1,...,z_n)'$, $(\nu_1,...,\nu_n)'$, and $\eta$, we can generate a graph from the distribution in \eqref{eqn:network-model}.

\subsection{Identification}\label{sec:ID_intuition}

Before explaining how we go from the ARD sample to the full sample, we explain identification of the parameters in the model.\footnote{Also see \cite{mccormick2015latent} for a discussion of identification as well as recommendations for the number of populations to fix based on the dimension of the hypersphere.} Here we provide a simple intuition, followed by a formal statement with proof in the Appendix.

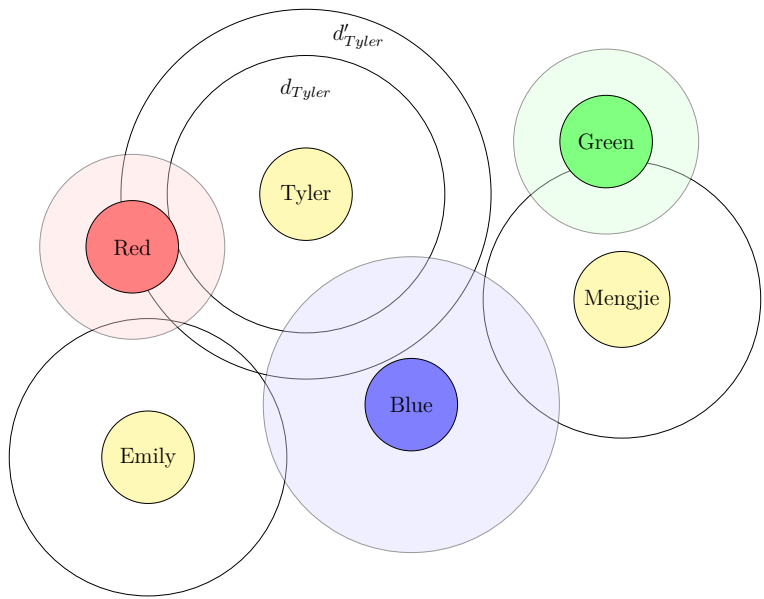
\begin{figure}[!h]
\scalebox{0.7}{\begin{tikzpicture}
\node[draw, circle, fill=yellow!35,  minimum size=50pt] at (0,0) (v1){$\text{Tyler}$};
\node[draw, circle,  minimum size=150pt] at (0,0) (v2){};
\node[draw, circle,  minimum size=200pt] at (0,0) (v2){};
\node[draw=none,  minimum size=20pt] at (0,2) (v2){$d_{Tyler}$};
\node[draw=none,  minimum size=20pt] at (1,3) (v2){$d_{Tyler}'$};
\node[draw, circle, fill=yellow!35, minimum size=50pt] at (-3,-5) (v3){$\text{Emily}$};
\node[draw, circle,  minimum size=150pt] at (-3,-5) (v4){};
\node[draw, circle, fill=yellow!35,  minimum size=50pt] at (6,-2) (v5){$\text{Mengjie}$};
\node[draw, circle,  minimum size=150pt] at (6,-2) (v6){};

\node[draw, circle, fill = red!15, opacity = 0.4, minimum size=100pt] at (-3.3,-1) (v6){};
\node[draw, circle, fill=red!50,  minimum size=50pt] at (-3.3,-1) (v5){$\text{Red}$};

\node[draw, circle, fill = blue!15, opacity = 0.4,  minimum size=160pt] at (2,-4) (v8){};
\node[draw, circle, fill=blue!50,  minimum size=50pt] at (2,-4) (v7){$\text{Blue}$};

\node[draw, circle, fill = green!15, opacity = 0.4,  minimum size=100pt] at (5.7,1) (v8){};
\node[draw, circle, fill=green!50,  minimum size=50pt] at (5.7,1) (v7){$\text{Green}$};

\end{tikzpicture}}
\caption{Identification of $\upsilon_k$ and $\eta_k$ for $k \in \{\text{Red}, \text{Blue}, \text{Green}\}$ holding fixed locations and degrees of nodes in the ARD sample. Identification of $\E[d_i]$ holding fixed locations and concentration parameters. \label{fig:ident}}
\end{figure}

Figure \ref{fig:ident} shows how the location $\upsilon_k$ and the concentration $\eta_k$ for category $k$ is intuitively identified assuming the latent geometry is a plane. Holding the location of three nodes fixed (here Tyler, Emily and Mengjie), and holding fixed their degree, the relative locations of categories (here Red, Green, and Blue) can be identified by placing their centers and controlling the concentration to match the Poisson rates observed in the ARD. To see that the concentrations of the Red, Green, and Blue trait groups are identified, consider what would happen if we changed the concentration of one of the groups.  If we increased the concentration of the Blue group (i.e. decreased the variance), then we would need to move Mengjie (and Tyler and Emily) closer to the Blue group to preserve the overlap between Emily's disc and the Blue group.  Moving Emily closer to the Blue group, though, necessitates moving her away from the Red group, reducing her overlap with the Red group.  We could try to compensate by decreasing the concentration (increasing the variance) of the Red group.  We can't do this, though, because doing so would change the overlap between Tyler's disc and the Red group.  Similarly the figure shows how the $\E[d_{Tyler}]$ can be identified holding fixed the location and concentration of the various categories, since this affects $\lambda_{Tyler,k}$.  Because the likelihood only depends on the latent space through the distances between individuals and groups, we fix the location of the center a small number of groups to address the invariance to distance-preserving rotations. 

The formal statement is as follows.
\begin{thm}\label{thm:identification}
For any $n$ by $K$ matrix of ARD responses ${\bf Y}$, we have that $\mathcal{L}(d_i,b_k,\zeta,\eta_k,\theta_{(z_i,\upsilon_k)};{\bf Y})=\mathcal{L}(d_i,b_k,\zeta,\eta'_k,\theta'_{(z_i,\upsilon_k)};{\bf Y})$ only if $\eta_k=\eta'_k$,  $\theta_{(z_i,\upsilon_k)}=\theta'_{(z_i,\upsilon_k)}$, $\zeta=\zeta'$, $\nu_i =\nu_i'$ and $z_i = z_i'$.
\end{thm}
We provide a formal proof of the theorem in Appendix \ref{sec:IDProof}.

\

\subsection{From ARD sample to Non-ARD sample} 
\label{sec:ARDtononARD}

Thus far we only have posteriors for our ARD sample $V_{ard}$.  We now turn to predicting $\nu_i$ and $z_i$ for $j \in V_{non}$. 
We use k-nearest neighbors to draw this distribution. Given demographic covariates $X_i$ for all $i \in V$, we define a distance between nodes in the feature space $d(X_i,X_j)$ for $i, j \in V$. For each $j \in V_{non}$, we pick $i' \in V_{ard}$ such that $d(X_{i'},X_j)$ is among the k smallest distances. We then take a weighted average of $\nu_{i'}$ and $z_{i'}$ with weights inversely proportional to $d(X_{i'},X_j)$, to estimate $\nu_j$ and $z_j$, respectively. We normalize $z_j$ such that $|z_j|=1$ to map it to the surface of the sphere.    Thus, we have described a framework that a researcher can use with only ARD data and demographic covariates to take a sample of draws from a network formation latent surface model.

\subsection{Drawing a graph}

We now describe the algorithm used to generate a distribution of graphs $\{{\bf g}_s\}_{s=1}^{S}$.  The algorithm for drawing graphs requires specifying the dimension of the latent hypersphere.  Throughout the paper we follow~\citet{mccormick2015latent} and use $p=2$, for a three-dimensional hypersphere.\footnote{We also investigate the performance of the method in real-world networks for $p=3$ in Appendix \ref{sec:p_3} and $p=4$ in Appendix \ref{sec:p_4}.}  This choice also facilitates visualizing latent structure.  The posterior distribution is not available in closed form.  We therefore use a Metropolis-within-Gibbs algorithm to obtain samples from the posterior.  In the description below the jumping scale is tuned adaptively throughout the course of sampling.  Specifically, every 50 draws we look at the acceptance rate of these draws and then adjust the scale of the jumping distribution.  We follow the guidelines given in ~\citet{gelman2013bayesian} and perform checks to ensure that our sampler has converged.

\begin{alg}[Drawing Graphs]

\ 

Input: $y_{ik} \;\; \forall  i \in V_{ard}$, $X_i \;\; \forall i \in V$.
\\Assume ARD groups, $k=1,...,K$, such that $K \geq p$.  We propose fitting the model as follows (noting that steps 1 \& 2 follow from \cite{mccormick2015latent}):
\begin{enumerate}
\item For a subset of the ARD groups, $k^{(s)}=1,...,K^{(s)}$, fix $\boldsymbol{\upsilon}_k^{(s)}$.  At each step we use these fixed positions in a Procrustes transformation\footnote{Procrustes transformations are a class of transformation that use rotation, translation, or uniform scaling.  Critically, they change the orientation and shape of an object but not the size.} (see~\citet{hoffrh2002}) to rotate the latent space back to a common orientation.
\item Repeat to convergence for $t=1,..., T$
\begin{enumerate}
\item For each $i$, update $z_i$ using a random walk Metropolis step with proposal $z_i^* \sim \mathcal{M}(z_i^{(t-1)}, \mbox{jumping scale})$.
Use the algorithm proposed by~\citet{wood94} to simulate proposals.

\item Update $\boldsymbol{\upsilon}_{k}$ using a conditionally conjugate Gibbs step~\citep{Mardia:1976vb,Guttorp:1988vl,Hornik:2013fy}. 
\item Update $d_i$ with a Metropolis step with \\ $\log(d_i^*)\sim\mbox{N}(\log(d_i)^{(t-1)}, (\mbox{\small jumping distribution scale})$.
\item Update $\beta$ with a Metropolis step with $\log(\beta^*)\sim\mbox{N}(\log(\beta)^{(t-1)}, (\mbox{\small jumping distribution scale})$.
\item Update $\eta_{k}$ with a Metropolis step with $\eta_{k}^*\sim\mbox{N}(\eta_{k}^{(t-1)}, (\mbox{\small jumping distribution scale})$.
\item Update $\zeta$ with a Metropolis step with $\zeta^*\sim\mbox{N}(\zeta^{(t-1)}, (\mbox{jumping distribution scale})$.
\item Update $\mu_\beta\sim\mbox{N}(\hat{\mu}_\beta,\sigma_\beta^2)$ where $\hat{\mu}_\beta=\sum_{k=1}^{K} \beta_k/K$.
\item Update $\sigma_\beta^2 \sim \mbox{Inv-}\chi^2(K-1,\hat{\sigma^2_\beta})$ where $\hat{\sigma^2_\beta}=\frac{1}{K-1}\sum_{k=1}^{K}(\beta_{k}-\mu_{\beta})^{2}$.
\item Update $\mu_d\sim\mbox{N}(\hat{\mu}_d,\sigma_d^2)$ where $\hat{\mu}_d=\sum_{i=1}^{n} d_i/n$.
\item Update $\sigma_d^2 \sim \mbox{Inv-}\chi^2(n-1,\hat{\sigma^2_d})$ where $\hat{\sigma^2_d}=\frac{1}{n-1}\sum_{i=1}^{n}(d_{i}-\mu_{d})^{2}$.
\end{enumerate}
\item Repeat for $s \in \{T/2+1,...,T\}$
\begin{enumerate}
\item Calculate $\nu_i^{t}  \;\; \forall i \in V_{ard}$ such that $\nu_i^{t}$ satisfies $(d_i)^{t} = \exp(\nu_i^{t}) \sum_i \exp(\nu_i^t) \left(\frac{C_{p+1}(0)}{C_{p+1}(\zeta)}\right)$.
\item Use method described in Section \ref{sec:ARDtononARD} to estimate $\nu_j^{t}$ and $z_j^t  \;\; \forall j \in V_{non}$. 
\item Sample graph ${\bf g}_t$ using the the procedure described below.
\end{enumerate}
\end{enumerate}

Output: $\{{\bf g}_s\}_{s=1}^{S}$

\end{alg}

To generate graphs, recall that the formation model has $\Pr(g_{ij} = 1 \vert \nu_i, \nu_j, \zeta, z_i, z_j) \propto \exp(\nu_i + \nu_j + \zeta z_i'z_j)$.  We estimate $\zeta$ and $z_i, z_j$ using the likelihood derived in~\cite{mccormick2015latent}.  The expression (\ref{eq:deg}) relates degree to the unobserved gregariousness parameters, $\nu_i$.  If we approximate $\E[\exp(\nu_j)]$ as the average of the $\nu_i$'s, then we can view (\ref{eq:deg}) as a system with $n$ equations and $n$ unknowns and obtain estimates for $\nu_i$ for each respondent. 

We then normalize the $\exp(\nu_i + \nu_j + \zeta {z_i^t}\; ' z_j^t)$ terms to produce probabilities.  Define \[\Pr(g_{ij}=1|z_i,z_j,\nu_i,\nu_j)=\frac{\exp(\nu_i+\nu_j+\zeta z_i'z_j)\sum_i \E[d_i]}{\sum_{i,j} \exp(\nu_i+\nu_j+\zeta z_i'z_j)}.\]  Normalizing in this way ensures 
$ \sum_i \E[d_i] \triangleq \sum_i \sum_j \Pr(g_{ij}=1|z_i,z_j,\nu_i,\nu_j).$  Since the formation model assumes that the propensities to form a ties between pairs are conditionally independent given the latent variables, we can now generate graphs by taking draws from a Bernoulli distribution for each pair with probability defined by $\Pr(g_{ij}=1|z_i,z_j,\nu_i,\nu_j)$.

\subsection{Discussion}
\subsubsection{Sensitivity to choice of prior distributions.}
A natural question in any Bayesian analysis is how the modelers' choices about prior distributions impact posterior inferences.  In our context, the priors are influential in two settings. First, as explained above, we put priors directly on the parameters of the ARD likelihood.  The ARD likelihood parameters then, in turn, determine the parameters for the network formation model.  To evaluate the influence of the prior distributions on our ability to estimate the parameters of the ARD likelihood (and therefore formation model), we conduct a series of experiments presented in Appendix~\ref{sec:priorsens}.  For the scalar and vector parameters (e.g., the individual degree, $d_i$) we examine the posterior distribution after varying the spread and center of the distribution of the prior.  For the latent space locations, recall that we fix some population centers for identification.  To ensure that our results are not sensitive to these choices, we perform experiments where we randomly choose both which ARD population centers we fix and where these groups are positioned on the sphere's surface.  

A second consideration in exploring our prior choices is the way that priors on the ARD likelihood parameters imply (via the formation model) priors on our network moments of interest.  That is, we do not explicitly put a prior on centrality.  The prior on centrality (and the other network moments) is, however, implied by the prior distribution placed on the parameters in the ARD likelihood.  Appendix~\ref{sec:priorsens} presents a second set of results that show how the priors used for our model relate to the network moments of interest.  We begin by simulating networks using the procedure above without any observed data.  That is, we generate a series of networks entirely from the specified prior distributions.  This series of networks demonstrates the wide range of possible networks that are supported by our formation model and the priors we specified. For context, we also plot the distribution of network moments from our estimated posterior distribution and from the observed data in Section~\ref{sec:data}.   

\subsubsection{Finite population and density}
We have provided a simple algorithm to go from ARD questions to draws from the posterior distribution of the graph that would have given rise to ARD answers by respondents with characteristics similar to those we observed in the data.  The model leverages a latent surface model similar to~\citet{hoffrh2002}, used in \cite{mccormick2015latent}, which is intimately related to the $\beta$-model studied in~\citet{chatterjeed2011} and ~\citet{graham2014econometric}.  One issue that has arisen from both the Bayesian and frequentest perspectives is the notion of density in the limit, or the rate at which the number of edges grows compared to the number of nodes.  The Bayesian paradigm uses the Aldous-Hoover Theorem~\citep{hoover1979relations, aldous1981representations} for node-exchangeable graphs to justify representing dependence in the network through latent variables.  This exchangeability assumption implies that a graph can be sparse if and only if it is empty~\citep{lovasz2006limits, diaconis2007graph, orbanz2015bayesian, crane2015framework}.  From a frequentist perspective,~\citet{chatterjeed2011} show that the individual fixed effects (corresponding to, for example, gregariousness) can only be consistently estimated when the network sequence is dense. 

In contrast to this previous work, however, we assume that our sample of egos arises from a population with fixed $n$.  That is, in our paradigm there is a network of finite size, $n$, and we observe a small $m$ number of actors.  We see the reliance on this assumption in, for example, our expression relating degree to the individual heterogeneity parameters, $\nu_i$.  Put a different way, there is no asymptotic sequence of networks. The number of edges in a graph still impacts estimation, however.  Even when the number of nodes is large, we do not expect $d_i$ to uniformly converge to $\E[d_i]$ if the graph is not dense.  This additional variability propagates through the model and inflates the posteriors of $\nu_i$. These may be quite poor in practice, though it is difficult to derive the finite sample distribution. Nonetheless, what this suggests is that in cases where the network is too sparse, the ARD approach may be uninformative, and the researcher will see this plainly. This is the case for two reasons. First, by definition, anyone in the ARD sample will know fewer alters with trait $k$ since the network has fewer links on average. Second, there will be too much variation in our location estimates and degree estimates, which then will also affect our node heterogeneity estimates. This means that when the researcher faces rather diffuse posteriors, the network may be too sparse to convey much information. We explore these issues in simulations below.

\section{How Well Does the Procedure Perform?}\label{sec:sims}

In this section we explore how well our procedure works under the assumption of correct specification. That is, we assume that the data-generating process is such that graphs are generated from the family of models described in \eqref{eqn:network-model}.  While taking a stand on the formation model permits tractability, it of course carries with it some well-understood limitations. We discuss these limitations and test the procedure in real-world network data from 75 Indian villages in Section \ref{sec:results}. In Section \ref{sec:application}, we further consider two different field experiments that used network data and ask whether using ARD alone would have allowed researchers to make similar conclusions.

Under the assumption of correct specification, and having demonstrated identification above, we cover two questions. The first is for which network statistics do we expect ARD to work well. That is, even if we knew the set of individual fixed effects $\nu_i$ and latent locations $\zeta_i$, when would we have sufficient information to recover the network statistics of interest or the economic parameters of interest in a regression. To do this, in Section \ref{sec:prob_realization} we develop the theory for a taxonomy of network features to classify when we would or would not expect recovery of the network features. We show a straightforward but informative result which says that if, for a sufficiently large graph, our statistic of interest for any random realization from the generating process will be close to its expectation, then we should expect the mean-squared error of our statistic to become negligible. We supplement this with simulations to show practical results as to which network statistics we can recover with low mean-squared error (MSE). Finally, we conduct a rich set of simulations to demonstrate across a number of network statistics how well the procedure works for a data-set that mimics real-world network data, in Section \ref{sec:core_results}.

Second, our simulations explore the sensitivity of our results to important features of the environment. Empirical network data may vary in terms of their degree distribution: how sparse they are (the number of links on average relative to the network size) and whether they exhibit thick right-tails (there are some nodes who are extremely well-connected relative to all others). As such, in Section \ref{sec:sparsity}, our simulations explore the efficacy of our procedure as we vary sparsity and the inclusion of hyper-popular nodes.

Further, the researchers can decide how many nodes to include in their ARD samples. Accordingly, in Section \ref{sec:network_size}, we look at how well the ARD procedures work as we vary the share of nodes that are sampled for the ARD questionnaire. We  simulate networks from what we call a rural environment (a smaller graph of 200-500 nodes) and an urban setting (thousands of nodes) and vary the share of nodes for which we have ARD. This exercise helps to provide guidance for research designs incorporating ARD.

\subsection{Theoretical insights on when ARD works}\label{sec:prob_realization}

We first provide theoretical insights on which network features will be amenable to estimation using ARD.  A core theme in this discussion is that the ARD model produces predicted probabilities of connections between pairs of individuals. To get network statistics, however, we must convert these probabilities into realizations of graphs and, therefore estimates of the expectations of graph characteristics across possible graphs. We investigate the impact of this feature of our procedure both theoretically and empirically. 

For the theoretical exercise, we assume that data arise from a formation model of the form presented in~(\ref{eqn:network-model}). In addition, we assume that the ARD procedure tightly identifies the model parameters.  \footnote{Recall our formal argument for identification in Section~\ref{sec:ID_intuition}.} These assumptions allow us to focus on when the expectation of the network statistic is sufficiently informative about any given graph realization. Under these assumptions, let $p_{ij}^{\theta_{0}}$ denote the probability that nodes $i$ and $j$ are linked under the data generating process with parameter vector $\theta_{0}$.

We separate our discussion into two cases: (1) the researcher has a single large network with $n$ nodes (or a handful of networks); (2) the researcher has many independent networks.

\subsubsection{Single Large Network}
\label{sec:singlebig}

We first consider the case where there is a single large network, and the researcher is interested in measuring a specific network statistic, $S_{i}\left({\bf g}\right)$ for node $i$ computed on graph ${\bf g}$.\footnote{This could easily be extended to functions of multiple nodes.} For the purposes of this argument, there is one actual realization of the graph, ${\bf g}^*$.  This realization is what we would have observed if we had collected information about all actual connections between members of the population, rather than collecting ARD. Importantly, the researcher collecting ARD cannot observe ${\bf g}^*$. This actual network realization does, however, come from a generative model that has parameters that can be estimated from the ARD.  The researcher  can, therefore, simulate graph realizations from the underlying data generating process under the true parameter vector, $\theta_0$, and construct an estimate for ${\rm E}\left[S_{i}\left({\bf g}\right)|\theta_0\right]$.  This expectation is over the possible graphs generated from the model with parameters $\theta_0$.  In practice, we will observe a $n\times K$ matrix of ARD, ${\bf Y}_n$, rather than $\theta_0$.  This expectation, then, is ${\rm E}\left[S_{i}\left({\bf g}\right)|{\bf Y}_n\right]$ or, if part of the graph is observed as part of the data generating process (through e.g. an egocentric strategy), ${\rm E}\left[S_{i}\left({\bf g}\right)|{\bf Y}_n, {\bf g}^{obs}\right]$, where ${\bf g}$ is missing completely at random with ${\bf g}=\left\{{\bf g}^{obs},{\bf g}^{unobs}\right\}$.  As we describe in Section~\ref{sec:ID_intuition}, the ARD data, ${\bf Y}_n$, are sufficient to identify the generative parameters, $\theta_0$. To simplify notation, we will omit the conditioning for the remainder of this section.  

To recap, if a researcher collected information about all links in the population, she could compute $S_i({\bf g}^*)$ directly.  With ARD, however, she can recover an expectation over graphs generated with a given set of parameters, ${\rm E}\left[S_{i}\left({\bf g}\right)\right]$. We are interested in cases in which knowing ${\rm E}\left[S_{i}\left({\bf g}\right)\right]$ is sufficient for learning about $S_i({\bf g}^*)$.  That is, cases where, if we can get a good estimate for ${\rm E}\left[S_{i}\left({\bf g}\right)\right]$ using ARD, we can say with confidence that we have recovered a statistic that is very similar to the statistic the researcher would have observed had she collected data on the entire graph.  More formally, for any realized graph, ${\bf g}$, does 

\[
S_{i}\left({\bf g}\right)\rightarrow_{p}{\rm E}\left[S_{i}\left({\bf g}\right)\right]?
\]

If this condition holds, then when the population of individuals, $n$, is large, the statistic of interest, $S_{i}\left({\bf g}\right)$, will be close to its expectation for any realization of the graph, including the one that is the researcher's population of interest, ${\bf g}^*$.  We have, therefore, that the statistic computed from the true graph and the statistic estimated using ARD are both close to the expectation and must, therefore, be close to each other and have small mean-squared error.  Similarly, if the statistic from a given realization does not converge to its expectation, then even after more nodes are observed, there is not increasing information, and thus the mean-squared error of the estimate should not shrink. The key feature of the result is that we do not need to know the exact structure of the graph that the researcher would have observed using a network census, ${\bf g}^*$.  Instead, we rely on the notion that the statistic will be close to its expectation for a sufficiently large graph and that this is true for any realization of the graph from a given generative process.    
  
We formalize this intuition using the straightforward proposition below. Though the proposition is uncomplicated to prove, it cements the condition required of the statistic of interest for us to reasonably expect that our ARD estimates will be similar to what a researcher would have observed by directly computing the statistic from the fully-elicited graph. Further, it serves to demystify how ARD can work to recover network statistics with such limited information on the graph. The information in ARD, by the arguments in Section~\ref{sec:ID_intuition}, is sufficient to estimate the parameters of the formation model. After proving the proposition, we provide examples of statistics where ARD should and should not perform well.  We demonstrate our result for these statistics mathematically and confirm our intuition through simulations in Section~\ref{sec:taxonomy_sims}.

\begin{prop}
\label{prop:main_taxonomy}Consider a sequence of distributions of graphs on $n$ nodes given by our afore-described model and $n\times K$
ARD ${\bf Y}$. Assume $\theta_{0}$ is known. Let $S_{i}\left({\bf g}^*\right)$
be the (unobserved) statistic of the underlying network and let $S_{i}\left({\bf g}\right)$
be the same statistic computed from graph ${\bf g}$, drawn from
the distribution with parameters $\theta_{0}$. Finally, assume that
\[
S_{i}\left({\bf g}\right)\rightarrow_{p}{\rm E}\left[S_{i}\left({\bf g}\right)\right].
\]
Then the MSE is
\[
{\rm E}\left[\left(S_{i}\left({\bf g}\right)-S_{i}\left({\bf g}^*\right)\right)^{2}\right]=o_{p}\left(1\right).
\]
\end{prop}

\

To clarify when this applies and when this fails, we provide several pedagogical examples. Our first example is a failure of Proposition \ref{prop:main_taxonomy}. 

\begin{cor}\label{cor:MSE_link}
\label{cor:link}Under the aforementioned assumptions, given an (unobserved) graph of interest, ${\bf g}^*$, and non-degenerate linking probabilities $0<p_{ij}^{\theta_{0}}<1$, then the MSE for $S_{i}\left({\bf g}\right) = g_{ij}$, a draw from the distribution of any single link $g_{ij}$ is given by
\[
{\rm E}\left[\left(g_{ij}-g^*_{ij}\right)^{2}\right]=p_{ij}^{\theta_{0}}\left(1-2g^*_{ij}\right)+g^*_{ij}.
\]
\end{cor}

Note that irrespective of $n$, this cannot tend to zero. When a link exists, the mean-squared error is $1-p_{ij}^{\theta_{0}}$ and when
it does not, the MSE is $p_{ij}^{\theta_{0}}$: these are just the complements of the Bernoulli probabilities.

\
\begin{cor}\label{cor:MSE_density_diffcent}
Under the aforementioned assumptions, given an (unobserved) graph of interest, ${\bf g}^*$, the MSE tends
to zero with probability approaching 1 for the following statistics:
\begin{enumerate}
\item Density (normalized degree):
\[
{\rm E}\left[\left(\frac{d_{i}\left({\bf g}\right)}{n}-\frac{d_{i}\left({\bf g}^*\right)}{n}\right)^{2}\right]=o_{p}\left(1\right).
\]
\item Diffusion centrality (nests eigenvector centrality and Katz-Bonacich
centrality) for parameter sequence $q_{n}=\frac{C}{n}$ and any $T$,
\[
{\rm E}\left[\left(DC_{i}\left({\bf g};q_{n},T\right)-DC_{i}\left({\bf g}^*;q_{n},T\right)\right)^{2}\right]=o_{p}\left(1\right).
\]
\end{enumerate}
\end{cor}

See Appendix \ref{sec:taxonomy_proofs} for proofs of the proposition and corollaries. 

A few remarks are worth mentioning. First, diffusion centrality is a more general form which nests eigenvector centrality when $q_{n}\geq\frac{1}{\lambda_{1}^{n}}$,
and because the maximal eigenvalue is on the order of $n$, this meets our condition. It also nests Katz-Bonacich centrality. In each of these, $T\rightarrow\infty$. It also captures a number of other features of finite-sample diffusion processes that have been used in applied
work. Each of these notions relate to the eigenvectors of the network \textendash{} objects that are ex-ante not obviously captured by the ARD procedure but ex-post work because the models are such that in large samples the statistics converge to their limits.

These results give us two practical extreme benchmarks. Our procedure should not perform well at all for estimating a realization of any given link in the network. In contrast, it should perform quite well for statistics such as degree or eigenvector centrality. Other statistics may fall somewhere in the middle of this spectrum.  For example, a notion of centrality such as betweenness, which relies on the specifics of the exact realized paths in the network, is unlikely to work particularly well because even for large $n$, the placement of specific nodes may radically change its value. Section \ref{sec:taxonomy_sims} explores these predictions empirically using simulations.

\subsubsection{Many Independent Networks}

Now consider the setting where the researcher has $R$ independent networks each of size $n_r$.  We'll take $n_r = n$ for simplicity, though the results presented here do not require this. We also have an ARD sample ${\bf Y}_{n,r}$ for every network $r=1,...,R$. Every network is generated from a network formation process with true parameter $\theta_{0,r}$.  In this case of many networks, we consider how well the ARD procedure performs when the researcher wants to learn about network properties, aggregating across the $R$ graphs.  This is the case we present empirically in Section~\ref{sec:application}.

Let $S^*_{r}:=S\left({\bf g}^*_{r}\right)$ be a network statistic from the $R$ unobserved graphs generating the ARD.  For any given graph from the data generating process, define $S_{r}:=S\left({\bf g}_{r}\right)$. For notational simplicity, we consider network-level statistics, but the argument can easily be extended to node, pair, or subset-based statistics. 

Assume the goal of the researcher is to estimate some model
\[
y_{r}=\alpha+\beta S^*_{r}+\epsilon_{r}
\]
and the economic parameter of interest is $\beta$. As before, $S^*_{r}$
is unobserved because ${\bf g}^*_{r}$ is unobserved and the researcher
must make do with ARD, ${\bf Y}_{r}$. The researcher instead estimates the expectation of the statistic given using ARD, $\bar{S}_{r}:= {\rm E}(S_{r})$.  The regression then becomes: 
\[
y_{r}=\alpha+\beta\bar{S}_{r}+u_{r}.
\]
The arguments in \cite{chandrasekharl2016} show that $\beta$ is still consistently estimated when using $\bar{S}_{r}$ as a regressor rather than $S_{r}$. We sketch out the argument here for completeness.  First, It is easy to expand the error term,
\begin{align*}
y_{r} & =\alpha+\beta\bar{S}_{r}+u_{r} 
  =\alpha+\beta\bar{S}_{r}+\left\{ \epsilon_{r}+\beta\left(S_{r}^*-\bar{S}_{r}\right)\right\} .
\end{align*}
By iterated expectations we can see that
\begin{align*}
{\rm E}\left[\bar{S}_{r}\left(S^*_{r}-\bar{S}_{r}\right)\right] & ={\rm E}\left[{\rm E}\left[\bar{S}_{r}\left(S^*_{r}-\bar{S}_{r}\right)\vert{\bf Y}_{r}\right]\right] 
  ={\rm E}\left[\bar{S}_{r}\left({\rm E}\left[S^*_{r}\vert{\bf Y}_{r}\right]-\bar{S}_{r}\right)\right] 
  ={\rm E}\left[\bar{S}_{r}\left(\bar{S}_{r}-\bar{S}_{r}\right)\right]=0.
\end{align*}
This means that under standard regularity conditions, we can consistently estimate $\beta$. The intuition is that the deviation of the conditional expectation $\bar{S}_{r}$ from $S_{r}$ is by definition orthogonal to the conditional expectation and independent across $r$. So one can think of the conditional expectation as an instrument of the true $S_{r}$ where the first-stage regression has a coefficient of 1.

Practically speaking, this means that even if we were interested in a regression of
\[
y_{12,r}=\alpha+\beta g_{12,r}+\epsilon_{r},
\]
where whether nodes 1 and 2 are linked affects some outcome variable of interest, and we are interested in this across all $R$ networks,
we can use $p_{12}^{\theta_{0}}:={\rm E}\left[g_{12,r}\vert{\bf Y}_{r}\right]$
instead in the regression to consistently estimate $\beta$. Note that in contrast to the single network case, where we were interested in recovering $g_{12}$ itself, and even with large $n$ the MSE would not tend to zero, here simply having the conditional expectation is enough to be able to estimate the economic slope of interest, $\beta$. Therefore, with many graphs, the ARD procedure should work well regardless of the properties of the given network statistic.

\subsubsection{MSE Simulation Results}\label{sec:taxonomy_sims}

We next explore the results for a single large graph through a simulation exercise.  We describe the simulation set-up we use here in full detail in the next section, but include the MSE results here as a demonstration that the intuition from the theoretical results in the previous section hold empirically.  For this simulation, we use graphs with 250 nodes, which is a similar size to the data we describe in Section~\ref{sec:data}, simulated from the data generating process in Equation~\ref{eqn:network-model}.  
In Figure \ref{fig:MSE}, we plot the mean squared errors of our estimation procedure across a range of network statistics  which are commonly used in applied economics. In order to make the MSEs comparable across statistics, we scale by $\frac{1}{\E[S_i]^2}$. Panel A focuses on node level statistics while Panel B focuses on graph-level statistics.

The node level statistics are as follows: (1) degree (the number of links); (2) eigenvector centrality (the $i$th entry of the eigenvector corresponding to the maximal eigenvalue of the adjacency matrix for node $i$); (3)  betweenness centrality (the share of shortest paths between all pairs $j$ and $k$ that pass through $i$); (4) closeness centrality (the average inverse distance from $i$ over all other nodes); (5)  clustering (the share of a node's links that are themselves linked);  (6) support (as defined in \cite{jacksonrt2012} -- whether linked nodes $ij$ have some $k$ as a link in common); (7) whether link $ij$ exists; (8)  closeness; (9) average path length;  and (10) the average distance from a randomly chosen ``seed'' (as in an information diffusion experiment).

The graph level statistics are as follows: (1) diameter; (2) average path length; (3) average proximity (average of inverse of shortest paths); (4) share of nodes in the giant component; (5) number of components; (6) maximal eigenvalue; (7) clustering; and (8) the share of links across the two groups relative to within the two groups where the cut is taken from the sign of the Fiedler eigenvector (this reflects latent  homophily in the graph).

\begin{figure}[htp]
\centering
\subfloat[Node-level]{
\includegraphics[width=.485\textwidth]{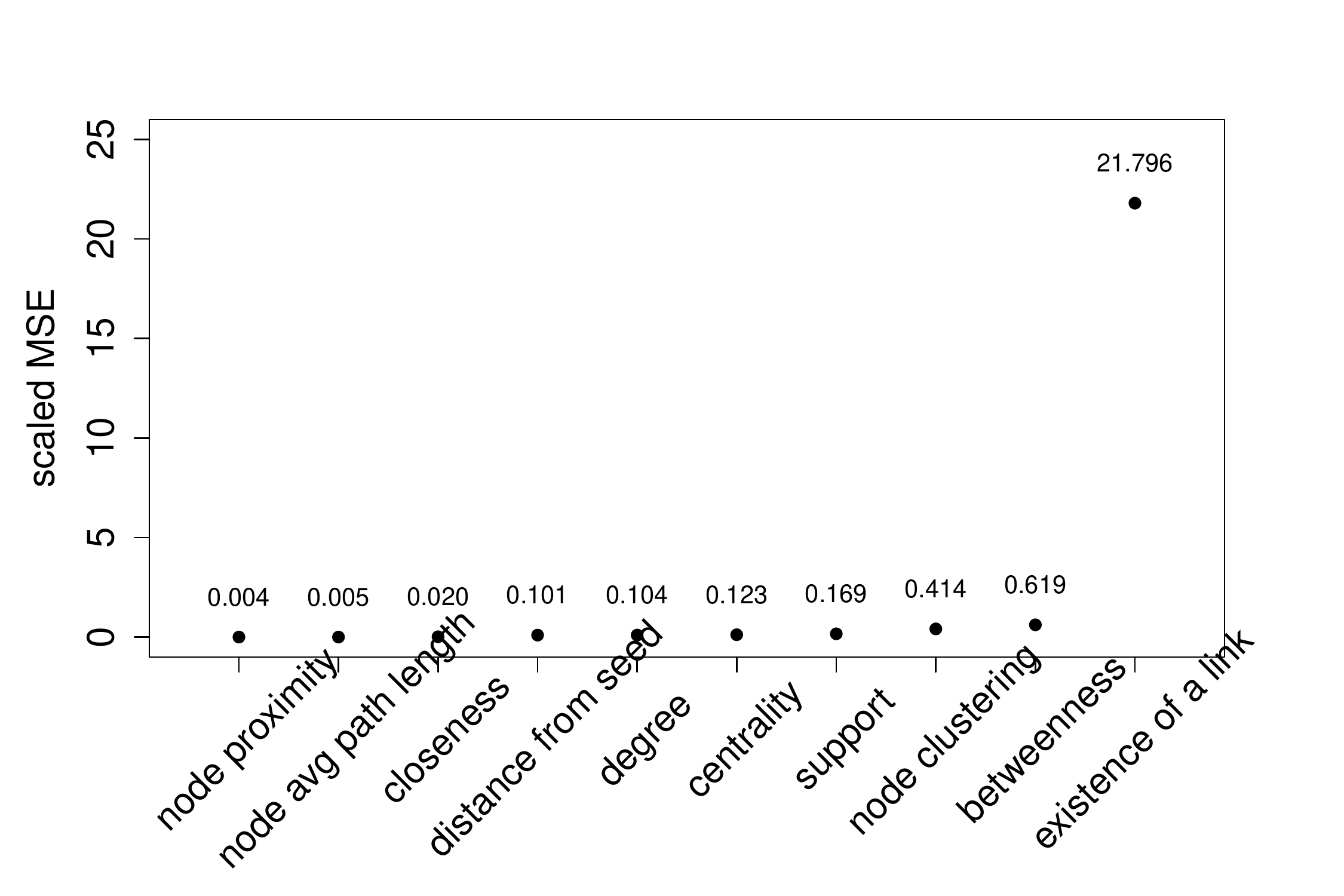}\qquad
}
\subfloat[Graph-level]{
\includegraphics[width=.485\textwidth]{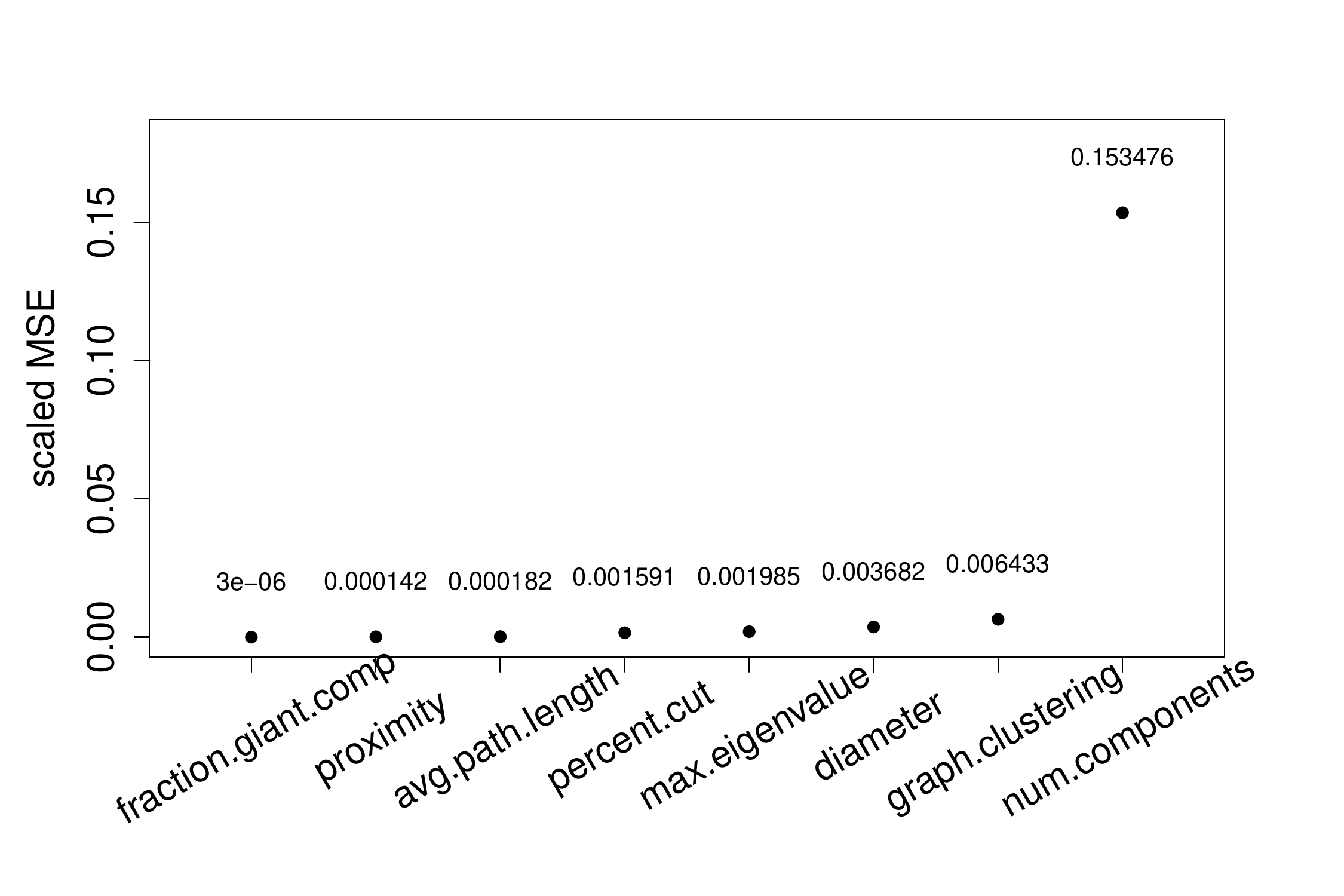}
}

\caption{Scaled MSE of node-level and graph-level network features. Each point in the figure represents the MSE across 250 simulations using graphs of size 250, a size comparable to the data we examine in Section~\ref{sec:data}.  These results corroborate the theoretical intuition devleped in Section~\ref{sec:singlebig}.  Note for example the small MSE for density and centrality measures, with worse performance for inferring a single edge, as predicted by our Corollaries.
}
\label{fig:MSE}
\end{figure}

Panel A of Figure \ref{fig:MSE} shows that the scaled MSEs in our simulations are quite small for most network statistics, including degree and (eigenvector) centrality, as predicted.  Strikingly, the scaled MSE for the estimates of the existence of a link is extremely large and matches the computation in Corollary \ref{cor:MSE_link}.  Moreover, as argued above, betweenness also performs worse than the other statistics.  

Panel B considers graph level statistics.  The scaled MSEs tend to be small for all but one network statistic -- the number of components in the graph. The number of components depends crucially on the existence of a small number of specific link realizations, calling upon the same intuition as the node-level existence of a link.

\subsection{Core Simulations}\label{sec:core_results}

We turn to a set of rich simulations which mimic real-world network data, but allow us to evaluate the efficacy of our procedure under correct specification.
\subsubsection{Simulation Model}

For each of our empirical investigations, we provide simulation evidence. We begin with a graph generated from the network formation model specified in Equation \ref{eqn:network-model} and simulate the ARD on that graph.

The simulation procedure is as follows:
\begin{enumerate}
\item We randomly generate $n$ locations on $\mathcal{S}^{p+1}$ uniformly at random to get $(z_i)_{i=1}^n$.
\item We randomly generate $\nu_i$ i.i.d. from a Normal distribution with parameters $\mu,\sigma^2$.
\item We generate a graph
  \[
  \Prob(g_{ij}=1|z_i,z_j,\nu_i,\nu_j)
  \]
\item We then pick $K$ features which we maintain to be binary.
  \begin{enumerate}
    \item Features are located with centers distributed uniformly at random over $\mathcal{S}^{p+1}$ at sites $\upsilon_k$.
    \item Each feature $k$ is distributed with concentration parameter $\eta_k$.
    \item A given site $i$ at location $z_i$ receives feature $k$ i.i.d. with probability $\Pr(i \in G_k) = \1\{u_{ik} < f(z_i \vert \upsilon_k,\eta_k)\}$ where $u_{ik}$ is a uniform random variable and $f(z_i \vert \upsilon_k,\eta_K)$ is the von Mises-Fisher density value at location $z_i$.
    \end{enumerate}
\item Constructed ARD responses are built using features of one's neighbors and the underlying graph.
\end{enumerate}

Unless otherwise stated, we set $n=250$, $\zeta=0.3$, $\mu=-1.27$, $\sigma=0.5$, and $K=12$, which are chosen to generate graphs that resemble our empirical network data in terms of average degree 20, clustering 0.13, proximity (defined as the  mean of the inverse of path lengths) 0.50, average path length 2.15, and the maximal eigenvalue 26.51 of the network. 

We then run our proposed procedure to estimate a range of network characteristics at both the individual- and node-level.

\subsubsection{Simulation results}

\begin{figure}[htp]
\centering
\subfloat[Degree]{
\includegraphics[width=.3\textwidth]{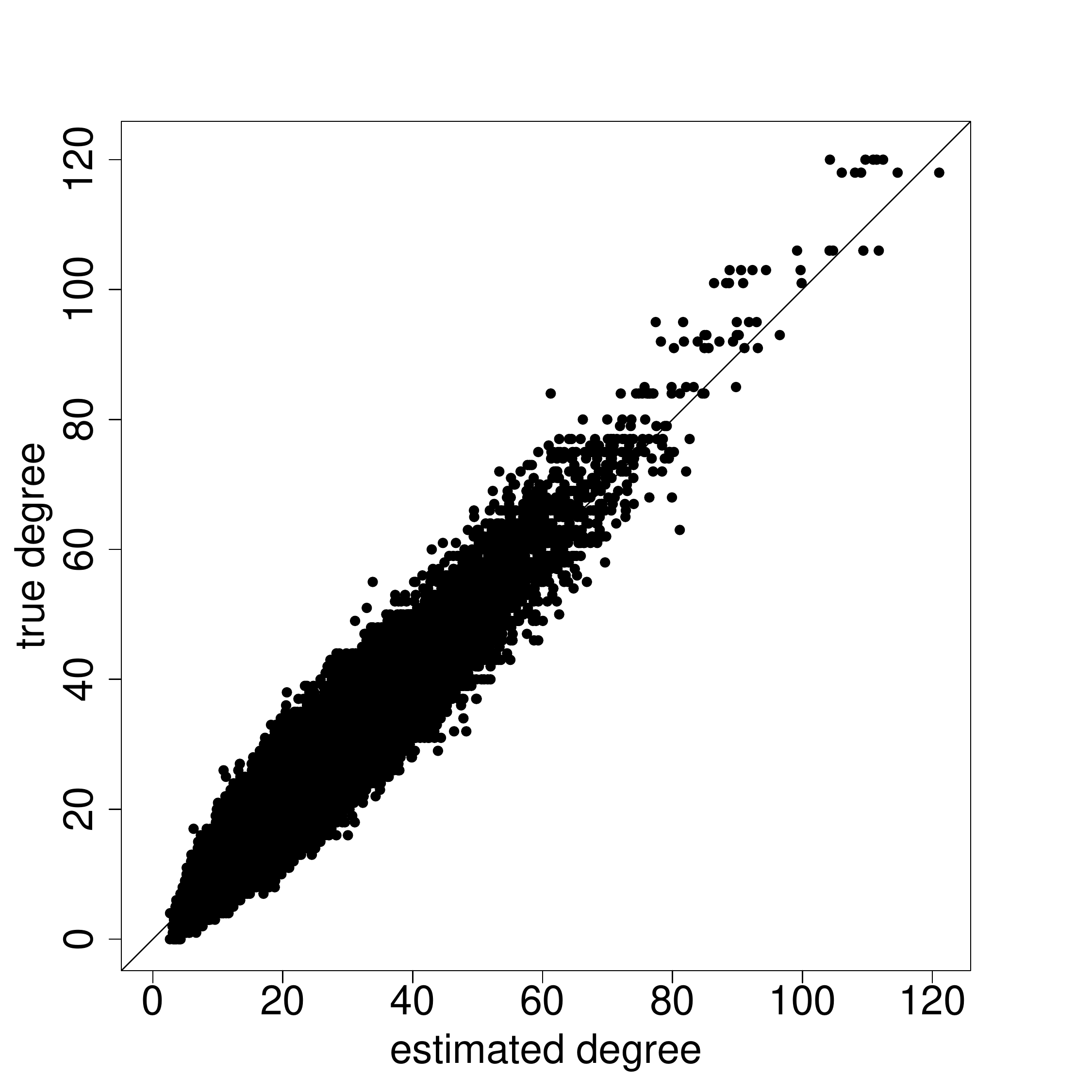}}
\subfloat[Eigenvector Centrality]{
\includegraphics[width=.3\textwidth]{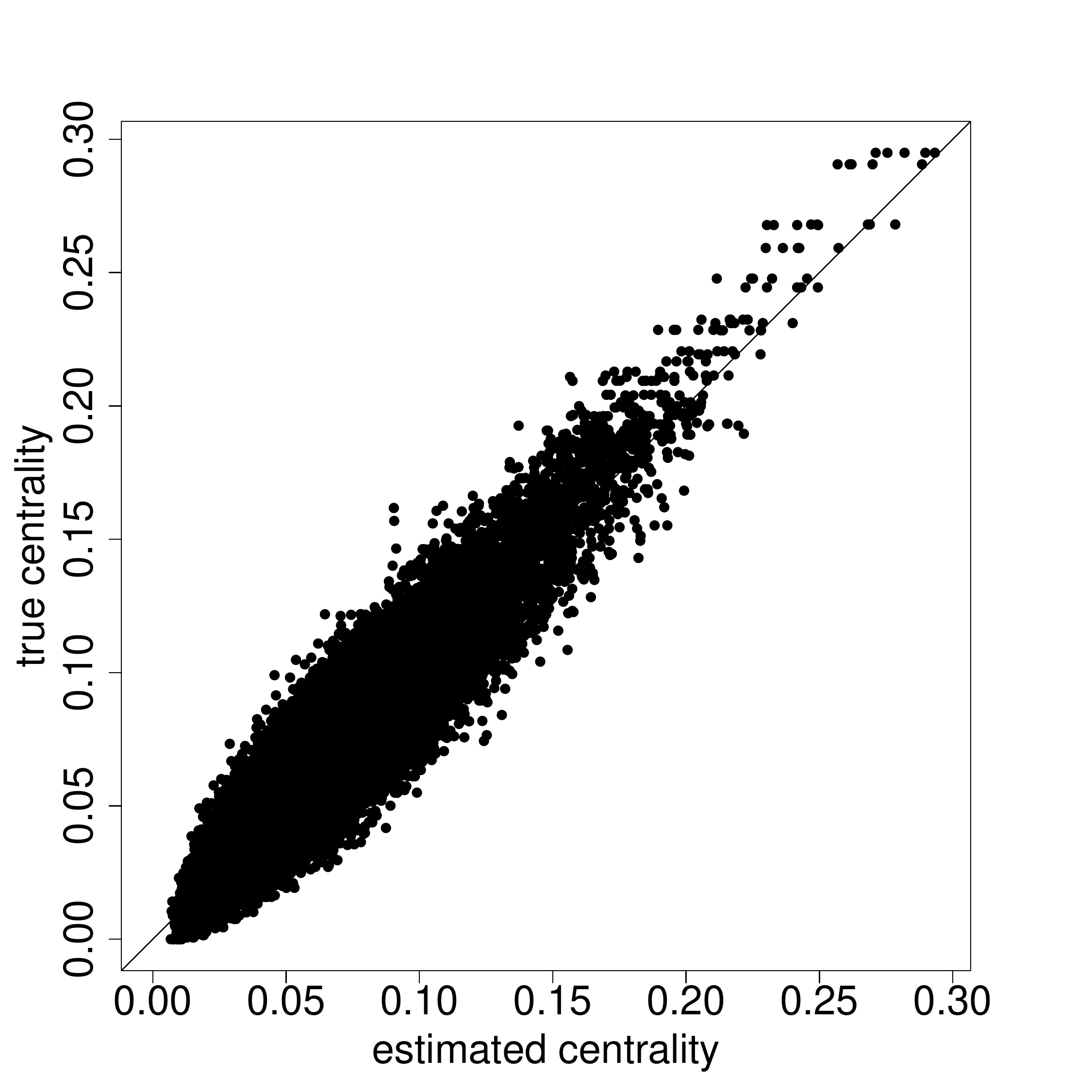}}
\subfloat[Maximum Eigenvalue]{
\includegraphics[width=.3\textwidth]{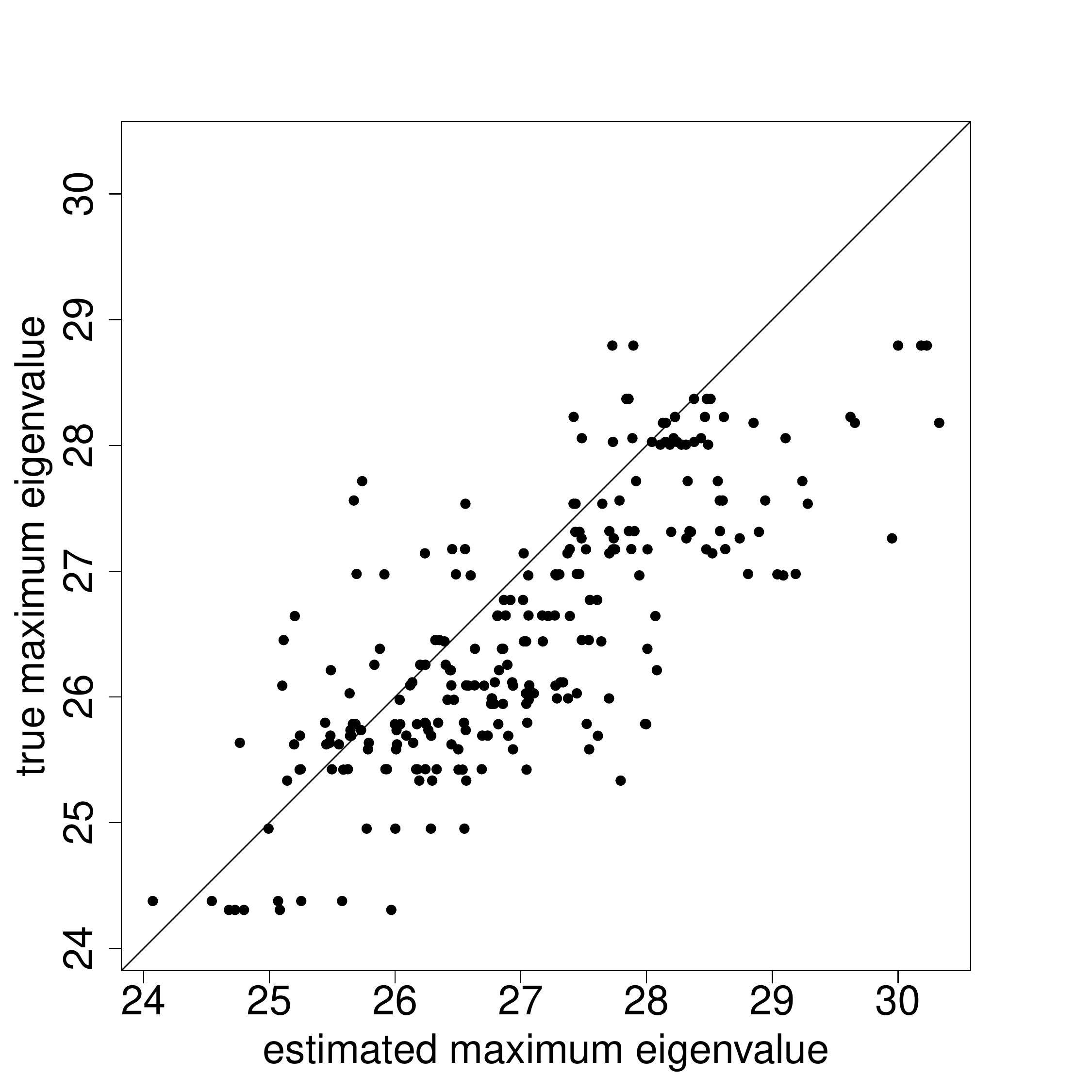}}

\smallskip

\subfloat[Proximity]{
\includegraphics[width=.3\textwidth]{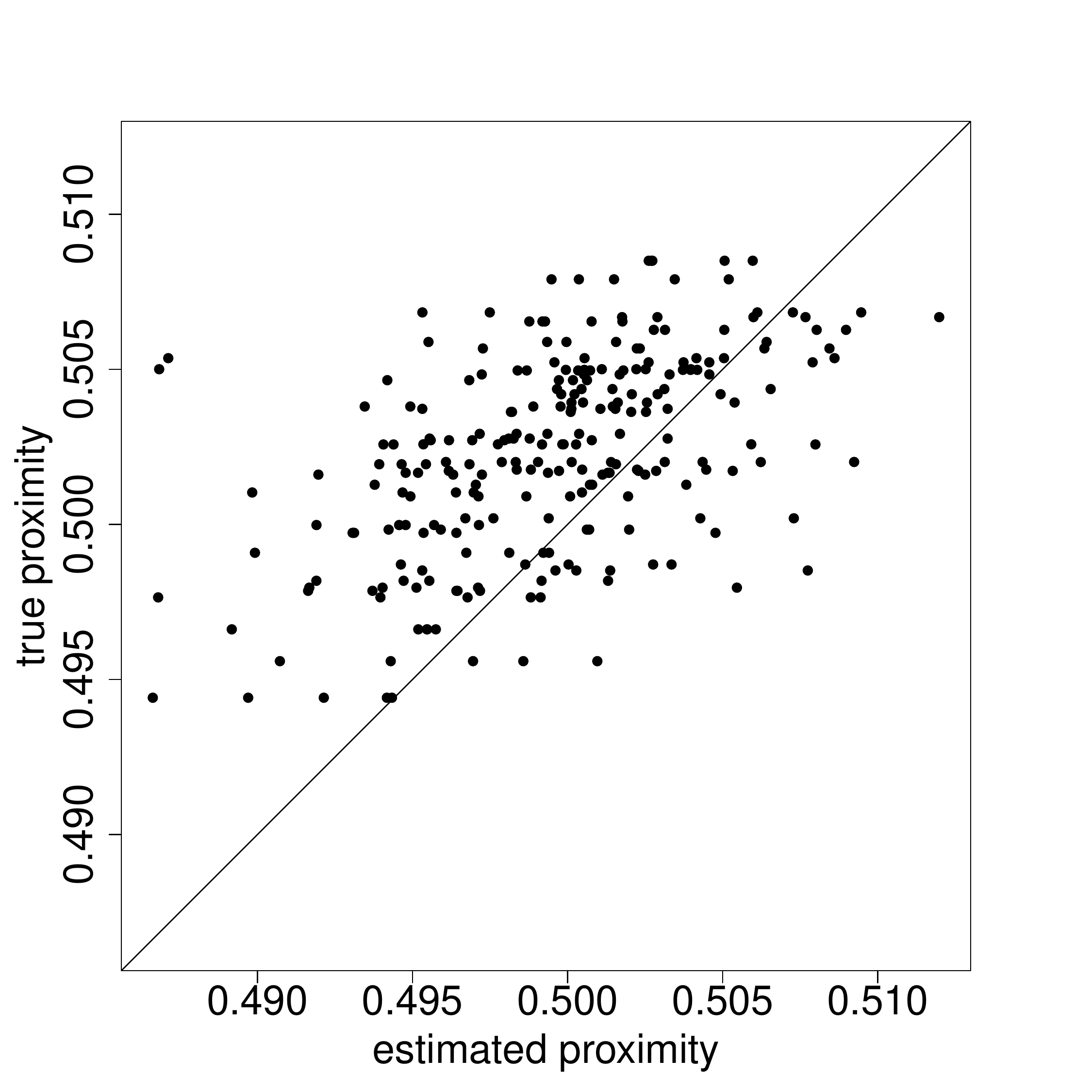}}
\subfloat[Network level Clustering]{
\includegraphics[width=.3\textwidth]{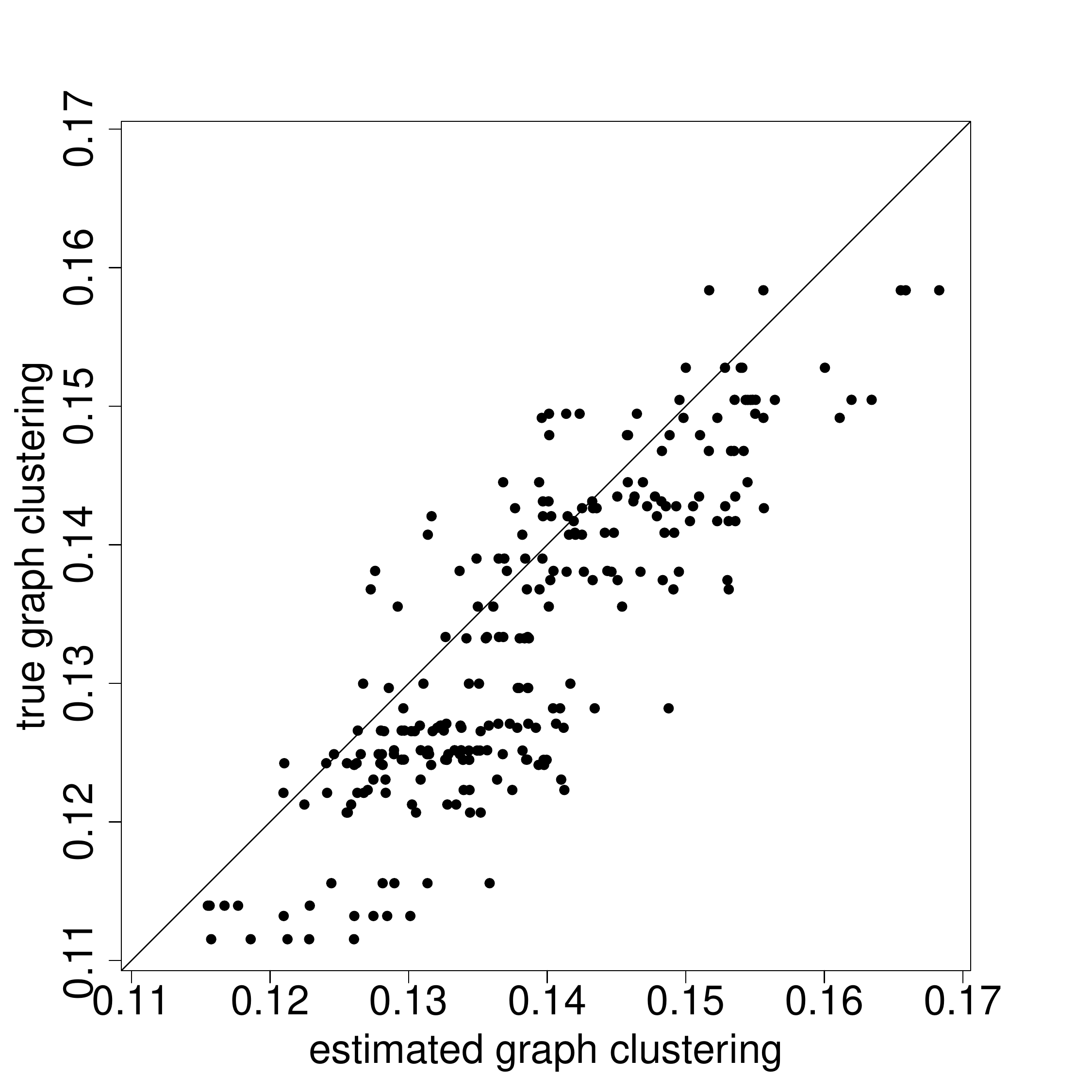}}
\subfloat[Eigenvector Cut]{
\includegraphics[width=.3\textwidth]{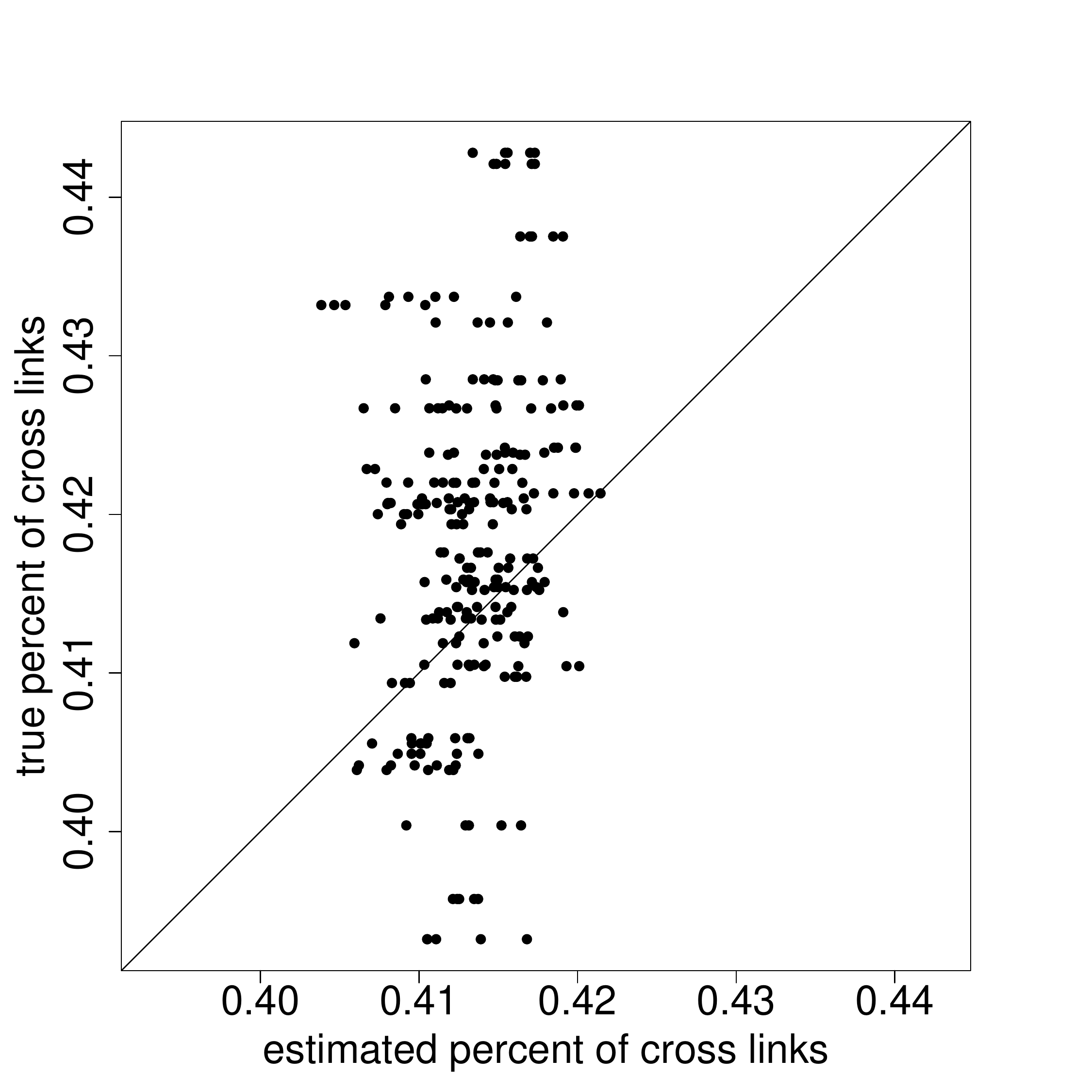}}

\caption{Node level and network level measures estimation for 250 simulations at core simulation set-up. These plots show scatterplots of estimated measure on the x-axis and true measure on the y-axis. There is a strong correlation between estimated statistic and statistic obtained from the true underlying graph, with the exception of eigenvector cut. The weak correlation in eigenvector cut comes from the fact that we sample individuals' and ARD subgroups' latent positions uniformly, as there is no strong separation of two groups in the true simulated graph.}
\label{fig:core_simulation}
\end{figure}

Figure \ref{fig:core_simulation} presents the results of our procedure using synthetic ARD data from graphs generated at the parameters specified above. We see that the procedure works well. Throughout the paper, we look at the degree, eigenvector centrality, and clustering at the node level, as well as the maximal eigenvalue, average path length, clustering, and eigenvector cut at the graph-level.\footnote{The eigenvector cut metric is defined by the eigenvector with the second smallest eigenvalue of the Laplacian matrix. Using the median of the eigenvector to partition the graph gives us two balanced groups of equal size. We plot the fraction of links that cross group boundaries.} The figures also display an additional set of network characteristics including betweenness centrality, closeness centrality, support and distance from ``seed'' at the node level as well as diameter, the fraction of links in the largest connected component and the number of components at the graph level.\footnote{We define support at the individual level to be the fraction of a node's links that are linked to at least one other link of that node. For distance from ``seed'', we arbitrarily choose one node in the graph and measure the minimum path length to that ``seed'' node for all other nodes.}

The figure shows strong correlations between the true value in the simulations and that predicted by the ARD sample for almost all of the statistics examined here. We do note that the correlation is weak in the case of eigenvector cut.  The eigenvector cut takes a narrow range of values in the underlying graph, however, because we simulate the locations of both individuals and groups uniformly across the surface of the sphere.  That is, there is no cut structure in the underlying formation model. Appendix \ref{sec:additional_plots_core} presents plots of additional network measures. There, we note that the estimates are quite close to the true values for several integer-valued statistics including diameter, fraction in the giant component, and the number of components. For these three measures, there is very little variation in the true measures.

\subsection{Varying sparsity}\label{sec:sparsity}

\subsubsection{Varying $\E[\nu_i]$}

We next explore the performance of our procedure when we vary sparsity -- the number of links relative to graph size.  
To do this, we hold all the parameters fixed, including $\zeta = 0.3$ at its original value, but vary the distribution of the node effects. In particular we change the mean of the effect $\mu$, with $\mu \in \{-1.96, -1.62, -1.27, -0.92, -0.58\}$. This varies the expected degree from 5 to 80, holding fixed $n=250$.

\begin{figure}[!h]
\centering
\includegraphics[width=.3\textwidth]{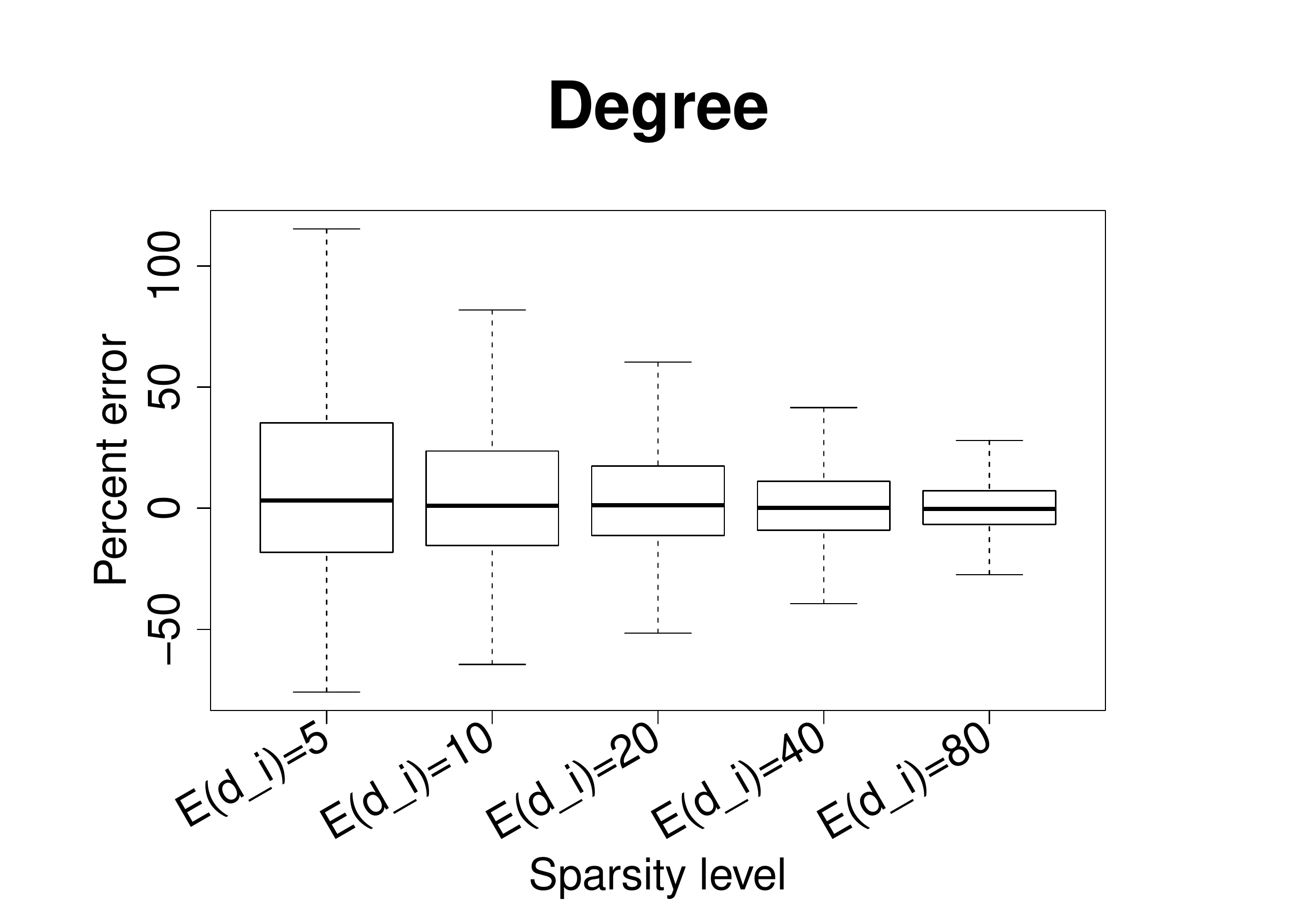}\quad
\includegraphics[width=.3\textwidth]{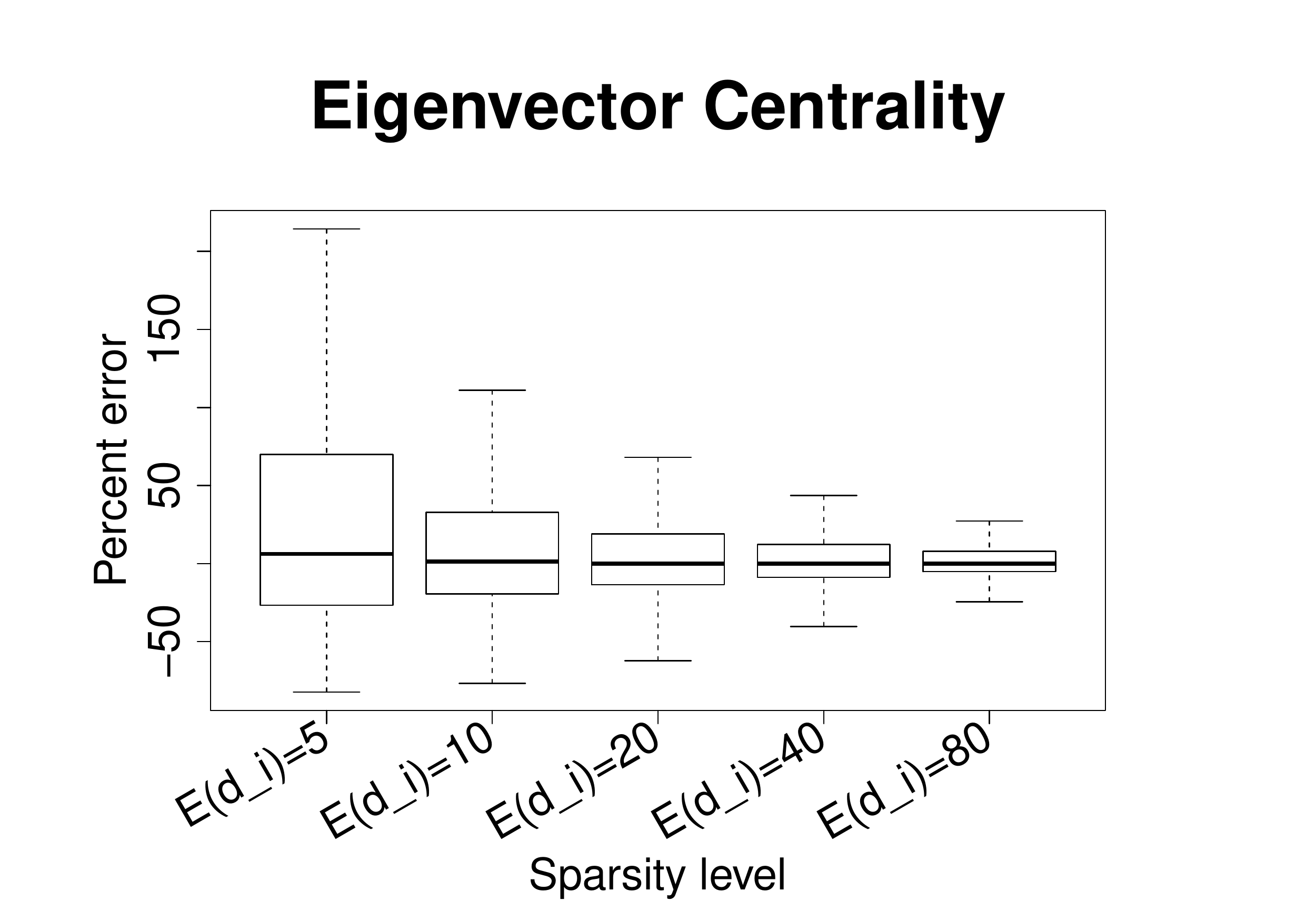}
\includegraphics[width=.3\textwidth]{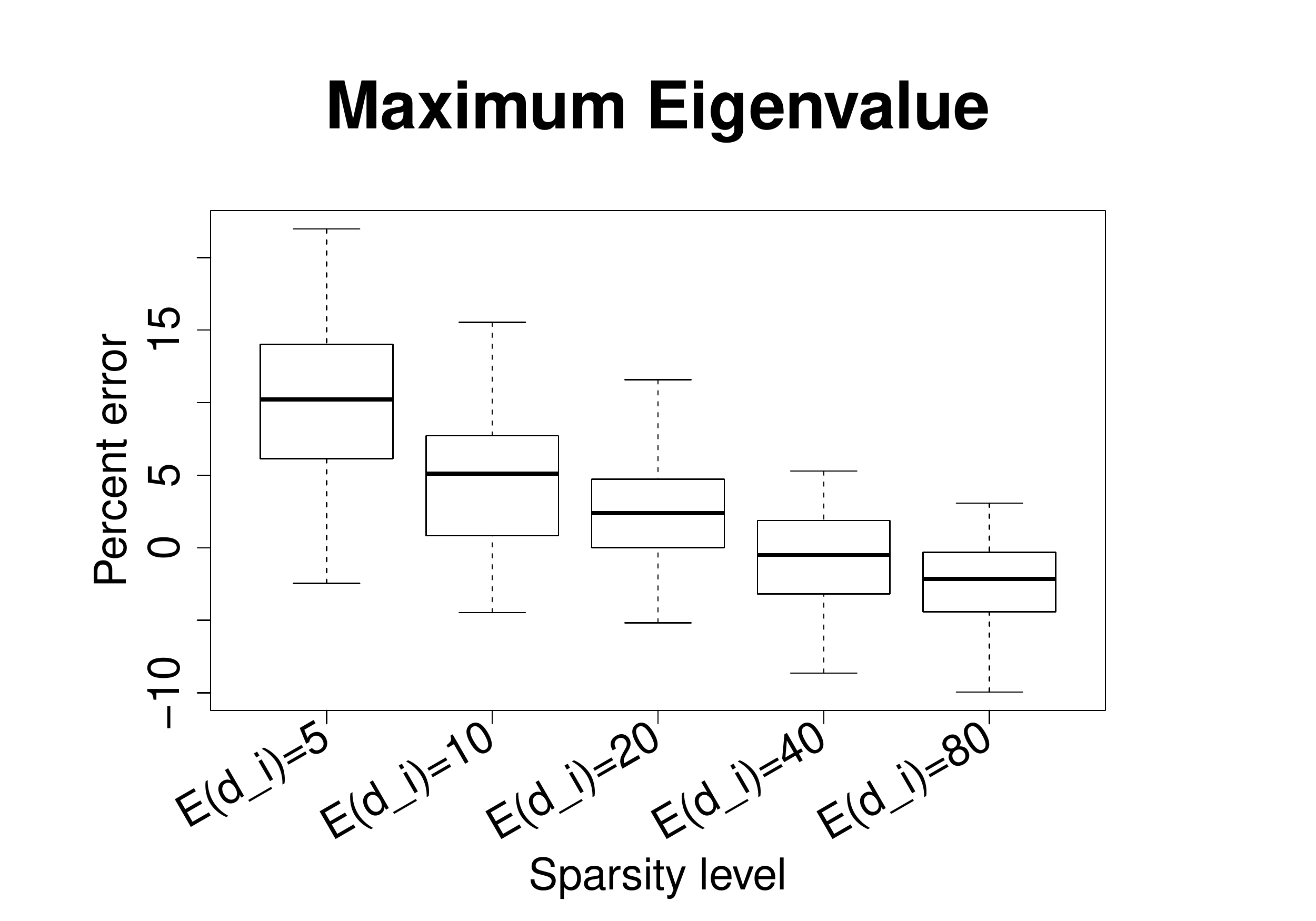}\quad

\medskip

\includegraphics[width=.3\textwidth]{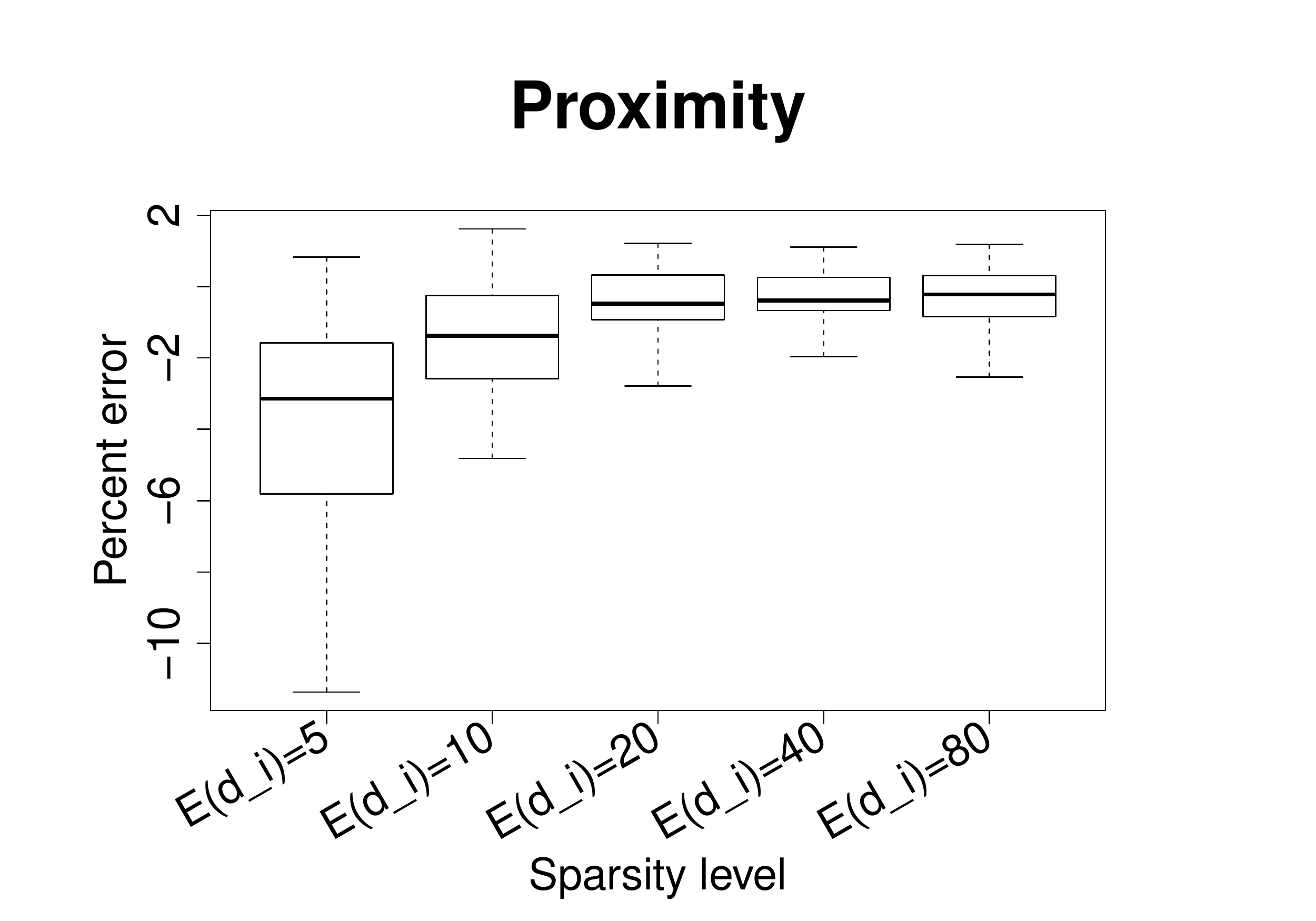}\quad
\includegraphics[width=.3\textwidth]{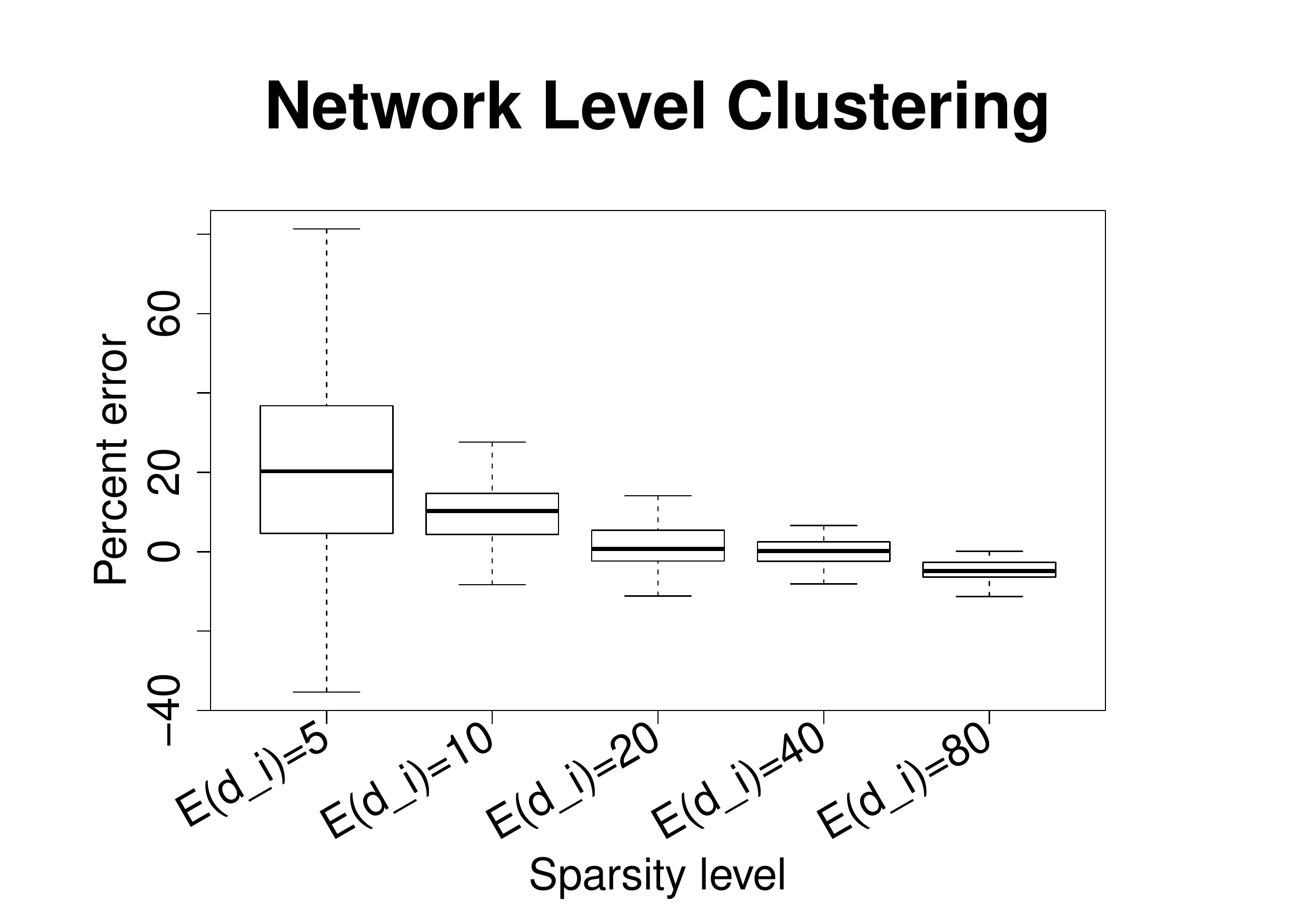}
\includegraphics[width=.3\textwidth]{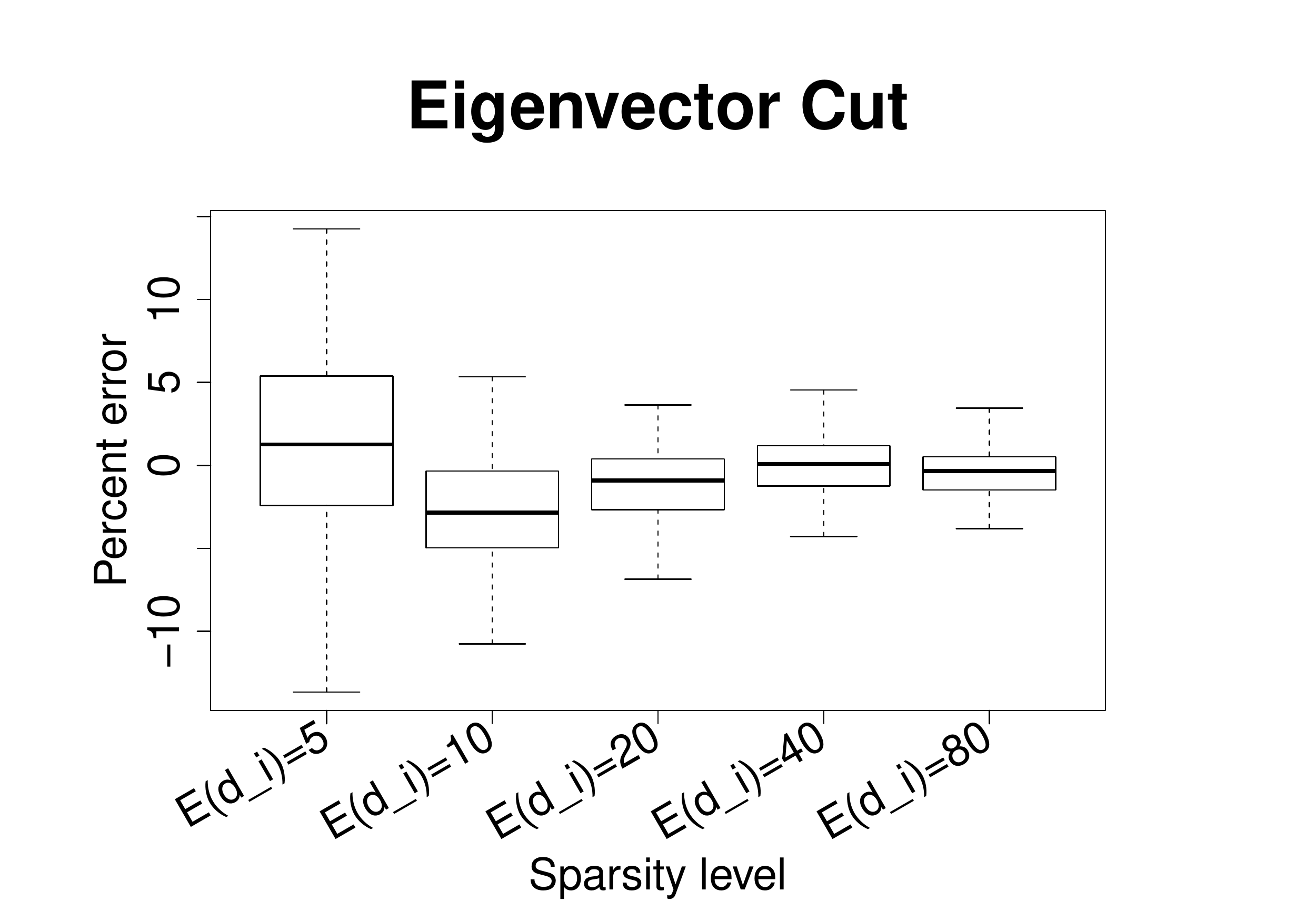}

\medskip

\includegraphics[width=.3\textwidth]{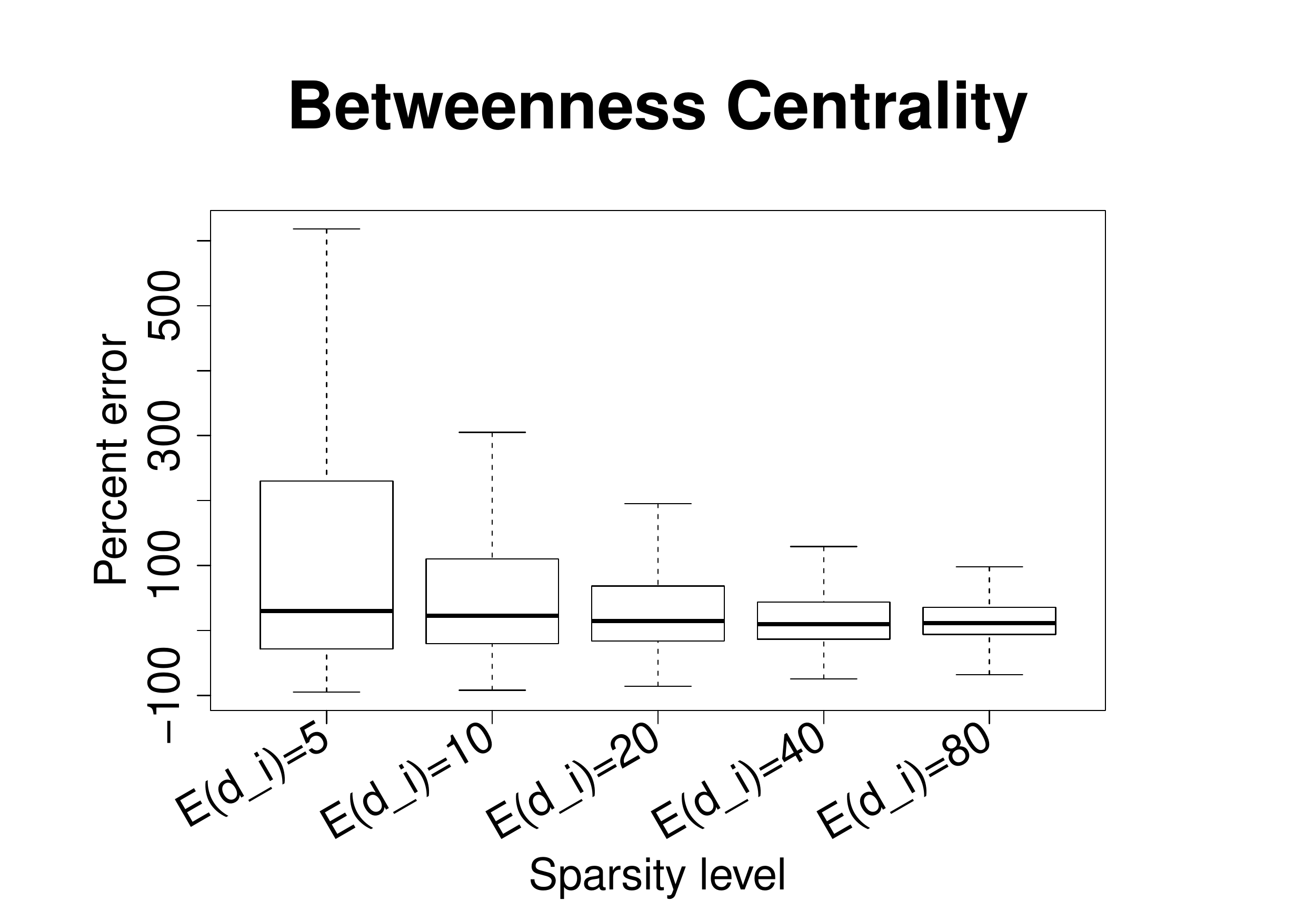}\quad
\includegraphics[width=.3\textwidth]{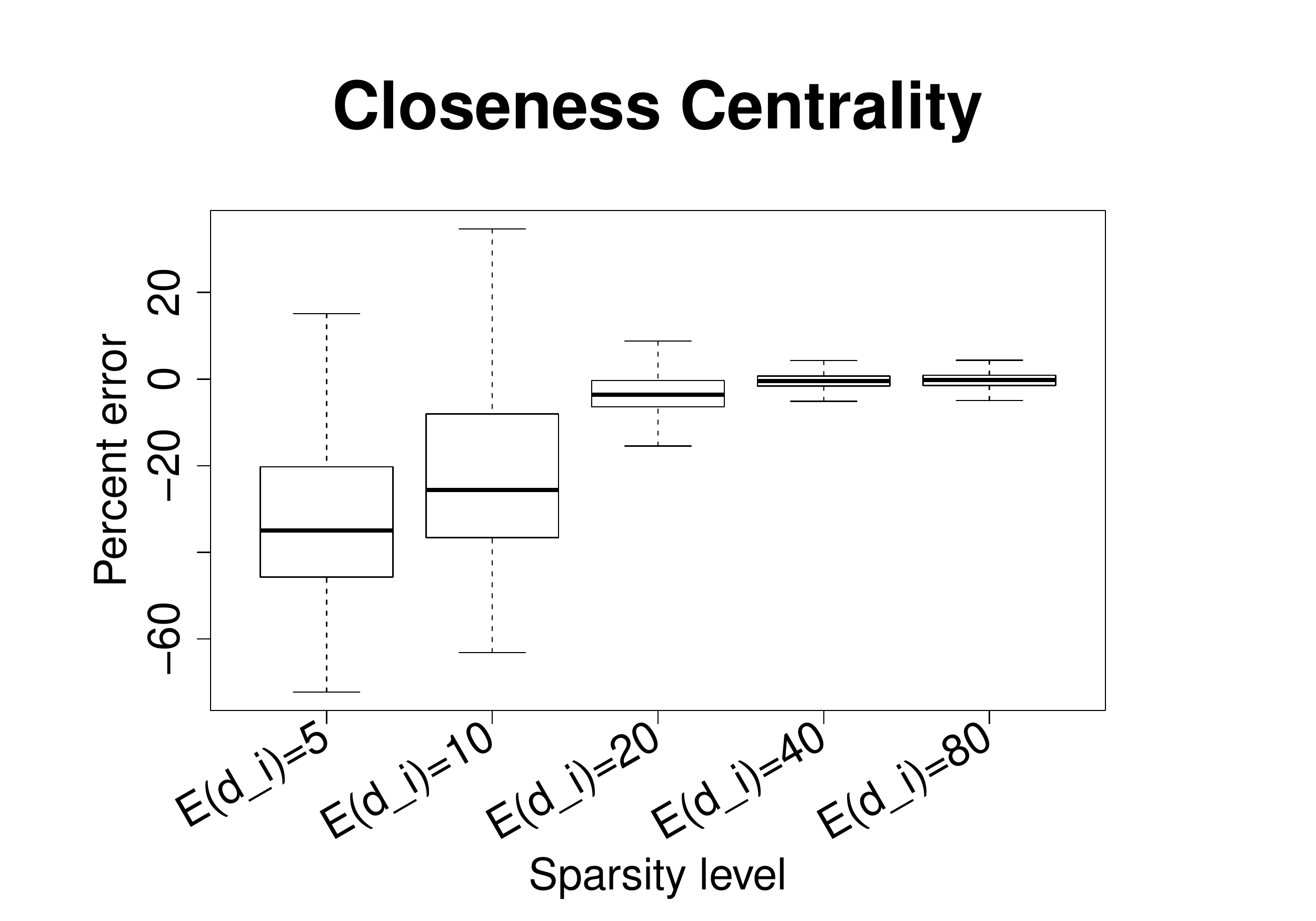}
\includegraphics[width=.3\textwidth]{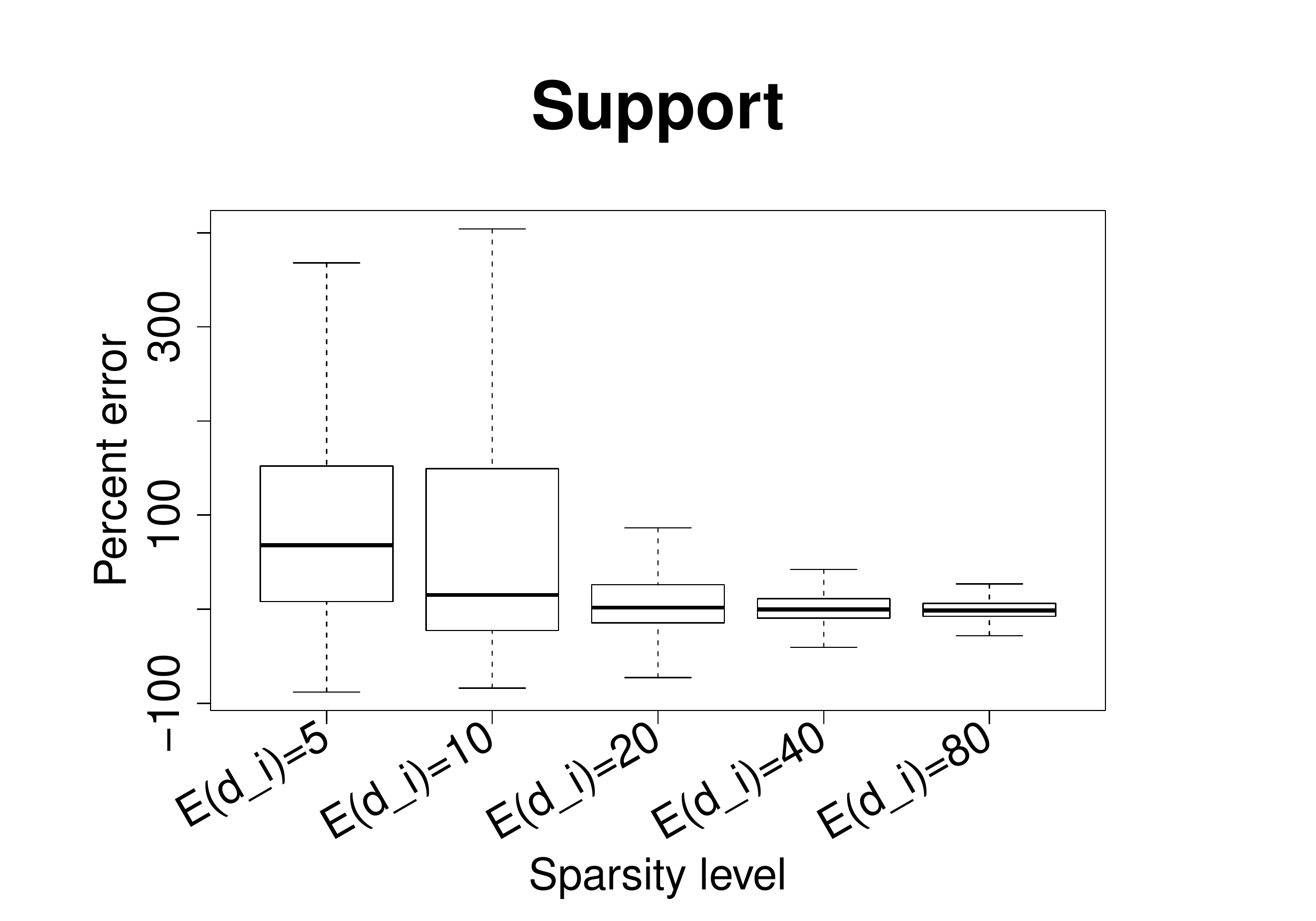}

\caption{Node level and network level measures estimation for 50 simulations at each sparsity level. The plots show boxplots of percentage errors for estimated statistic, with outliers not shown on the graph. For node level measures, the bias is near zero at all sparsity levels except for closeness and support, and variance decrease with decreasing sparsity. 
For closeness, bias increases with increasing sparsity, because the true graph is more likely to have disconnect components as sparsity increases. For network level measures, the bias is overall small. Even for network level clustering estimation at the most sparse level, the middle 50$\%$ has less than 40 percent error.}
\label{fig:sparsity}
\end{figure}

We define the percentage error as the difference between the estimated and true measure divided by the true measure. At each sparsity level, we pool simulations and make plots of mean $\pm$ standard deviation of percent error. Figure \ref{fig:sparsity} shows how well our algorithm estimates these measures at varying sparsity levels. As the graph becomes less sparse, we have smaller bias and variation in the estimation of degree and centrality. For maximum eigenvalue, proximity, and clustering, the bias in estimation has a monotone pattern. For proximity and clustering, we have less variation as the graph becomes less sparse. For eigenvector cut, the bias is very small at all sparsity levels and the variation decreases as the graph becomes less sparse.

\subsubsection{Sparse with thick tails}

Our next exercise is to approximate networks that exhibit heavy tails. That is, the network may mostly be sparse but some nodes may have extremely high degree. To operationalize this, we hold all the parameters fixed as before, but now draw $\nu_i$ from a Normal distribution with $\mu=-0.92, \sigma=0.3$ with probability $\lambda$ and from a Normal distribution with $\mu=-1.96,\sigma=0.3$ with probability $1-\lambda$. The high centrality nodes have, on average, expected degrees of 40, while the rest have, on average, expected degrees of 5. We pick $\lambda=0.1$ so the average number of high centrality nodes is 25, but the actual number may vary in each simulation. The goal of this exercise is to study whether we can pick out which members of the network have high eigenvector centrality, which is important in a diffusion process for instance, even though the graph is extremely sparse. 

\begin{table}[ht!]
\centering

\scalebox{0.8}{\begin{tabular}{lllll}
\multicolumn{2}{l}{\multirow{2}{*}{}}                                            & \multicolumn{2}{l}{Estimated top decile}                 & \multirow{2}{*}{}        \\ \cline{3-4}
\multicolumn{2}{l|}{}                                                             & \multicolumn{1}{l|}{Yes}   & \multicolumn{1}{l|}{No}     &                          \\ \cline{2-5} 
\multicolumn{1}{l|}{\multirow{2}{*}{True top decile}} & \multicolumn{1}{l|}{Yes} & \multicolumn{1}{l|}{18.16} & \multicolumn{1}{l|}{6.84}   & \multicolumn{1}{l|}{25}  \\ \cline{2-5} 
\multicolumn{1}{l|}{}                                 & \multicolumn{1}{l|}{No}  & \multicolumn{1}{l|}{6.84}  & \multicolumn{1}{l|}{218.16} & \multicolumn{1}{l|}{225} \\ \cline{2-5} 
\multicolumn{2}{l|}{}                                                             & \multicolumn{1}{l|}{25}    & \multicolumn{1}{l|}{225}    & \multicolumn{1}{l|}{250} \\ \cline{3-5} 
\vspace{2mm}
\end{tabular}}
\caption{Confusion matrix of top decile eigenvector centrality estimation\label{table:heavytails}}
\begin{tablenotes}
Notes: The confusion matrix reports how well the method picks the top decile of eigenvector central nodes. The rows represent true instances, while the columns represent predicted instances. Thus, the diagonal of the matrix reports true positives and true negatives, while the off-diagonal elements capture mislabeled instances. 
\end{tablenotes}
\end{table}

Table \ref{table:heavytails} shows the confusion matrix for this exercise---which presents true positives, true negatives, false positives, and false negatives---for the top decile eigenvector centrality estimation average over 50 simulations. With a $73\%$ true positive rate and a $27\%$ false positive rate, we successfully recover the majority of high centrality nodes. We note that the actual number of high centrality nodes varies in each simulation, which results in some noise in our estimation. 
\

\subsection{Varying network size and sampling share}\label{sec:network_size}

\begin{figure}[htp]
\centering

\includegraphics[width=.4\textwidth]{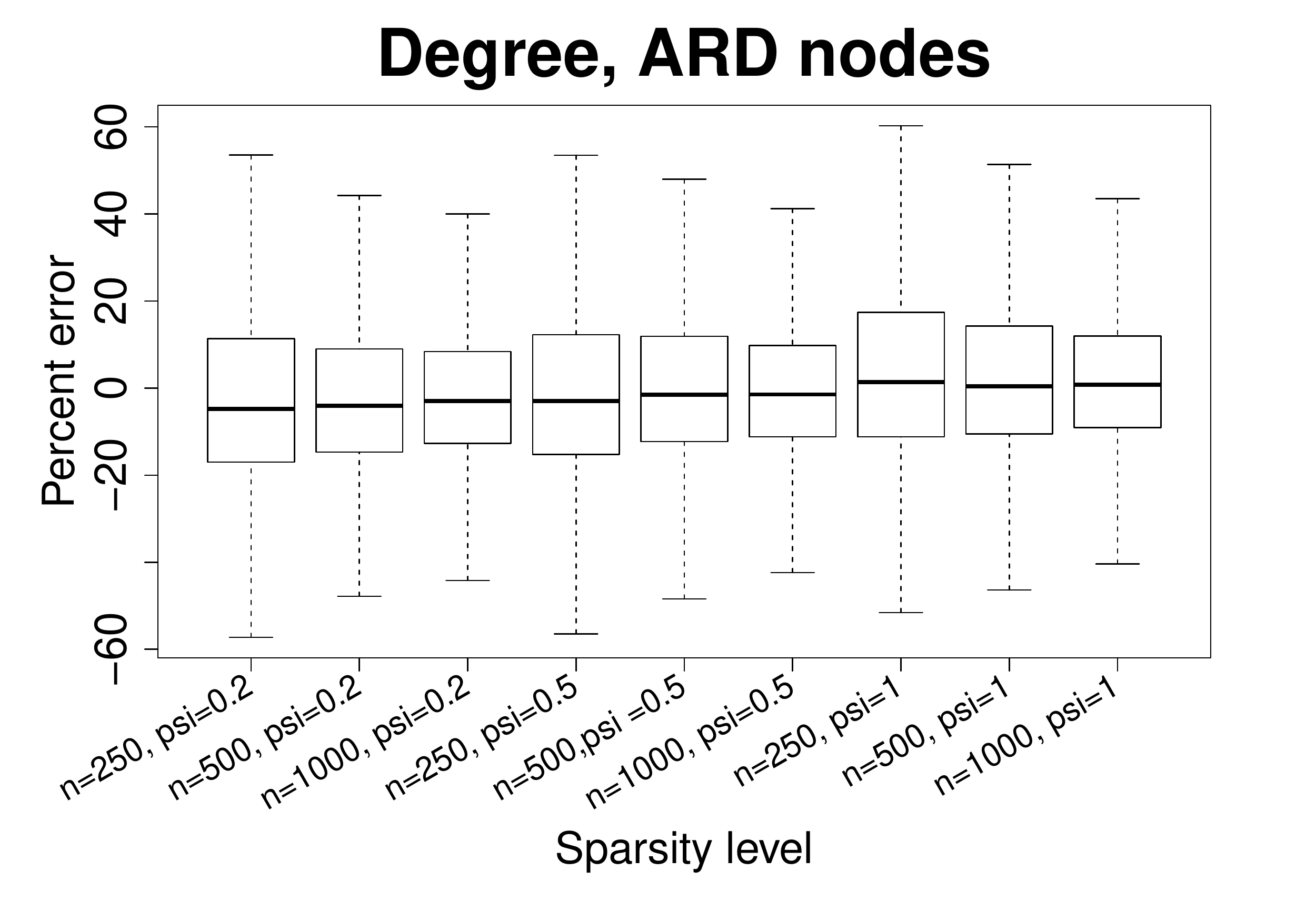}\quad
\includegraphics[width=.4\textwidth]{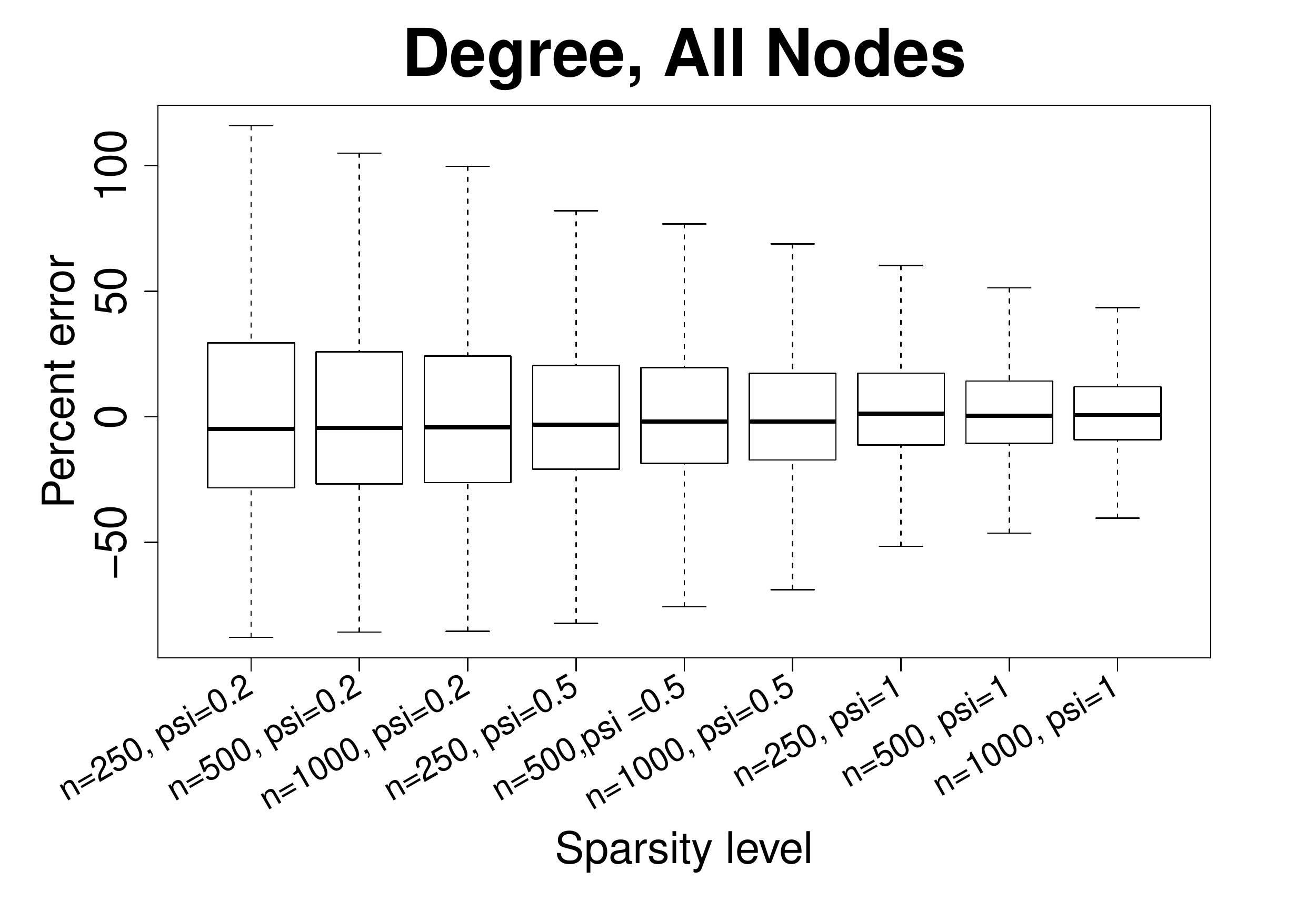}

\medskip

\includegraphics[width=.4\textwidth]{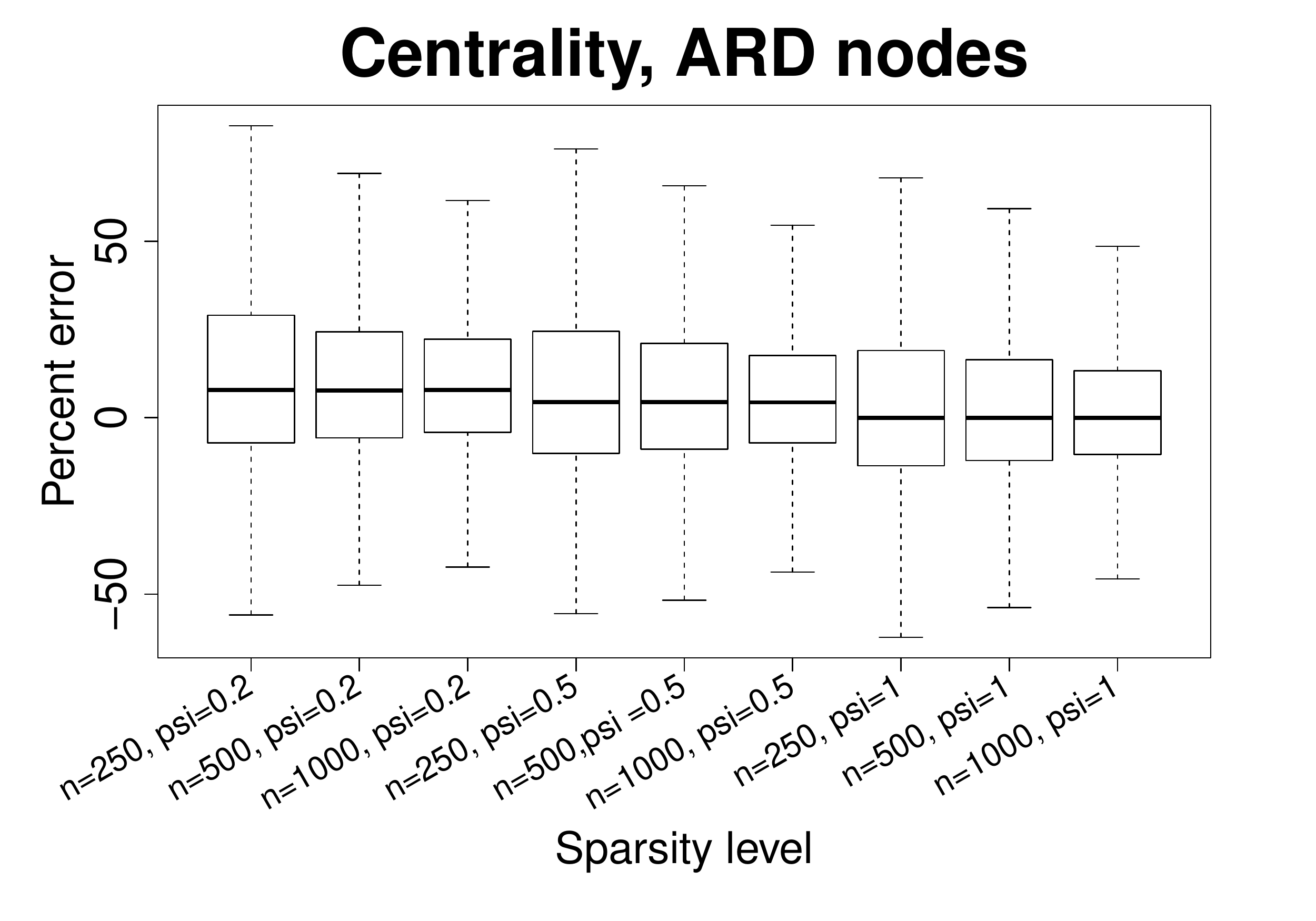}\quad
\includegraphics[width=.4\textwidth]{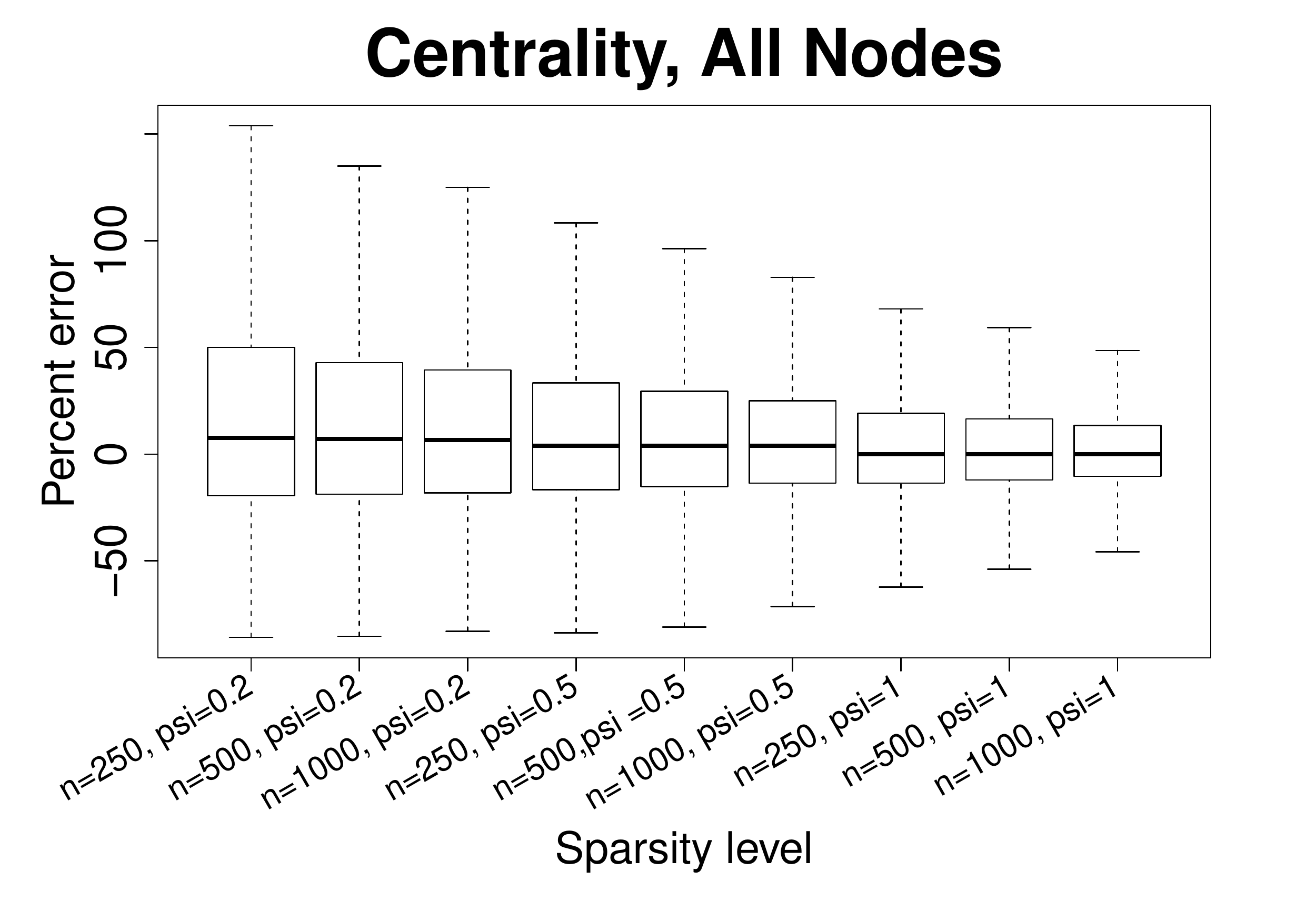}

\medskip

\includegraphics[width=.4\textwidth]{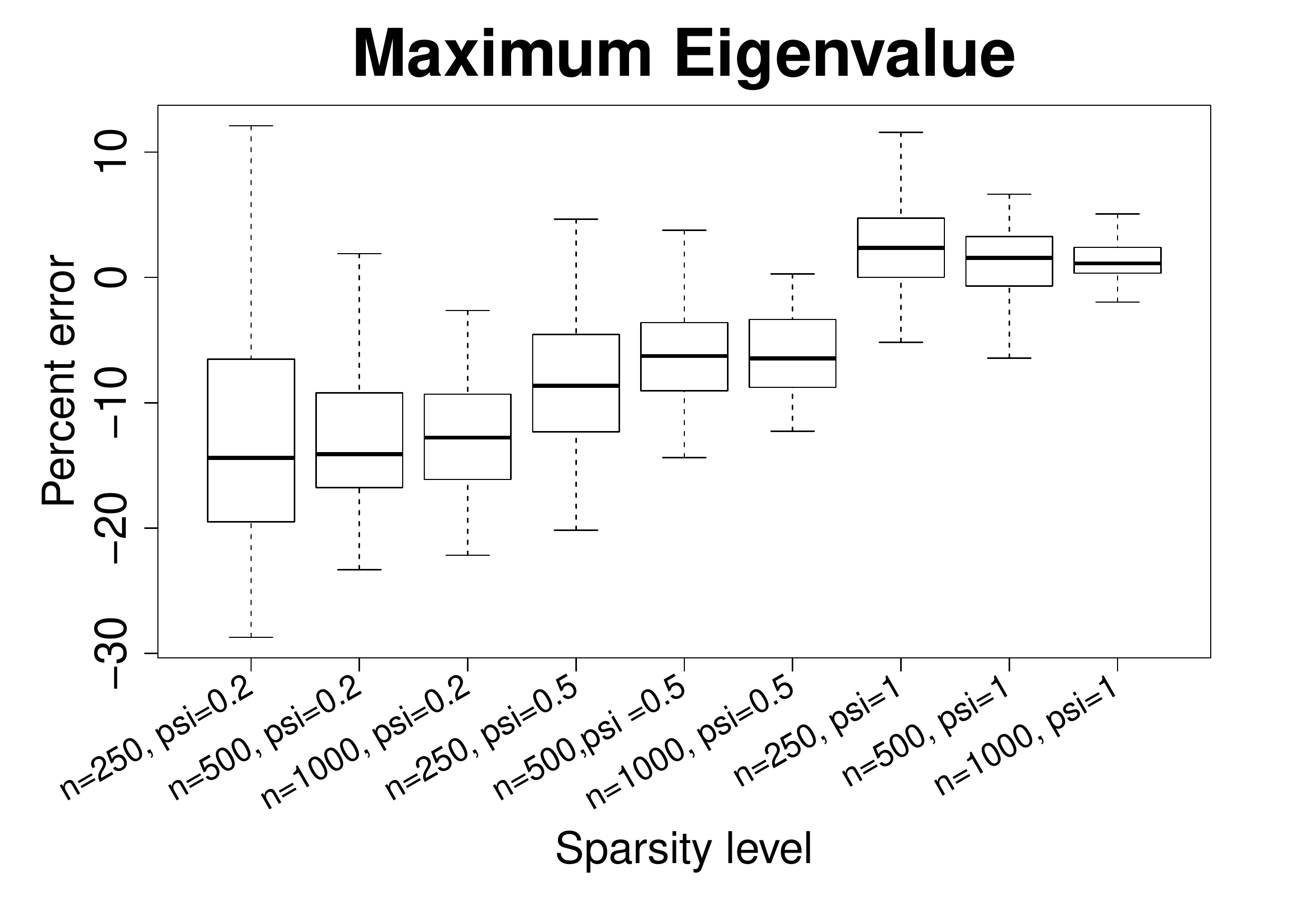}\quad
\includegraphics[width=.4\textwidth]{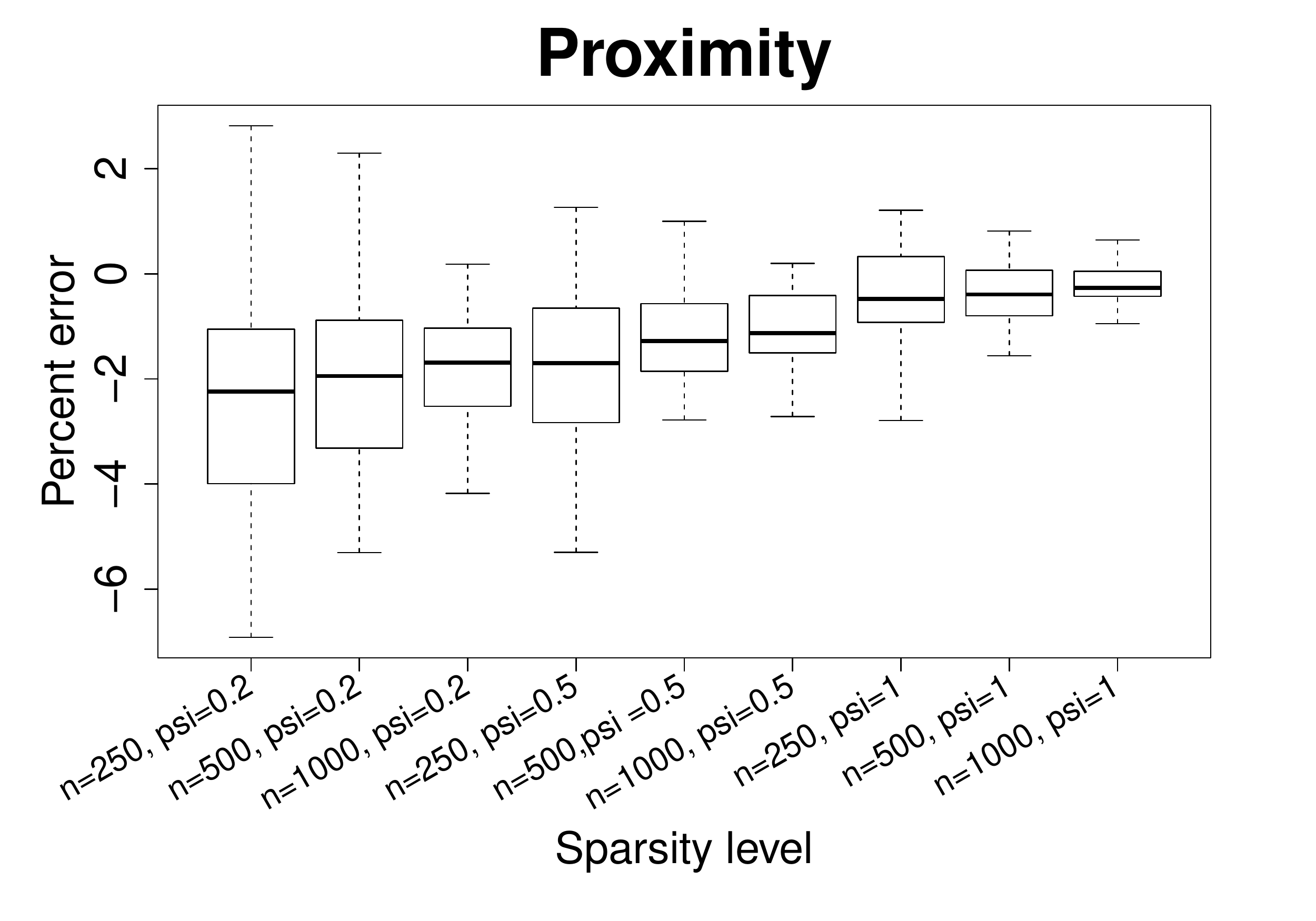}\quad

\medskip

\includegraphics[width=.4\textwidth]{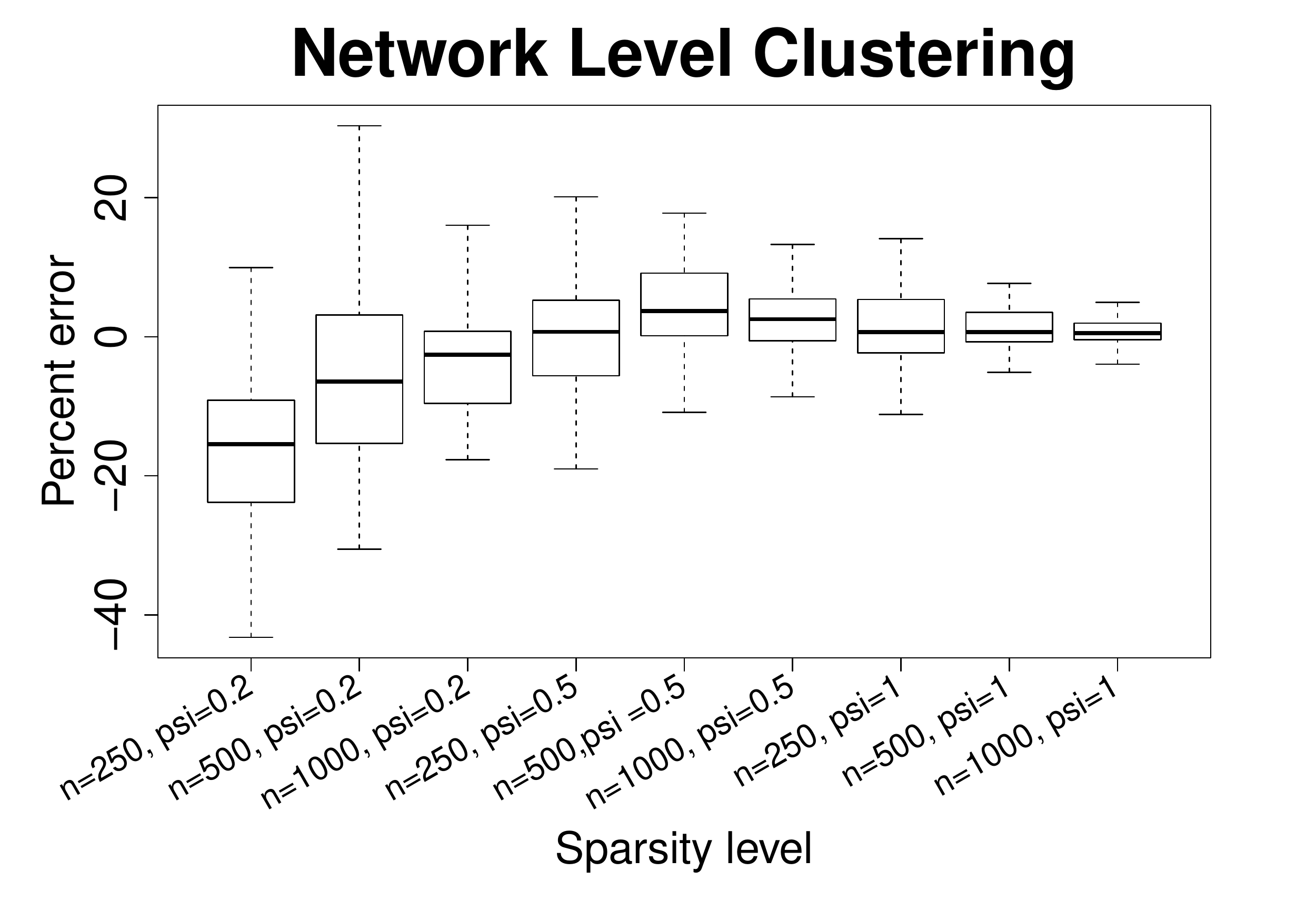}
\includegraphics[width=.4\textwidth]{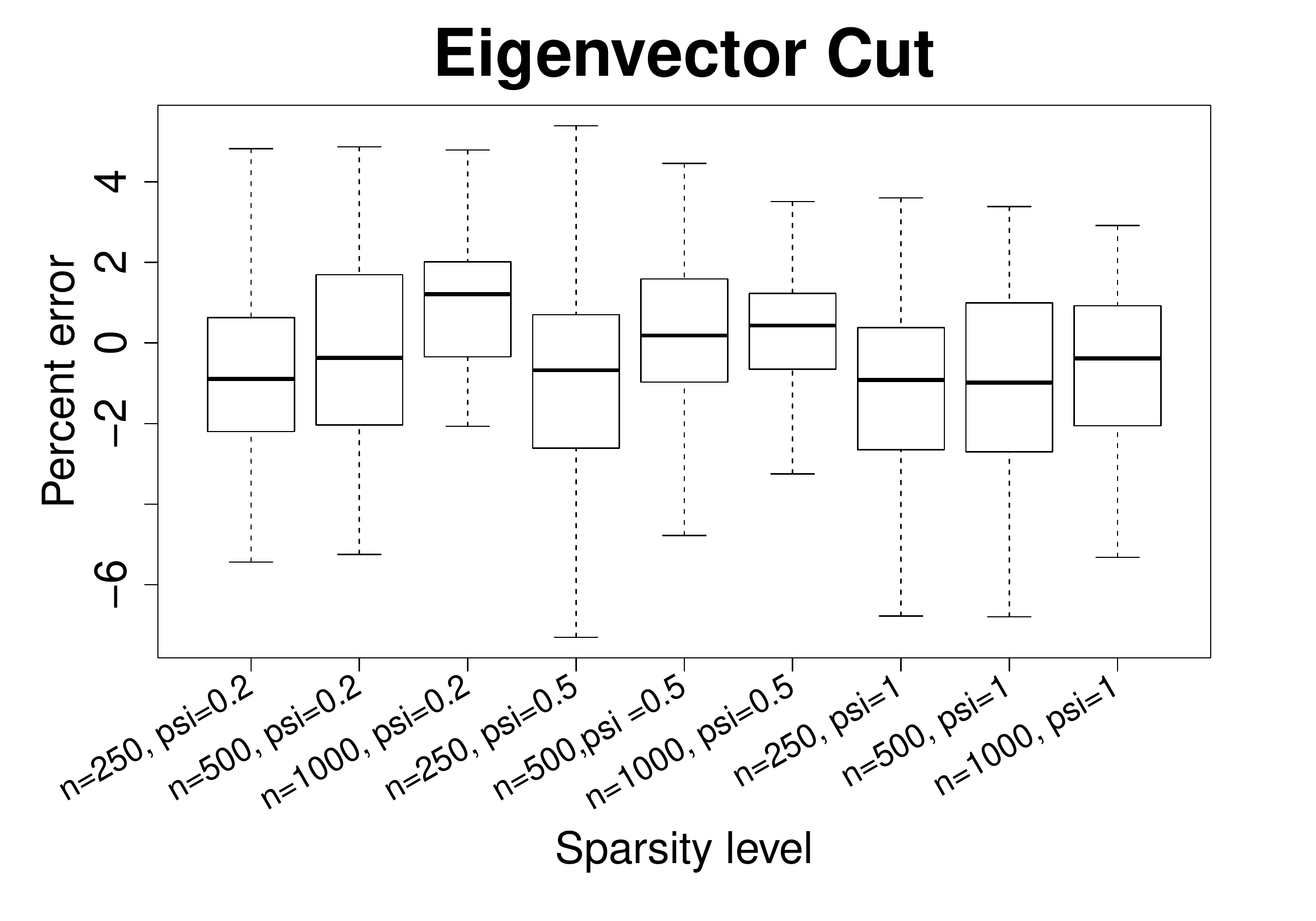}

\caption{Node level and network level measures estimation for 50 simulations at each combination of $\psi \in \{0.2,0.5,1\}$ and $n \in \{250, 500, 1000\}$. The plots show boxplots of percentage errors for estimated statistic, with outliers not shown on the graph. The typical bias for node level statistic estimation is near zero at all levels of $\psi$ and $n$, and variance decreases as we increase $\psi$ and $n$. Our estimation of network level statistics improve with increasing $\psi$ and $n$, with the exception of eigenvector cut. The estimated percentage of cross links has low bias and variance at all levels of $\psi$ and $n$.}
\label{fig:psivary}
\end{figure}

Next we study what happens as we move from what we call a rural environment to an urban environment, exploring  what happens as the number of nodes in the population gets larger, and when we have to reduce the ARD sampling share.  In particular we vary $n \in \{250, 500, 1000\}$. We also vary the share in the ARD sample, $\psi \in \{0.2,0.5,1\}$. When $\psi<1$, we sample demographic features $X$ for all nodes with $X_{i1} \sim N(\nu_i,\sigma)$. We construct $X_{i2}$ such that $X_{i2}$ is in one of eight categories depending on the sign of each coordinate of $z_i$.

Figure \ref{fig:psivary} presents estimation results when we vary $n$ and $\psi$. When $\psi$ is fixed, in general we have less bias and variation as we increase $n$. When $n$ is fixed, performances of degree and centrality estimation on ARD nodes are similar at various $\psi$. As we expect, increasing the share of ARD nodes increases the precision of node level measures estimation for all nodes.
\

We underestimate maximum eigenvalue when we do not have 100\% ARD sampling, and we overestimate maximum eigenvalue when we do have a $100\%$ ARD sample. We underestimate average path length at all $n$ and $\psi$; the bias in estimation decreases as we increase $n$ and $\psi$. Our estimation of network level clustering is within $20\%$ of true value most of the time, and our estimation of the percentage of cross links using eigenvector cut is mostly within $5\%$ of the true values.

\

\section{Simulations with Real-World Networks}\label{sec:results}

The goal of this section is to take the technique to the field and see how well, in a real, empirically-relevant context, we might have done using ARD in place of full network data. After all, our ARD technique can only do as well as the latent surface model specified in Equation \ref{eqn:network-model} does at capturing network structure.\footnote{Here, we remind the reader that ARD information and other insights from our method could, in principle, be applied to other network formation models that may better suit certain applications. For instance, the sub-graph generated models (SUGMs) discussed in \cite{chandrasekharj2012} allow for violations of the ``triangle inequality'' in latent space to generate a different distribution of triangles among nodes. We conjecture it is straightforward to identify SUGMs through ARD.} 

Our choice of a parametric model clearly has implications for the performance of the method and carries with it some of the limitations of random geometric graph sorts of models: conditional on locations on the surface, it is unlikely for very distant nodes to ever link, making so-called ``short-cuts'' rather rare events. Further, clustering in the network (e.g. homophily based on a given characteristic) is accomplished through the positions of particular individuals in the latent space.\footnote{If a node is more likely to link to those whose locations are nearby, and the network neighbor is also more likely to link to those with close latent locations, then the initial node is also on average going to be in relatively close proximity to the neighbor's friends on the latent space, leading to a higher linking probability and higher levels of clustering.} If there is a clear cleavage in the network (and the ARD questions asked on the survey also make it possible to detect this), then our model will generate graphs that faithfully reflect this distinction. If, however, there is a weak preference for connection within rather than between groups, this will be more difficult to detect.

\subsection{Setting and Data} 
\label{sec:data}
We aim to show the potential for ARD to be used in place of detailed social network maps. To do this, we begin with the rich network data collected by \cite{banerjeegossip}. This consists of network data from 89\% of 16,476 households across 75 villages in Karnataka, India. Thus, in the undirected, unweighted graph, we have information about 98\% of all potential links. The survey asks about 12 types of interactions: (1) whose house the respondent visits; (2) who visits the respondent's house; (3) kin in the village; (4) non-relatives with whom the respondent socializes; (5) who provides help on medical decisions; (6) from whom the respondent borrows money; (7) to whom the respondent lends money; (8) from whom the respondent borrows material goods such as kerosene or rice; (9) to whom the respondent lends such material goods; (10) from whom the respondent receives advice before an important decision; (11) to whom the respondent gives advice; and (12) with whom the respondent goes to temple, mosque or church. We use a graph which is undirected and unweighted, taking a link as the union over all the above dimensions. The ratio of average degree over network size ranges from 0.04 to 0.21, with a median of 0.08. The sparsity level is the same as our core simulation, where ratio of expected degree over network size is 20/250=0.08.

We asked 12 additional questions in a follow-up survey 12 months later to a random sample of approximately 30\% of households, covering traits such as owning a tractor, having met with an accident, illness incidents, birth of twins, educational attainment and family composition. We use 8 of  these 12 traits as the basis for the ARD analysis. The other four questions are deleted because they are rare or non-informative of sampled households' positions in the network.

Our first goal is a proof of concept for the use of ARD and the latent distance model to generate a posterior distribution for each graph.  To do this, we construct ARD responses for the 30\% sample: what would be the aggregate counts these respondents would have given us had we asked them ARD questions? It also allows us to abstract from errors in knowledge or in recall by survey respondents.\footnote{For example, we know the tractor ownership of each individual in the 30\% sample.  We can then construct the number of links of each ARD respondent to others in the ARD sample who have a tractor. This gives us the ARD responses for the induced subgraph.  To estimate the number of links to tractor-owning households in the full graph, we can simply scale by the sampling rate.} 

For what follows, the 29\% of the households with supplemental surveys form our ARD sample, while the remaining 71\% of households are non-ARD nodes. Because we construct ARD responses for households who answer supplemental surveys in each village, the actual percentage of households with constructed ARD responses varies by village. One village only has a 6.7\% sampling rate and therefore gets dropped, increasing the sampling rate across all villages used to 30\%. Recall that we observe a set of demographic covariates collected in the census of \cite{banerjeegossip} for all nodes and we can use these covariates to predict $\nu_i$ and $z_i$ for nodes not in the ARD sample.

\subsection{Network Level Results}

We begin by looking at the same network-level statistics that we have focused on throughout the paper: $\lambda_1(g)$, social proximity, clustering, and eigenvector cut. 

\begin{figure}[htp]
\centering
\subfloat[$\lambda_1(g)$]{
\includegraphics[width=.3\textwidth]{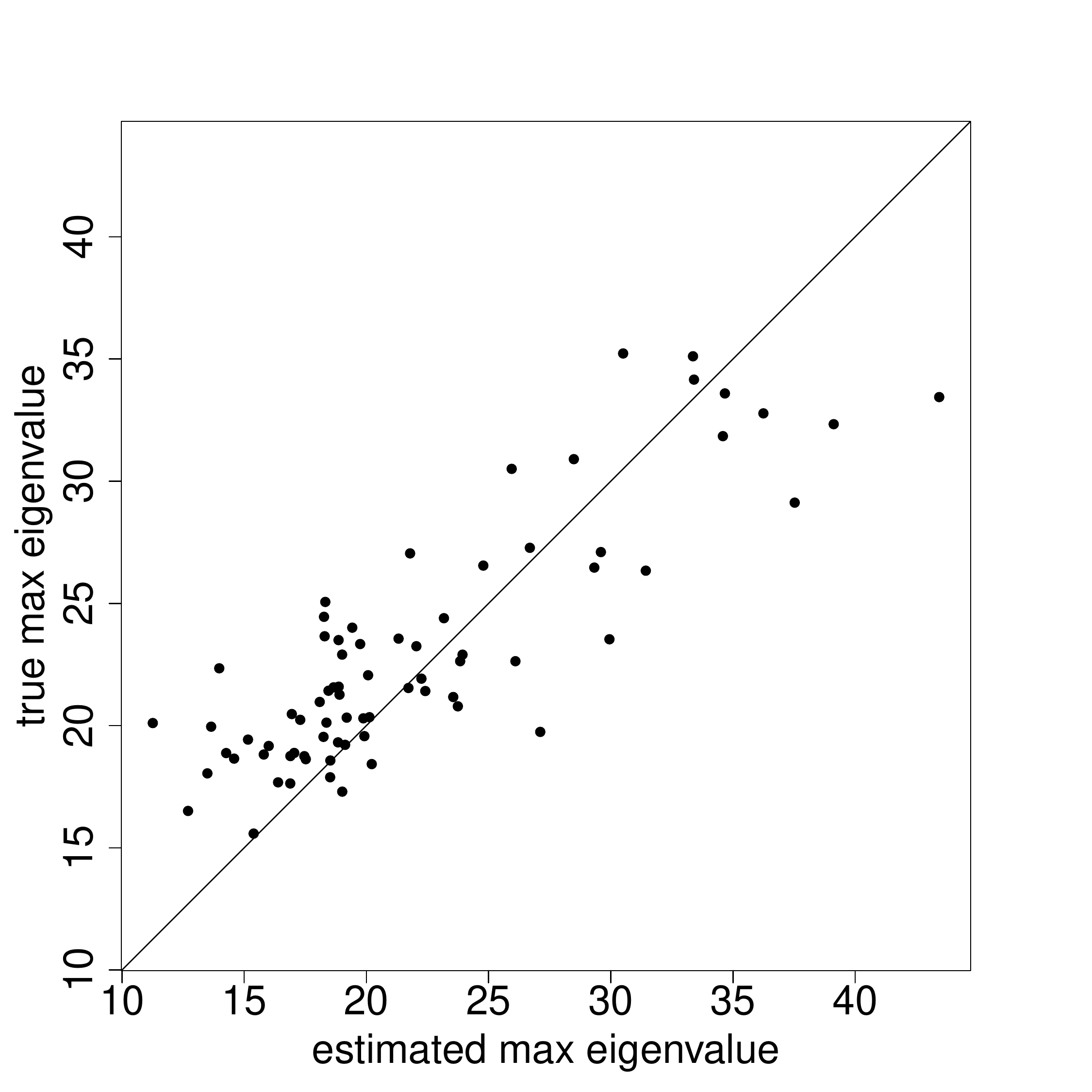}\qquad
}
\subfloat[Social Proximity]{
\includegraphics[width=.3\textwidth]{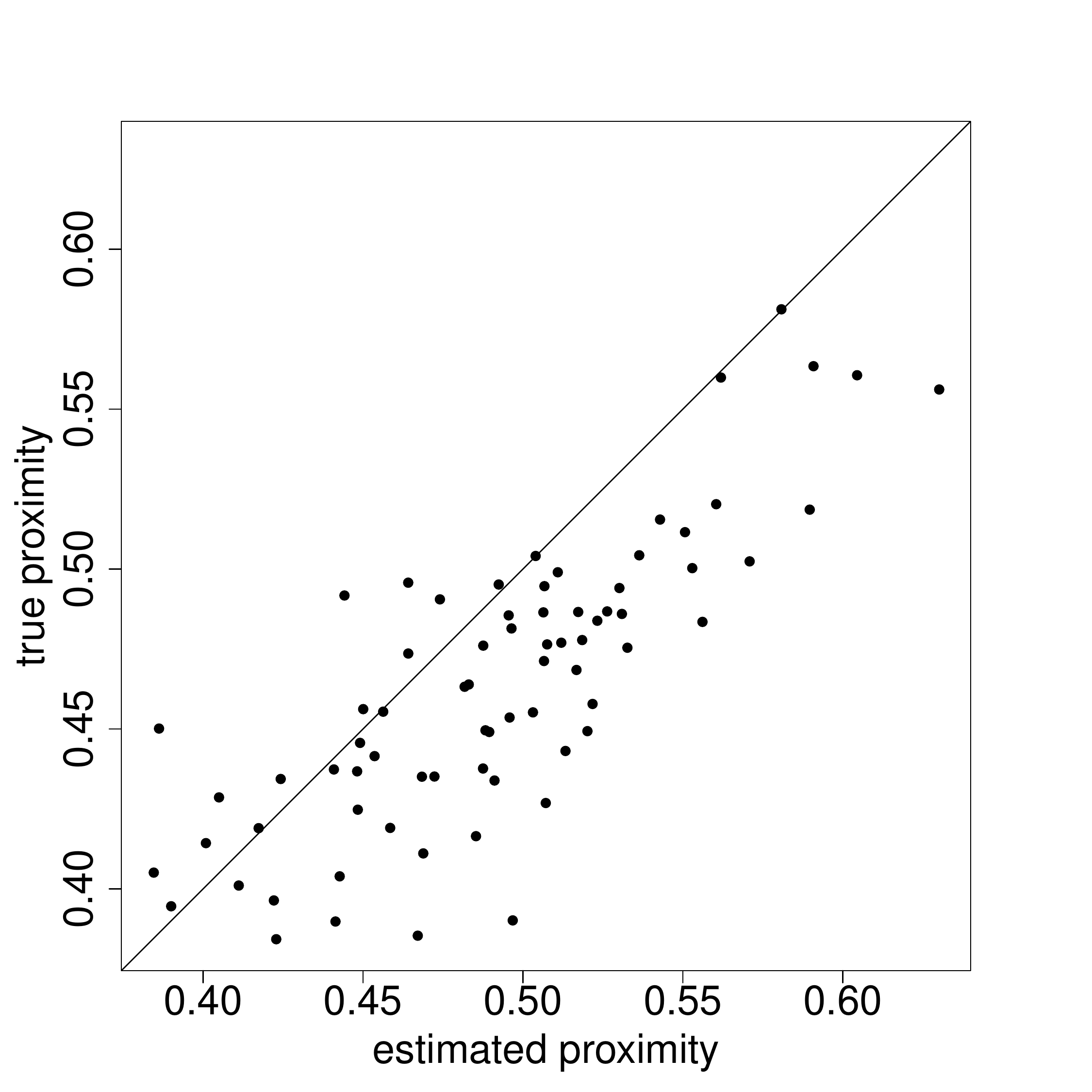}
}
\medskip

\subfloat[Clustering]{
\includegraphics[width=.3\textwidth]{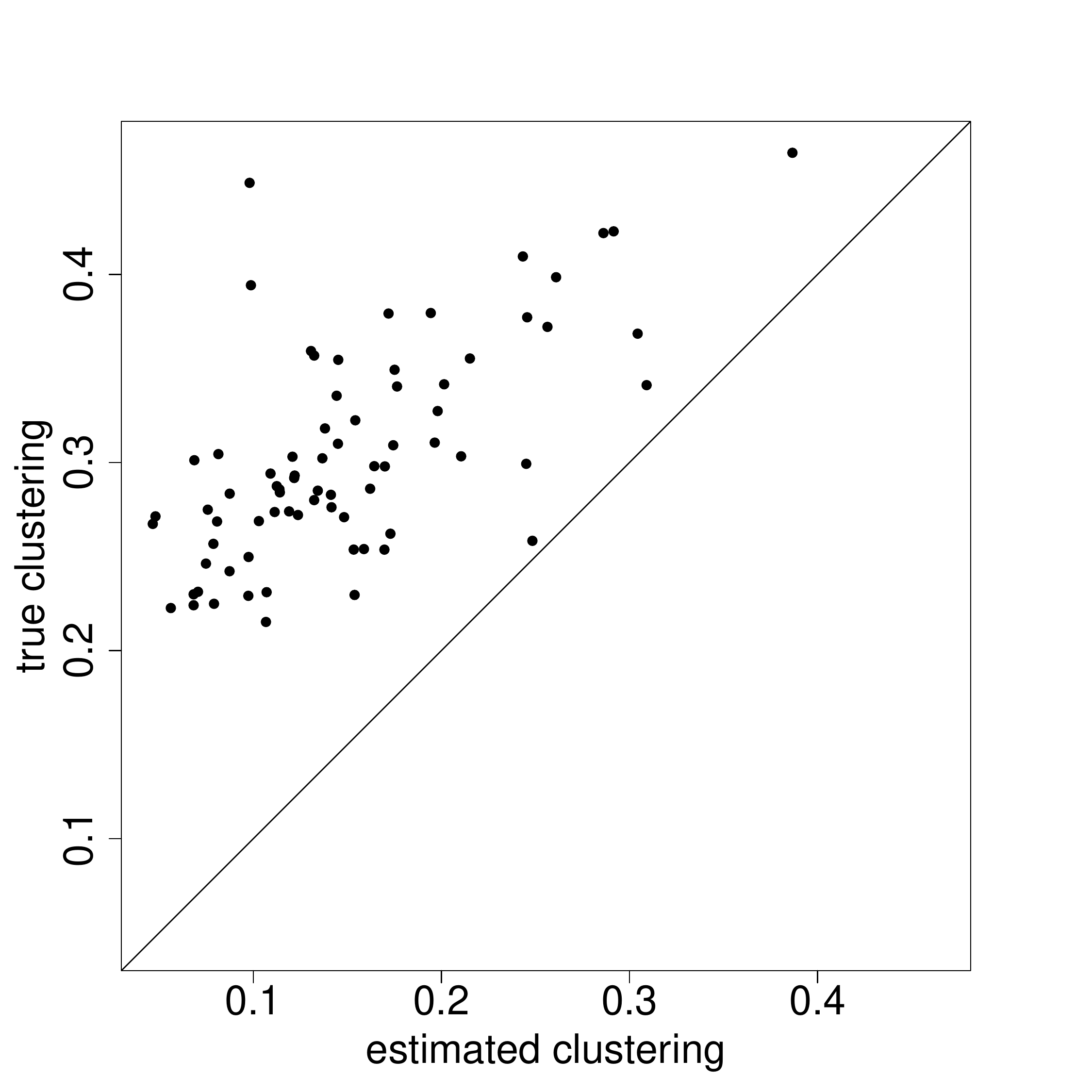}
}
\subfloat[Eigenvector Cut]{
\includegraphics[width=.3\textwidth]{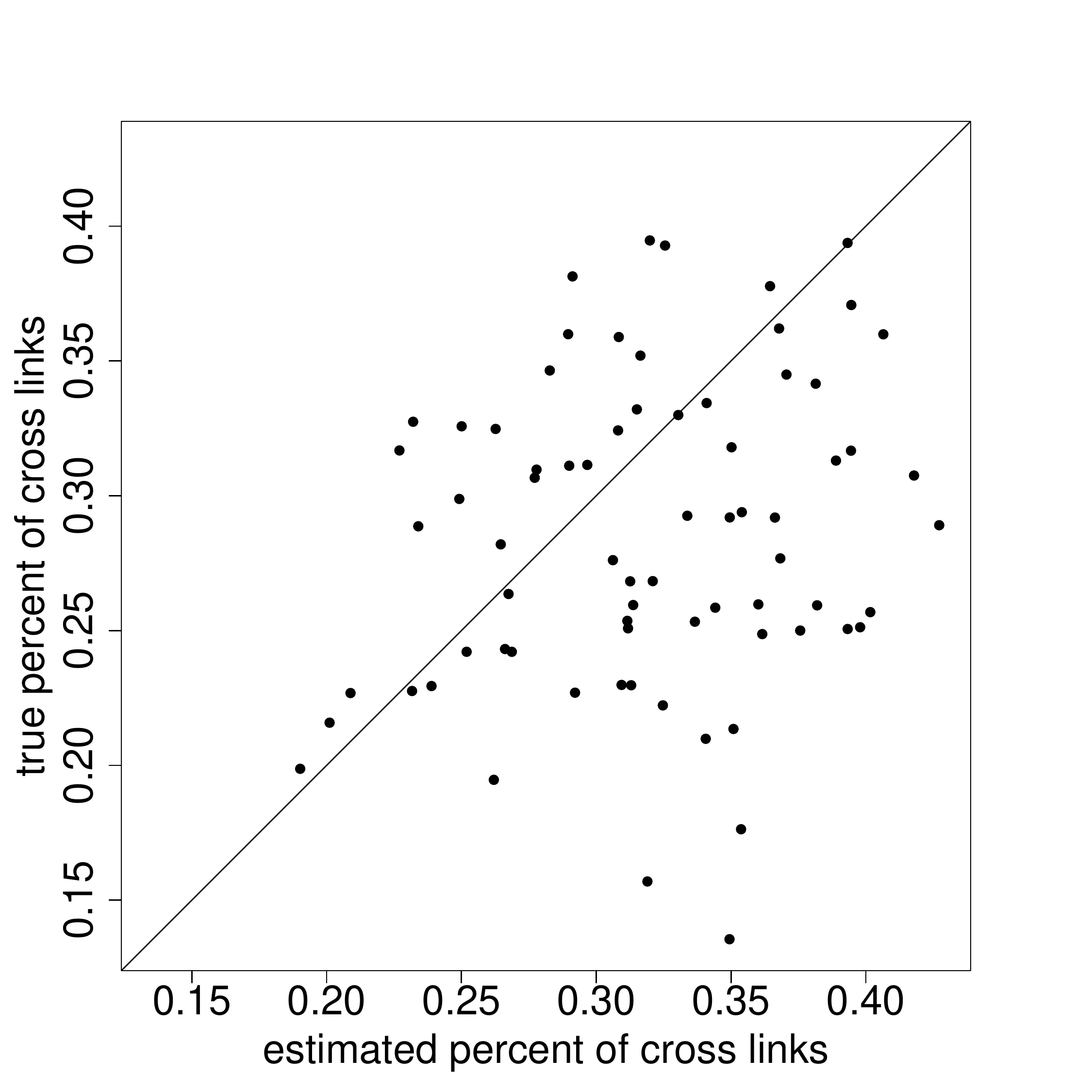}
}

\caption{Network level measures estimation for households in villages in Karnataka.  These plots show scatterplots across all villages with the estimated network level measure on the x-axis and the measure from the true underlying graph on the y-axis.  There is correlation between the estimated statistic and the true statistic, even though there is some bias for clustering.}
\label{fig:indian_network}
\end{figure}

Figure \ref{fig:indian_network} plots the results.\footnote{See Appendix \ref{sec:additional_plots_karnataka} for plots of additional network statistics at both the graph and node levels.} In particular, each panel plots the posterior mean for the network statistic in question against the true value in the data, for each of the 75 villages. We see, rather remarkably, that these global network features are rather well-captured by the ARD procedure. The procedure is weakest for clustering but note that though there is clearly a bias, it is small and out-performs many off-the-shelf models of network formation \citep{chandrasekharj2012}.\footnote{See Table 2 in the paper which compares the implied network level statistics (e.g., eigenvector cut, maximal eigenvalue, clustering, average path length) when we fit (1) a conditional edge independent model flexibly using a rich set of covariates and (2) the same conditional edge independent model but adding in node-level fixed effects (i.e., the model of \cite{graham2014econometric}). Both of these miss across the board in terms of the relevant network statistics. Finally, prior work has demonstrated theoretical failures of consistency of ERGMs when links and triangles are introduced (as would be needed to model realistic data) and also slow (exponential) mixing times for MCMCs used in estimation \citep{shalizi2013consistency,bhamidi2011mixing}. Therefore, our model out-performs, both theoretically and through simulations, conditional edge independent models, \cite{graham2014econometric} which adds fixed effects but no latent locations, and non-trivial ERGMs.}

\subsection{Node Level Results}

Next we turn to node-level results. We again focus on degree, eigenvector centrality, and clustering.

Figure \ref{fig:indian_ARD} presents the results for the ARD sample and Figure \ref{fig:indian_Full} presents the results for the entire sample. We see from Figure \ref{fig:indian_ARD} that the estimated degree, eigenvector centrality, and clustering coefficient are strongly correlated with the true values in the data (Panels A, B, C). Furthermore, in Panels D, E, and F we plot the percent error averaged over all nodes in the sample by village, plotted by village ordered by standard deviation of percent errors.

\begin{figure}[!h]
\centering
\subfloat[Degree]{
\includegraphics[width=.3\textwidth]{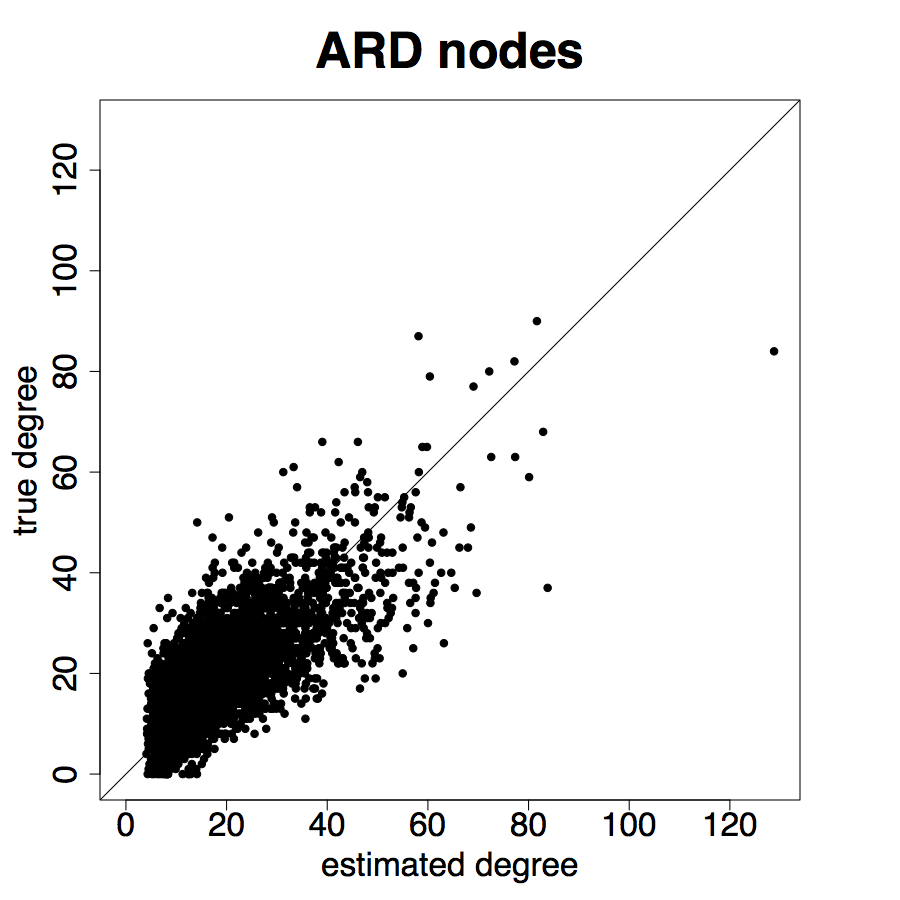}

}
\subfloat[Eigenvector Centrality]{
\includegraphics[width=.3\textwidth]{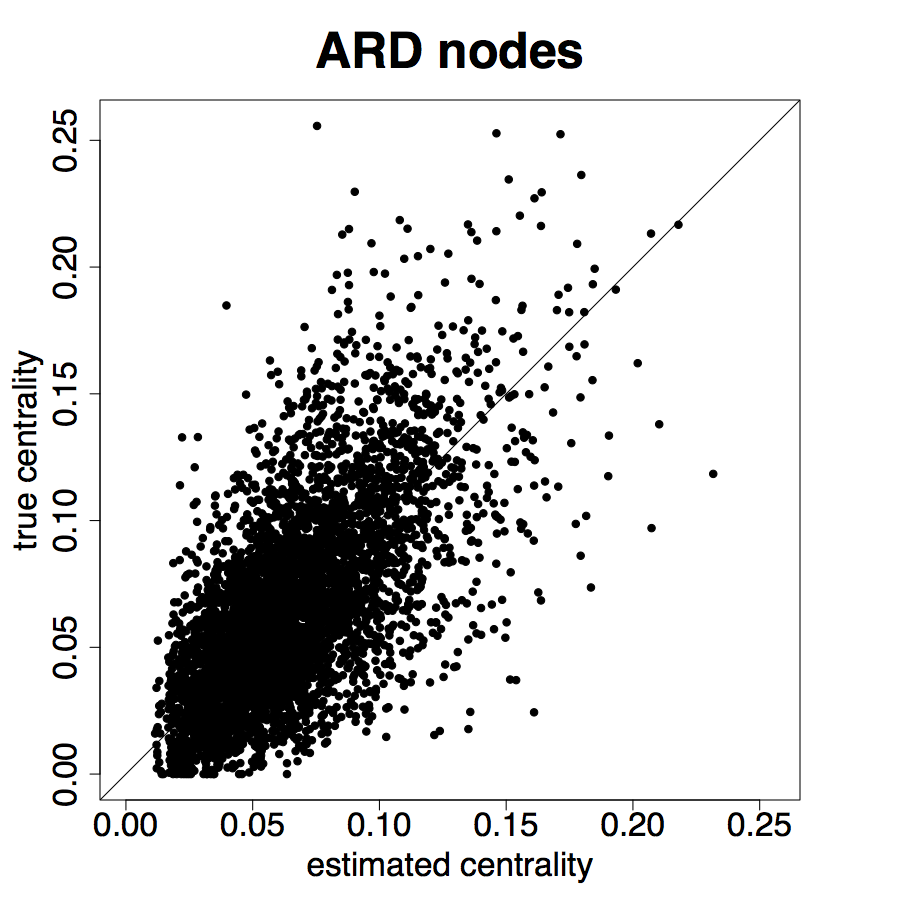}
}
\subfloat[Clustering]{
\includegraphics[width=.3\textwidth]{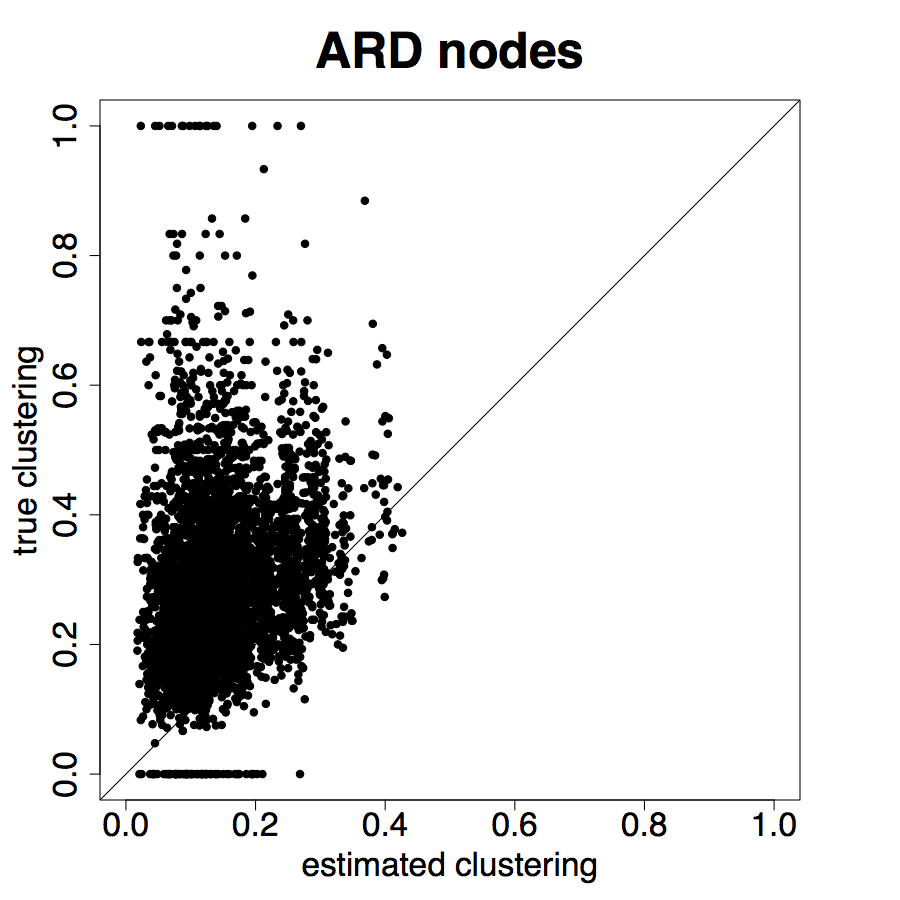}

}

\medskip

\subfloat[Degree]{
\includegraphics[width=.3\textwidth]{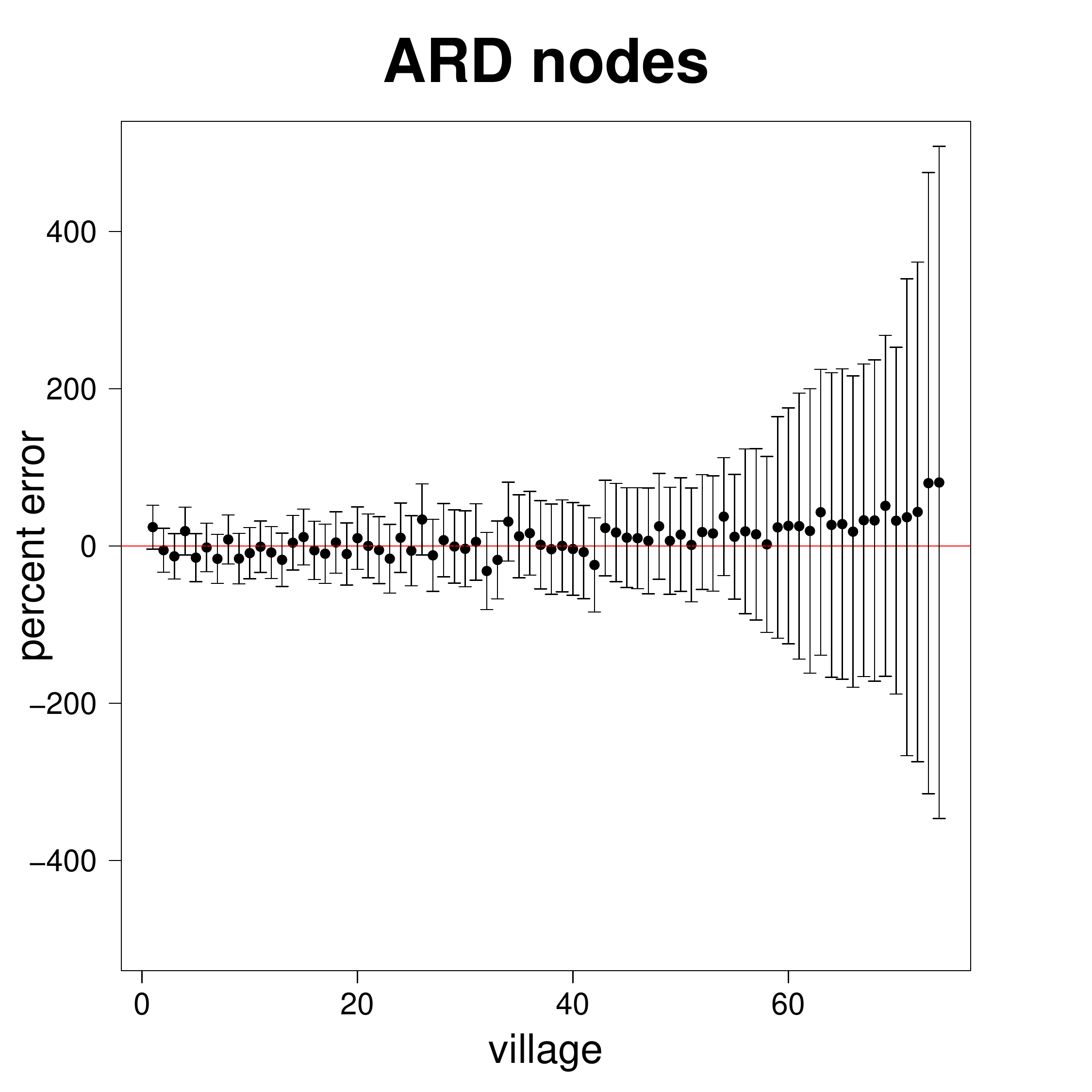}

}
\subfloat[Eigenvector Centrality]{
\includegraphics[width=.3\textwidth]{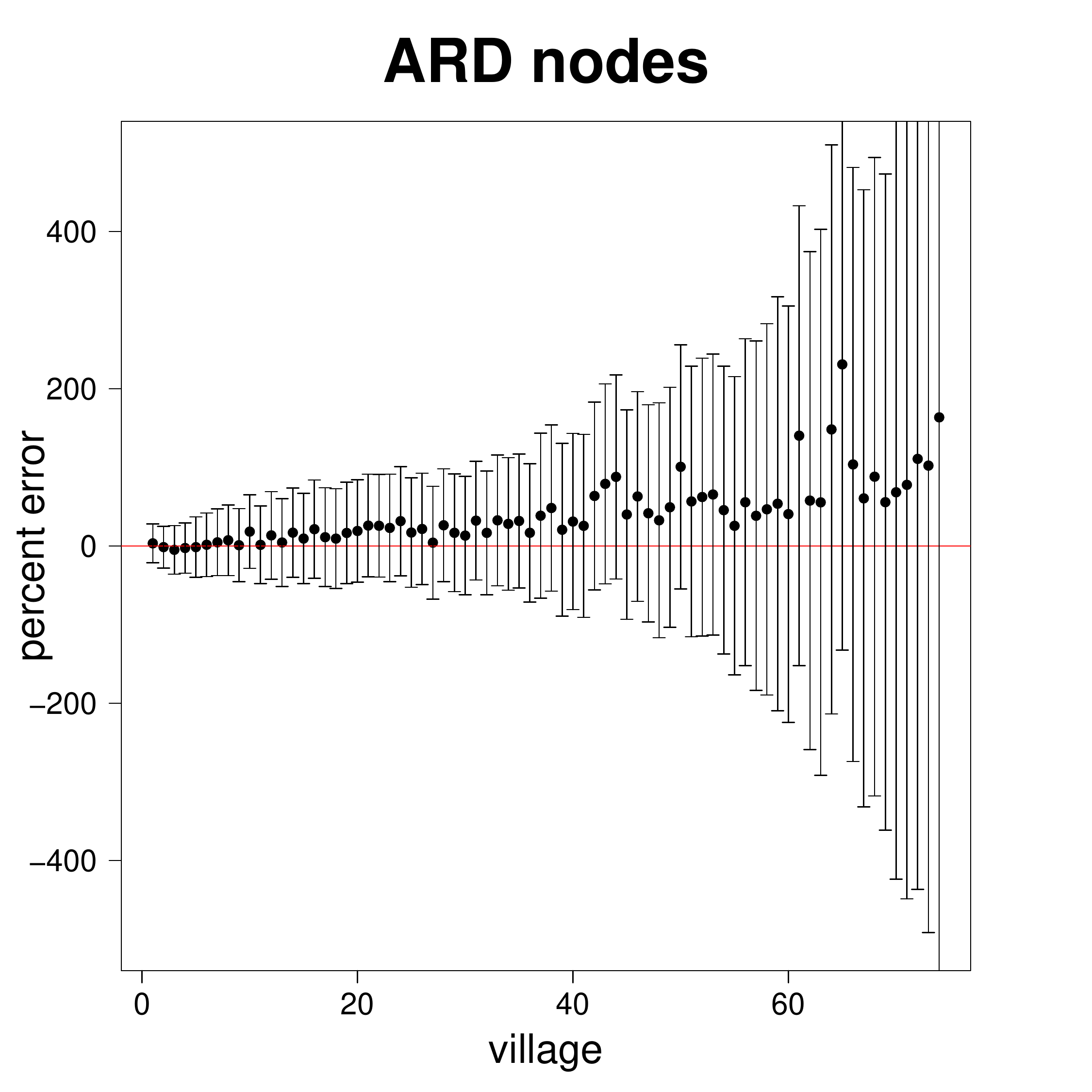}
}
\subfloat[Clustering]{
\includegraphics[width=.3\textwidth]{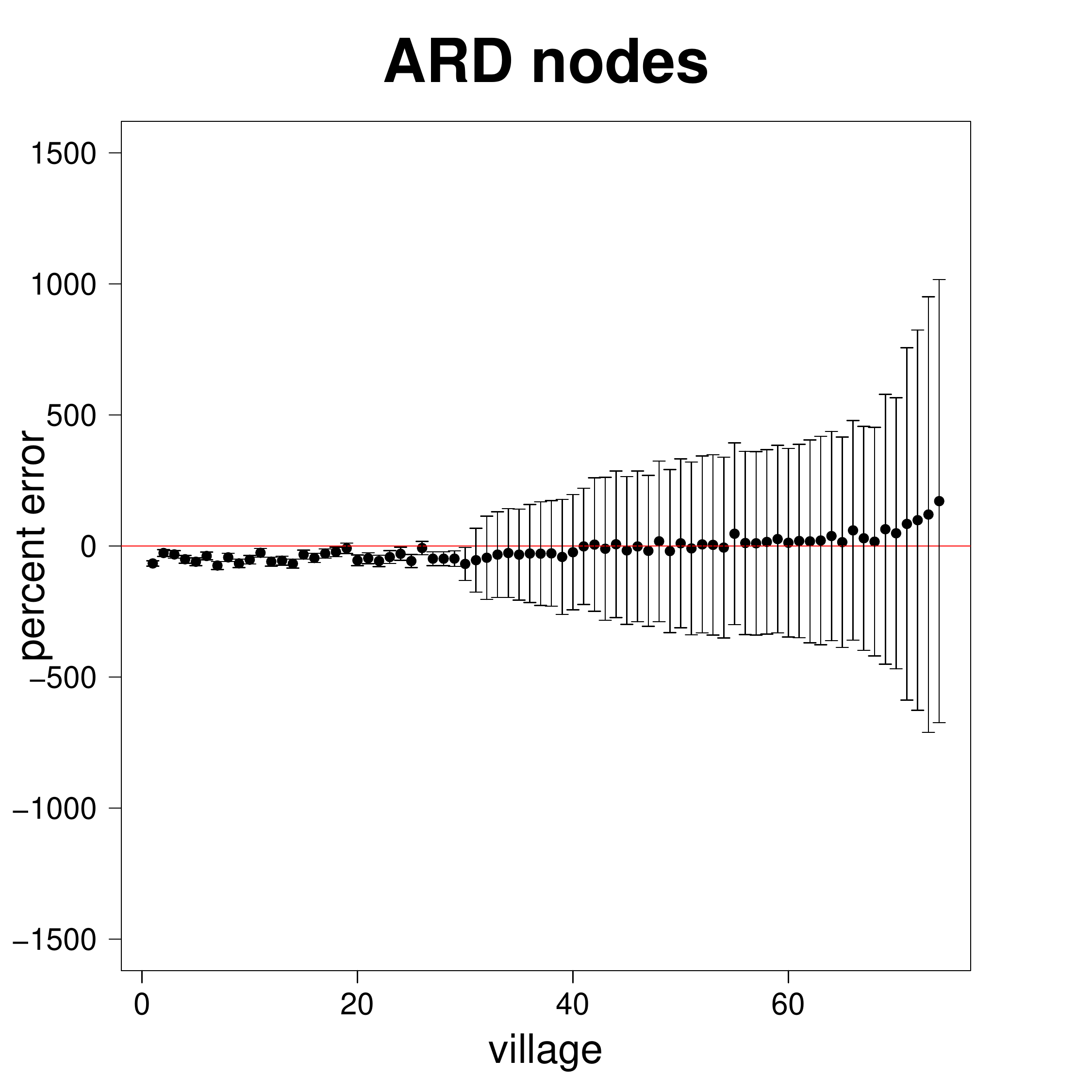}

}

\caption{Node level measures estimation for households in villages in Karnataka.  These plots show results using only nodes with ARD.  Plots in the top row show scatterplots across all villages with the estimated node level measure on the x-axis and the measure from the true underlying graph on the y-axis.  
The bottom row shows mean $ +/- $ standard deviation of percent errors of the estimated node level measure across all villages.  We see that overall there is strong correlation between the statistic on the underlying graph and the one estimated with ARD, with the exception of clustering. With clustering as a measure of triadic closure and the specified form of our generative model, it is not surprising that node level clustering estimation is a little weak.}
\label{fig:indian_ARD}
\end{figure}

Table \ref{table:heavytails_India_ARD}, Panel A presents a confusion matrix to look at the probability that a node picked by a researcher using ARD is in the top decile of the centrality distribution, which is a $47\%$ true positive rate. For comparison, this is a comparable rate to that in \cite{banerjeegossip} using the ``gossip survey'' technique to elicit nominations from the village as to who is central if the nominee is also a social or political leader in the village.

\begin{table}[!h]
\centering
\subfloat[ARD Nodes\label{tab:confARD}]{
\scalebox{0.7}{\begin{tabular}{lllll}
\multicolumn{2}{l}{\multirow{2}{*}{}}                                            & \multicolumn{2}{l}{Estimated top decile}                 & \multirow{2}{*}{}        \\ \cline{3-4}
\multicolumn{2}{l|}{}                                                             & \multicolumn{1}{l|}{Yes}   & \multicolumn{1}{l|}{No}     &                          \\ \cline{2-5} 
\multicolumn{1}{l|}{\multirow{2}{*}{True top decile}} & \multicolumn{1}{l|}{Yes} & \multicolumn{1}{l|}{234} & \multicolumn{1}{l|}{271}   & \multicolumn{1}{l|}{505}  \\ \cline{2-5} 
\multicolumn{1}{l|}{}                                 & \multicolumn{1}{l|}{No}  & \multicolumn{1}{l|}{271}  & \multicolumn{1}{l|}{4012} & \multicolumn{1}{l|}{4283} \\ \cline{2-5} 
\multicolumn{2}{l|}{}                                                             & \multicolumn{1}{l|}{505}    & \multicolumn{1}{l|}{4283}    & \multicolumn{1}{l|}{4788} \\ \cline{3-5} 
\vspace{2mm}
\end{tabular}}
}
\subfloat[All Nodes\label{tab:confAll}]{
\scalebox{0.7}{\begin{tabular}{lllll}
\multicolumn{2}{l}{\multirow{2}{*}{}}                                            & \multicolumn{2}{l}{Estimated top decile}                 & \multirow{2}{*}{}        \\ \cline{3-4}
\multicolumn{2}{l|}{}                                                             & \multicolumn{1}{l|}{Yes}   & \multicolumn{1}{l|}{No}     &                          \\ \cline{2-5} 
\multicolumn{1}{l|}{\multirow{2}{*}{True top decile}} & \multicolumn{1}{l|}{Yes} & \multicolumn{1}{l|}{470} & \multicolumn{1}{l|}{1167}   & \multicolumn{1}{l|}{1637}  \\ \cline{2-5} 
\multicolumn{1}{l|}{}                                 & \multicolumn{1}{l|}{No}  & \multicolumn{1}{l|}{1167}  & \multicolumn{1}{l|}{13262} & \multicolumn{1}{l|}{14429} \\ \cline{2-5} 
\multicolumn{2}{l|}{}                                                             & \multicolumn{1}{l|}{1637}    & \multicolumn{1}{l|}{14429}    & \multicolumn{1}{l|}{16066} \\ \cline{3-5} 
\vspace{2mm}
\end{tabular}}
}
\caption{Confusion matrix of top decile eigenvector centrality estimation for ARD nodes (Panel A) and all nodes (Panel B)}
\label{table:heavytails_India_ARD}

\end{table}
Figure \ref{fig:indian_Full} repeats the above results for the entire sample. The results are largely similar to the ARD sample alone, though clearly there is more noise, as expected, when including the non-ARD sample.

\begin{figure}[!h]
\centering
\subfloat[Degree]{
\includegraphics[width=.3\textwidth]{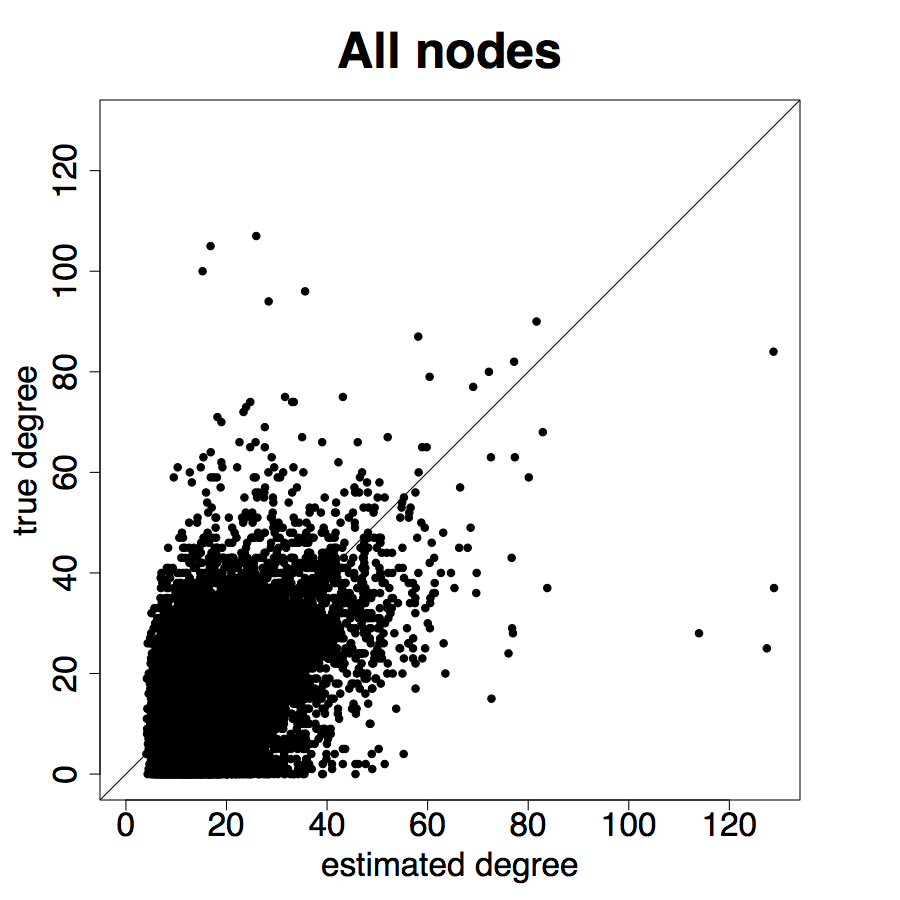}}
\subfloat[Eigenvector Centrality]{
\includegraphics[width=.3\textwidth]{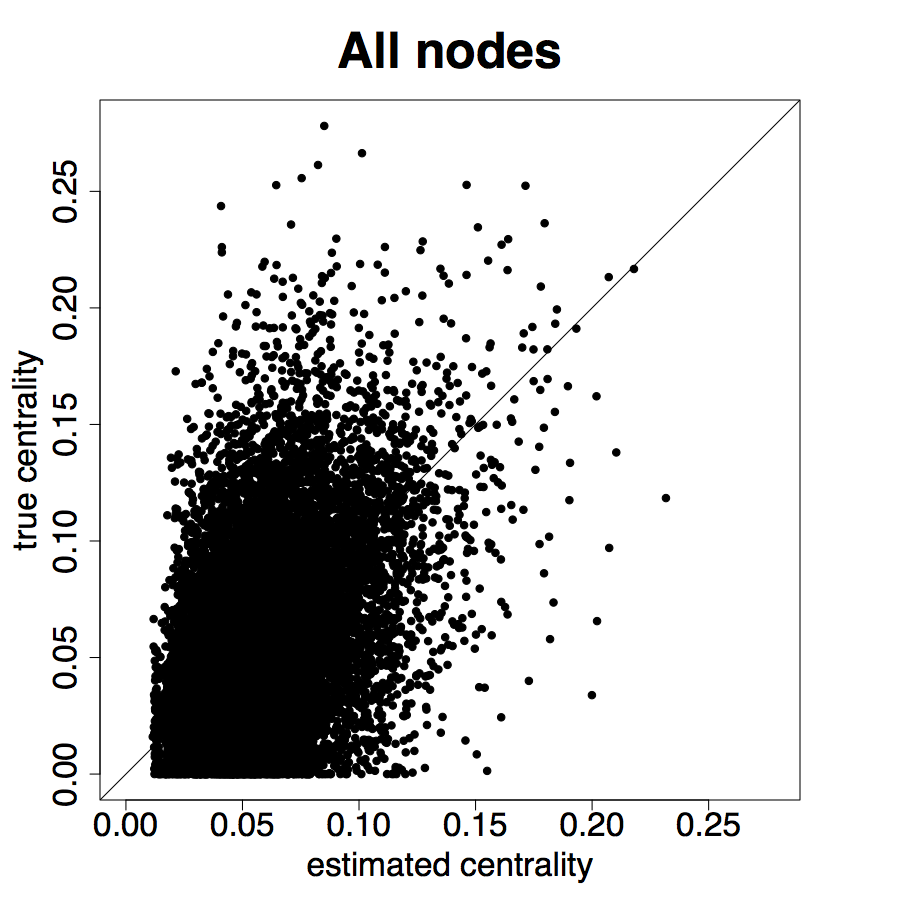}
}
\subfloat[Clustering]{
\includegraphics[width=.3\textwidth]{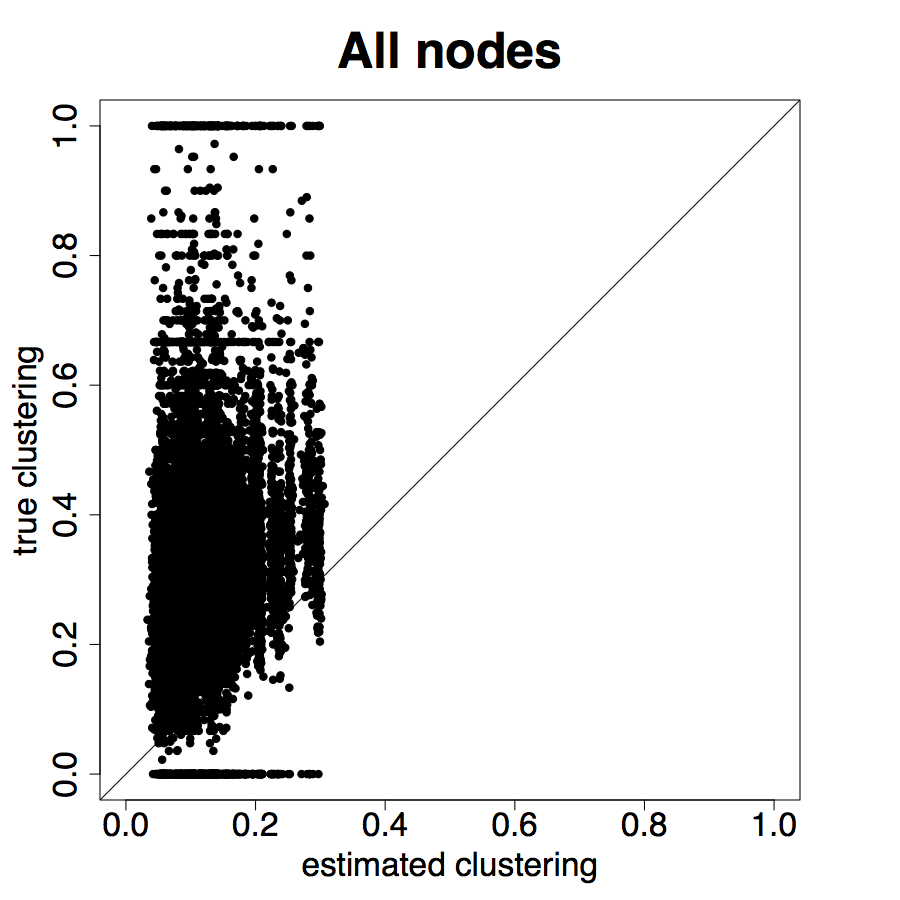}
}

\medskip

\subfloat[Degree]{
\includegraphics[width=.3\textwidth]{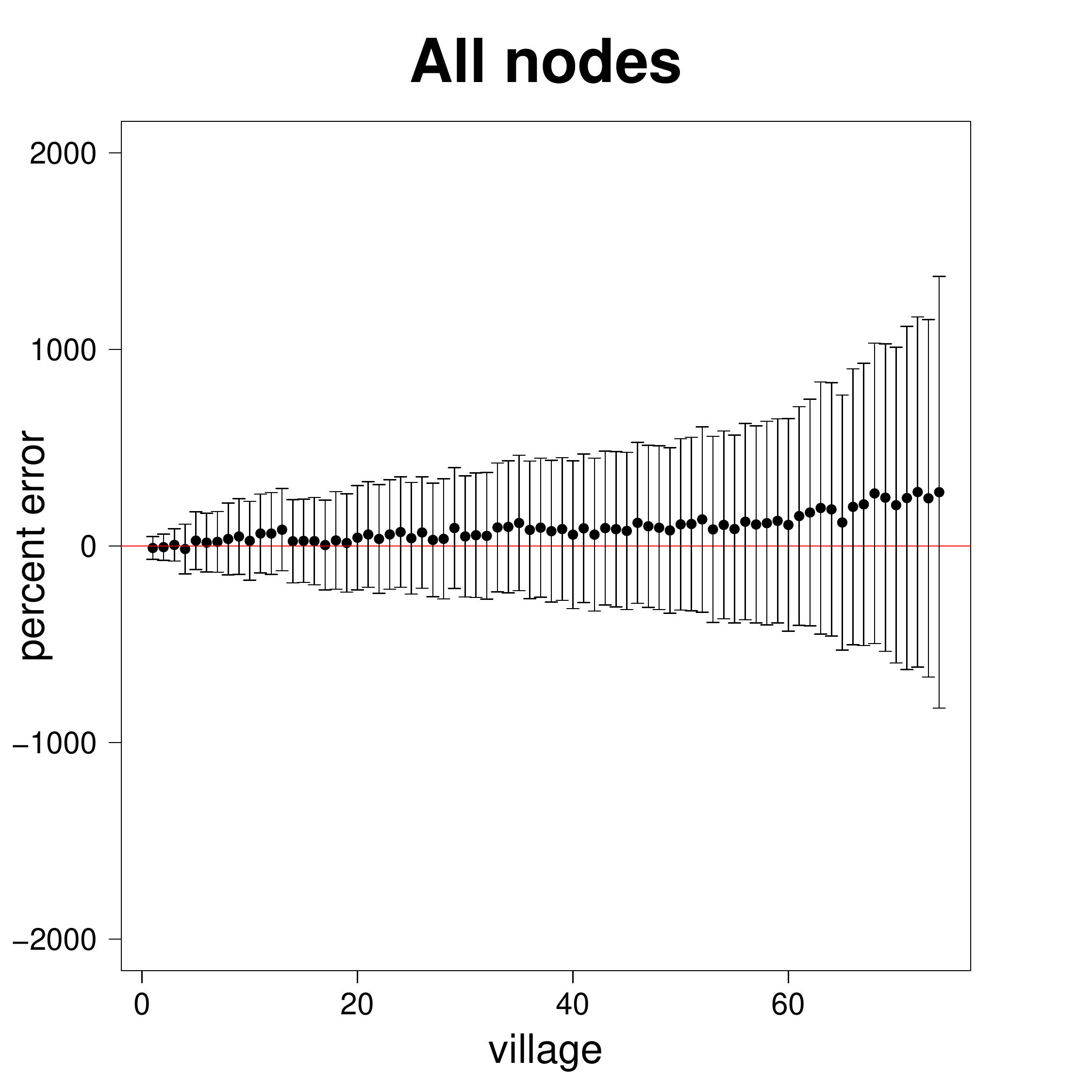}
}
\subfloat[Eigenvector Centrality]{
\includegraphics[width=.3\textwidth]{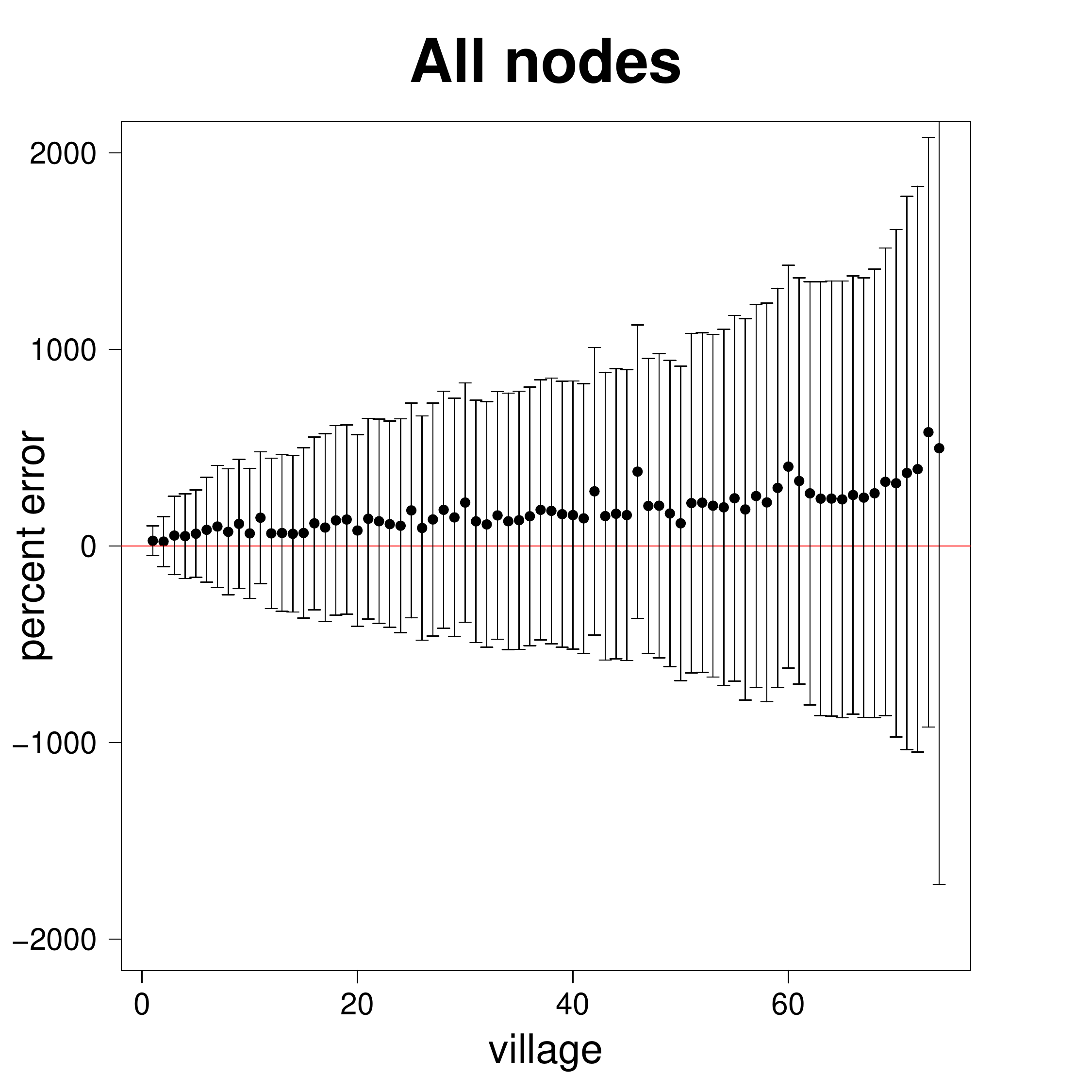}
}
\subfloat[Clustering]{
\includegraphics[width=.3\textwidth]{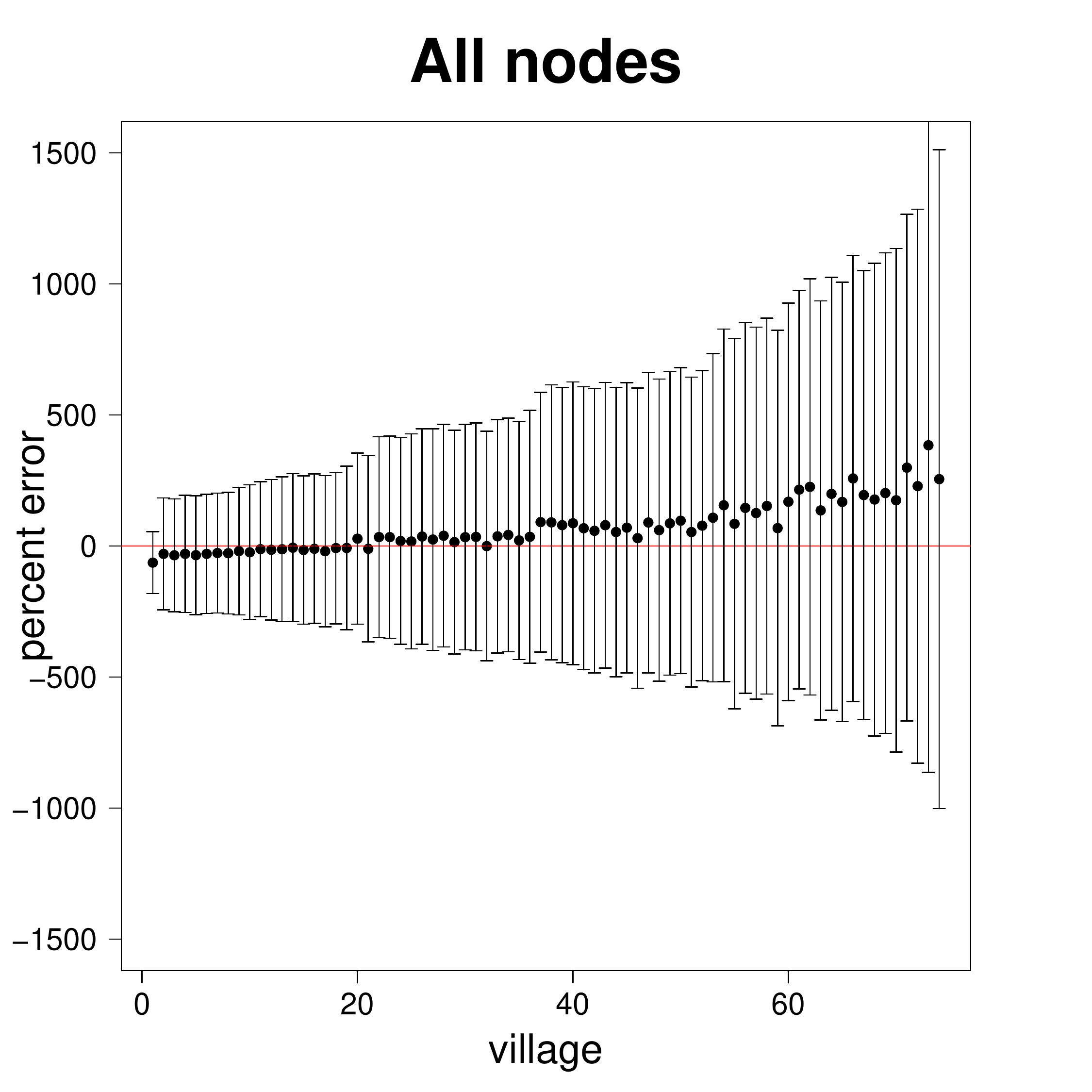}
}

\caption{Node level measures estimation for households in villages in Karnataka.  These plots show results using all nodes. Plots in the top row show scatterplots across all villages with the estimated node level measure on the x-axis and the measure from the true underlying graph on the y-axis.  The bottom row shows mean $ +/- $ standard deviation of percent errors of the estimated node level measure across all villages.  We see that overall there is weak correlation between the statistic on the underlying graph and the one estimated with ARD. The weak correlation for non-ARD nodes comes from  the noisy mapping from demographic covariates to $\nu$ and $z_i$.}
\label{fig:indian_Full}
\end{figure}

Table \ref{table:heavytails_India_ARD} Panel B presents the confusion matrix for the entire sample, with a $29\%$ true positive rate. We have a $16\%$ true positive rate even when we pick top decile centrality nodes from non-ARD sample. For context, this is about as high as a non-nominated leader \citep{banerjeegossip}, whom a microfinance institution might specifically pick to diffuse information widely.

\subsection{Discussion}

Taken together, our results suggest that ARD with the latent distance model and the procedure proposed here is a useful tool because the researcher will have reasonable estimates of a number of network features. As is unsurprising for a model of the form specified here, it is a little bit weak when it comes to clustering.

\section{Empirical Applications}\label{sec:application}

We now present two empirical applications that use ARD techniques. They build upon prior work by the authors, in part. The goal is to illustrate here that a researcher could have done this sort of economic analysis using ARD only, equipped with our method.

The first example looks at what would have happened if the researchers had obtained ARD for an experiment on savings and reputation. The second example actually looks at a setting where survey ARD was collected.

\subsection{Encouraging savings behavior in rural Karnataka}

\

Our first application builds on \cite{SavingsMonitors}. The
authors study social reputation through the lens of savings. In a
field experiment, savers set 6-month targets for themselves. They
do so knowing they may be assigned a \textquotedblleft monitor,\textquotedblright{}
a villager who will be notified biweekly about their progress. Progressing
towards a self-set target exhibits more responsibility, providing
an avenue for the saver to build reputation with the monitor and others
in the community. In 30 villages, monitors are randomly assigned to
a subset of savers. This generates variation in the position of the
monitor in the network. Because the monitor is free to talk to others,
information about the saver\textquoteright s progress and reputation
may spread. A signaling model on a network guides the analysis: if
the saver is more central, information can spread more widely, and
if the saver is more proximate to the monitor, information likely
spreads to those with whom the saver is more likely to interact in
the future. For saver $i$ and monitor $j$, the model shows that the network matters for signaling through the quantity\footnote{Formally, \cite{SavingsMonitors} show \[
q_{ij} = \frac{1}{n}\sum_k p_{jk} \sum_k p_{ik}+ n \cdot \cov(p_{\cdot i}. p_{\cdot j})
\]
Here $p_{ij} \propto \left[\sum_{t=1}^T (\theta g)^t\right]$ is the probability that a unit of information that begins with $i$ is sent to $j$, where transmission across each link happens with probability $\theta$. \cite{banerjeegossip} shows that for sufficiently high $T$, $\sum_k p_{jk}$ converges to the eigenvector centrality of $j$. \cite{SavingsMonitors} shows that in equilibrium, only when $q_{ij}$ is sufficiently high does the saver actually save.
}
\[
q_{ij} = \frac{1}{n}\text{Monitor Centrality}\times \text{Saver Centrality} + n \cdot \text{Proximity of Saver-Monitor}.
\]

\cite{SavingsMonitors} have near-full network data (from the \cite{banerjeegossip} sample), allowing them to calculate $q_{i,j}$. They find that randomly-selected monitors increase household savings across all accounts by 35\%. Consistent with the model, a one-standard deviation increase in $q_{ij}$ 
leads to an additional 29.6\% increase in total savings. Additionally, 15 months after the end of our savings period, they show that reputational information
spread: randomly selected individuals surveyed about savers in the
study were more likely to have updated correctly about a saver\textquoteright s
responsibility when the saver was randomly assigned a more central
monitor. Moreover, the savings increase persisted, and in the intervening 15 months, monitored savers were better able to cope with shocks. 

\begin{table}[ht!]
\centering
\caption{Log total savings across all household accounts regressed on monitor signaling value\label{tab:SM_rev}}
\scalebox{0.65}{\begin{threeparttable} 
\begin{tabular}{lcc} \hline
 & (1) & (2) \\
 & Log Total Ending Savings & Log Total Ending Savings \\ \hline
 &  &  \\
Signaling value of monitor with full network data ($q_{ij}$), Standardized & 0.248 &  \\
 & (0.0931) &  \\
Predicted signaling value of monitor with ARD ($q_{ij}$), Standardized &  & 0.181 \\
 &  & (0.0888) \\
 &  &  \\
Observations & 422 & 422 \\
 Number of villages & 30 & 30 \\ \hline
\end{tabular}

\begin{tablenotes} Notes: Standard deviation of village-level block bootstrap in parentheses.
\end{tablenotes} 
\end{threeparttable}} 
\end{table}

How would our conclusions have changed if \cite{SavingsMonitors} only had access to ARD and not the full network maps? Table \ref{tab:SM_rev} presents regressions of the log of total household savings across all household accounts against the model-based measure of how much signaling value the monitor provides the saver, $q_{ij}$. We construct ARD estimates by taking samples from the posterior distribution and then using the average estimated $q_{ij}$ across those posterior draws.  In the experiment we showed that a one standard deviation increase in $q_{ij}$ due to random assignment of the monitor led to a 24.8\% increase in total household savings (column 1). In column 2 we show that even if we did not have the network data, if we had ARD alone for a 30\% sample, we would have had a very similar conclusion, inferring that a one standard deviation increase in predicted $q_{ij}$ corresponds to a 18.1\% increase in total household savings across all accounts. Said differently, we could have used ARD questions to easily pick good monitor-saver pairs.  

\begin{table}[ht!]
\centering
\caption{Beliefs about savers and monitor centrality\label{tab:SM_rev_rep}}
\scalebox{0.65}{\begin{threeparttable} 
\begin{tabular}{lcc} \hline
 & (1) & (2) \\
 & Belief about saver's responsibility & Belief about saver's responsibility \\ \hline
 &  &  \\
Monitor centrality with full network data, Standardized & 0.0500 &  \\
 & (0.0142) &  \\
Predicted monitor centrality with ARD, Standardized &  & 0.0340 \\
 &  & (0.0161) \\
 &  &  \\
Observations & 4,743 & 4,743 \\
 Number of villages & 30 & 30 \\ \hline
\end{tabular}

\begin{tablenotes} Notes: Standard deviation of village-level block bootstrap in parentheses. ``Responsibility'' is constructed as 1(Saver reached goal)*1(Respondent indicates saver is good or very good at meeting goals) + (1-1(Saver reached goal))*1(Respondent indicates saver is mediocre, bad or very bad at meeting goals).  See~\cite{SavingsMonitors} for further details.
\end{tablenotes} 
\end{threeparttable}} 
\end{table}
As a further examination of our approach, we repeat the same exercise using another specification from~\cite{SavingsMonitors}.  Table~\ref{tab:SM_rev_rep} shows the results of a regression where we the outcome is the respondent's belief about the saver's responsibility and the regressor is the monitor's centrality.  Observing the complete network, a unit increase in the monitor's centrality corresponds to about a 5\% increase respondent's belief about saver responsibility.  Using ARD, we would estimate an increase of about 3.4\%, leading (as in the previous example) to the same substantive conclusions.

This application also gives us an opportunity visualize how network characteristics map to the latent space representation. In Figure \ref{fig:SM_pics}, we plot the locations and concentrations of the ARD traits for four sample villages that were part of the \cite{SavingsMonitors} savings study. We then overlay the positions in the latent space of the individuals participating in the experiment as monitors, depicted as rings.  The size of the ring depicts the monitor's eigenvector centrality.  Finally, we color the monitor rings to indicate the savings performance of the saver to whom each monitor was randomly allocated -- darker shades depict higher levels of savings.

\begin{figure}[!h]
\centering
\subfloat[Village 13]{
\includegraphics[width=.45\textwidth]{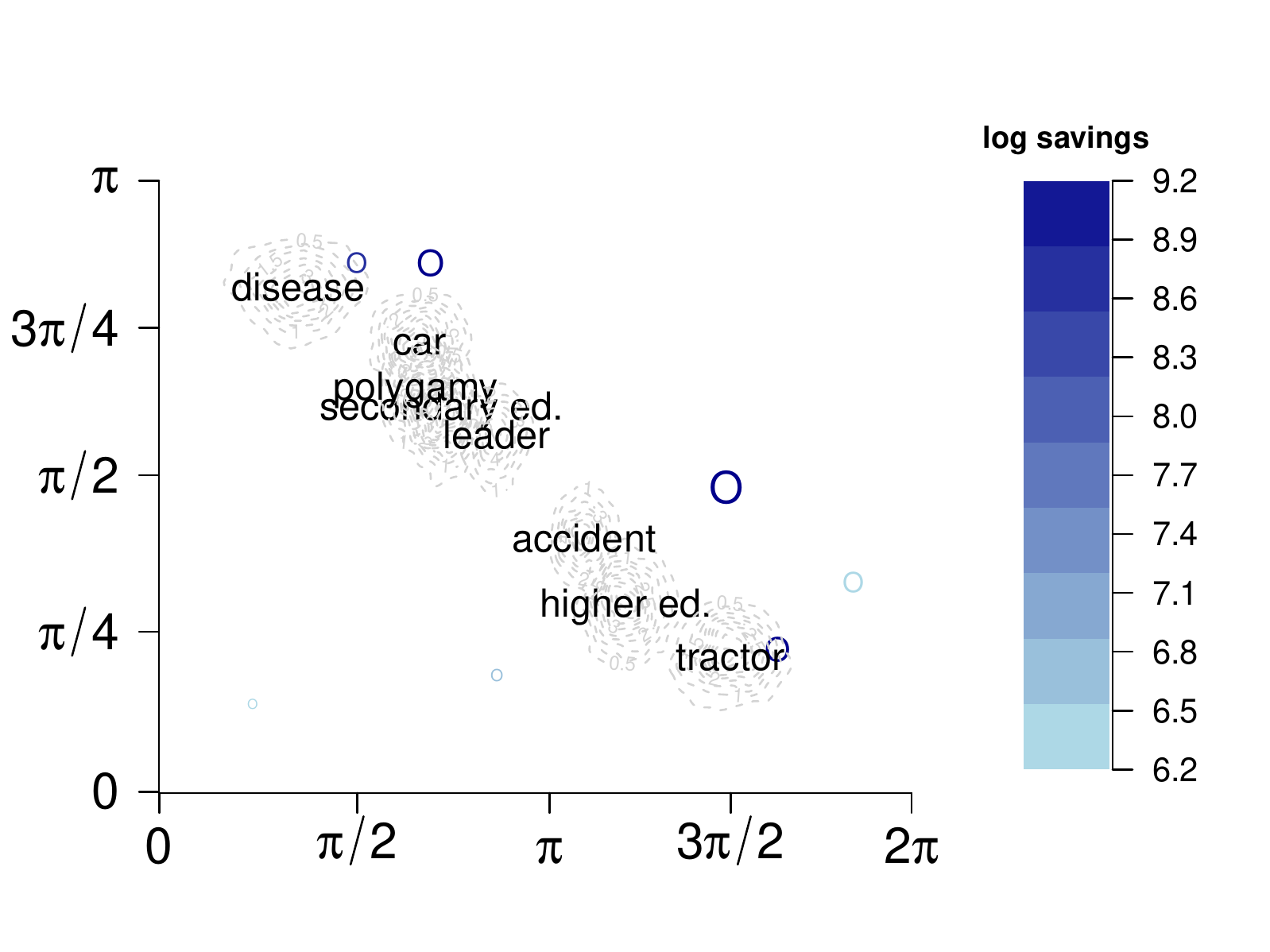}}
\subfloat[Village 22]{
\includegraphics[width=.45\textwidth]{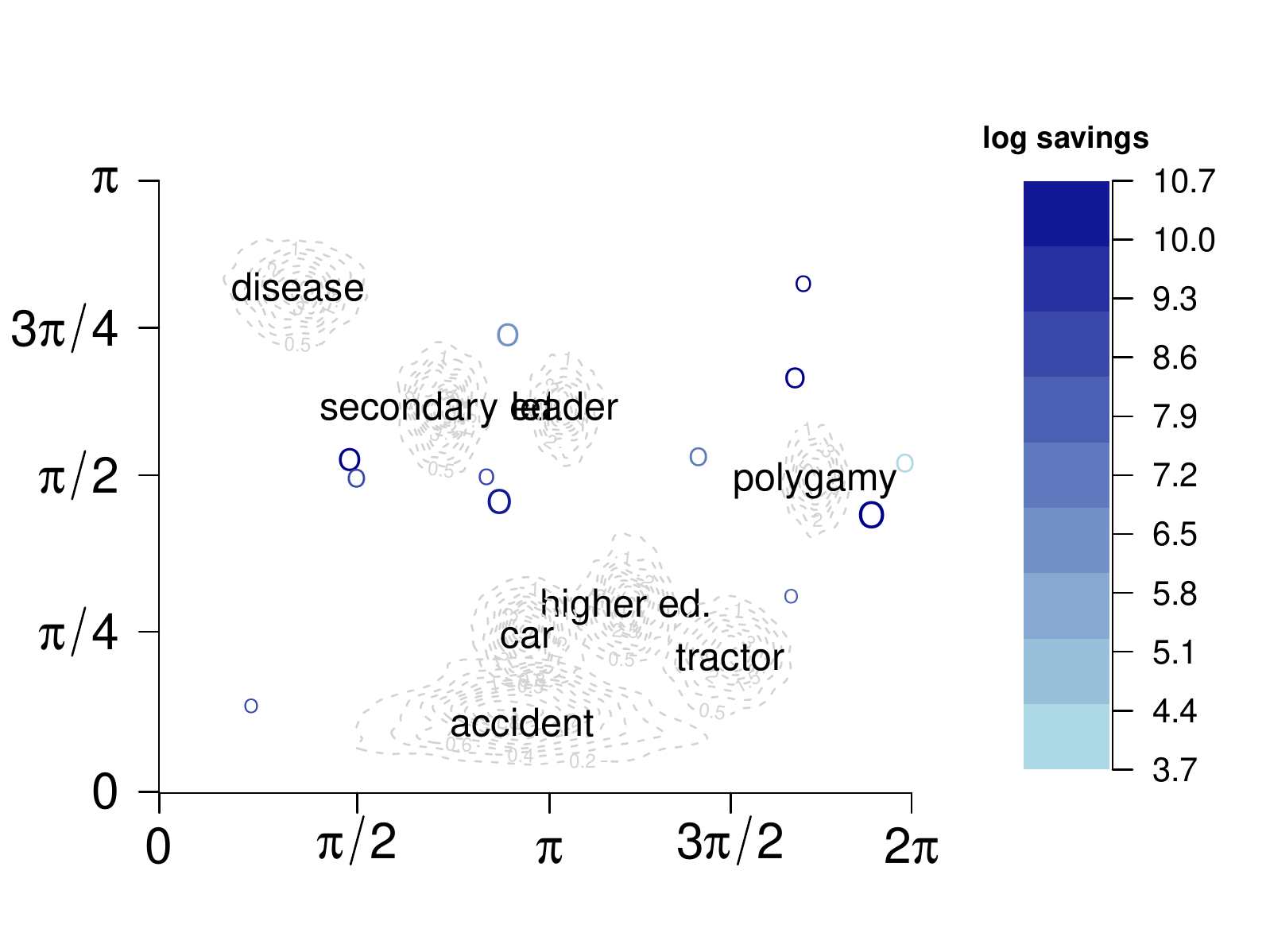}
}
\smallskip

\subfloat[Village 25]{
\includegraphics[width=.45\textwidth]{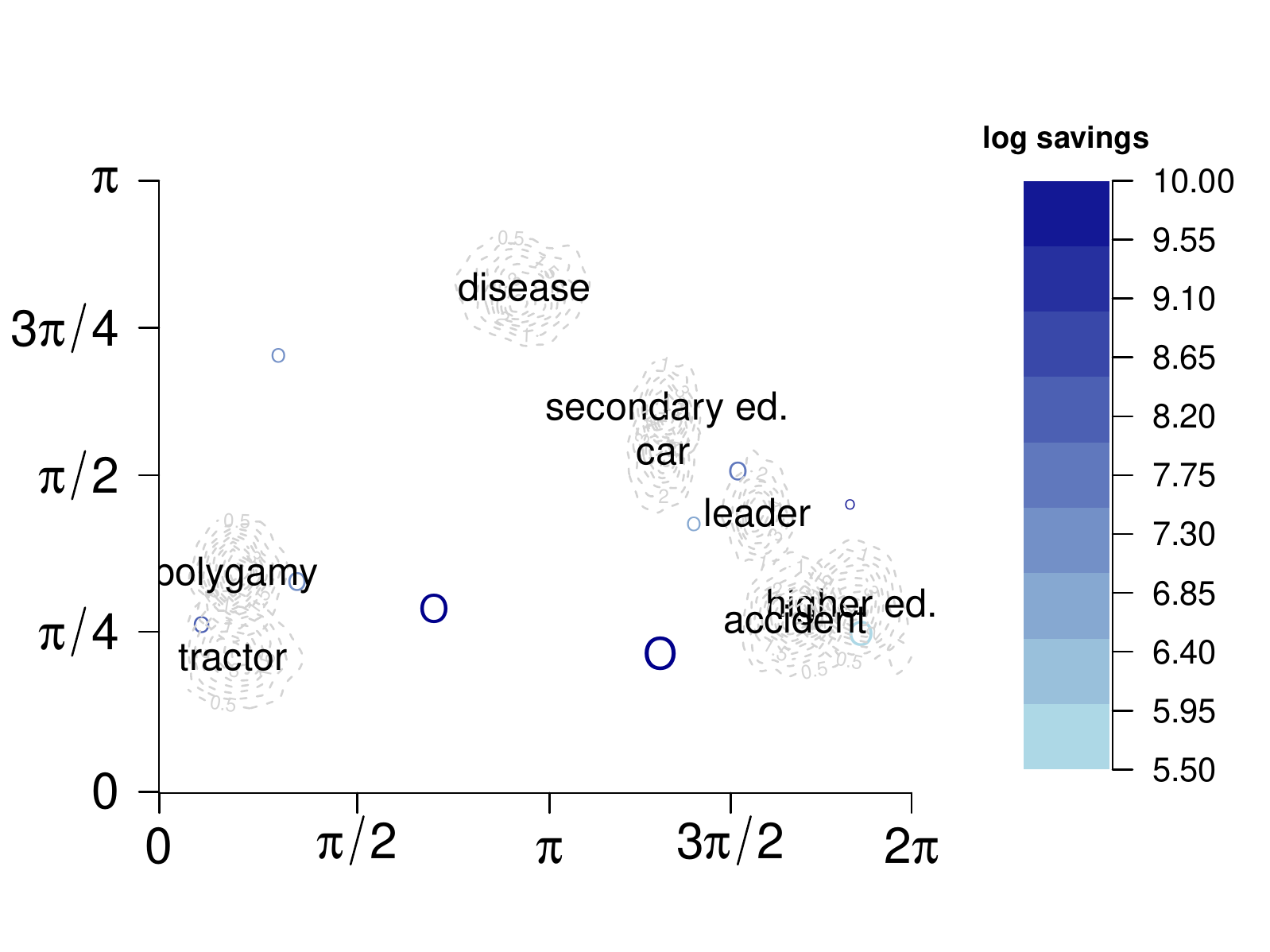}
}
\subfloat[Village 54]{
\includegraphics[width=.45\textwidth]{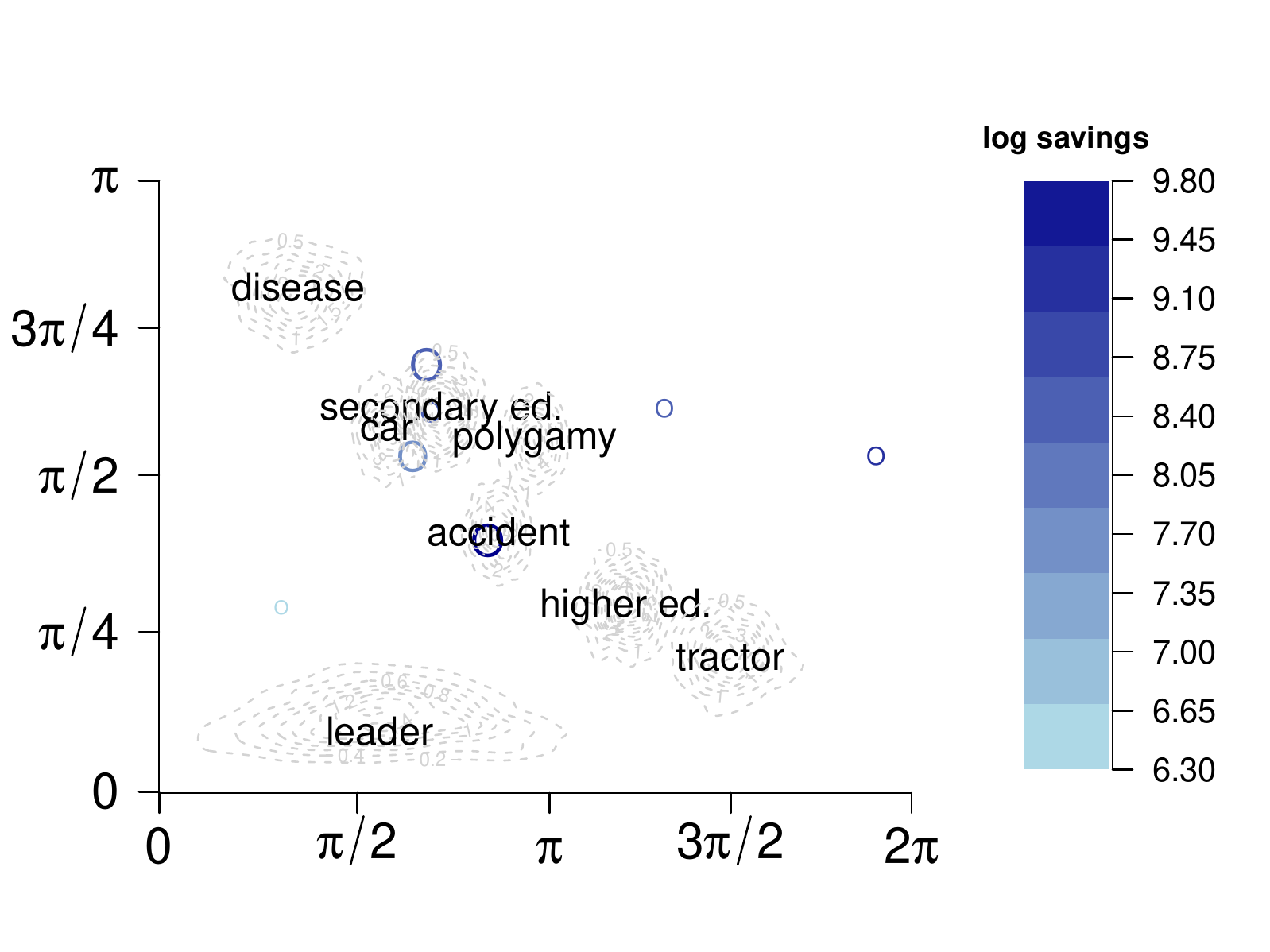}
}

\caption{Sample latent locations of randomly assigned monitors by centrality and the savings of the their respective savers. This illustrates the pattern that more central monitors corresponded to higher levels of savings.
}
\label{fig:SM_pics}
\end{figure}

As \cite{SavingsMonitors} find, there appears to be a relationship between monitor centrality (here denoted by larger rings) and the saver's performance (here given by darker colors). This is consistent with the theory that more central monitors under the signaling model generate larger incentives for the saver to save. Furthermore, the visualization demonstrates that the larger rings tend to be located closer to the centers of traits or between centers of traits. That is, they are closer to the center of masses of clusters of types of individuals. This makes sense as this means that the latent location of a central monitor will tend to be closer to many more other individuals, ceteris paribus.

\subsection{Impact of microfinance in Hyderabad}\label{sec:hyderabad}

\

The goal of our final example is to demonstrate to the reader a context in which we collected and use only ARD survey questions in our analysis. We first demonstrate that the researcher could have obtained the same conclusions using the ARD instead of the network data that was collected in this study. But because the network data was incomplete (specifically the authors only measured degree -- the number of links but not the identities -- and support -- how many links had a friend in common), the researchers could not ask how their intervention impacted the network more generally. Using ARD techniques, we show what conclusions the researchers could have learned about how the network was affected by the intervention only using the ARD survey data and estimates from the surveys of each neighborhood's average degree.

This example concerns the introduction of microfinance in Hyderabad, India. A recent literature has examined the effects that introducing microfinance to previously unbanked communities can have ambiguous and heterogenous effects on the underlying social and economic networks that facilitate informal risk-sharing. On the one hand, as in \cite{feigenberg2013economic}, links may be built between microfinance members and there may be an increased incentive to build links to relend \citep{kinnan2012kinship}. On the other hand, the fact that individuals who have now become banked have less of a need to rely on informal insurance may nudge them to break links with others, and this can have local or even general equilibrium effects on the network, which can reduce density and increase paths among all nodes \citep{banerjeecdj2016change}.

In \cite{banerjee2015miracle}, the authors study a randomized controlled trial where microfinance was introduced randomly
to 52 out of 104 neighborhoods in Hyderabad. \cite*{banerjeebdk2016} look at longer run outcomes 6 years after the intervention. This example is useful for two reasons. First, it is an urban setting where the researchers have no hope of obtaining full network data.\footnote{We thank an anonymous referee for noting that we could also tweak our surveys in urban settings to measure ARD responses separately within the respondent's own neighborhood and also across neighborhoods.  While mapping an entire urban space likely requires an infeasible number of surveys, putting some structure on relationships within and across neighborhoods might allow for better urban network maps. We leave such an application to future work.}
Second, it shows how we may measure the effect of economic interventions
on social network structure, as predicted by theory, despite not having
network data. 

In the original paper, \cite{banerjeebdk2016} measure each node's within-neighborhood degree and support, defined as the fraction of links between the respondent and a connection such that there exists a third person who is linked to both nodes in the pair. They find that both degree and support decrease with the treatment. Note that they did not get any subgraph data since the links were not matched to a household listing: degree and support can be thought of as just two numbers. 

\cite{banerjeebdk2016} also collected ARD data, which we use here. In particular, a sample
of approximately 55 nodes in every neighborhood was surveyed and demographic covariates as well as ARD were collected for this entire sample. As before, we fit a network formation model using the ARD data and this sample of nodes.\footnote{In this application we use the survey responses for degree and input each graph's estimated average degree directly into the model.} A complete list of ARD questions used in this survey is in Appendix~\ref{sec:ardqlist}. 

We explore whether microfinance affects network structure by regressing

\[
y_{v}\left(g\right)=\alpha+\beta\text{Treatment}_{v}+\epsilon_{v}
\]
where $v$ indexes neighborhood and $\text{Treatment}_{v}$ is a dummy for treatment neighborhoods. Our outcome variable $y_{v}\left(g\right)$ of interest is the rate of support.

Theory is silent on whether density should increase or reduce,
whether triadic closure (clustering or support) should increase or
reduce, which can depend on a number of things: for instance, whether relending or autarky forces affect the incentives to maintain risk-sharing links \citep{jacksonrt2012}.

\

\begin{table}[!h]
\centering
\caption{Network statistics regressed on treatment \label{tab:Tables20}}
\scalebox{0.7}{\begin{threeparttable} 
\begin{tabular}{lccc} \hline
 & (1) & (2) & (3) \\
 & Percent Supported (Data) & Percent Supported (Estimate) & Graph-level Proximity (Estimate) \\ \hline
 &  &  &  \\
Treatment Neighborhood & -0.0655 & -0.0892 & -0.0463 \\
 & (0.0318) & (0.0532) & (0.0144) \\
 &  &  &  \\
Constant & 0.4427 & 0.4404 & 0.4485 \\
 & (0.0644) & (0.0935) & (0.0096) \\
 &  &  &  \\
 Mean of the response variable&  0.3880&0.3129&0.4238\\
  &  &  &  \\
Observations & 3,514 & 3,598 & 62 \\ \hline
\end{tabular}

\begin{tablenotes} Notes: Standard deviation of village-level block bootstrap in parentheses. Sample includes neighborhoods with estimated sampling rate $\geq 20 \%$. For large number of excluded low sampling rate neighborhoods, the population count is top-coded at 500 households. For these very large neighborhoods, we calculate the sampling rate using a population of 500. The outcome variable of columns 1 and 2 is the share of links that are supported and in column 3 it is the average proximity in the graph.
\end{tablenotes} 
\end{threeparttable}} 
\end{table}

Table \ref{tab:Tables20} reports the regression results.  Column 1 replicates the specification from \cite{banerjeebdk2016} that past exposure decreased support.  Column 2 presents the same regression, but using estimated support.  The estimates of the treatment effects along with the levels of support (the regression constant) are quite similar.  We view this exercise as a ``validation'' of the ARD-based model.  Given that the estimated treatment effect looks quite similar using the different support measures, in Column 3, we present the results of a graph-level regression, using proximity (the average inverse path length in the network) as the outcome variable.  Note that it was not possible for the authors to collect such a statistic using their surveys.  We find that estimated proximity decreases, meaning that the decline in links due to microfinance exposure lead to larger average distances between households in the community. This exercise demonstrates how our method may be useful to researchers seeking to study the evolution of networks, without requiring full network data.

\

\section{Cost Savings Using ARD}\label{sec:cost}

We have demonstrated that our approach for estimating network statistics has the potential to serve as a replacement for the collection of full network data. Namely, we show above that we can replicate the findings of \cite{SavingsMonitors} and \cite{banerjeebdk2016} with our ARD-based estimates alone. While it is always preferable to collect the underlying graph data, one important benefit from ARD is that it is substantially easier and cheaper to collect. 

Table \ref{tab:CostTab} presents a comparison of the costs associated with a full network survey with those of an ARD exercise for a target sample of 120 villages. Panel A summarizes the major differences in the budget assumptions between the two methods. We assume that a census is conducted in both methodologies, though household members need only be enumerated in the full network surveys.  We also assume that the full network data is collected from 100\% of households, while the ARD protocol samples from 30\% of households. Importantly, the ARD method does not require the time consuming matching of a household's reported links with the enumerated census. Given these assumptions, Panel B of Table \ref{tab:CostTab} shows that ARD is substantially cheaper, costing approximately 80\% less than the full network surveys.

\begin{table}[!h]
\centering
\caption{Cost Comparison: Full Network vs. ARD Surveys \label{tab:CostTab}}

\begin{threeparttable}
\includegraphics[scale = 0.75]{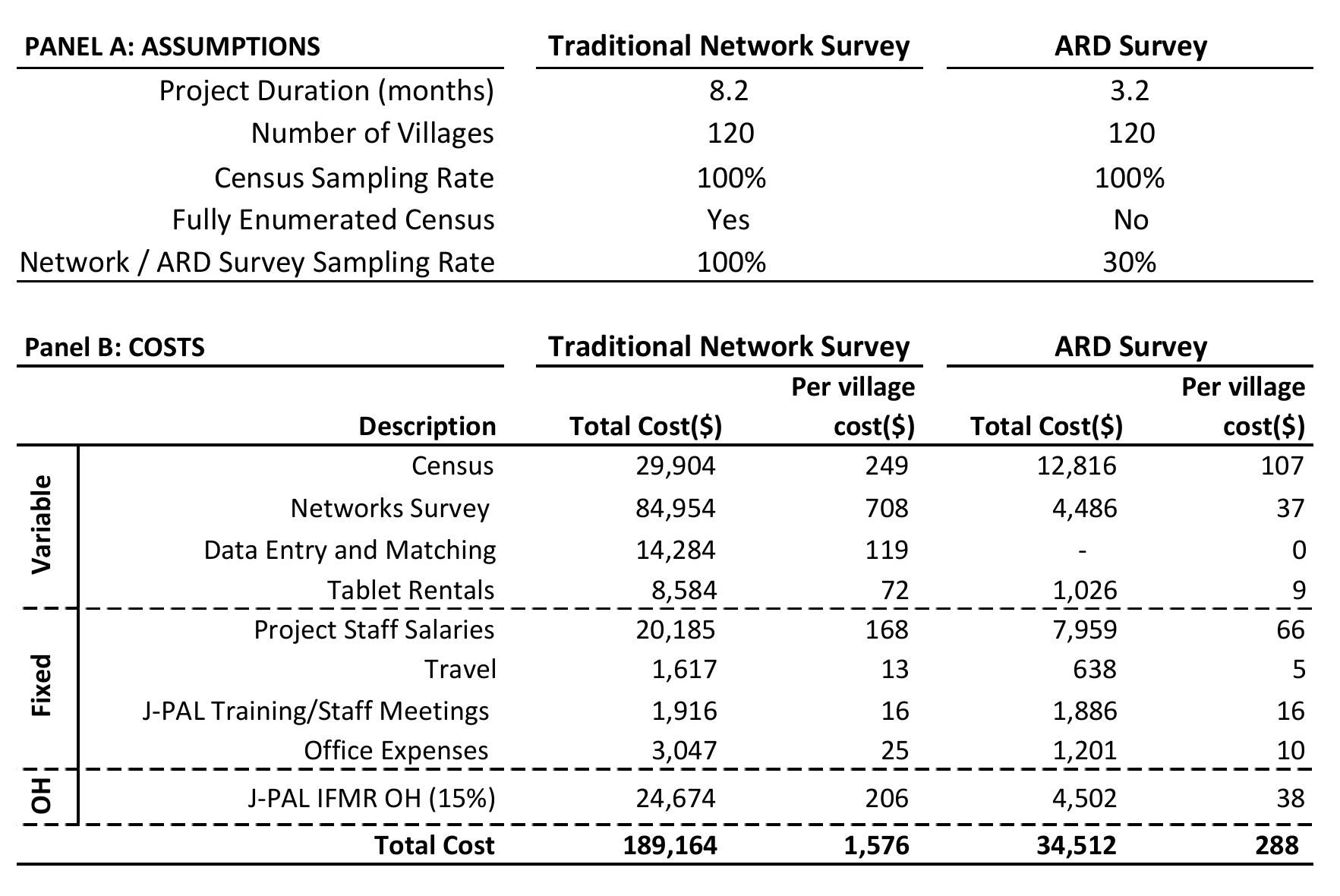}
\begin{tablenotes} Notes: This cost comparison was prepared by J-PAL South Asia, the organization that implemented the network surveys for \cite{banerjeecdj2013} in Karnataka, India. 
\end{tablenotes} 
\end{threeparttable}
\end{table}

In Figure \ref{fig:CostFig}, we show that these dramatic cost reductions are not only a bi-product of the 30\% sampling rate assumption. Even with 100\% sampling, ARD surveys are still over 70\% cheaper than the full network alternative. This sample budget highlights that using ARD estimates could indeed expand the feasibility of empirical network research.

\begin{figure}[!h]
\centering
\includegraphics[scale = 0.4]{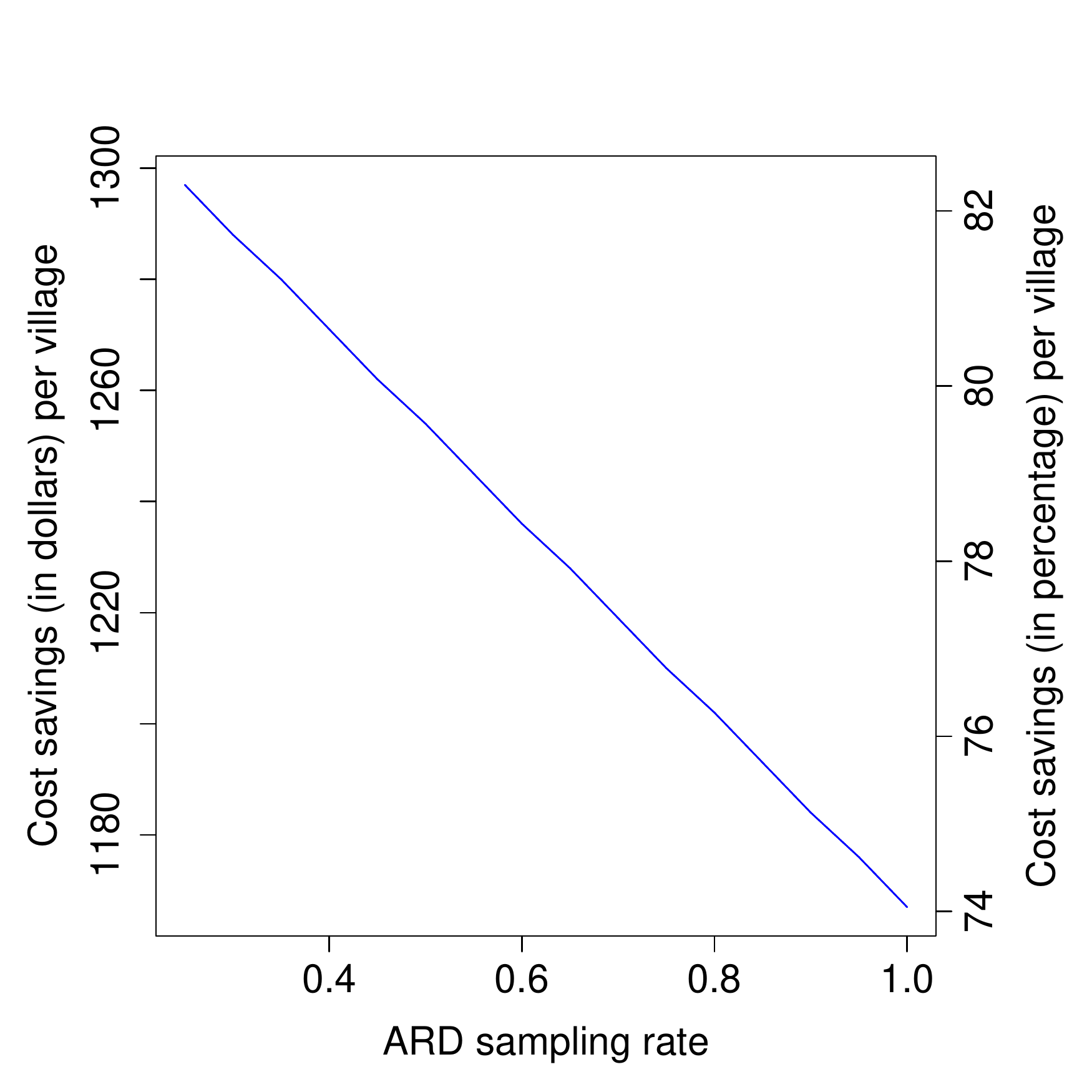}

\caption{Cost Savings of ARD vs. Full Network Surveys by ARD Sampling Rate}
\label{fig:CostFig}
\end{figure}

It should go without saying that should a researcher be able to afford it, full network data is the gold-standard, and even partial network data could help being used in conjunction with ARD. The findings of this paper suggests that the \cite{hoff2008modeling} model is good enough at capturing relevant features of the network. Therefore, while the network formation model can be estimated using ARD, certainly having more information about a subgraph will aid the researcher in both estimating the network formation model and integrating over the missing data in order to recover features of interest to the researcher as argued in  \cite{chandrasekharl2016}.

\section{Conclusion}\label{sec:conclusion}

We have shown that by adding a very simple set of questions to standard survey instruments, researchers and policymakers can retrieve powerful information about the underlying social network structure. This information is easy to obtain in standard instruments and therefore can be employed in a cost-effective way.

There is a prior literature as to whether a researcher could simply ask individuals from the network. For instance, \cite{banerjeegossip} shows that simply asking "gossip" questions can be used to identify eigenvector central individuals. However, there are no results for other features such as  such as clustering, path length, cut in the network, and so on.\footnote{Note that part of the insight in \cite{banerjeegossip} was to realize that eigenvector seems complicated but if you know who you hear gossip about frequently, this mechanically corresponds to central individuals. This is a unique trait for  centrality, not all statistics.} Further, we have reason to believe this sort of procedure likely would not work for other network features. For instance \cite{friedkin1983}, \cite{krackhardt1987,krackhardt2014}, among others in sociology, and also our own work in \citet*{bct}, all document such biases. They show that network knowledge decays in distance, that degrees are systematically misestimated and that individuals are more likely to think their friends are friends, among other things.

We suggest a simple blueprint for researchers and policymakers in the field to obtain network data. If possible, researchers should add five to ten ARD questions to the census as a standard demographic variable that would be recorded just like geographic data. If not, then researchers should at least ask ARD questions for a sample of respondents. We discuss how one might collect ARD data for use in our model in Appendix \ref{sec:blueprint}.

There are several avenues for future research. The first would involve optimizing and standardizing ARD question design. What sorts of ARD questions should be asked? What would provide the most information to make better inferences about network structure? This has been in part the subject of work by, for example, \cite{feehan2016quantity} in the sociology and epidemiology literatures. Another avenue for future work builds upon the recent interest in trying to control for unobservables that both drive network structure and outcome variables of interest, the ARD approach might allow us to identify and control for latent variables.  Yet another direction would provide guidelines for picking the dimension of the latent space.  In particular, we could use fraction of overlap between traits to restrict the set of feasible latent dimensions.\footnote{To see the intuition for this, consider the case where there are three groups A, B, and C.  Each of these groups would need to be placed on a sphere in such a way as to reflect the overlaps between individuals in one or more of the groups (a person who is a member of A and B should go in the disc of both groups, for example.  The configuration implied by these overlaps may not be possible in all dimensions.~\citet{fosdick2019multiresolution} point out a similar restriction arising because of the triangle inequality for latent spaces on the plane.}

A final avenue for future research involves looking beyond the survey network setting. Predominantly, the literature on ARD has been focused on surveyed social networks. However, we note here that our entire framework readily extends to any network context where the researchers naturally have aggregated data about links between nodes and categories of other nodes. To see this, consider the two most common economic network applications outside of social networks: inter-sectoral linkages \citep{acemoglu2012network,barrot2016input,carvalho2016supply} and banking \citep{acemoglu2015systemic,elliott2014financial,gandy2016bayesian, gandy2017adjustable,upper2004estimating}. 

Let us consider the simple example of a dataset where the researcher has a sample of firms and input-output data. So the researcher sees a collection of firms and then transactions the firm has with other (sub-)sectors. One can reinterpret this as simply
``How many links does the firm have to firms with trait $k$?"
where many links will now just be a weighted (by, for example the volume of trade) conditional degree instead of a conditional degree and trait $k$ is just (sub-)sector $k$. This is just ARD for a weighted and directed graph.\footnote{The model presented above in the paper is for cases when the underlying network is unweighted (binary) and undirected.  The formation model we use is un-normalized, however, making the extension to the weighted case straightforward.  One could extend the method to address directed graphs by introducing an asymmetric distance measure as suggested in, for example, \cite{hoffrh2002}.}

What this immediately implies is that questions of interest such as whether firm-level shocks propagate or get absorbed in their production networks (e.g., \cite{barrot2016input}) or whether if theory suggest that certain supply chains should be more robust than others to shocks, could be probed even with limited ARD data, using the techniques developed in this paper. There is nothing specific to survey network data in our statistical framework, rather it applies more broadly to any context where there are measurements of aggregate interactions between connected units.

Similarly, if we consider a dataset where the researcher sees aggregated data from bank loans, where the bilateral inter-bank loan is unavailable, but aggregated loans are (e.g., by type of bank), the methodology applies once again. Thus, our technique suggests an avenue for regulators and agencies, such as the Federal Reserve, to release anonymized data in aggregates that still allow researchers to get at important network economic questions.

\bibliographystyle{ecta}
\bibliography{networks}

\appendix

\
\clearpage

\section{Proofs}\label{sec:proofs}
\subsection{Identification}\label{sec:IDProof}
\

In this section, we formally discuss identification.  Essentially, we need three latent group centers to be fixed and to have distinct positions on the hypersphere.  We also need to know the trait status of at least some individuals and for there to be at least some individuals with more than one trait.  This is sufficient to identify the parameters governing the locations of each of the types and the concentration parameters. If we assume that trait status is unrelated to gregariousness (which is necessary for the derivation of the likelihood anyway) then we can identify the coefficient zeta.  Based on zeta and degree (which is identified as described in~\citet{mccormick2015latent} using the latent trait group sizes) we can identify the individual gregariousness parameters.   All that is left are the individual level latent positions, which we show can be identified based on the previously described parameters.

We begin by defining terms necessary to describe the spherical geometry and then provide the necessary conditions.  Throughout the proofs here we will assume a latent sphere. We now proceed to out definitions and conditions.~\\ 

\noindent{\bf Definitions.}
\begin{itemize}
    \item A sphere path consists of the points where a plane going through the origin intersects the sphere.
    \item Two points are antipodal if there are indefinitely many great circles passing through them.
\end{itemize}

\noindent{\bf Conditions:}
\begin{enumerate}
    \item The centers of the von Mises-Fisher distributions representing three of the alter groups are fixed.
    \item The fixed points are not antipodal.
    \item The fixed points are not on one great circle.
    \item For some $k,k'$, $\eta_k \neq \eta_{k'}$.
\end{enumerate}

\

\begin{proof}[Proof of Theorem \ref{thm:identification}]
Under the above conditions, this is a direct corollary to Propositions \ref{prop:latent-distribution}, \ref{prop:nu}, and \ref{prop:z}.
\end{proof}

\

\begin{prop}\label{prop:latent-distribution}
Considering the conditions above, fixing $\upsilon_k$ for $k=1,2,3$ such that all three are not on a great circle, trait centers $\upsilon_k$ for $k=4,...,K$, concentration parameters $\eta_k$ for $k=1,...,K$, and $\zeta$ are identified.
\end{prop}

\begin{proof} The von Mises-Fisher distribution is a symmetric unimodal distribution with probability mass declining in distance from the center, $\upsilon$, tuned by concentration parameter $\eta$. For each individual we know their latent trait group(s).  This is a fundamental distinction between our setting and that of~\citet{mccormick2015latent}, who typically do not assume this information is known.  We can think of the positions of each individuals as draws from one or more of the von Mises-Fisher distributions on the sphere.  An individual who belongs to two trait group has to be at the intersection of the densities of the two trait groups.  Knowing the fraction of individuals who have both traits, therefore, intuitively tells us something about the overlap between the densities of the two trait groups.  Throughout this proof keep in mind that we are not using the specific locations of individuals (which we only show is identified in a subsequent proposition), but rather the density defined by the overlap between trait groups. 

More formally, define the lens, $\ell(A,B)$, as the expected share of individuals drawn from this distribution who have traits $A$ and $B$. Equivalently, we can think of this as the volume of the overlap between the densities of the two distributions for all individuals up to a pre-specified, but arbitrary\footnote{We could define the lens for example as the are of the overlap in bands that represent that 95th percentile of the distribution.  We need to specify a cutoff because the densities are continuous across the surface. The choice is arbitrary so long as the discs are sufficiently wide to include the overlap between densities.}, cumulative probability.  In general let $\ell(A_1,...,A_k)$ denote the expected share of individuals drawn who have all traits. We can treat all lenses as observed in the data because for a large $m$, we know the traits that every node has.

For notational convenience and without loss of generality, we will assume that the fixed group centers correspond to the first three latent trait groups, $\upsilon_1, \upsilon_2, \upsilon_3$. Observe that this immediately implies all three $\eta_k$ for $k=1,...,3$ are identified. For the sake of argument assume that $\eta_1$ is known. Then from $\ell(1,2)$ we have that $\eta_2$ is identified. Given $\eta_2$, from $\ell(2,3)$, we have $\eta_3$ identified. But we can of course identify $\eta_1$ similarly from $\eta_3$.  This logic applies because we can map the overlapping section, $\ell(1,2)$, into specific values of the cumulative distribution function of the von Mises-Fisher distributions.  If we change $\eta_2$, then the location of individuals' latent positions that are draws from this distribution must also change.  Changing these locations changes the boundary of $\ell(1,2)$.  Similarly, changing the boundary of $\ell(1,2)$ implies a change in the densities of the von Mises-Fisher distributions for the first and second traits.  Since the centers of these distributions are fixed any change in the distribution must come through the concentration parameter.   

Further, this solution is unique. To see this, assume that we are at some unique solution $\eta_1,\eta_2,\eta_3$. Consider an alternative value of any combination of concentration parameters. Clearly all concentration parameters cannot increase because then the lenses would not match the true lenses. Consider then the case where at least one $\eta_k$ declines. In this case, if $\eta_{k'}$ were not to increase, then $\ell(k,k')$ would not match the expectation observed in the data. Consequently, $\eta_{k'}$ must increase. In this case, should $\eta_{k'}$ increase, then $\eta_{k''}$ must decline to preserve $\ell(k',k'')$. But in this case, the lens $\ell(k,k'')$ must increase as both concentration parameters have declined. Therefore the solution is unique.

To see why $\zeta$ is identified, consider any two $k,k'$ with $\eta_k \neq \eta_{k'}$. Because we know the respective von Mises-Fisher distributions for each trait, we can compute the ratios of the expectations of \eqref{eq:lambda} conditional on each type $k$ and $k'$, plugging in for $d_i$ from \eqref{eq:deg}. Because the individual effects are drawn independently of trait by assumption, all terms that depend on $\nu_i$ drop since the distribution of $\nu_i$ is independent of trait type, so they have the same expectations irrespective of $k$ or $k'$. As such
\[
\frac{\E_i[\lambda_{ik} \vert i\in G_k]}{\E_j[\lambda_{jk} \vert j\in G_{k'}]} = f(b_k,b_{k'},\eta_k,\eta_{k'},\zeta)
\]
where the right hand side is a known function that  comes from taking these ratios. The only unknown is $\zeta$. There is a unique solution to the equation---we leave the algebra to the reader---but can be seen from the fact that the link probability is monotonically declining in $\zeta$ and faster for lower $\eta_k$, holding all else fixed, so the ratio term also is monotone in $\zeta$.\end{proof}

\

\begin{prop}\label{prop:nu}
Considering the conditions above, $\nu_i$ for $i=1,...,m$, individual gregariousness effects for the entire  ARD sample, are identified.
\end{prop}
\begin{proof}
By Proposition \ref{prop:latent-distribution}, the $\upsilon_k$ and $\eta_k$ and $\zeta$ are identified. By \eqref{eq:lambda}, $d_i$ can be obtained and by \eqref{eq:deg} we have for every $i=1,...,m$ in the ARD sample an equation relating the fixed effect $\nu_i$ to the degree. We have $m$ equations and $m$ unknowns.

To see why the solution is unique consider fixing for the moment some $\nu_1$ without loss of generality. In this case, we can write
\( \nu_i = h_i \nu_1 \) for every $i$, where $h_i$ is the ratio of the degrees between person $i$ and person 1. Then we can write
\[
\exp(\nu_1)(\frac{1}{n}\sum_i \exp(h_i \nu_1)) = \frac{d_1}{m \cdot \frac{C_{p+1}(0)}{C_{p+1}(\zeta)}}. 
\]
This is a monotone function in $\nu_1$ and has a unique solution, which then identifies the remainder of the $\nu_i$ as well scaling by $h_i$.
\end{proof}

\

\begin{prop}\label{prop:z}
Considering the conditions above, the latent locations  $z_i$ for $i=1,...,m$ for the entire  ARD sample, are identified.
\end{prop}
\begin{proof} From Propositions \ref{prop:latent-distribution} and \ref{prop:nu}, we have identified all parameters except for $z_i$.  To show this result, we first state two results from spherical geometry.  The proofs of these results are available in standard texts (e.g.~\citet{Biringer:2015}).

{Result:}{\it The sphere path between two points is unique unless the points are antipodal.}

{Result:}{\it There are exactly three isomorphisms for spherical geometry.}

The first result defines a unique distance from each respondent latent position and at least two of the three latent group means.  A respondent position can be  antipodal with one of the three fixed groups, but then cannot be with the two others because the three groups cannot be antipodal. 

The second result limits the number of possible operations that threaten identifiability.  Recall that, if an operation changes the latent distance between an point and the center of a group, then the operation will also change the likelihood.  Thus, if we show that we cannot perform any of the three possible distance preserving transformations on the sphere after fixing group centers, then we have also completed the proof. 

We consider two cases, the first takes and arbitrary point that is not antipodal to any of the latent centers, whereas the second case considers any point that is antipodal with one latent center. 

{\bf Case 1.}
Since we fix three centers which are not on a great circle, we cannot do any reflections of points without changing the distance to one of the centers. For rotations, consider centers $\upsilon_1$ and $\upsilon_1$, and a point $z_i$. Since $\upsilon_1$ and $\upsilon_2$ are not antipodes, if we rotate $z_i$ around center $\upsilon_1$ and keep $d(z_i,\upsilon_1)$ the same, it is possible that $d(z_i,\upsilon_2)$ changes. The points $z_i, z_i'$ such that $d(z_i,\upsilon_1)=d(z_i',\upsilon_1)$ and $d(z_i,\upsilon_2)=d(z_i',\upsilon_2)$ are reflections over the plane that intersects $\upsilon_1$ and $\upsilon_2$ in a great circle. $z_i$ and $z_i'$ have equal distance to any point on this great circle, and unequal distance to any point not on this great circle. Since the third center $\upsilon_3$ is not on this the great circle that intersects $\upsilon_1$ and $\upsilon_2$, $d(z_i,\upsilon_3) \neq d(z_i',\upsilon_3)$.

{\bf Case 2.}
When we change the point's position, then the distance between that point and the antipodal latent center decreases.

This completes the proof.
\end{proof}
%
%
%
%
\subsection{Taxonomy}\label{sec:taxonomy_proofs}
We present the proofs for the taxonomical results.

\

\begin{proof}[Proof of Proposition \ref{prop:main_taxonomy}]
Observe that
\begin{align*}
{\rm E}\left[\left(S_{i}\left({\bf g}\right)-S_{i}\left({\bf g}^*\right)\right)^{2}\right] & ={\rm E}\left[\left(\left(S_{i}\left({\bf g}\right)-{\rm E}\left[S_{i}\left({\bf g}\right)\right]\right)+\left({\rm E}\left[S_{i}\left({\bf g}\right)\right]-S_{i}\left({\bf g}^*\right)\right)\right)^{2}\right]\\
 & \leq{\rm E}\left[\left(S_{i}\left({\bf g}\right)-{\rm E}\left[S_{i}\left({\bf g}\right)\right]\right)^{2}\right]\\
 & +2{\rm E}\left[\left(S_{i}\left({\bf g}\right)-{\rm E}\left[S_{i}\left({\bf g}\right)\right]\right)\right]\left({\rm E}\left[S_{i}\left({\bf g}\right)\right]-S_{i}\left({\bf g}^*\right)\right)\\
 & +\left(S_{i}\left({\bf g}^*\right)-{\rm E}\left[S_{i}\left({\bf g}\right)\right]\right)^{2}.
\end{align*}
We can readily see that each of these terms are $o_{p}\left(1\right)$.
\end{proof}

\begin{proof}[Proof of Corollary \ref{cor:MSE_link}]
This is straightforward to calculate:
\begin{align*}
{\rm E}\left[\left(g_{ij}-g^*_{ij}\right)^{2}\right] & ={\rm E}\left[g_{ij}^{2}-2g_{ij}g^*_{ij}\right]+g_{ij}^{2}\\
 & =p_{ij}^{\theta_{0}}\left(1-2g^*_{ij}\right)+g^*_{ij}
\end{align*}
which completes the proof.
\end{proof}

\

\begin{proof}[Proof of Corollary \ref{cor:MSE_density_diffcent}]
For degree, one can check that
\[
\sum_{k}\frac{p_{ij}^{\theta_{0}}\left(1-p_{ij}^{\theta_{0}}\right)}{k^{2}}\leq\sum_{k}\frac{1}{k^{2}}\rightarrow0
\]
so the Kolmogorov condition is satisfied and
\[
\frac{d_{i}}{n}-\frac{{\rm E}\left[d_{i}\right]}{n}\rightarrow_{a.s.}0
\]
which satisfies the conditions of Proposition \ref{prop:main_taxonomy}.

For diffusion centrality, recall that
\begin{align*}
DC_{i}\left({\bf g};q_{n},T\right): & =\sum_{j}\left[\sum_{t=1}^{T}\left(q_{n}{\bf g}\right)^{t}\right]_{ij}\\
 & =\sum_{j}\sum_{t=1}^{T}\frac{C^{t}}{n^{t}}\sum_{j_{1},...,j_{t-1}}g_{ij_{1}}\cdots g_{j_{t-1}j}.
\end{align*}
It is easy to check analogous to the degree term, for any $t$,
\[
\frac{1}{n^{t}}\sum_{j}\sum_{j_{1},...,j_{t-1}}g_{ij_{1}}\cdots g_{j_{t-1}j},
\]
which has variance at most $\prod_{s=1}^{t}p_{j_{s-1}j_{s}}\left(1-\prod_{s=1}^{t}p_{j_{s-1}j_{s}}\right)\leq1$
for any summand, with $j_{0}=i$ and $j_{s}=j$. The Kolmogorov condition
again applies and so every term converges to its expectation.
\end{proof}

\clearpage

\section{Implementation Appendix}\label{sec:blueprint}

\subsection{Cookbook}
The goal of this section is to provide a researcher or policymaker with a practical blueprint for collecting the required data and implementing our latent distance model. We propose this method in situations when the researchers want to estimate social network characteristics but when full social network maps are either infeasible or prohibitively expensive to collect.

In our preferred implementation, the researchers would collect a census of all members of the graph of interest. This approach might be feasible in settings such as a rural village, where typically there is enumeration done and basic demographics are taken for all nodes. However, we recognize that censuses might not be feasible in all settings such as a large urban slum. We include a discussion of such settings in Section \ref{sec:Census_Infeasible}.

We envision researchers conducting the following steps:
\begin{enumerate}
    \item {\bf Design ARD survey questions:} The first step is to choose which traits to use. This choice will depend on the context of the specific empirical setting.  But generally-speaking, the traits should satisfy the following criteria:
    \begin{itemize}
        \item The traits should satisfy the core assumptions of the model: that in a latent space sense they are located predominantly in one region (the distribution of individuals' latent positions is single-peaked). See Section \ref{sec:QuestionDesign} for a more detailed discussion of this. 
        \item The traits should likely be observable by others (because eliciting the information in a survey relies on the observations of the respondent) and should not be subject to much measurement error (respondents should not know so many people with the trait that it is difficult for them to recall everyone, for example).
        \item The number of traits should not be very long, both to avoid survey fatigue and keep costs low.\footnote{However, recall that the method requires fixing the positions of three groups on the surface.  Therefore, the number should be larger than five.}
    \end{itemize}
    \item {\bf Conduct census survey:} The census survey should include the following parts:
    \begin{itemize}
        \item The ARD traits: Knowing this allows the researcher to calculate the population share of nodes in the graph with each trait $k$.
        \item An additional set of demographic characteristics denoted by $X$. The vector of $X$s allows for the researcher to predict the latent locations of the nodes not included in the ARD survey sample.   
    \end{itemize}
    See Section \ref{sec:CensusBihar} for a sample census questionnaire.
    \item {\bf Conduct ARD survey:}
    The researchers will need to decide what share $\psi$ of households will be surveyed. This is simply a budgetary computation, but we suggest that $\psi\geq 0.2$. 
    
    The ARD survey should contain:
    \begin{itemize}
        \item Link enumeration: This step is useful for providing a clear way to define a link, to aid in the interpretation of the ARD counts, and to decrease measurement error in the ARD counts. This step also gives a direct estimate of each node's degree in the sample. If the friend list methodology is not possible, we discuss how the procedure changes in Section \ref{sec:No_DegreeList}.
        \item ARD responses: For every trait in the ARD list, ask the subject to count within the enumerated list of links, how many have each trait.

    \end{itemize}
    See Section \ref{sec:ARDBihar} for a sample ARD questionnaire.
    
    \item {\bf Run the ARD estimation procedure using inputs from the surveys:} Section \ref{sec:ARD_Codes} details how to download and execute all of the ARD estimation codes in R. 
    \item {\bf Estimate the economic parameter of interest:} See Section \ref{sec:ARD_Codes} for details on the estimation procedure.

\end{enumerate}

\subsubsection{Census Infeasible}\label{sec:Census_Infeasible}

In this subsection we assume that the researcher does not have access
to a census of the population and has a vector of attributes for every
unit (e.g., household or individual) in the population. Intuitively,
the core difference between this context and the prior context is
that the researcher does not have the population share by type from
the census itself. This is the case for the Hyderabad urban example in Section \ref{sec:hyderabad}.
\begin{enumerate}
\item[(3)] {\bf ARD survey:} If there is no census, then the researcher should ask every node in the ARD sample whether they have each trait. Then these sampled observations can be used to compute estimates of the population
shares.
\item[(4)] {\bf ARD estimation procedure:} Without a census, one cannot follow the procedure in Section \ref{sec:ARDtononARD} to estimate the locations of the non=ARD nodes.  Instead:
\begin{itemize}
    \item For the 1-$\psi$ share of non-ARD nodes, draw node locations based on latent trait distributions observed in the ARD sample. 
    \item When drawing graphs, use the estimated latent locations based on the ARD responses for the $m$ ARD survey nodes and from this procedure for the remaining $n-m$ nodes. 
\end{itemize}
\end{enumerate}

\subsubsection{Link enumeration is infeasible}\label{sec:No_DegreeList}

\begin{enumerate}
\item[(3)] {\bf ARD survey:} Ask the subject to reflect on their friends (or links in whatever
manner the researcher is trying to collect data).
\begin{itemize}
\item This can be recorded by the enumerators. The number of links gives the degree for each ARD node. 
\item If the number of links is expected to be too large for respondents to reliably count, use a N-Sum like method (see e.g.~\cite{mccormick2010many}).
\end{itemize}
\item[(4)] {\bf ARD estimation procedure:} The one difference in the estimation procedure is that the expected degree of each node needs to be estimated (see Equation \ref{eq:lambda}), rather than taken directly from survey responses. The code is built to accommodate this case.  
\end{enumerate}
\

\subsection{Discussion of Question Design}\label{sec:QuestionDesign}

Here, we discuss how to choose ARD traits to enable us to construct a good image of the network. While we leave a precise characterization of optimal questions to future research, we nonetheless can offer practical insights to aid in ARD survey design. 

Conceptually ARD traits are those which, under the model, organize the latent space into regions such that nodes with certain traits are more likely to be towards the centers of those regions. Recognizing that under the model, nodes are linked as a function of their distance in this latent space, nodes are more liked to be linked to other nodes with similar such traits. This gives some insight as to which ARD features may be useful to organize the latent space.

Then when we ask, ``how many of your friends have been gored by a bull" or ``how many of your friends have multiple wives," those that have a positive count of this are going to have to be located somewhere close to the (latent, unknown) location of the cluster of people with this kind of experience. The reason is because we assume that the network that exists forms from the model in Equation \ref{eqn:network-model}, so it is most likely that someone who knows a friend that got gored by a bull and another person who has a friend who got gored by a bull are then likely to be in the same part of the latent space.
What this means is that we do not need traits that ``drive" the latent space per se, but traits that are informative. So a bad example might be a trait where it is peppered throughout the village. Not everyone does it, but many groups do, and so many people at very different points in the latent space are likely to have known someone who has this trait. As such, both (1) how many friends have ever experienced crop loss due to a drought and (2)  how many friends do you know who have twins (in a rural setting where IVF is uncommon) would presumably be uninformative. However, something where a subcommunity engages in a practice (multiple wives) would be a better trait.

In sum, a good way to think about a useful trait, in our view, is one that is ``single peaked". It should be a characteristic that is likely to be held by one group, not distributed throughout.  Furthermore because traits are used to triangulate the latent space, ones that are not essentially redundant should be chosen. If the traits essentially identify the same set of people (e.g., how many friends are Muslim?; how many friends have ever gone to a mosque?), then clearly they do not add value.

\clearpage

\subsection{Sample Census Questionnaire}\label{sec:CensusBihar}

\subsection*{IDENTIFICATION}

\begin{enumerate}
\item Date of Interview
\item Surveyor Name
\item District Name
\item SubDistrict Name
\item Village Name
\item HHID 
\item GPS of the HH (marked automatically)
\end{enumerate}

\subsection*{HOUSEHOLD IDENTIFIERS}
\begin{enumerate}
\item What is the name of the respondent ?
\item What is the name of the Household Head ?
\item What is the caste of the household head?
\item What is the sub-caste of the household head?
\item Does the Household have an electricity connection?
\item What type of roofing material does the household have?
\item Does the Household own land ?
\item Does the Household have a toilet ?
\end{enumerate}

\subsection*{ARD TRAITS}

\begin{enumerate}
\item Does the House have 2 or greater than 2 floors ? 
\item Does the respondent own a kirana shop / tea/ sweets shop/PDS shop?
\item Has any member in your household migrated to another city for labor or construction
work in the last 2 years?
\item Does any member in your household own a bike?
\item Does the respondents' house have iron/steel gates?
\item Has any member in the household passed the 12th Standard?
\item Does anyone in your household own a goat/hen?
\item Is any member in the household a government Employee?
\item Does anyone in your household have a smart phone?
\item Did any adult in the household have typhoid, malaria, or cholera in
the past six months ?
\item Does the house have 5 or greater than 5 children below the age of
18 ?
\item Is anyone in your Household a member of religious or cultural committee
at the village level ?
\end{enumerate}

\clearpage

\subsection{Sample ARD Questionnaire}\label{sec:ARDBihar}

\subsection*{IDENTIFICATION}

\begin{enumerate}
\item Enter HHID
\item Village Name
\item District Name 
\item SubDistrict Name 
\item Gram Panchayat Name 
\item Name of the Respondent
\end{enumerate}

\subsection*{FRIEND LIST ELICITATION}

\textit{Instruction to Enumerator: Note down the list of names as
given by the respondent. As you note down the names make sure that
names that are repeated are marked, so that at the end of the 3 questions,
we have a list of unique friends}
\begin{enumerate}
\item Tell the names of the Household Heads of those families in this
village whose house you visit or who visit your house frequently or
with whom you socialize frequently ? 
\item Tell the names of Household Heads of those families in this village
who give you advice/ or to whom you give advice on farming/health/financial
issues? \textit{(Ask each part separately)}
\item If you urgently needed kerosene/charcoal, rice or money, who do you
borrow them from or who borrows it from you? \textit{(Ask each part
separately)}
\end{enumerate}

\subsection*{ARD}

\textit{Instruction to Enumerator: }
\begin{itemize}
\item \textit{Inform the respondent of the name and no of friends that have
been named in the previous section}
\item \textit{Tell the respondent that the questions in this section pertain
to the friends named in the previous section}
\end{itemize}
Out of all the households whose name you took in the previous section,
how many have the following traits :
\begin{enumerate}
\item No of floors in the house are greater than or equal to 2?
\item No of households out of your friend list who own a kirana shop/ tea/
sweets shop/ PDS?
\item No of households out of your friend lis wherein any member migrated
to another city for labor/construction work in the last 2 years?
\item No of households among your friend list who own a bike?
\item No of households among your friend list whose house have iron/steel
gates?
\item No of households among your friend list wherein any member has passed
the 12th Standard?
\item No of households among your friend list which own goats/ hen?
\item No of households among your friend list where in any member is a government
Employee?
\item No of households among your friend list where someone has a smart
phone?
\item No of households among your friend list where any adult has had typhoid,
malaria, or cholera in the past six months?
\item No of households among your friend list which have 5 or greater than
5 children below the age of 18?
\item No of households among your friend list where anyone is a member of
religious or cultural committee at the village level?
\item No of households among your friend list who belong to the Scheduled
Caste?
\end{enumerate}

\clearpage
\subsection{Estimation using ARD}\label{sec:ARD_Codes}

This section presents an abridged walk-through as to how to use our estimation procedure. We assume the researcher has csv or xls data and is familiar with Stata and (to a lesser degree) R~\citep{Rcitation}. We walk the researcher step-by-step moving from the raw data, through Stata, through R (with code provided) which outputs csv files, back to Stata in order to conduct estimation of interest. A more detailed walk-through that explains all intermediate code is provided in Section \ref{sec:ARD_Codes_Long}.

\begin{enumerate}
    \item Download ARD code: \url{https://github.com/MengjiePan/BCMP}    
    
    \item Format survey data in the following manner:
    \begin{itemize}
        \item Create a dataset(csv,xls)  that is $m$ ARD nodes by $K$ ARD responses for each village and save each file as ARD\_SURVEY\_i.csv
        \item Create a dataset that is $n$ nodes by the $K$ ARD-trait covariates from the census for each village and save each file as ARD\_CENSUS\_i.csv
        \item Create a dataset that is $m$ ARD nodes by the $L$ covariates from the census (e.g., GPS, household identifiers). Create another dataset that is $n-m$ Non ARD nodes by the $L$ covariates from the census(same covariates as used for ARD Nodes). Use $L$ covariates of these two datasets in a distance function to create a $n-m$ by $m$ dataset. This will be used in k-nearest neighbours algorithm. Save each file as distance\_i.csv

    \end{itemize}
    
	\includegraphics[scale = 0.8]{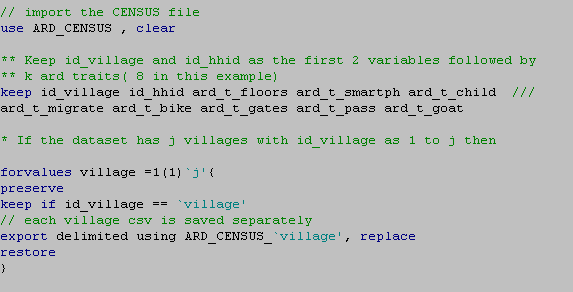}
	
	\includegraphics[scale = 0.8]{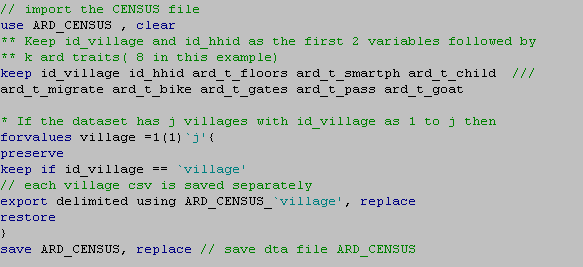}
		
	\item Copy the downloaded R files in the same folder. The folder structure should be as shown in the figure below(for 4 villages)
	
	\includegraphics{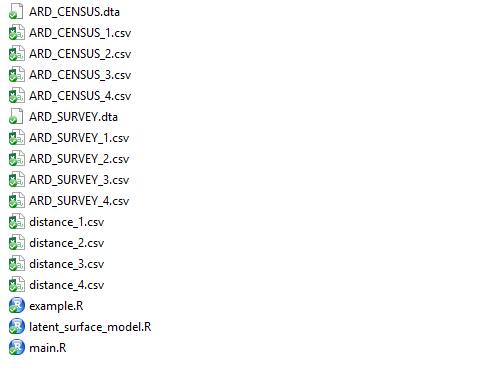}

	\item Open the file example.R	
	\item Download R Packages - \texttt{igraph}\citep{csardi2006igraph} , \texttt{movMF}\citep{hornik2014movmf}, \texttt{xlsx}\citep{dragulescu2018package} (if the datasets are in xls), \texttt{readstata13}\citep{readstata13} (if the datasets are in Stata 13,14) [example.R downloads these packages]
	\item Enter the path to the folder in variable \texttt{r\_folder} (Line 24).
	
	\includegraphics{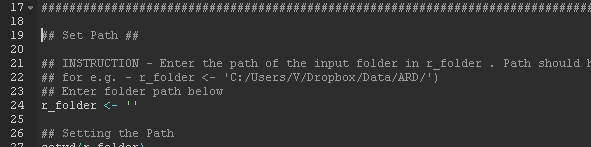} 
	
	\item	Run the R Script example.R. Output should be generated in Folder \texttt{OUT} in the current folder.

    \item Import the network characteristics that have been generated in folder \texttt{OUT}
    
    \includegraphics{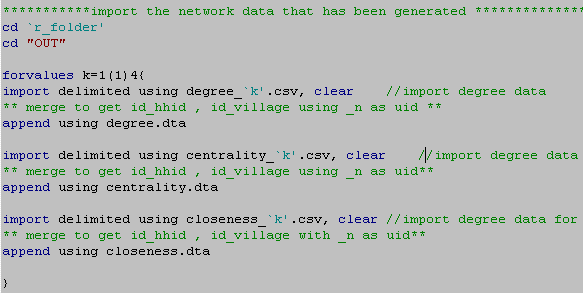}
    
    \item Import the graph simulations that have been generated from folder \texttt{OUT/SIMULATION}
    
    \item Conduct economic estimation of interest. For instance,
    \[
    y_{iv} = \alpha + \beta \frac{1}{B}\sum_{b=1}^BS(g)_{iv,b} + \epsilon_{iv},
    \]
    to estimate $\beta$, which is the parameter of interest in this example, where $i$ is a node and $v$ is the independent network for $v=1,...,V$ networks in the sample.
    
     \includegraphics{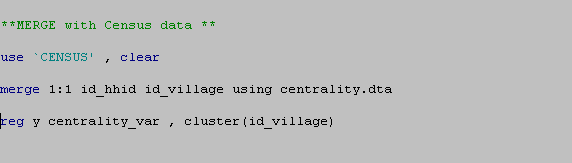}
     
\end{enumerate}

\clearpage

\begin{center}
    {\large {\sc {\bf Online Appendix: Not for Publication}}}
\end{center}

\section{Detailed Estimation Procedure}\label{sec:ARD_Codes_Long}

This section presents the detailed walk-through for the estimation procedure.

\begin{enumerate}
    \item Download ARD code: \url{https://github.com/MengjiePan/BCMP}    
    
    \item Format survey data in the following manner:
    \begin{itemize}
        \item Create a dataset(csv,xls)  that is $m$ ARD nodes by $K$ ARD responses for each village and save each file as ARD\_SURVEY\_i.csv
        \item Create a dataset that is $n$ nodes by the $K$ ARD-trait covariates from the census for each village and save each file as ARD\_CENSUS\_i.csv
        \item Create a dataset that is $m$ ARD nodes by the $L$ covariates from the census (e.g., GPS, household identifiers). Create another dataset that is $n-m$ Non ARD nodes by the $L$ covariates from the census(same covariates as used for ARD Nodes). Use $L$ covariates of these two datasets in a distance function to create a $n-m$ by $m$ dataset. This will be used in k-nearest neighbours algorithm. Save each file as distance\_i.csv

    \end{itemize}
    
	\includegraphics[scale = 0.8]{STATA_1}
	
	\includegraphics[scale = 0.8]{STATA_2}
		
	\item Copy the downloaded R files in the same folder. The folder structure should be as shown in the figure below(for 4 villages)
	
	\includegraphics{Folder}

	\item Open the file example.R	
	\item Download R Packages - \texttt{igraph}\citep{csardi2006igraph} , \texttt{movMF}\citep{hornik2014movmf}, \texttt{xlsx}\citep{dragulescu2018package} (if the datasets are in xls), \texttt{readstata13}\citep{readstata13} (if the datasets are in Stata 13,14) [example.R downloads these packages]
	\item Enter the path to the folder in variable \texttt{r\_folder} (Line 24). \textbf{Running the R Script example.R now } should generate the ARD Output in the Folder \texttt{OUT} in the current folder. The steps given next explain the process in detail through code snippets.

	\includegraphics{PATH}
	
     \item Preparing the datasets for constructing ARD :
      \begin{itemize}
      
      	\item The datasets created in Step 2 are imported(Line 36-54) and are named \texttt{ard\_survey}, \texttt{ard\_census} and \texttt{distance.all} respectively
      	\item Calculate the value of variable \texttt{total.prop} - \textit{fraction of ties in the network that are made with members of group k, summed over K groups} using \texttt{example.R} (Line no 69-80). The variable \texttt{villagei} stores the \texttt{ard\_census} traits.
        \end{itemize}
        
    	\includegraphics{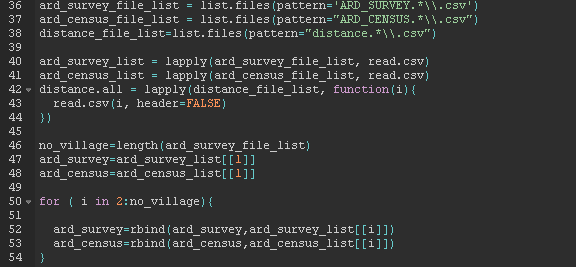}
    	
	    \includegraphics{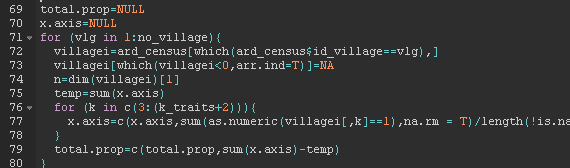}

    \item Estimate the parameters of the model: $(\nu_i,z_i)_{i=1}^m$ for the $m$ ARD households, $\zeta$, $(\upsilon_k,\eta_k)_{k=1}^m$ (the latent trait distribution location and concentration parameters).
    \begin{itemize}
        \item Use \texttt{example.R} to call(Line 93) \texttt{main.R}, which calls(Line 23) function  \texttt{f.metro} in \texttt{latent\_surface\_model.R}. 
        \item The call to function \texttt{main} of \texttt{main.R} on Line 93 requires 4 input variables
	            \begin{itemize}
	            	\item \texttt{y} - use the \texttt{ard\_survey} dataset that has been imported
	            	\item \texttt{total.prop} - Calculated in Step 4 
	            	\item \texttt{muk.fix } - the positions of fixed variables calculated in Line 126-127 of \texttt{example.R}
	            	\item \texttt{distance.matrix} - use the \texttt{distance.all} dataset that has been imported
		         \end{itemize}

        \item The Output of the call to \texttt{f.metro} is stored in variable \texttt{posterior} of \texttt{main.R}

    \end{itemize}
    
       	\includegraphics{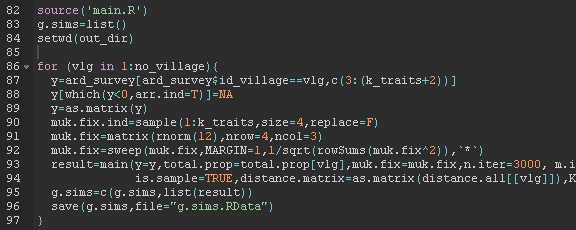}

      	\includegraphics{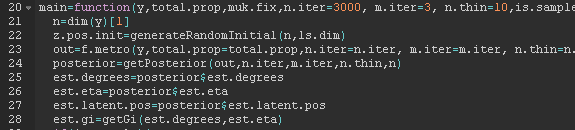}
	
	\item Estimate $\nu_i$ and $z_i$ for the $n-m$ nodes that are in the census but not the ARD sample.
	\begin{itemize}
		\item  \texttt{main.R} (Line no 30) calls function \texttt{getPosteriorAllnodes} in \texttt{main.R}. The call to the function takes variable \texttt{distance.matrix} as an input(which had been passed to function \texttt{main} from \texttt{example.R} in Step 5)

		\item Output is stored in variable \texttt{posteriorAll}. The estimated latent positions $z_i$ are stored as an attribute of \texttt{posteriorAll} as \texttt{est.latent.pos.all}
		\item \texttt{getPosteriorAllnodes} estimates  $\nu_i$ and $z_i$ using $k$-means from \texttt{distance.matrix} variable. This variable  has been calculated using the $K+L$ covariates for the $m$ nodes in the ARD sample
		and  $n-m$ Non-ARD nodes 
	\end{itemize}
			\includegraphics{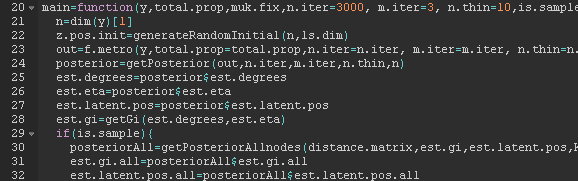}
			
			\includegraphics{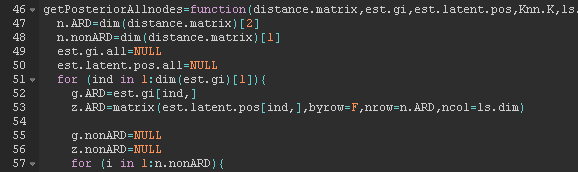}

	\item Draw a set of $b=1,\ldots,B$ draws from the network formation probability model (now with estimated parameters for all nodes) from the posterior distribution.		
        \begin{itemize}
        	\item Use \texttt{main.R} (Line no 33) to call function \texttt{simulate.graph.all} . The output is stored in variable \texttt{g.sims} . \texttt{simulate.graph.all} calls(Line 108) \texttt{simulate.graph.once} for each run.

            \item Draw a parameter vector ${\bf \theta}$ (all the above parameters) from the posterior.
            \item Draw a graph $g_b$ given ${\bf \theta}_b$. (Line 130 - function \texttt{simulate.graph.once})
        \end{itemize}    
        
		 \includegraphics{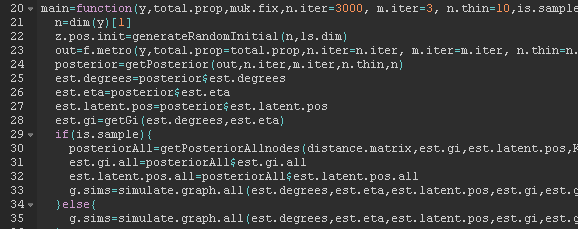}
		 
		 \includegraphics{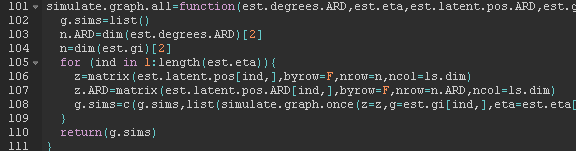}
		 
 		 \includegraphics{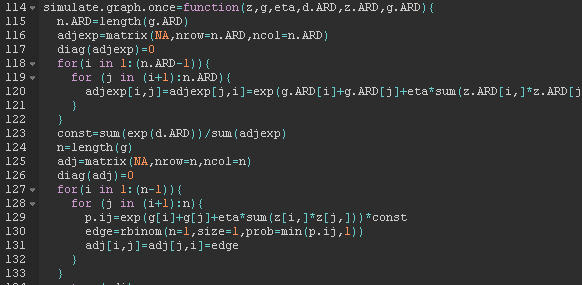}
    \item Compute network statistics of interest $S(g_b)$ for each draw $g_b$ for $b=1,...,B$.
    \begin{itemize}
        \item Construct your own desired functions
        \item Or use a suggested code \texttt{example.R} (Line no 115-144)
    
        \end{itemize}
        \includegraphics{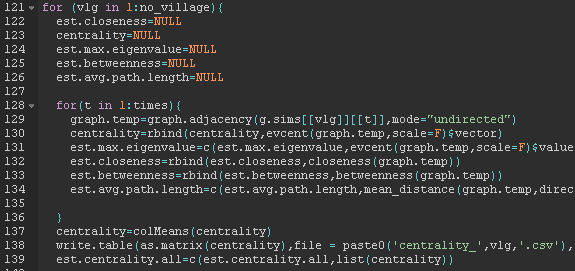}
    \item Import the network characteristics that have been generated in folder \texttt{OUT}
    
    \includegraphics{STATA_OUT1}
    
    \item Import the graph simulations that have been generated from folder \texttt{OUT/SIMULATION}
    
    \item Conduct economic estimation of interest. For instance,
    \[
    y_{iv} = \alpha + \beta \frac{1}{B}\sum_{b=1}^BS(g)_{iv,b} + \epsilon_{iv},
    \]
    to estimate $\beta$, which is the parameter of interest in this example, where $i$ is a node and $v$ is the independent network for $v=1,...,V$ networks in the sample.
    
     \includegraphics{STATA_OUT2}
     
\end{enumerate}

\clearpage

\section{ARD Questions from \cite*{banerjeebdk2016}}\label{sec:ARD_questions}
\label{sec:ardqlist}

This section presents the ARD questions used in \cite*{banerjeebdk2016} that we use in Section \ref{sec:hyderabad}. 

 How many other households do you know in your neighborhood ...
\begin{enumerate}
\item where a woman has ever given birth to twins?
\item where there is a permanent government employee?
\item where there are 5 or more children?
\item where any child has studied past 10th standard?
\item where any adult has had typhoid, malaria, or cholera in the past six months?
\item where any adult has been arrested by the police?
\item where at least one woman has had a second marriage?
\item where at least one man currently has more than one wife?
\end{enumerate} 

\clearpage

\section{Comparing Latent Model to a Beta Model}
\label{sec:betacompare}
\setcounter{table}{0}
\renewcommand{\thetable}{E.\arabic{table}}
\setcounter{figure}{0}
\renewcommand{\thefigure}{E.\arabic{figure}}

\

We compare our model to the beta model to illustrate how adding latent positions to our model fitting procedure affects the precision of our estimation. 
\

To fit a beta model, we first run \citep{mccormick2010many} to get a posterior distribution of estimated degrees for ARD nodes. Then taking $\zeta=0$ in Equation \eqref{eq:deg}, we get a posterior distribution of $\nu_i$. As with the latent case, we generate graph using $\Pr(g_{ij}=1|\nu_i,\nu_j) \propto \exp(\nu_i) \exp(\nu_j)$ and average measures over simulated graphs.

\begin{figure}[!h]
\centering
\includegraphics[width=\textwidth]{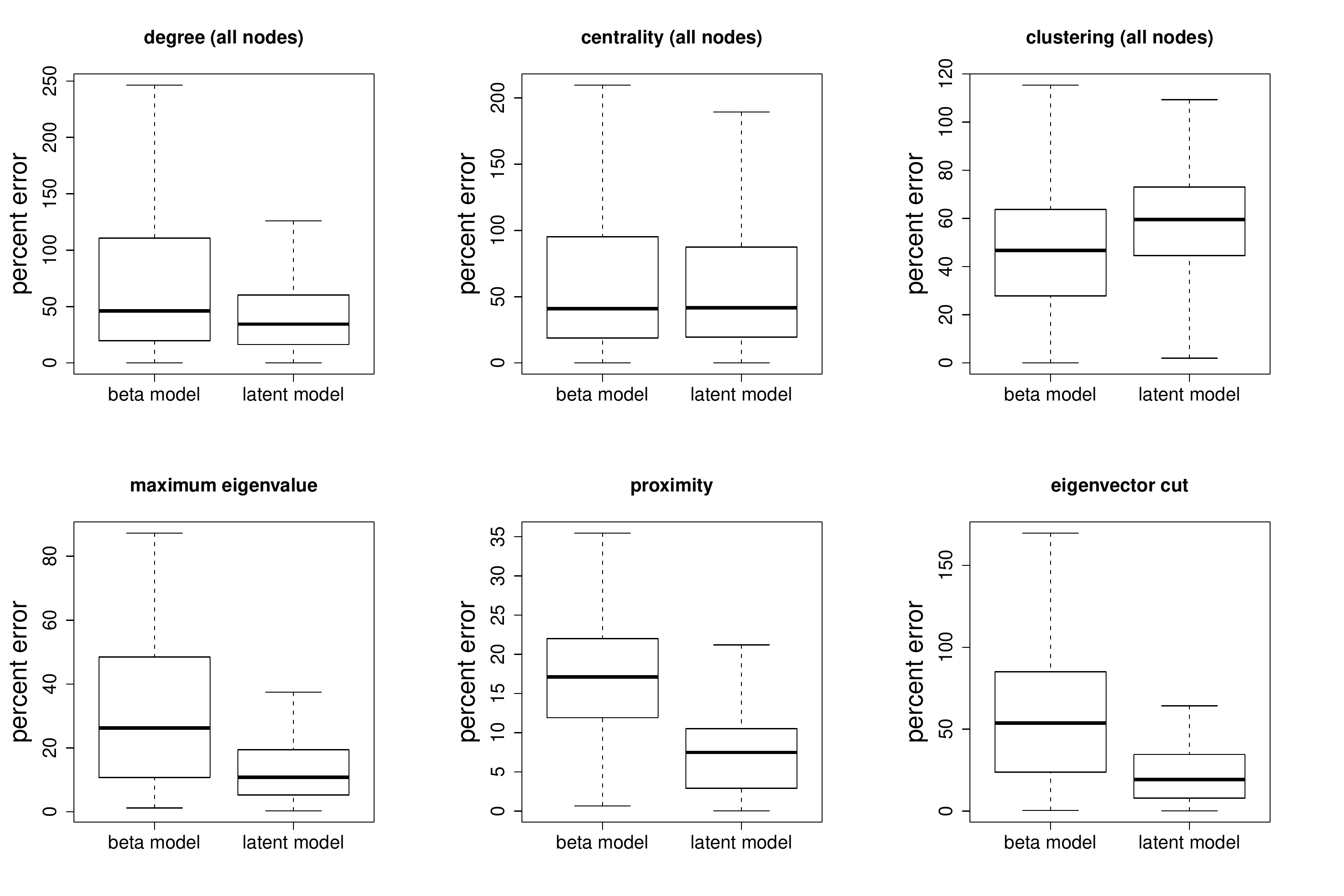}
\caption{Comparison of using beta model and latent model to estimate node level measures for all nodes and network level measures. These plots show boxplot of absolute percentage error for each statistic. Latent model outperforms beta model on all network level measures, and has similar performances on node level measures.}
\label{fig:boxplot}
\end{figure}

We compare beta model and latent model on degree, centrality, and clustering estimation on all nodes, as well as maximum eigenvalue, proximity, and eigenvector cut. Because the absolute percentage errors are very right skewed, we present boxplots that show the distribution for each measure (Figure \ref{fig:boxplot}), with outliers omitted from the plot. The beta model performs slightly better in estimating clustering, but performs worse in degree, proximity, maximum eigenvalue, and eigenvector cut. The mean absolute percentage error with eigenvector cut using latent model is approximately two thirds of the one using beta model. This illustrates one advantage of using a latent surface model. The propensity of forming an edge not only depends on the popularity of two nodes, but also on their distance on the latent surface. So the simulated graphs resemble the true graph's partitioning better than the simulated graphs from a beta model.

\clearpage

\section{Prior Experiments}
\label{sec:priorsens}
\setcounter{table}{0}
\renewcommand{\thetable}{F.\arabic{table}}
\setcounter{figure}{0}
\renewcommand{\thefigure}{F.\arabic{figure}}

We show how the choice of priors and fixed subpopulations affect our results.  The priors we use in Section \ref{sec:results} are: uniform hyperpriors for $\mu_d, \sigma_d^2$, Gamma(0.5,0.5) for $\zeta$, and Gamma(5,0.1) for $\eta_k$, and this is what ``base model'' in Figures \ref{fig:boxplot_mu}-\ref{fig:boxplot_mukfix} refers to. We have experimented with the following alternate priors: $\mu_d \sim \mathcal{N} (0,5)$, $\mu_d \sim \mathcal{N} (2,5)$, and $\mu_d \sim \mathcal{N} (4,5)$; $\sigma_d^2$ follow inverse-chi-squared distribution with parameters (1,0.5) and (1,3); $\zeta \sim$ Gamma(2,0.5) and Uniform(0.001,10); $\eta_k \sim$ Gamma(10,0.1) and Uniform(0.1,150).

We perform two types of sensitivity analyses.  First, we show that the quality of our estimates is consistent across a wide set of choices of prior values.  Second, we directly examine the influence of the prior by comparing three sets of densities: the density in the observed Karnataka data, the posterior density, and the density from the prior.  Additionally we consider two different ways of fixing positions of a subset of subpopulations on the latent space. In Section \ref{sec:results} we fix subpopulations based on their caste information and the fact that people in the same caste are more likely to know each other. Here we experiment with choosing randomly positions and which subpopulations to fix (``mukfixRandom'' in Figure \ref{fig:boxplot_mukfix}), as well as intentionally fixing subpopulations very close to each other (``mukfixClose'' in Figure \ref{fig:boxplot_mukfix}).

Similar to Figure \ref{fig:boxplot}, Figures \ref{fig:boxplot_mu}-\ref{fig:boxplot_mukfix} show the distribution of absolute percentage errors for each measure with outliers omitted. We see from these figures that changing priors and fixed subpopulations have no impact on the performances of our proposed method, although the prior on $\zeta$ has slight impact on the estimation of maximum eigenvalue.

Moving now to our second set of sensitivity analyses, Figure~\ref{fig:prior_density} shows density plots for five different network features.  In each of the plots we see three histograms.  The green histograms represent the density of the network feature that arises from the prior distribution choices we use in Section~\ref{sec:results}.  These densities arise from generating networks from the prior distributions.  That is, they describe the types of networks our formation model would produce in the absence of data.  As a contrast, we plot the densities from the (estimated) posterior, which includes information from both the prior and from ARD constructed using the Karnataka data.  For comparison, we also included the observed density from the Karnataka data, or the ``true'' density.

\begin{figure}[!h]
\centering
\includegraphics[width=\textwidth]{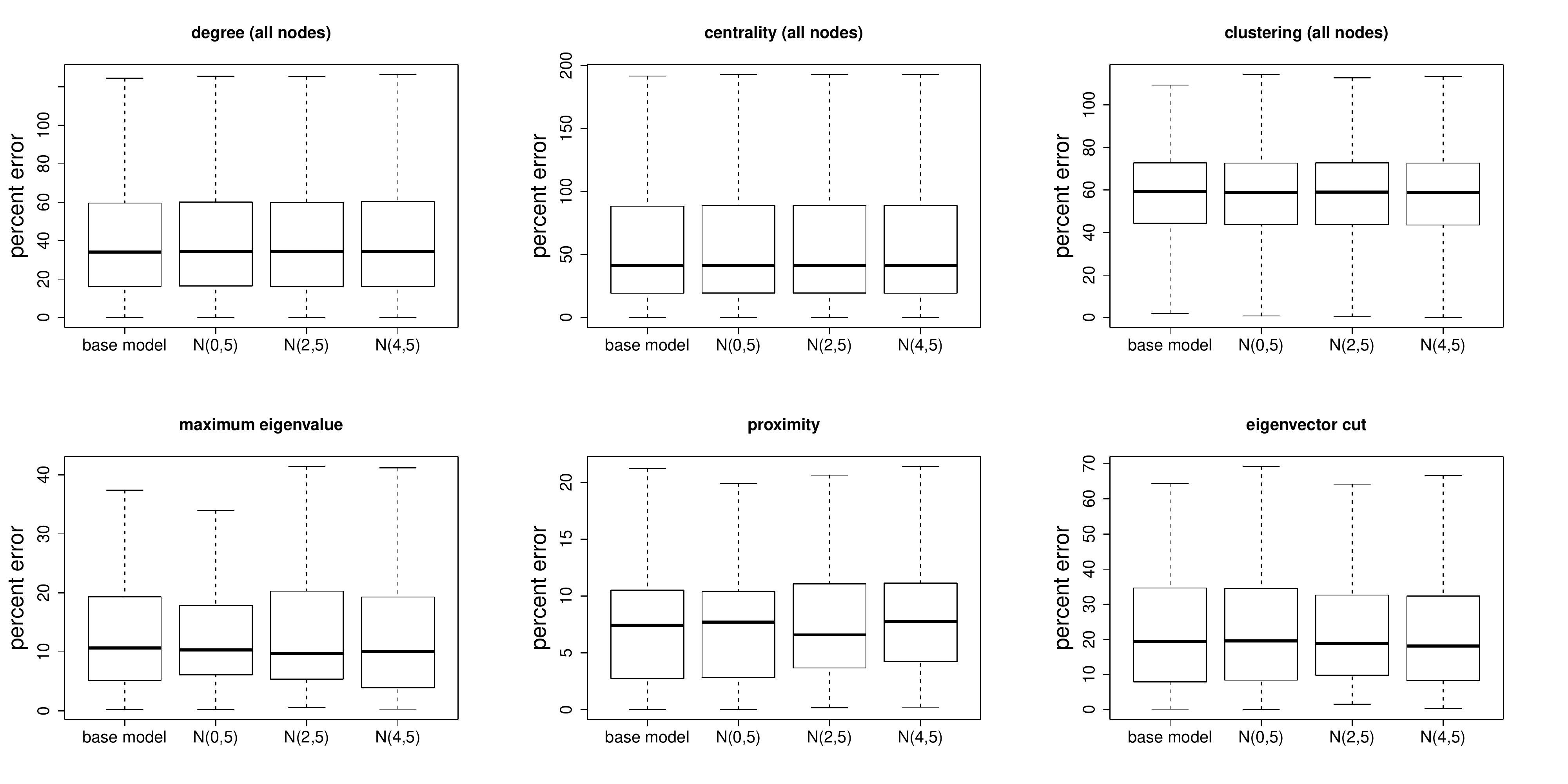}
\caption{Comparison of using uniform, $\mathcal{N} (0,5), \mathcal{N} (2,5), \mathcal{N} (4,5)$ priors for hyperparameter $\mu_d$ to estimate node level measures for all nodes and network level measures. These plots show boxplot of absolute percentage error for each statistic. Prior of $\mu_d$ do not have an impact on the results.}
\label{fig:boxplot_mu}
\end{figure}

\begin{figure}[!h]
\centering
\includegraphics[width=\textwidth]{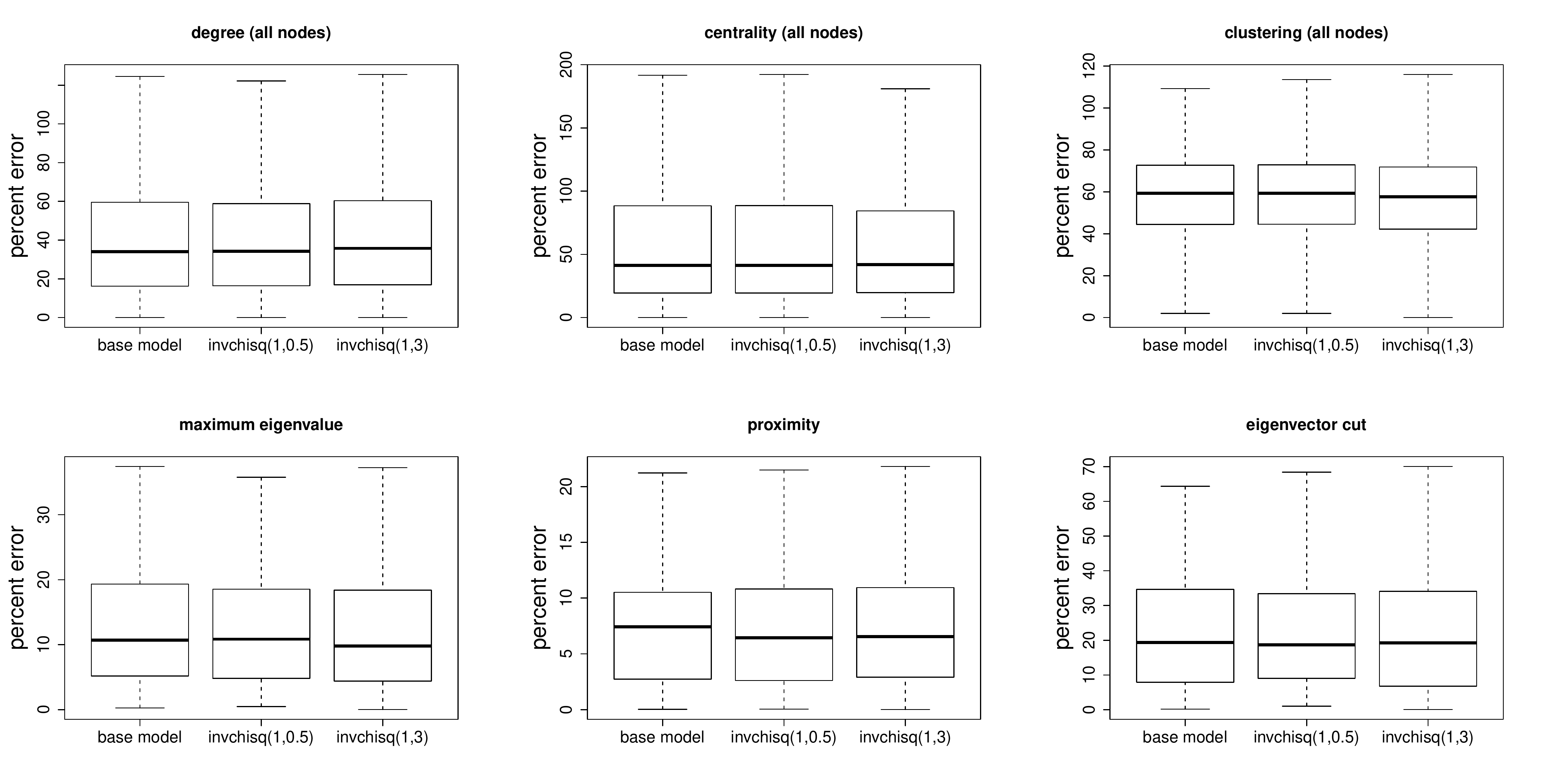}
\caption{Comparison of using uniform, inverse-chi-squared distribution with parameters (1,0.5) and (1,3) priors for hyperparameter $\sigma_d^2$ to estimate node level measures for all nodes and network level measures. These plots show boxplot of absolute percentage error for each statistic. Prior of $\sigma_d^2$ do not have an impact on the results.}
\label{fig:boxplot_sd}
\end{figure}

\begin{figure}[!h]
\centering
\includegraphics[width=\textwidth]{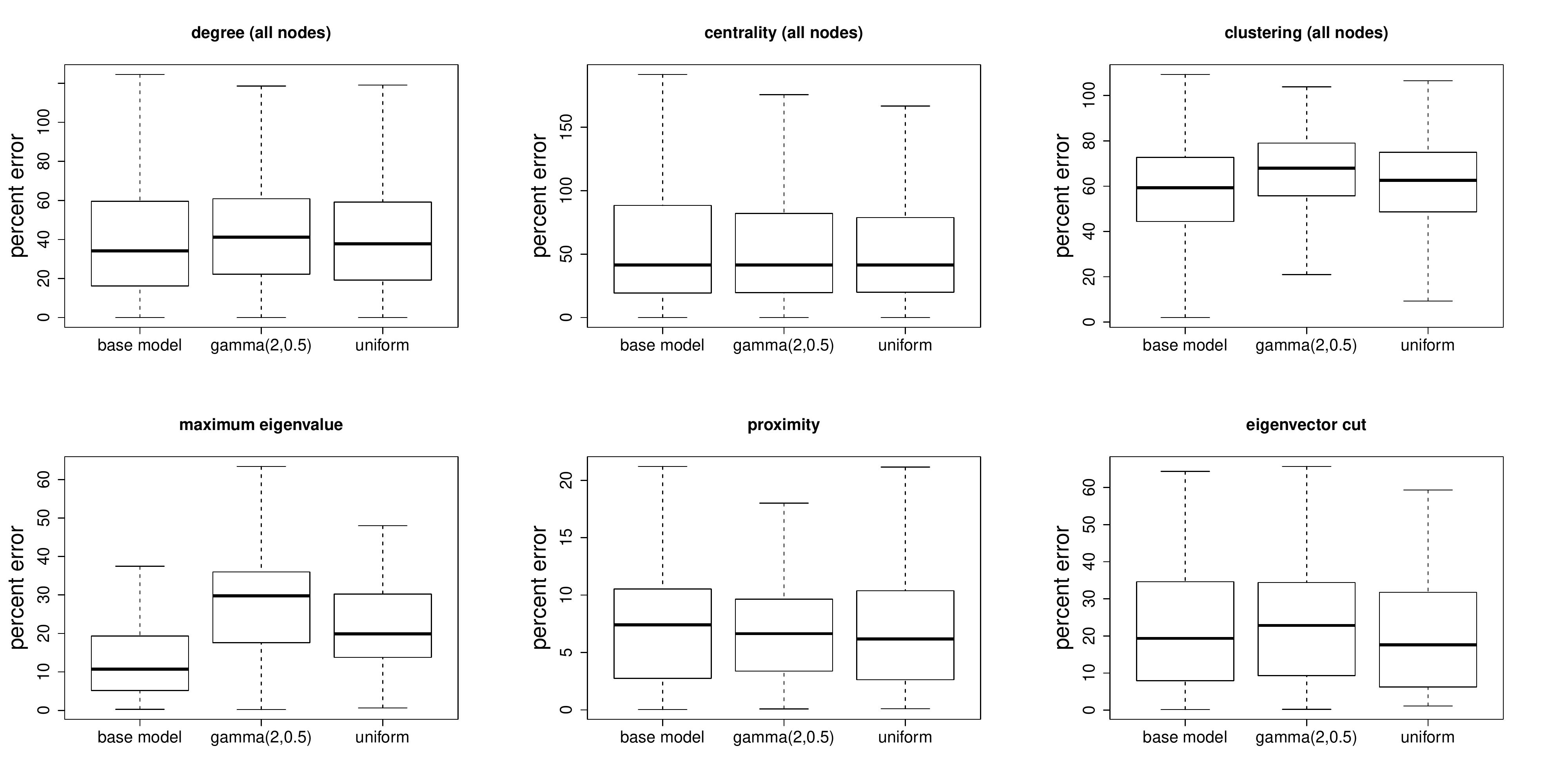}
\caption{Comparison of using Gamma(0.5,0.5), Gamma(2,0.5) and Uniform(0.001,10) priors for $\zeta$ to estimate node level measures for all nodes and network level measures. These plots show boxplot of absolute percentage error for each statistic. Prior of $\zeta$ impacts maximum eigenvalue and clustering slightly.}
\label{fig:boxplot_eta}
\end{figure}

\begin{figure}[!h]
\centering
\includegraphics[width=\textwidth]{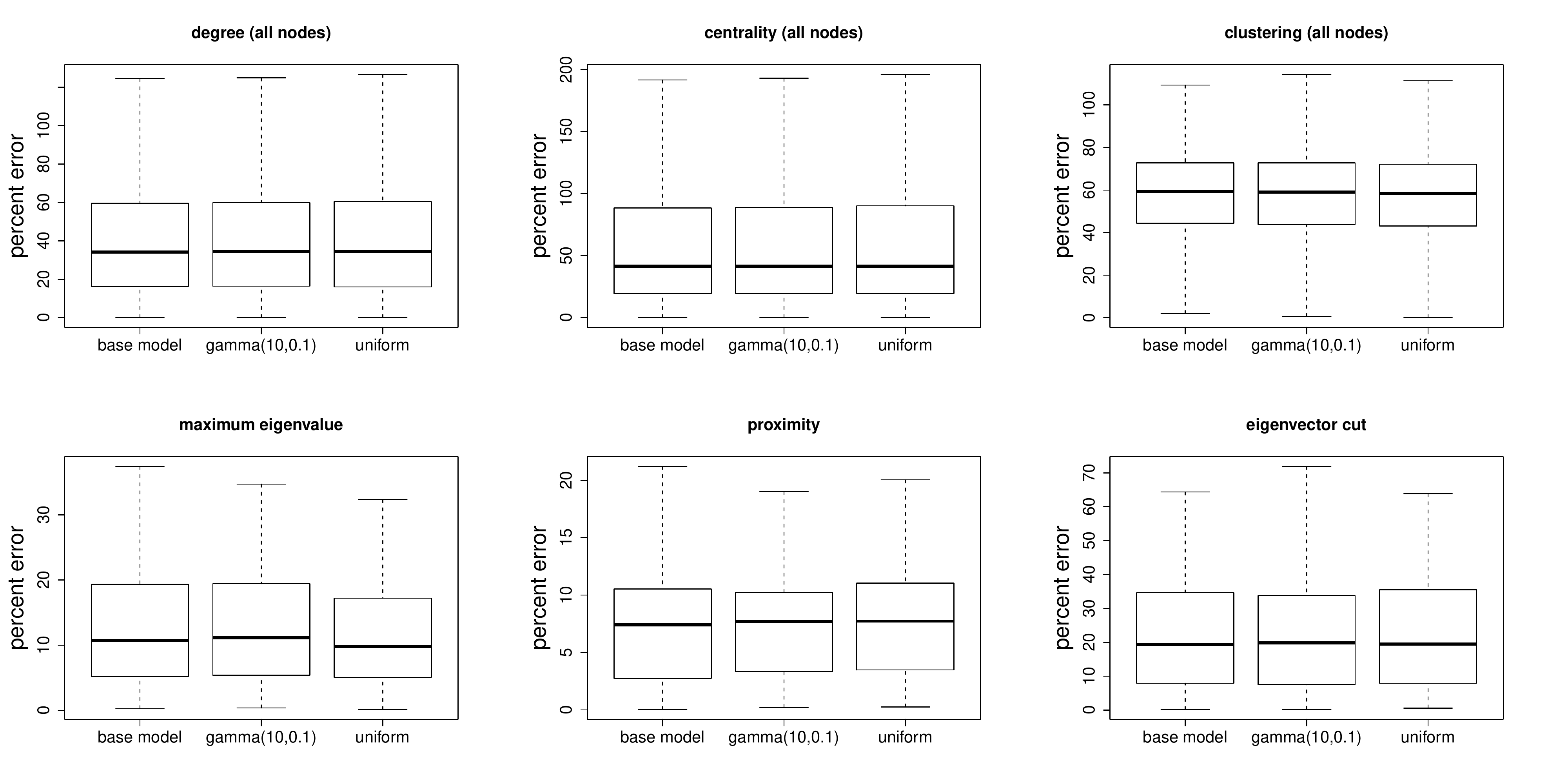}
\caption{Comparison of using Gamma(5,0.1), Gamma(10,0.1) and Uniform(0.1,150) priors for $\eta_k$ to estimate node level measures for all nodes and network level measures. These plots show boxplot of absolute percentage error for each statistic. Prior of $\eta_k$ do not have an impact on the results.}
\label{fig:boxplot_etak}
\end{figure}

\begin{figure}[!h]
\centering
\includegraphics[width=\textwidth]{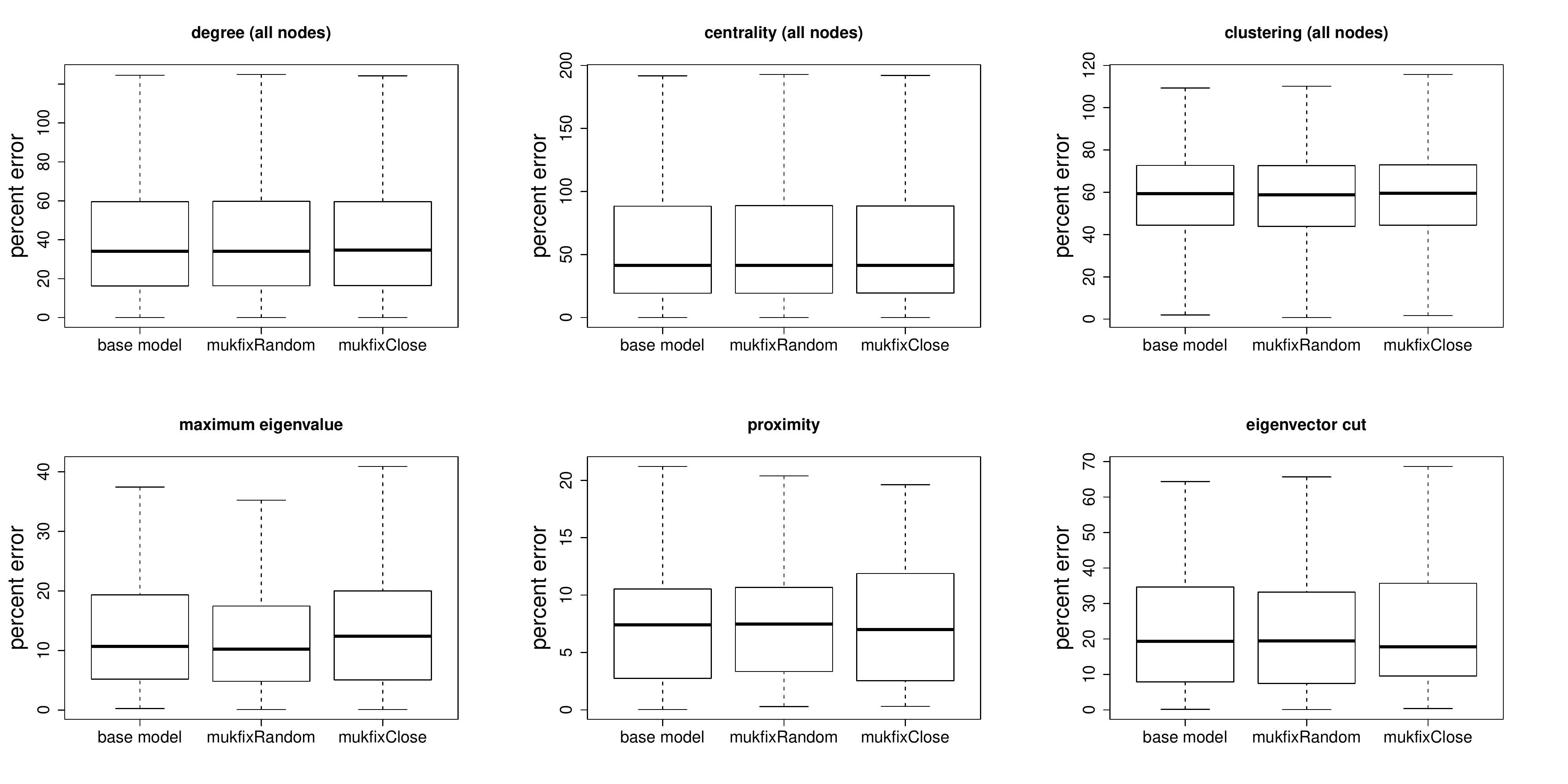}
\caption{Comparison of results from models fixing subpopulations based on caste information, fixing subpopulations randomly, and intentionally fixing subpopulations close. These plots show boxplot of absolute percentage error for node level measures for all nodes and network level measures. These three ways of fixing a subset of subpopulations do not have an impact on the results.}
\label{fig:boxplot_mukfix}
\end{figure}

\begin{figure}[!h]
\centering
\subfloat{
\includegraphics[width=.3\textwidth]{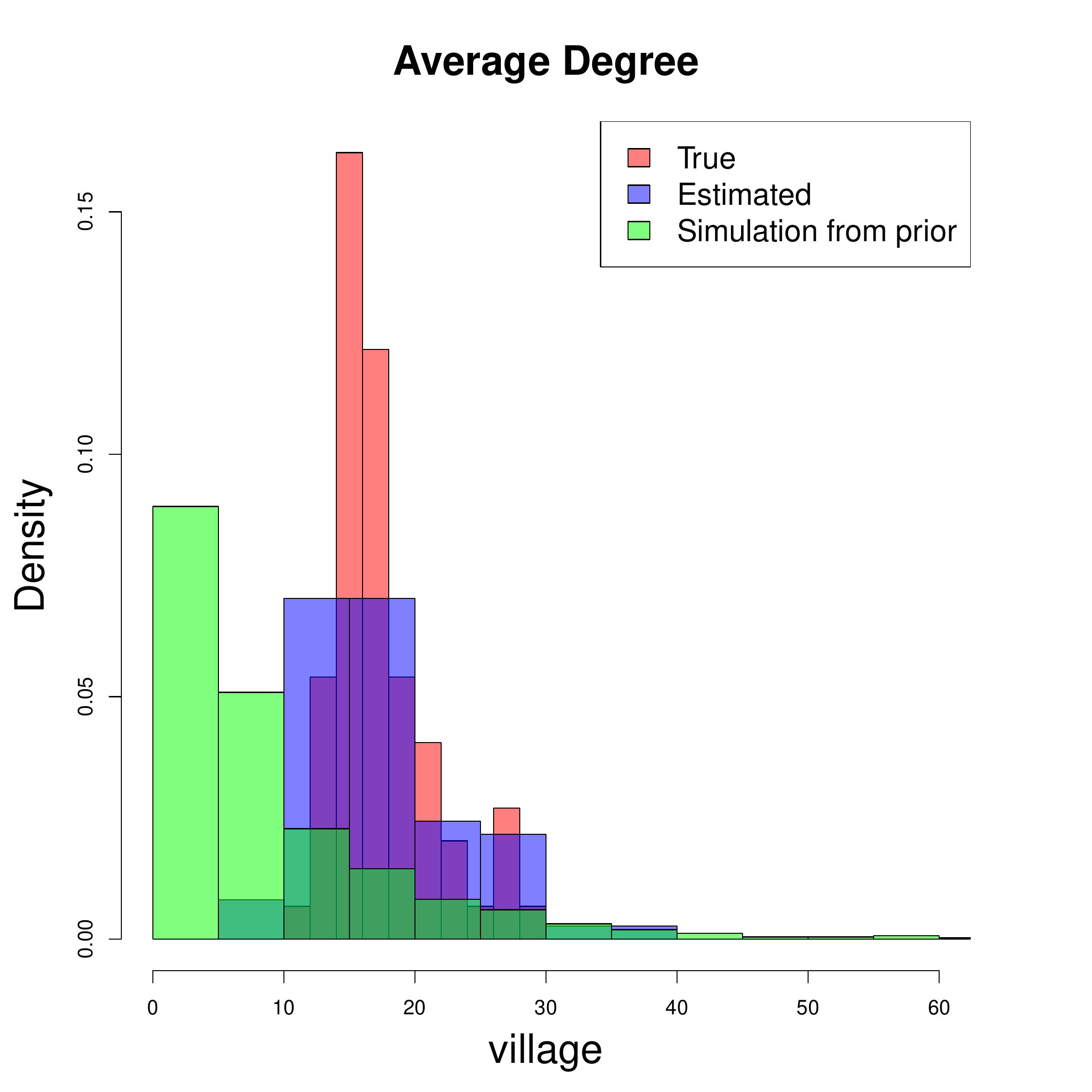}}
\subfloat{
\includegraphics[width=.3\textwidth]{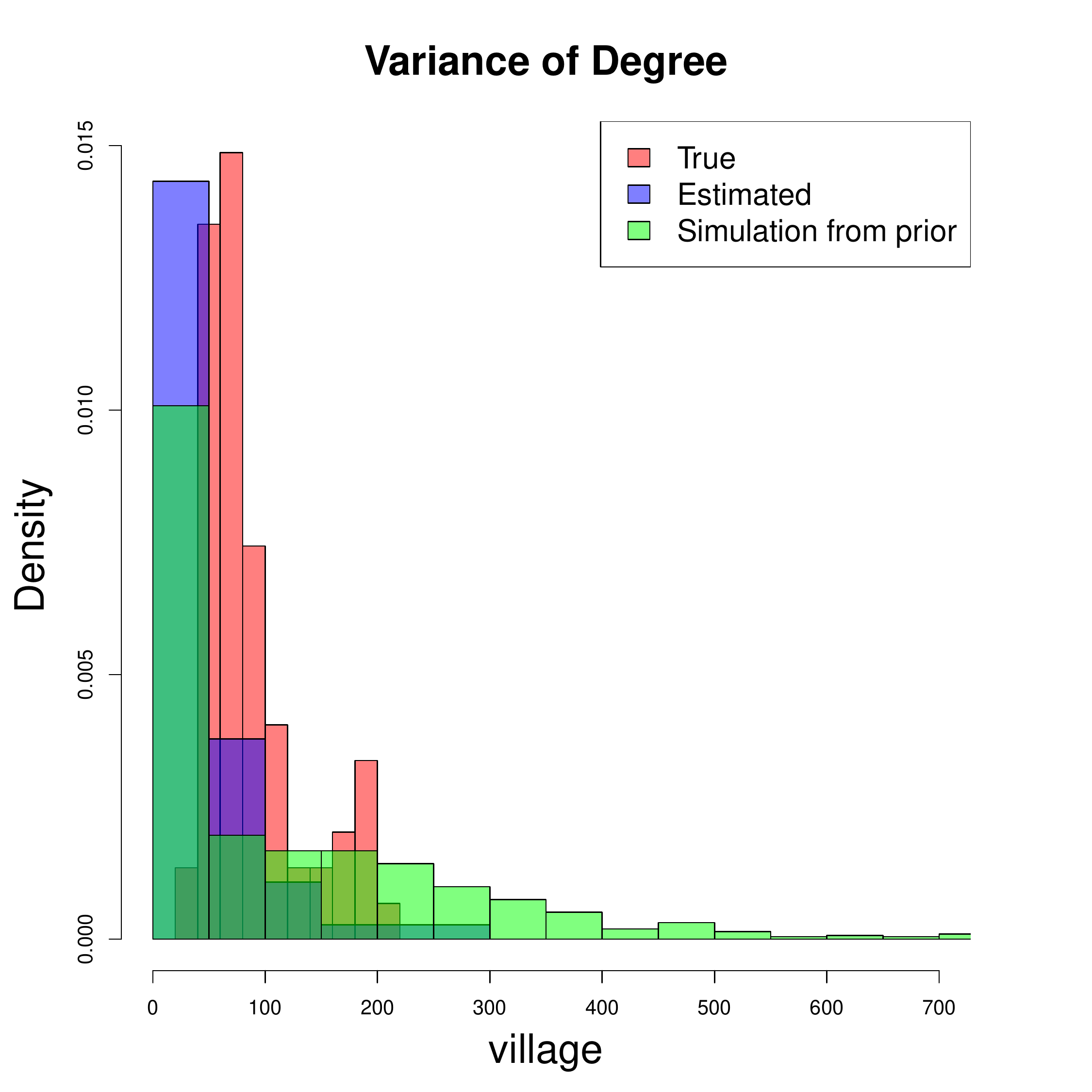}
}
\subfloat{
\includegraphics[width=.3\textwidth]{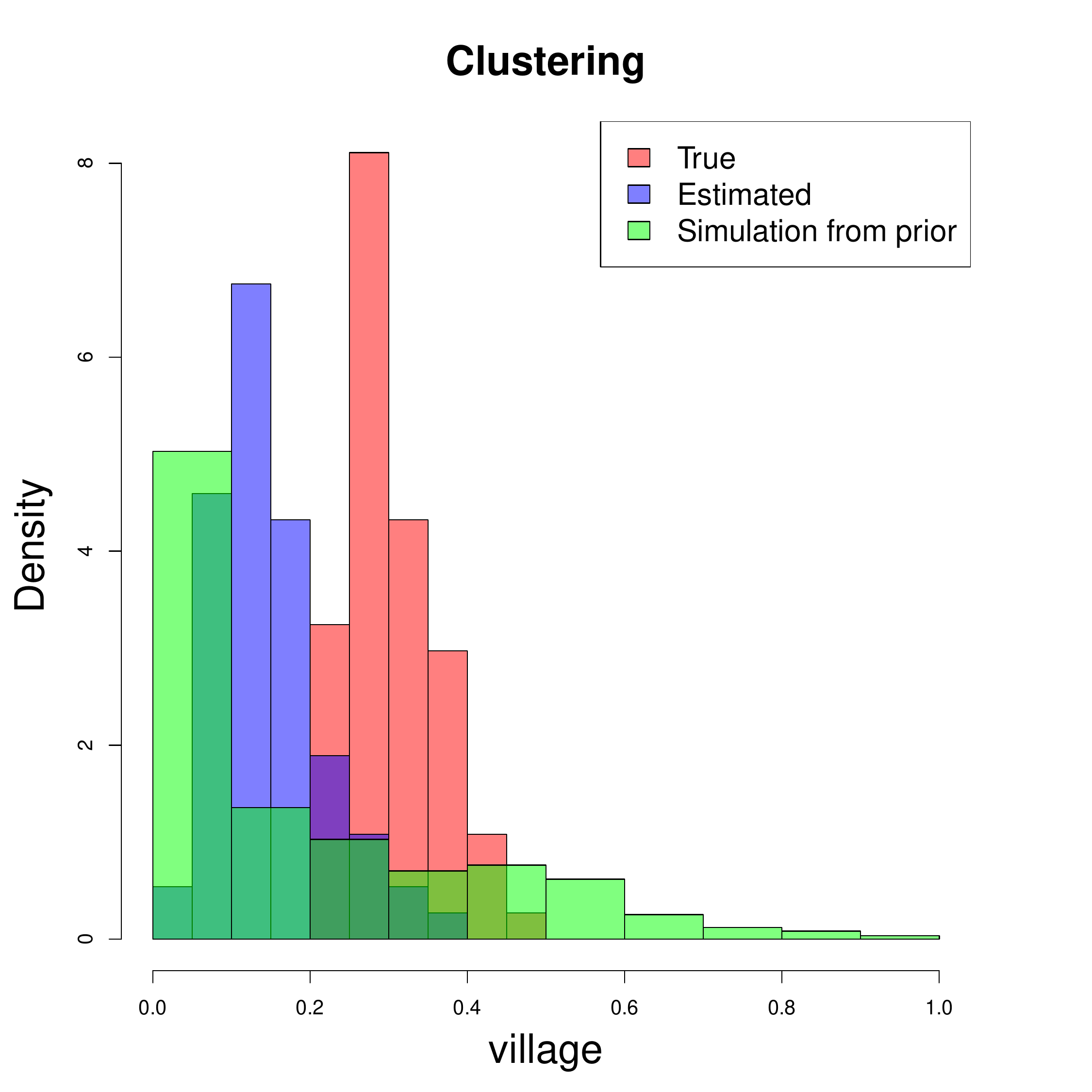}
}

\medskip

\subfloat{
\includegraphics[width=.3\textwidth]{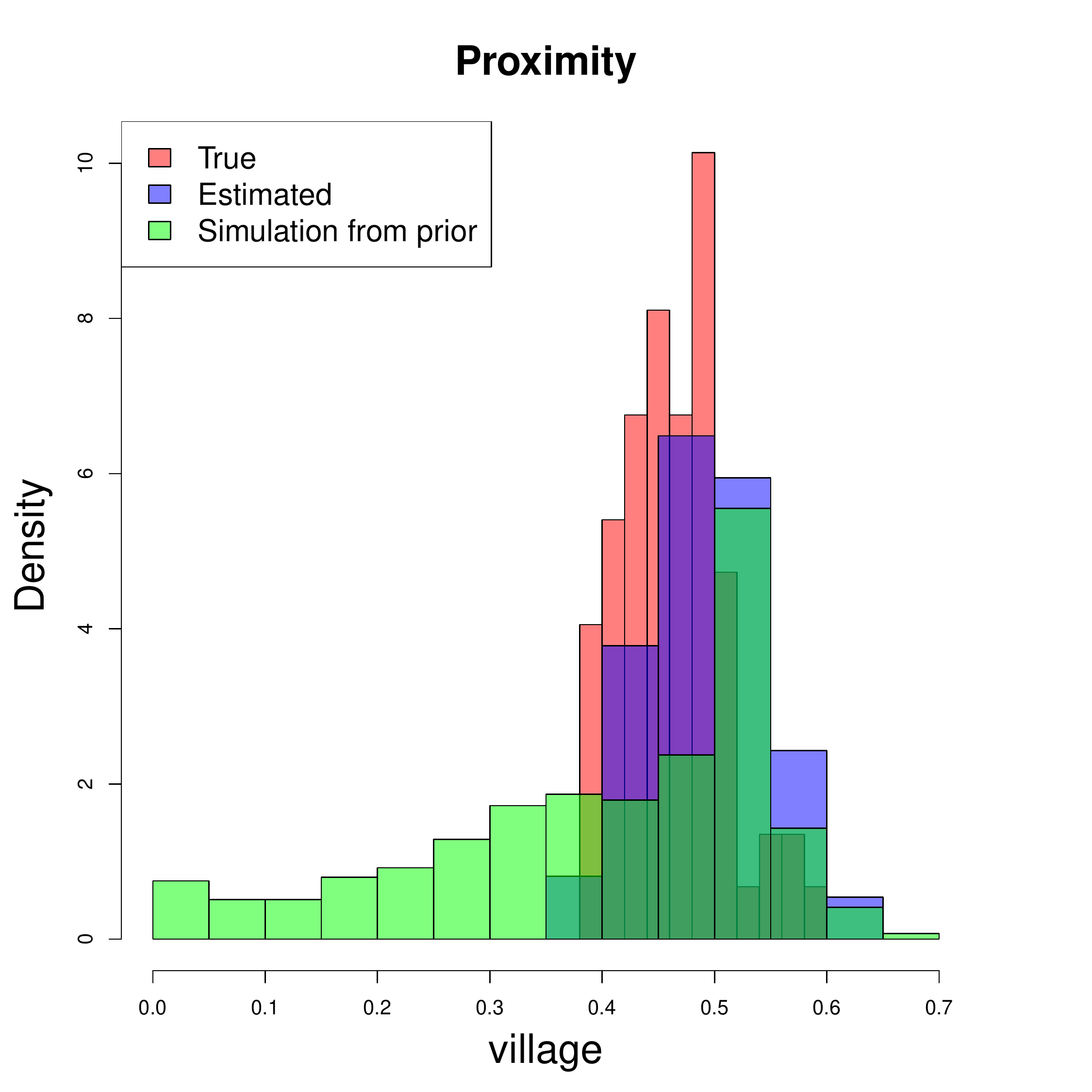}
}
\subfloat{
\includegraphics[width=.3\textwidth]{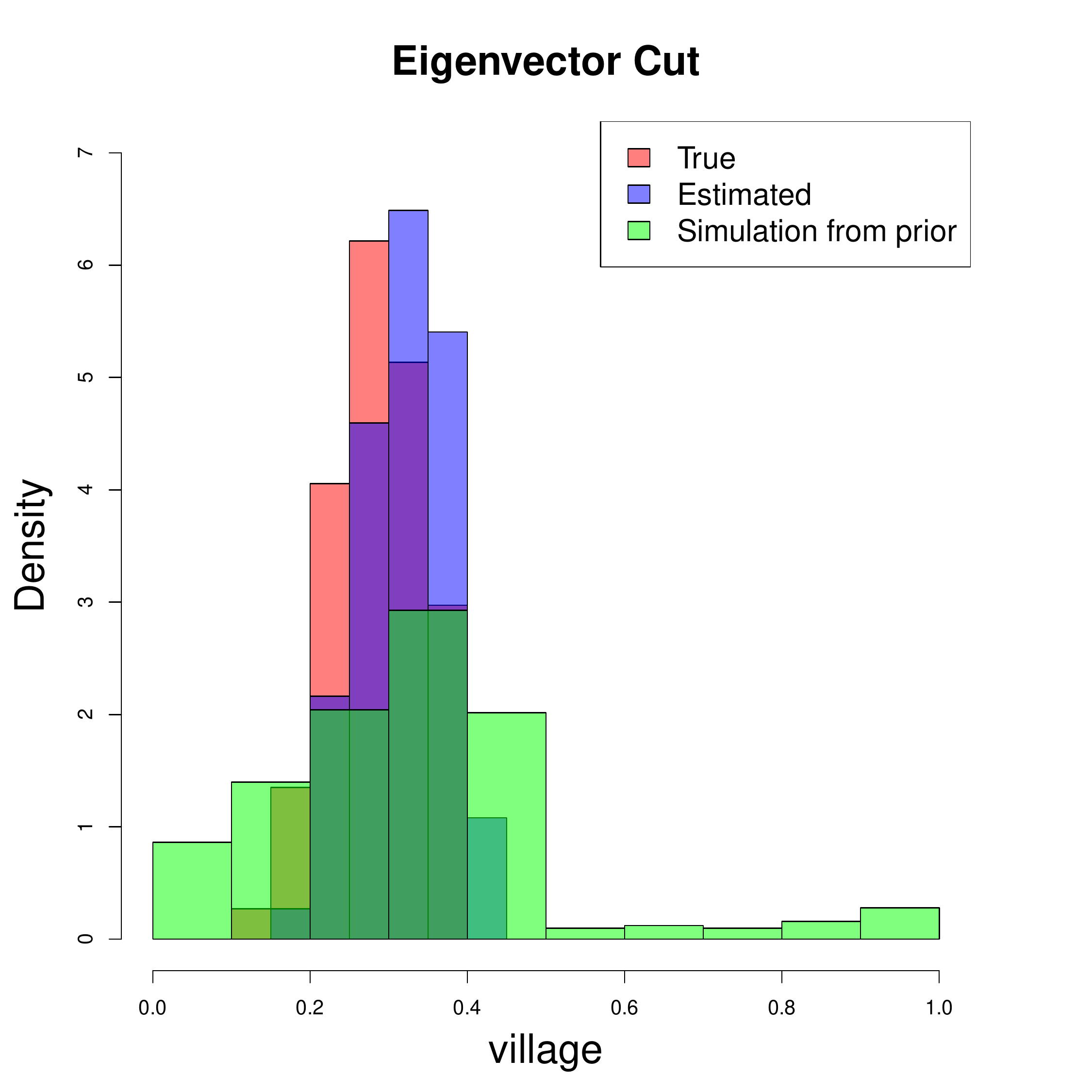}
}

\caption{Density of average degree, variance of degree, network-level clustering, proximity, and eigenvector cut.  The histograms labeled ``True'' show the density observed in the Karnataka networks.  The ``Estimated'' histograms are the density estimated from fitting our model to this data and the ``Prior'' histograms are the density from networks simulated using our chosen prior distributions.  Overall, the ``Prior'' histograms have higher variance and are, in many cases, centered in different places than estimated (posterior) densities, indicating that information in the ARD data are driving estimation, rather than the prior.}
\label{fig:prior_density}
\end{figure}

\clearpage

\section{Simulating non-uniform feature centers}
\label{sec:nonuniformcenters}
\setcounter{figure}{0}
\renewcommand{\thefigure}{G.\arabic{figure}}

In this section, we show that in the case where some of the feature centers are clustered in latent space, our method is able to achieve the same results as in core simulation (Section \ref{sec:core_results}). To simulate non-uniform feature centers, we first simulate 4 out of the 12 centers uniformly randomly. Then for each of the 4 centers, we simulate two additional centers that are close to it in the latent space. Specifically, let $\mathbf{\upsilon}_1$,...,$\mathbf{\upsilon}_4$ be the first four uniform centers. Then we sample $\mathbf{\upsilon}_5$, $\mathbf{\upsilon}_6 \sim \mathcal{M}(\mathbf{\upsilon}_{1},20)$, $\mathbf{\upsilon}_7$, $\mathbf{\upsilon}_8 \sim \mathcal{M}(\mathbf{\upsilon}_{2},20)$, $\mathbf{\upsilon}_9$, $\mathbf{\upsilon}_{10} \sim \mathcal{M}(\mathbf{\upsilon}_{3},20)$, and $\mathbf{\upsilon}_{11}$, $\mathbf{\upsilon}_{12} \sim \mathcal{M}(\mathbf{\upsilon}_{4},20)$. We choose the variance parameter to be 20 by trial and error such that the centers are clustered together. The simulation setup for the rest parameters are the same as in Section \ref{sec:sims}. Figure \ref{fig:nonuniformcenters} shows that with non-uniform features centers, the estimated measures and true measures are tightly correlated.

\begin{figure}[!h]
\centering
\subfloat{
\includegraphics[width=.3\textwidth]{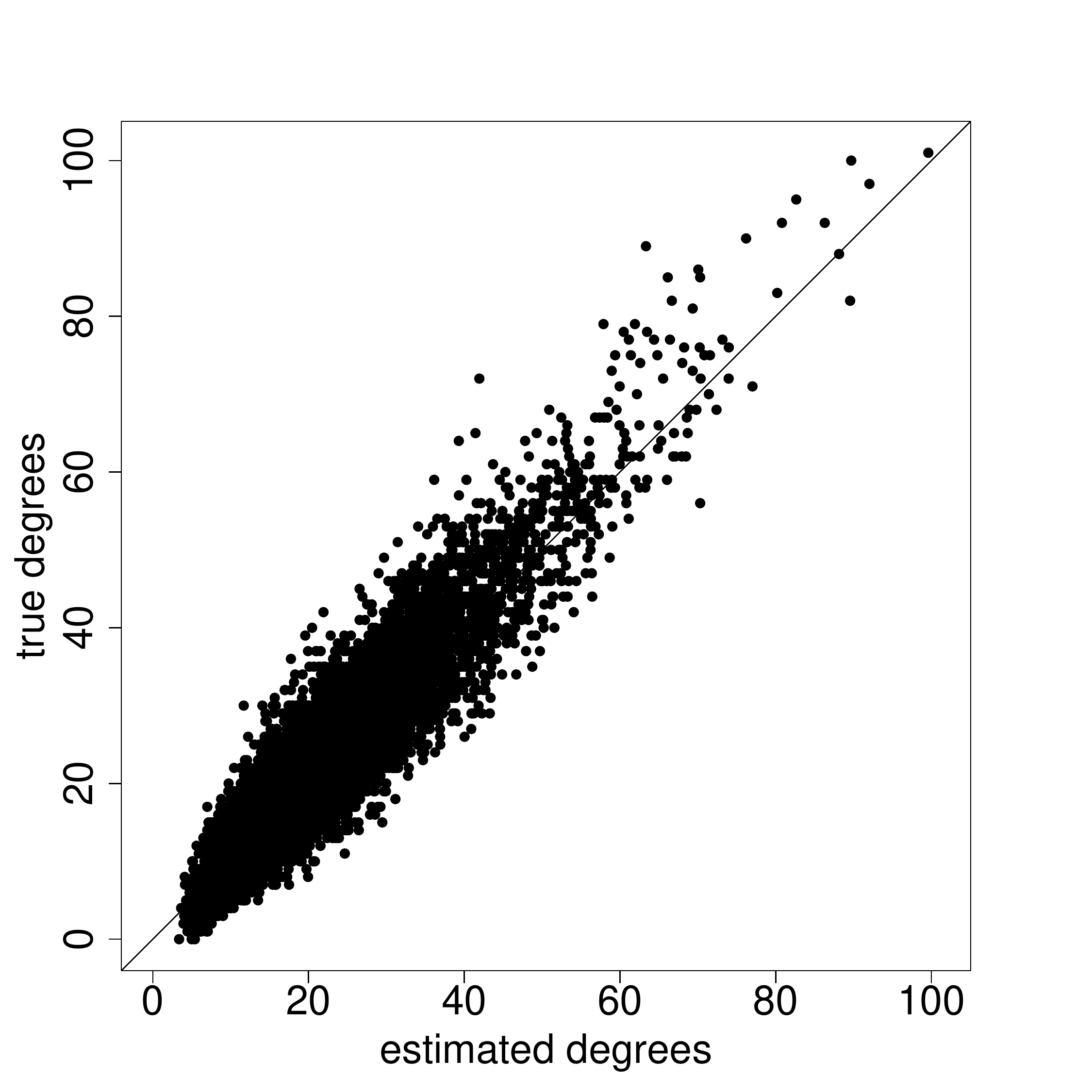}}
\subfloat{
\includegraphics[width=.3\textwidth]{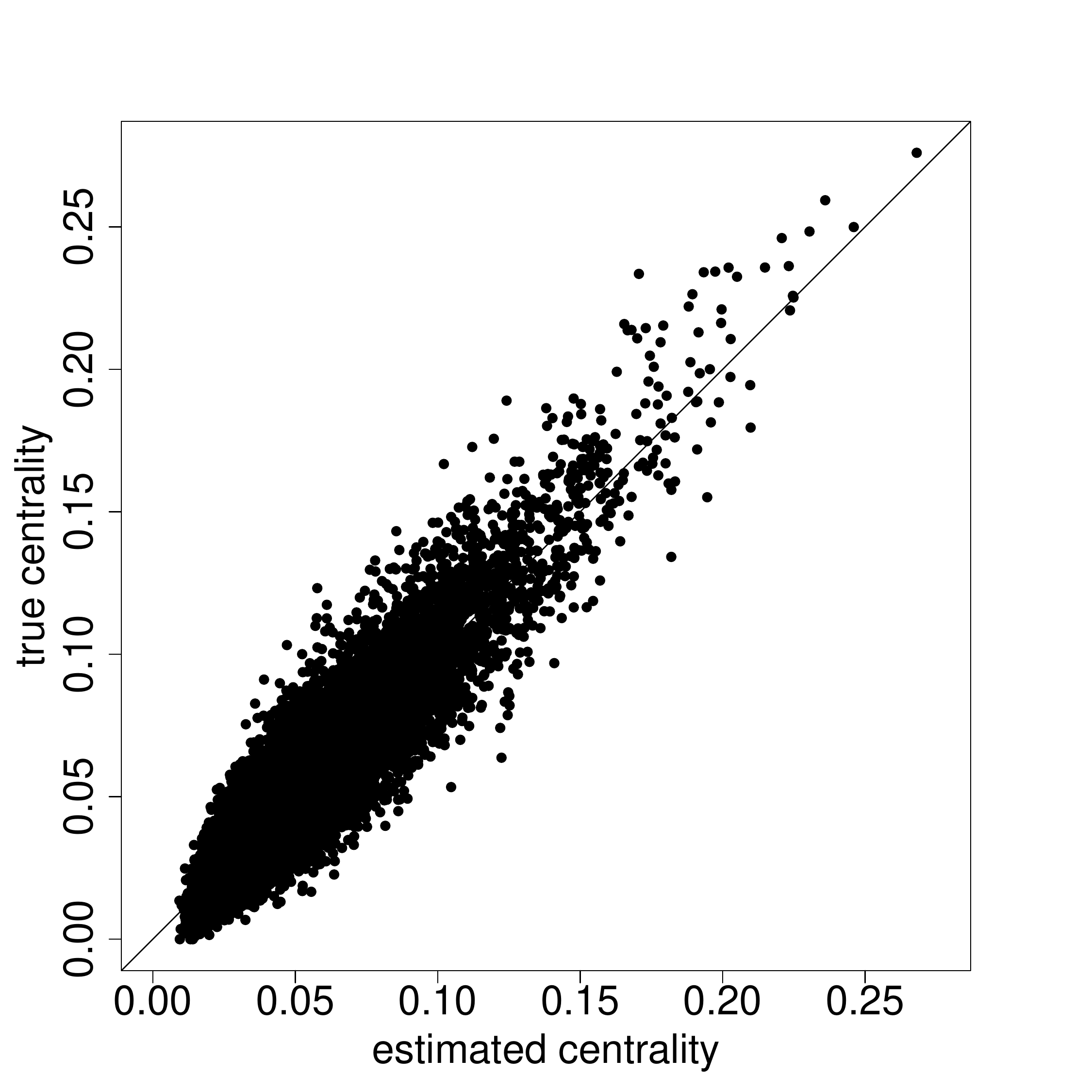}
}
\subfloat{
\includegraphics[width=.3\textwidth]{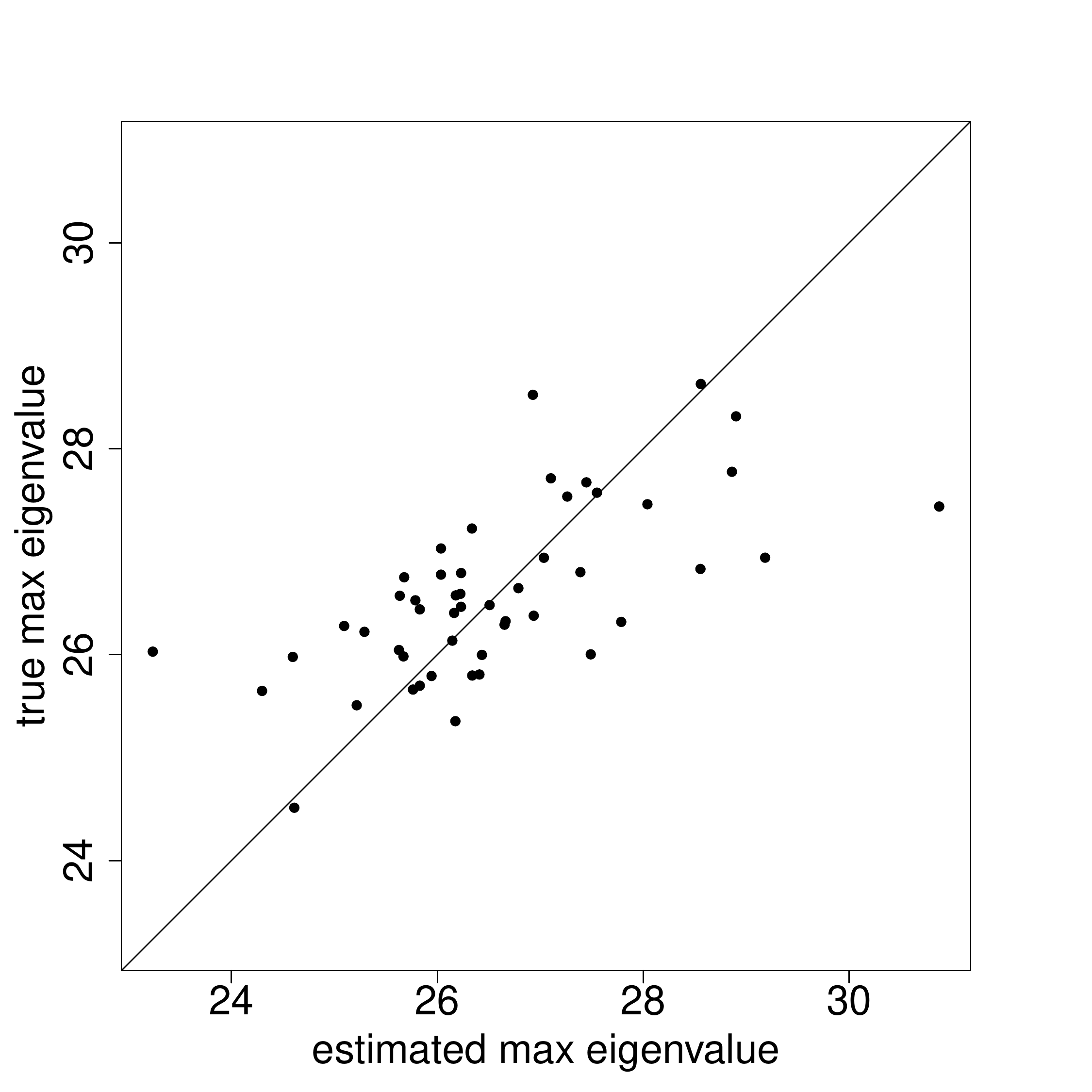}
}

\medskip

\subfloat{
\includegraphics[width=.3\textwidth]{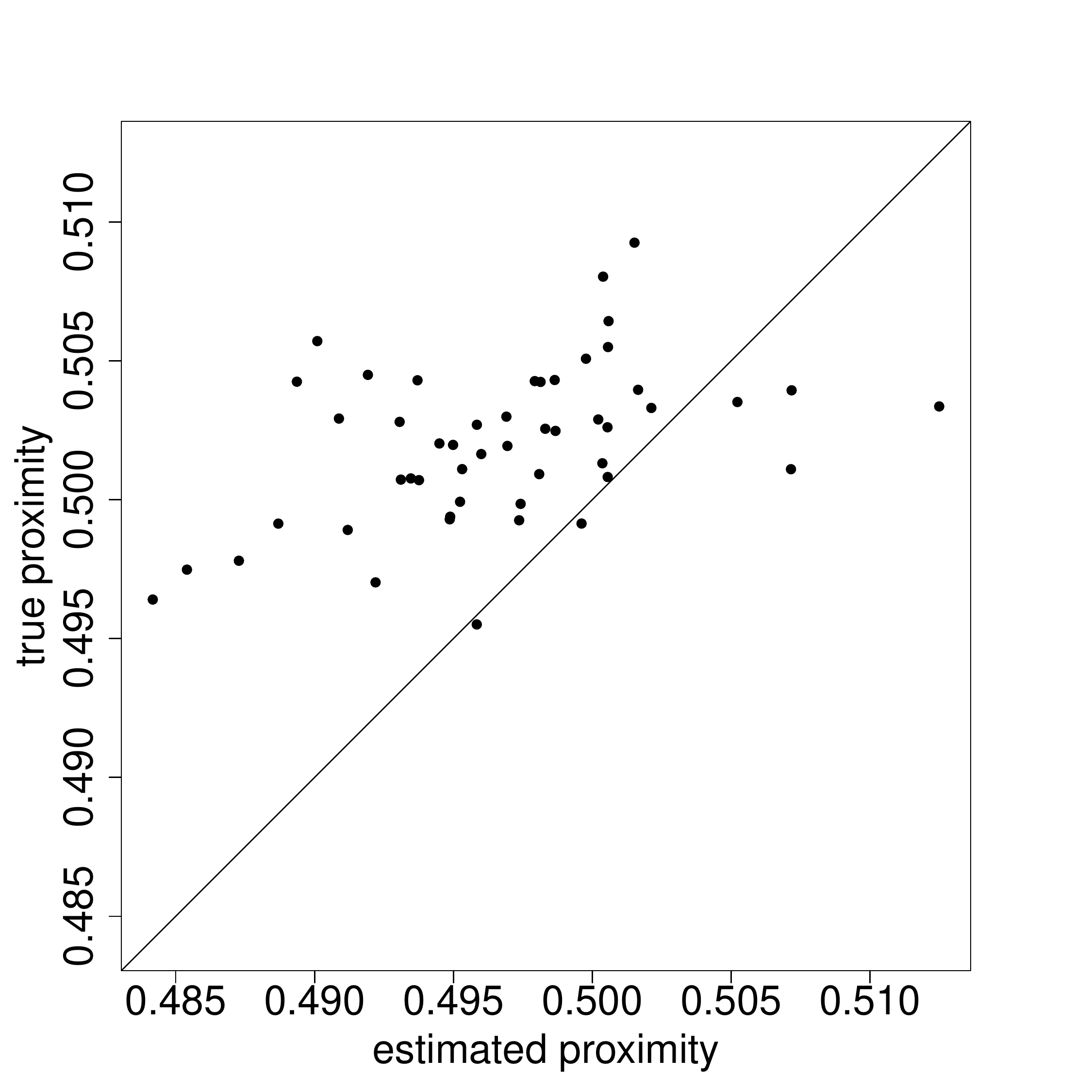}
}
\subfloat{
\includegraphics[width=.3\textwidth]{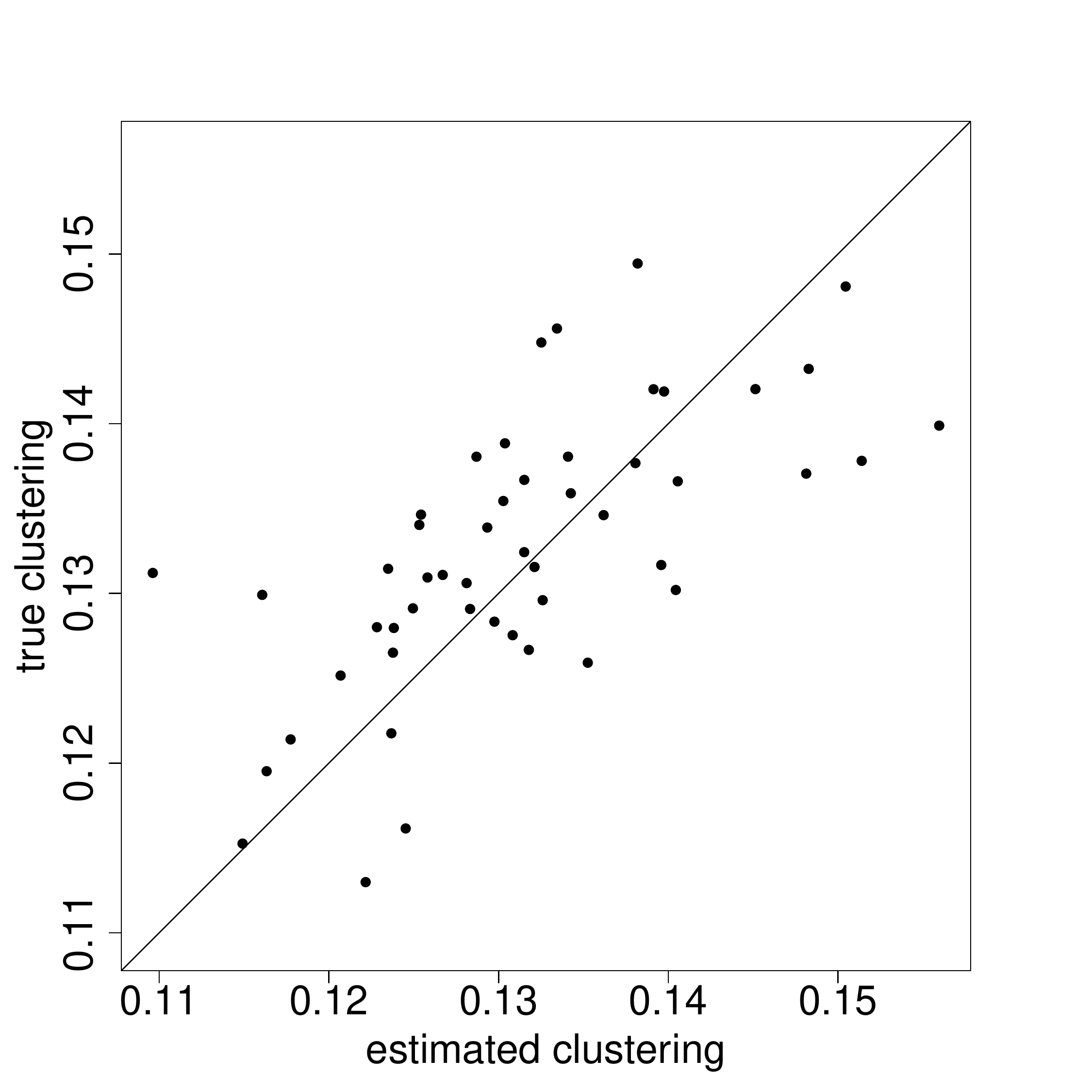}
}
\subfloat{
\includegraphics[width=.3\textwidth]{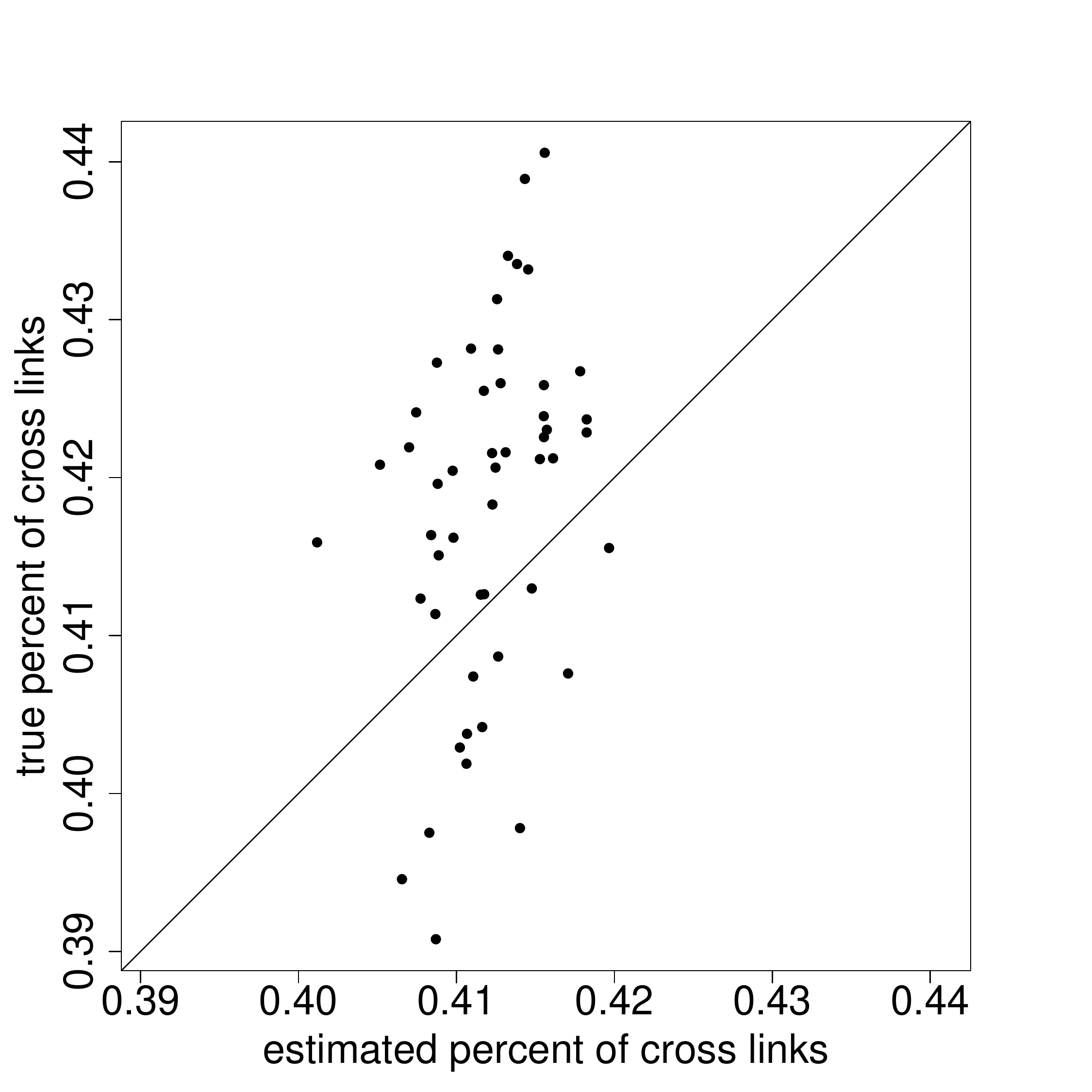}
}

\caption{Network and individual level measures estimation for 50 simulations at core simulation set-up, with the twist of simulating non-uniform feature centers. These plots show scatterplots of estimated measure on the x-axis and true measure on the y-axis. There is a strong correlation between estimated statistic and statistic obtained from the true underlying graph, with the exception of eigenvector cut. The results are similar to ones when simulating 12 feature centers uniformly.}
\label{fig:nonuniformcenters}
\end{figure}

\clearpage
\section{Scatterplots on additional measures for core simulation}
\label{sec:additional_plots_core}
\setcounter{figure}{0}
\renewcommand{\thefigure}{H.\arabic{figure}}

\begin{figure}[!h]
\centering
\subfloat[Closeness]{
\includegraphics[width=.3\textwidth]{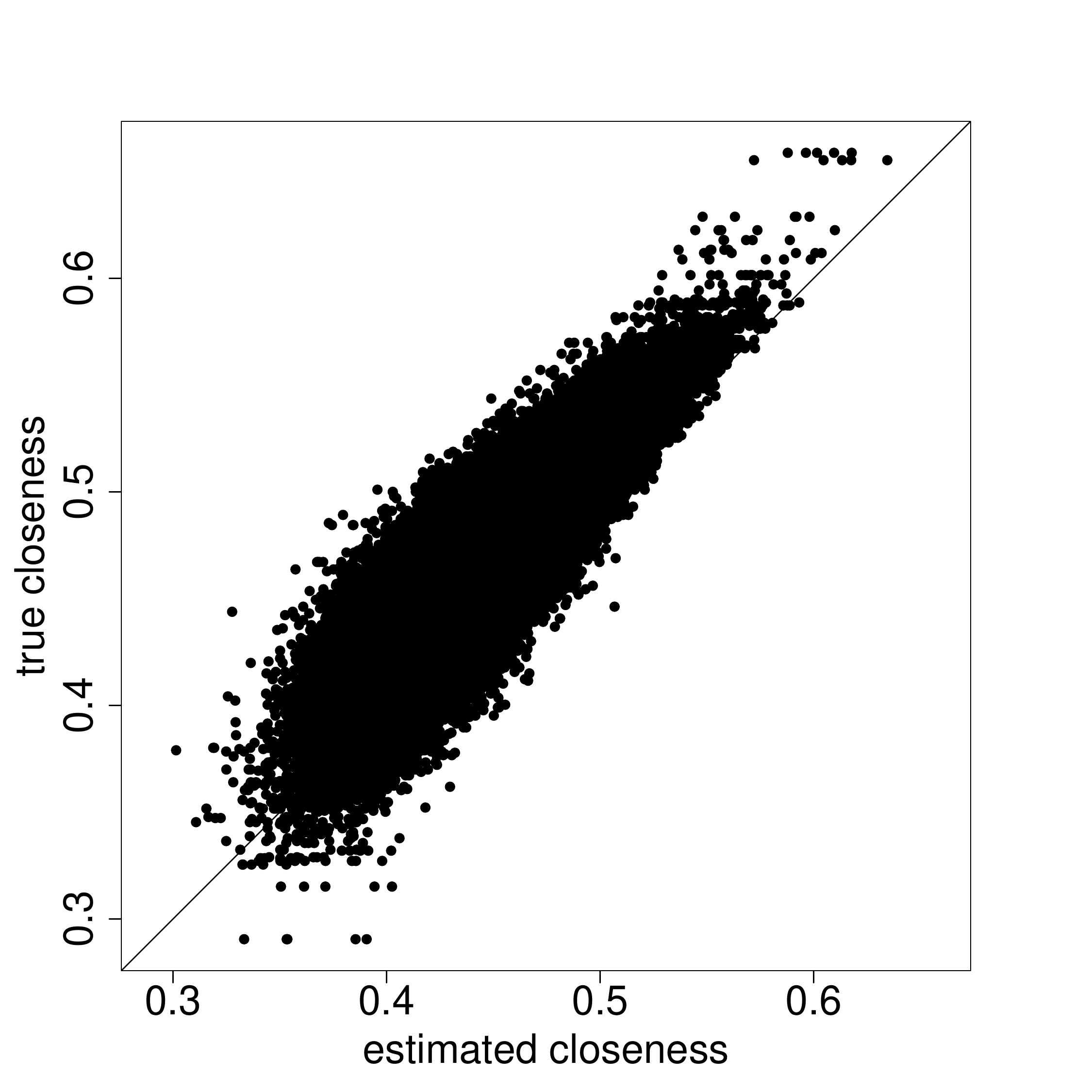}}
\subfloat[Betweenness]{
\includegraphics[width=.3\textwidth]{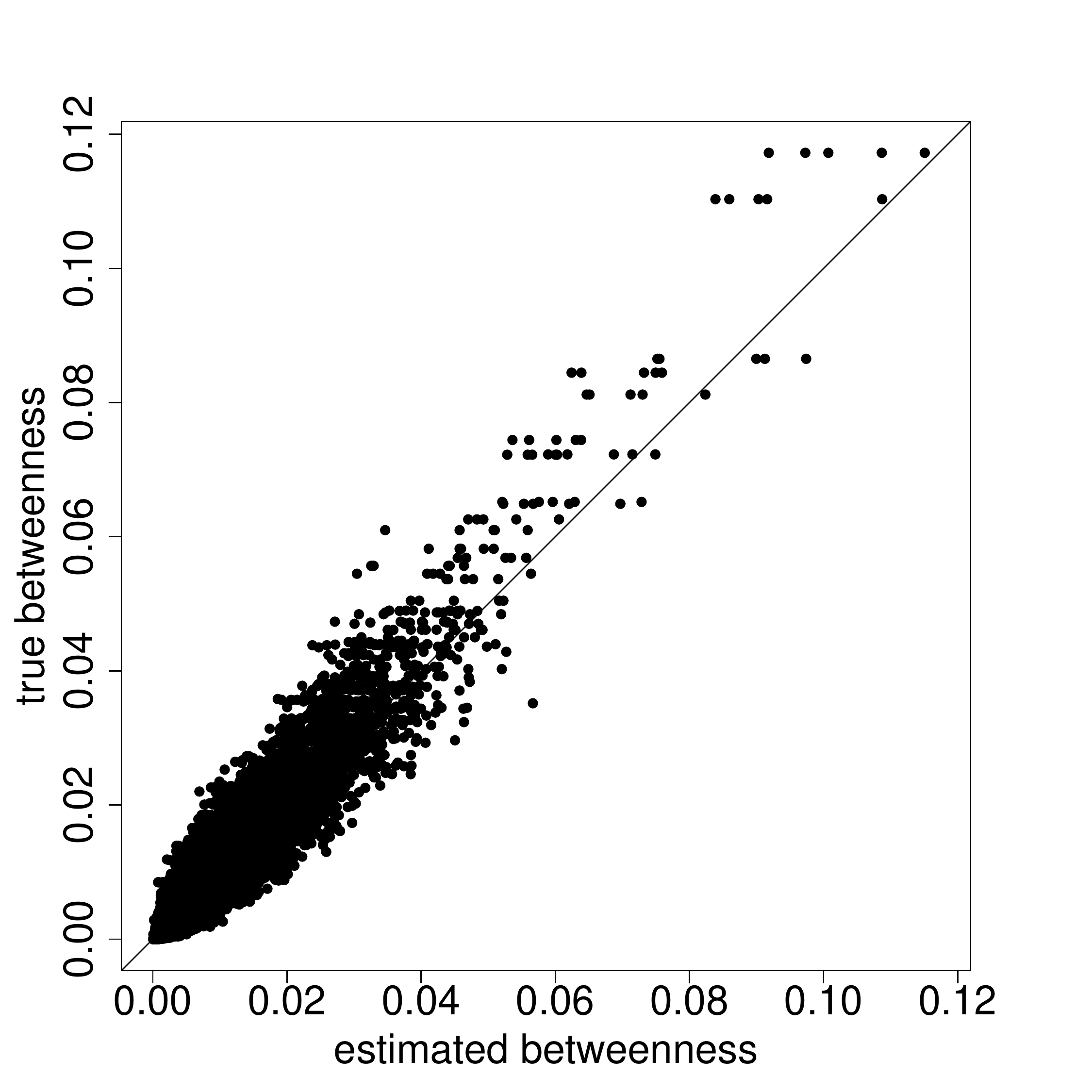}
}
\subfloat[Support]{
\includegraphics[width=.3\textwidth]{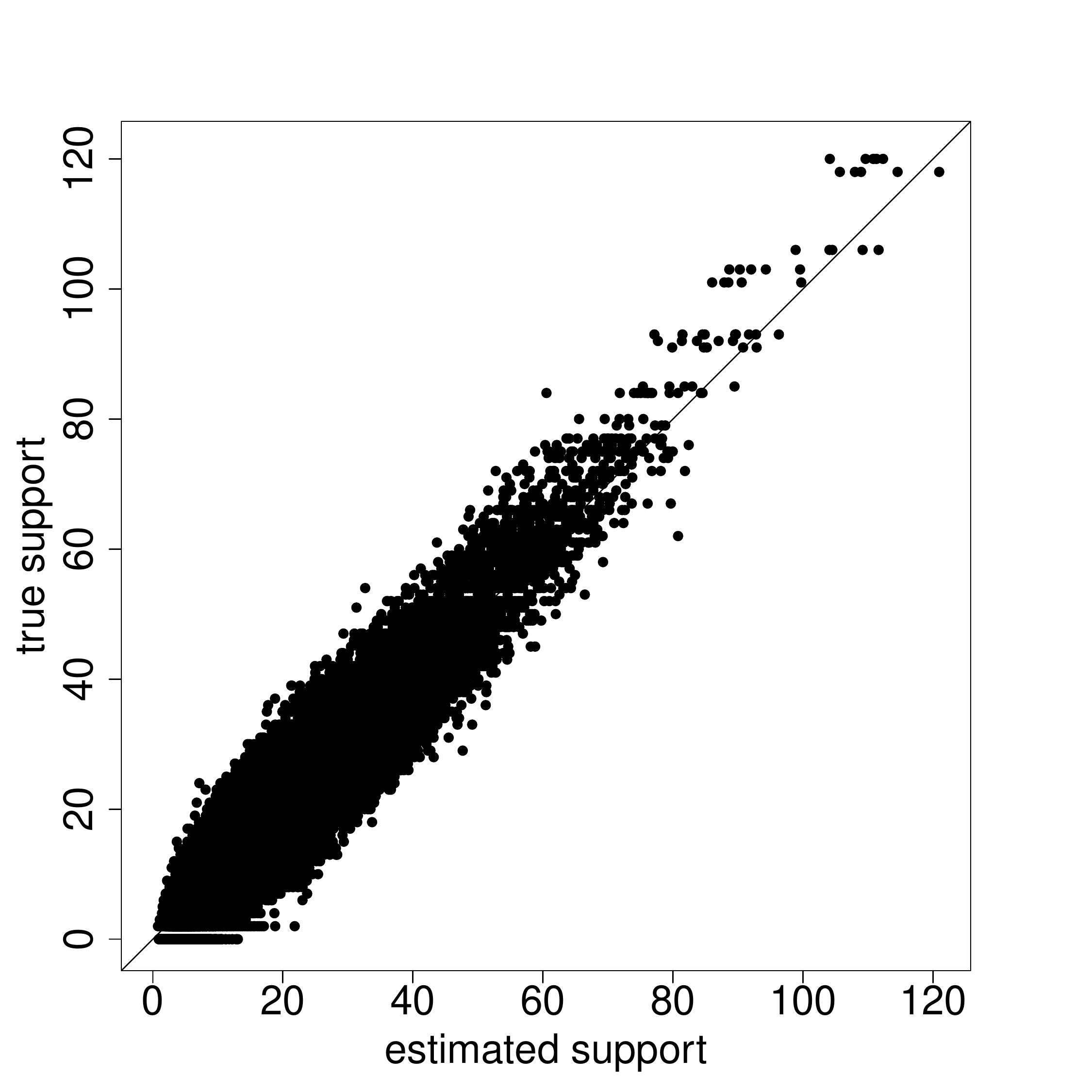}
}

\smallskip
\subfloat[Distance from seed]{
\includegraphics[width=.3\textwidth]{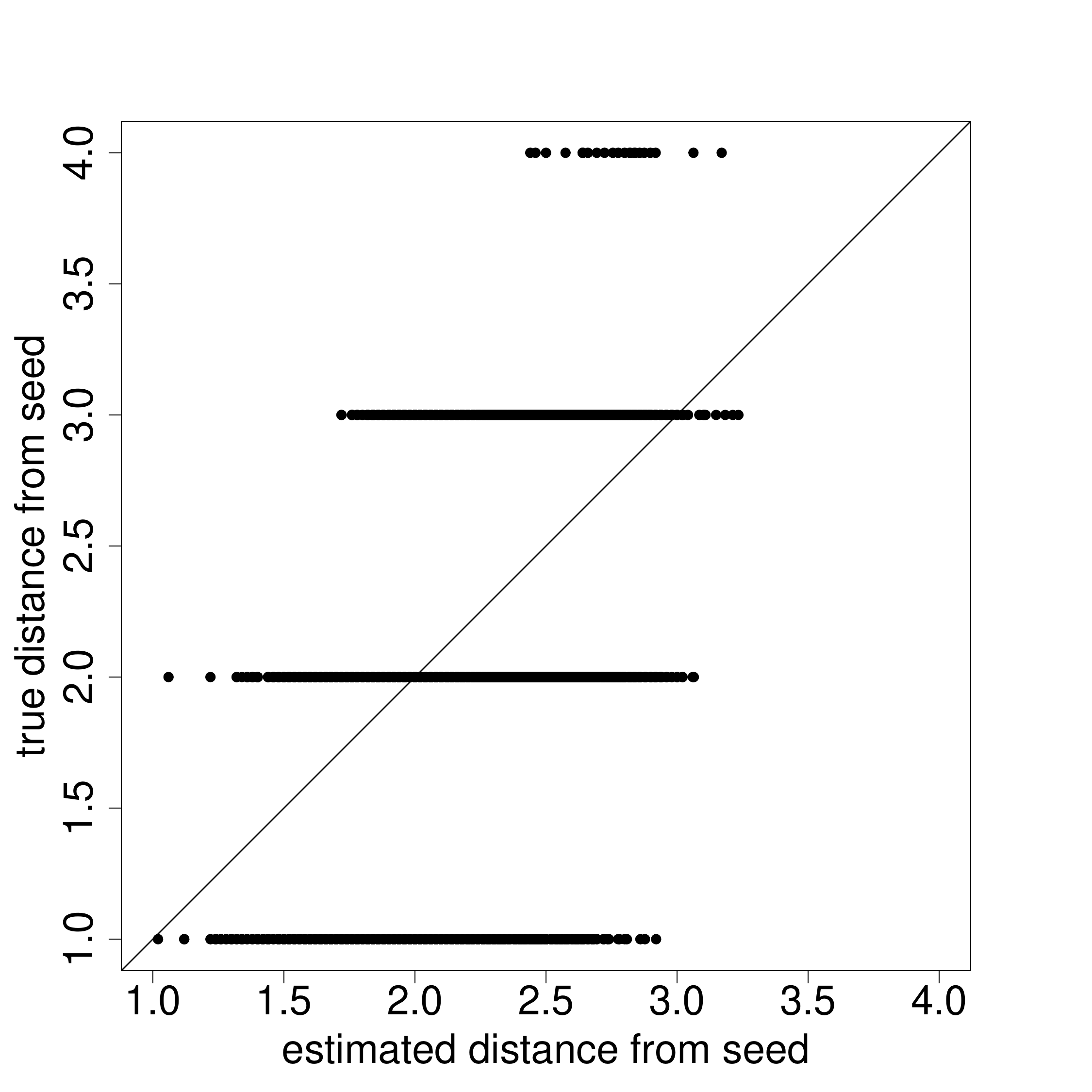}}
\subfloat[Node level clustering]{
\includegraphics[width=.3\textwidth]{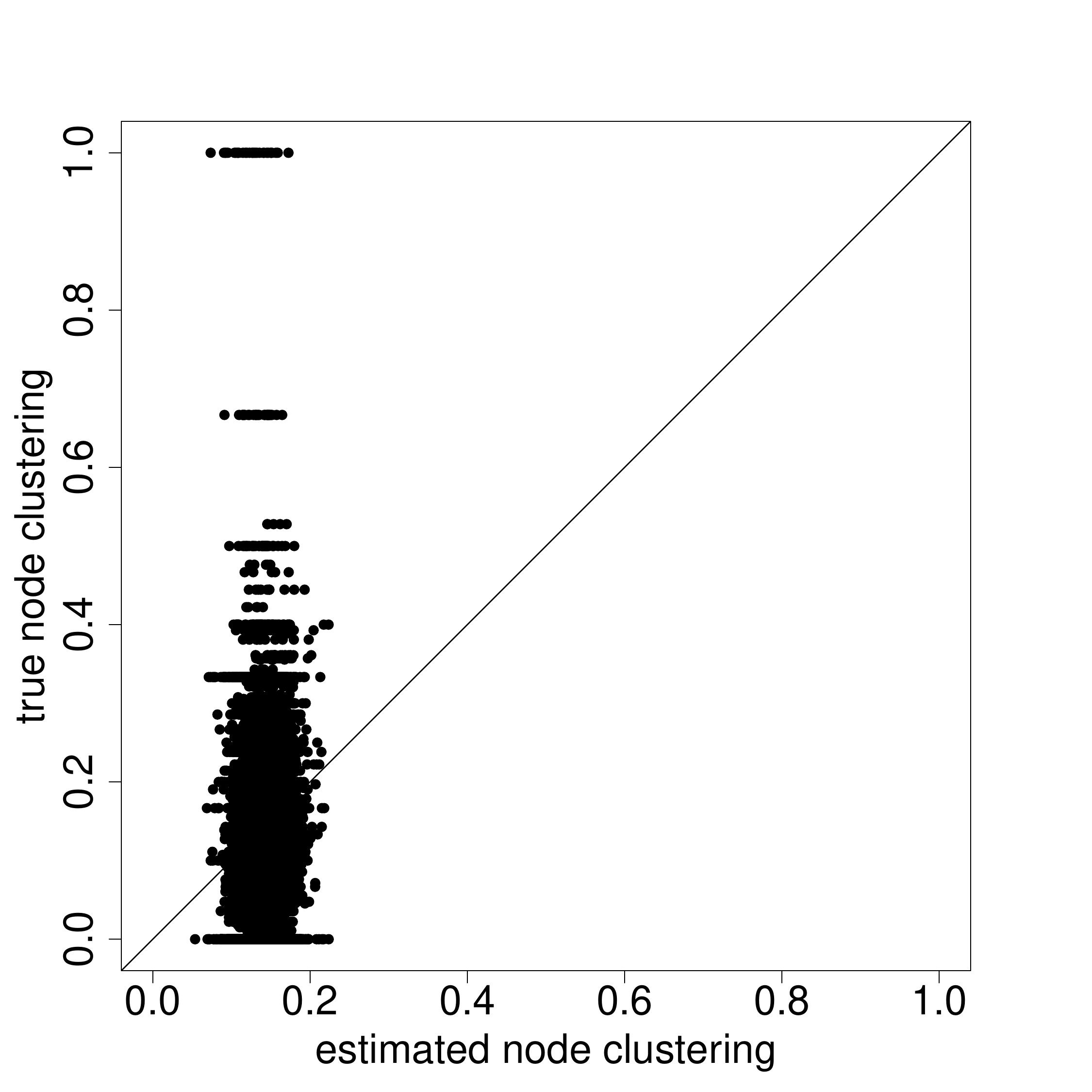}}
\subfloat[Treated neighborhood share]{
\includegraphics[width=.3\textwidth]{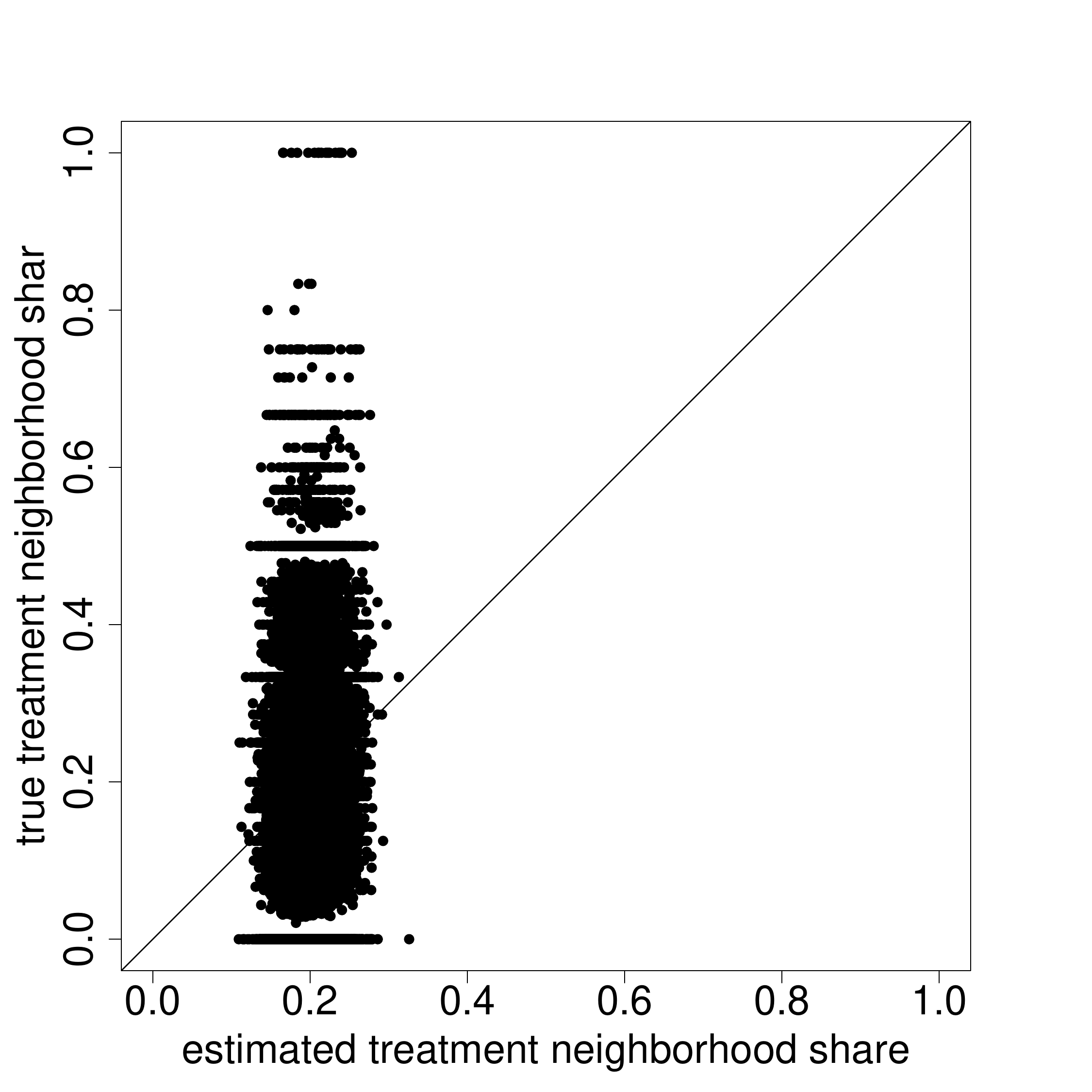}}

\caption{Node level measures estimation for 250 simulations at core simulation set-up. These plots show scatterplots of estimated measure on the x-axis and true measure on the y-axis. We remove the nodes where the simulated true graph is not fully connected. For betweenness, closeness, and support, there is a strong correlation between estimated statistic and statistic obtained from the true underlying graph. The weak correlation in the node-level clustering measure is an artifact of weak clustering in underlying ``true'' model. }
\label{fig:core_sims_node_appendix}
\end{figure}

\begin{figure}[!h]
\centering
\subfloat[Number of components]{
\includegraphics[width=.3\textwidth]{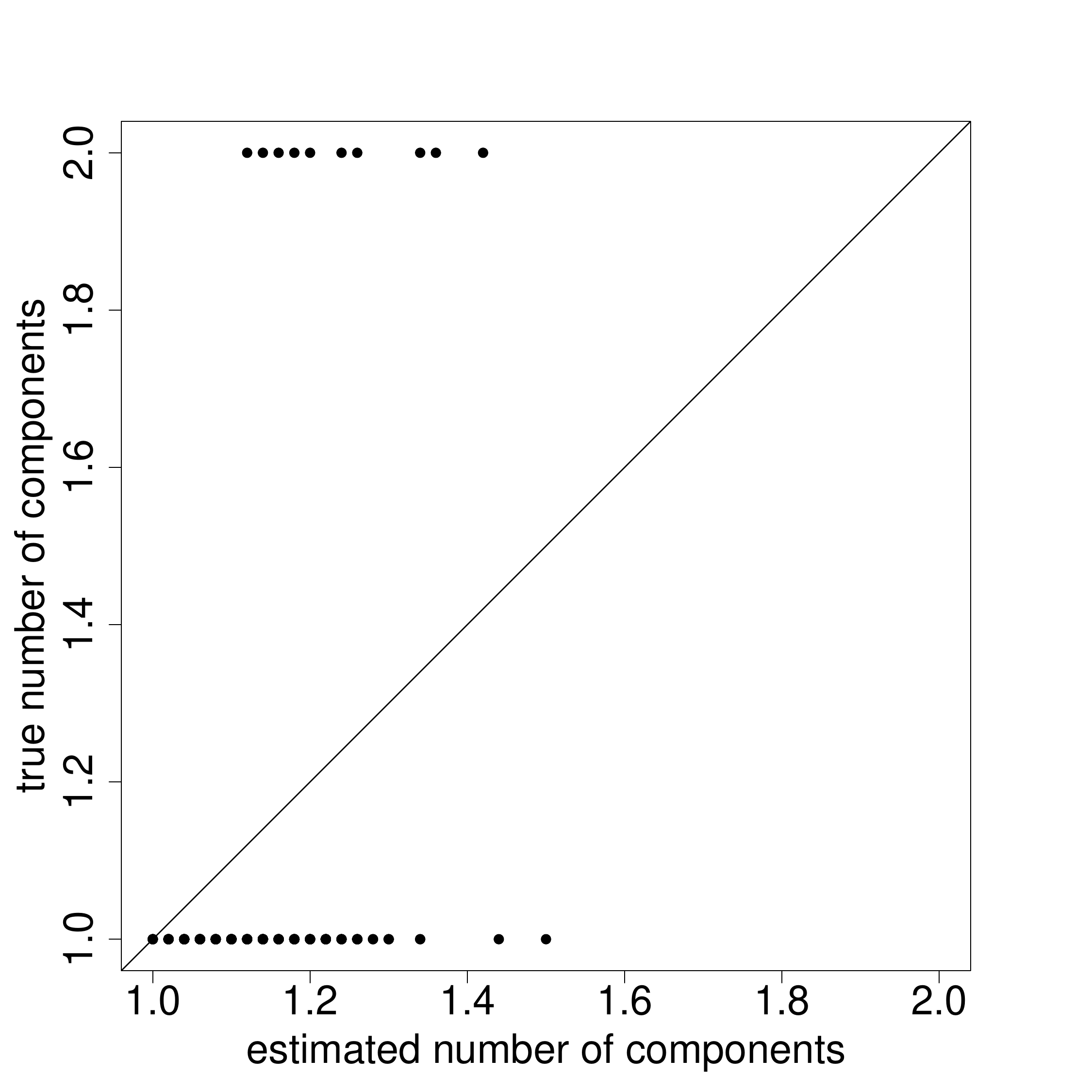}}
\subfloat[Average path length]{
\includegraphics[width=.3\textwidth]{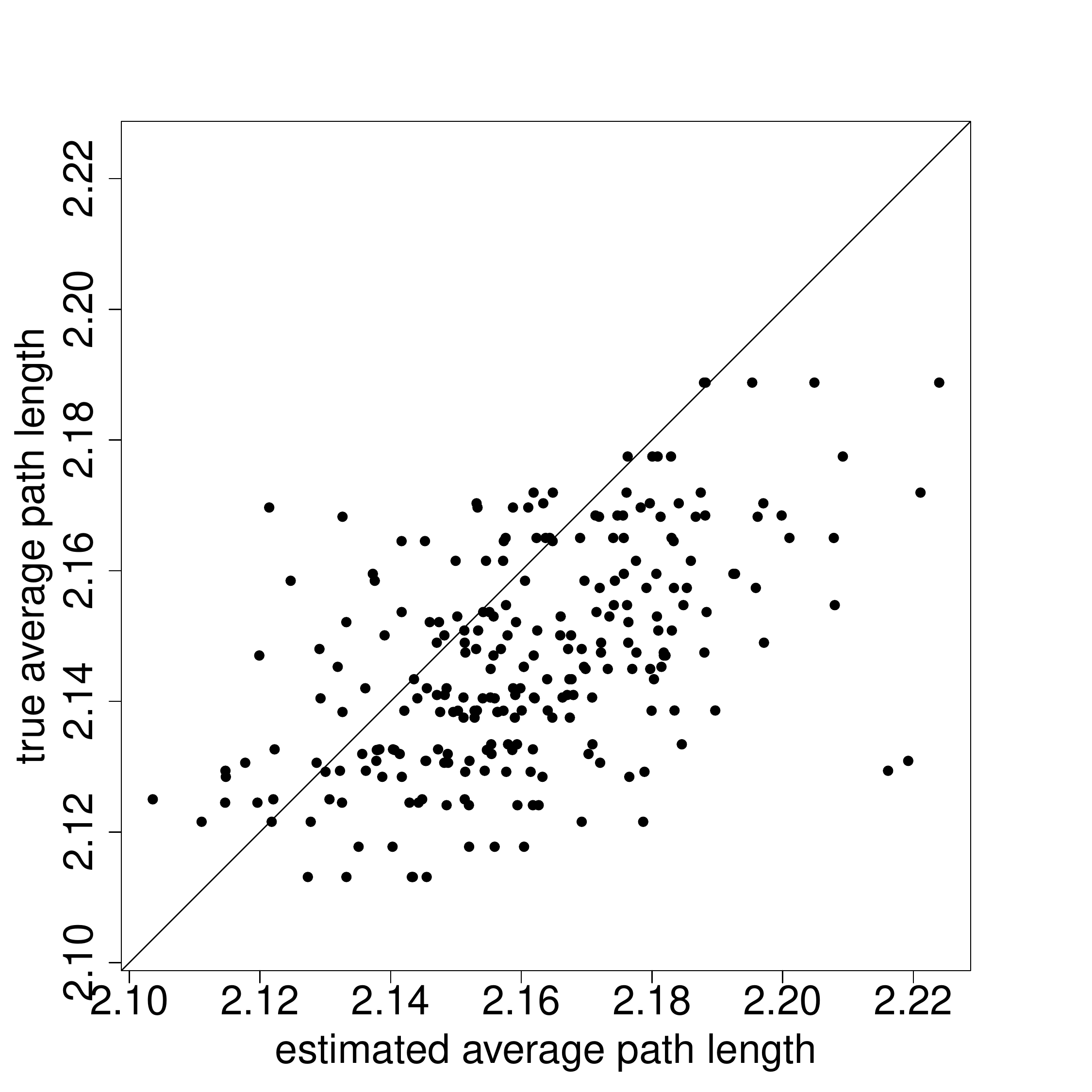}}

\smallskip
\subfloat[Diameter]{
\includegraphics[width=.3\textwidth]{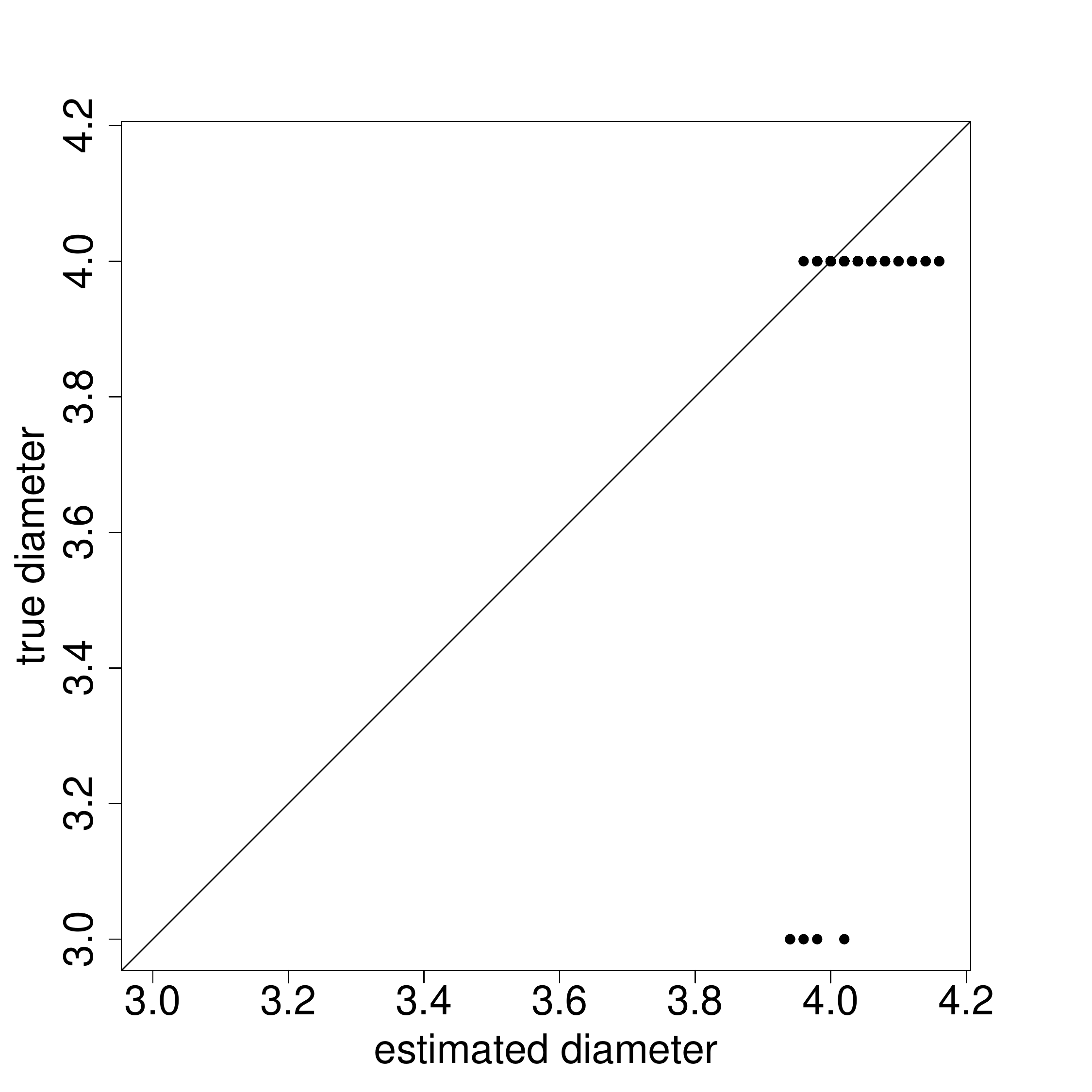}}
\subfloat[Fraction of giant component]{
\includegraphics[width=.3\textwidth]{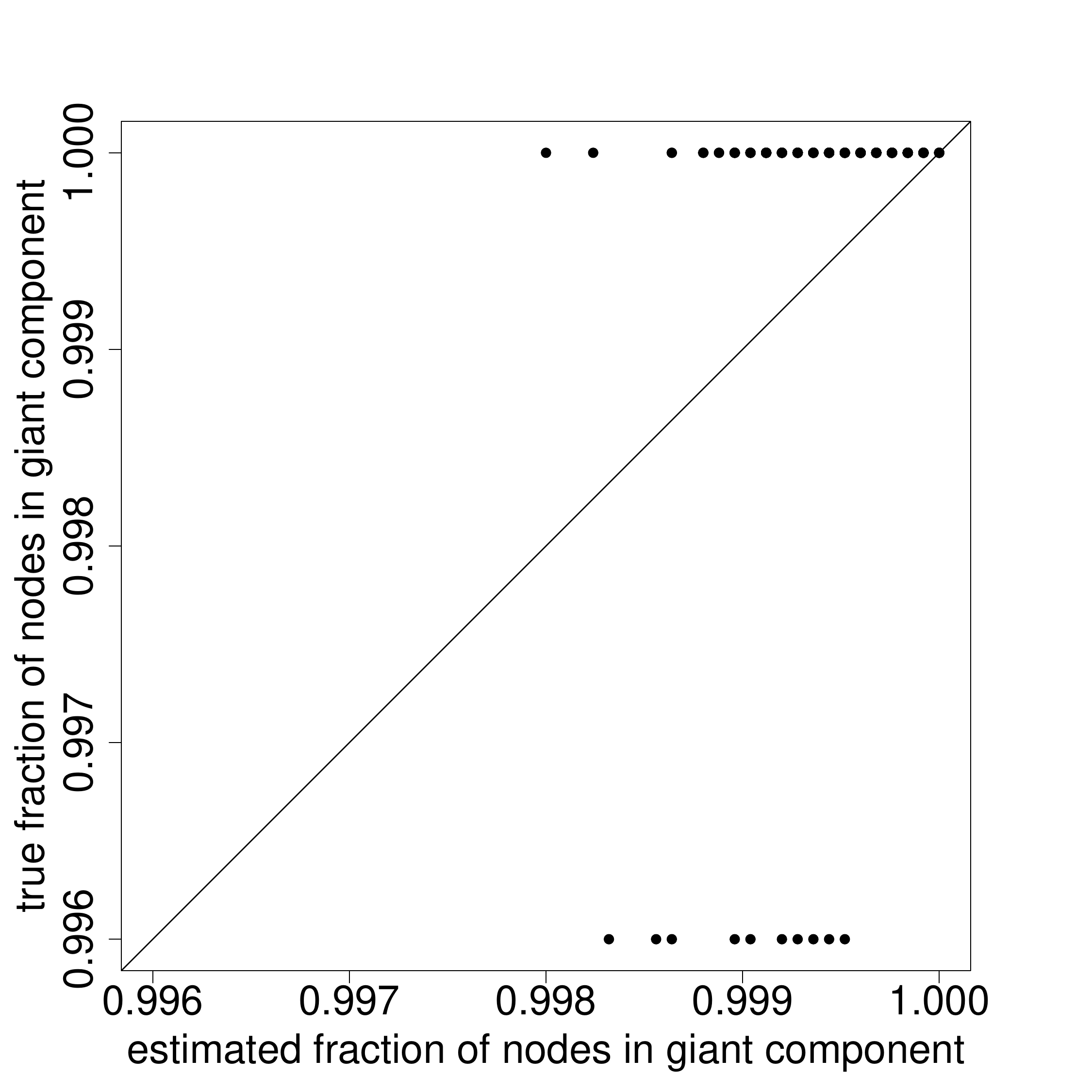}}

\caption{Network level measures estimation for 250 simulations at core simulation set-up. These plots show scatterplots of estimated measure on the x-axis and true measure on the y-axis. For not fully connected graph, diameter is the diameter of the giant component, and average path length is taken over all finite path lengths. For average path length, there is a strong correlation between estimated statistic and statistic obtained from the true underlying graph. For all other measures, the weakness comes from the fact that there is not much variation in the true measure based on our sampling strategy.}
\label{fig:core_sims_network_appendix}
\end{figure}

\clearpage
\section{Scatterplots on additional measures for Karnataka villages}
\label{sec:additional_plots_karnataka}
\setcounter{figure}{0}
\renewcommand{\thefigure}{I.\arabic{figure}}

\begin{figure}[!h]
\centering
\subfloat[Closeness]{
\includegraphics[width=.3\textwidth]{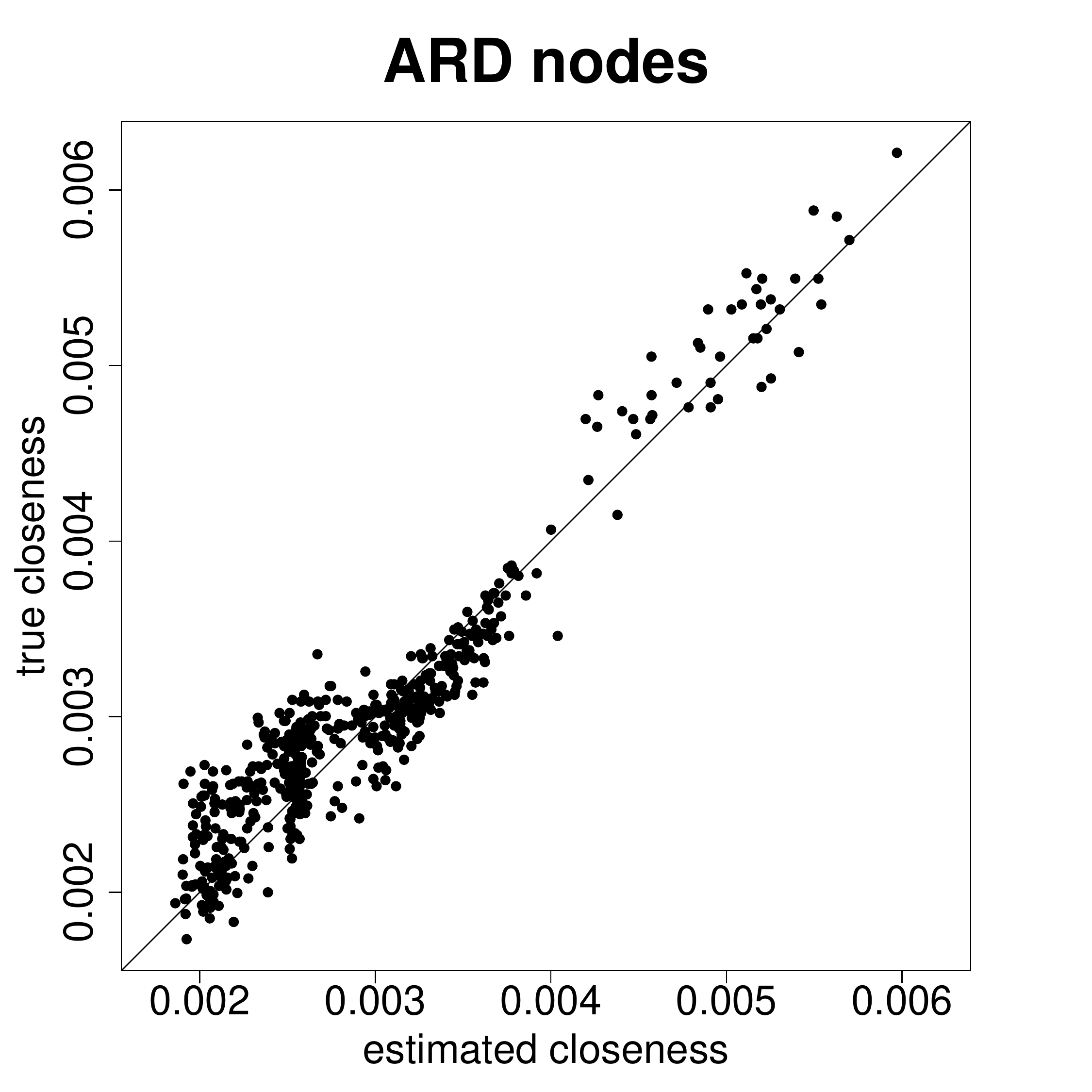}}
\subfloat[Betweenness]{
\includegraphics[width=.3\textwidth]{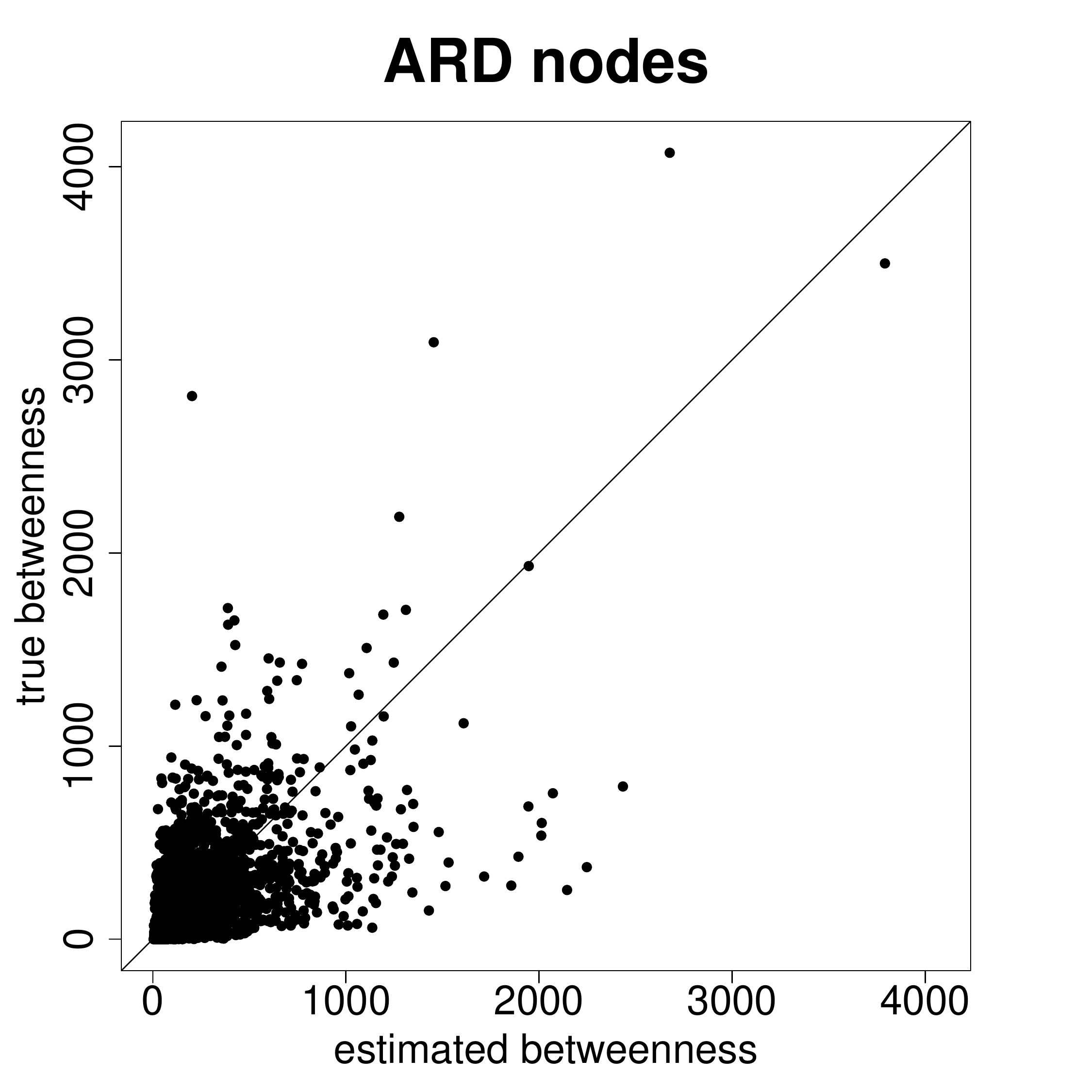}
}
\subfloat[Support]{
\includegraphics[width=.3\textwidth]{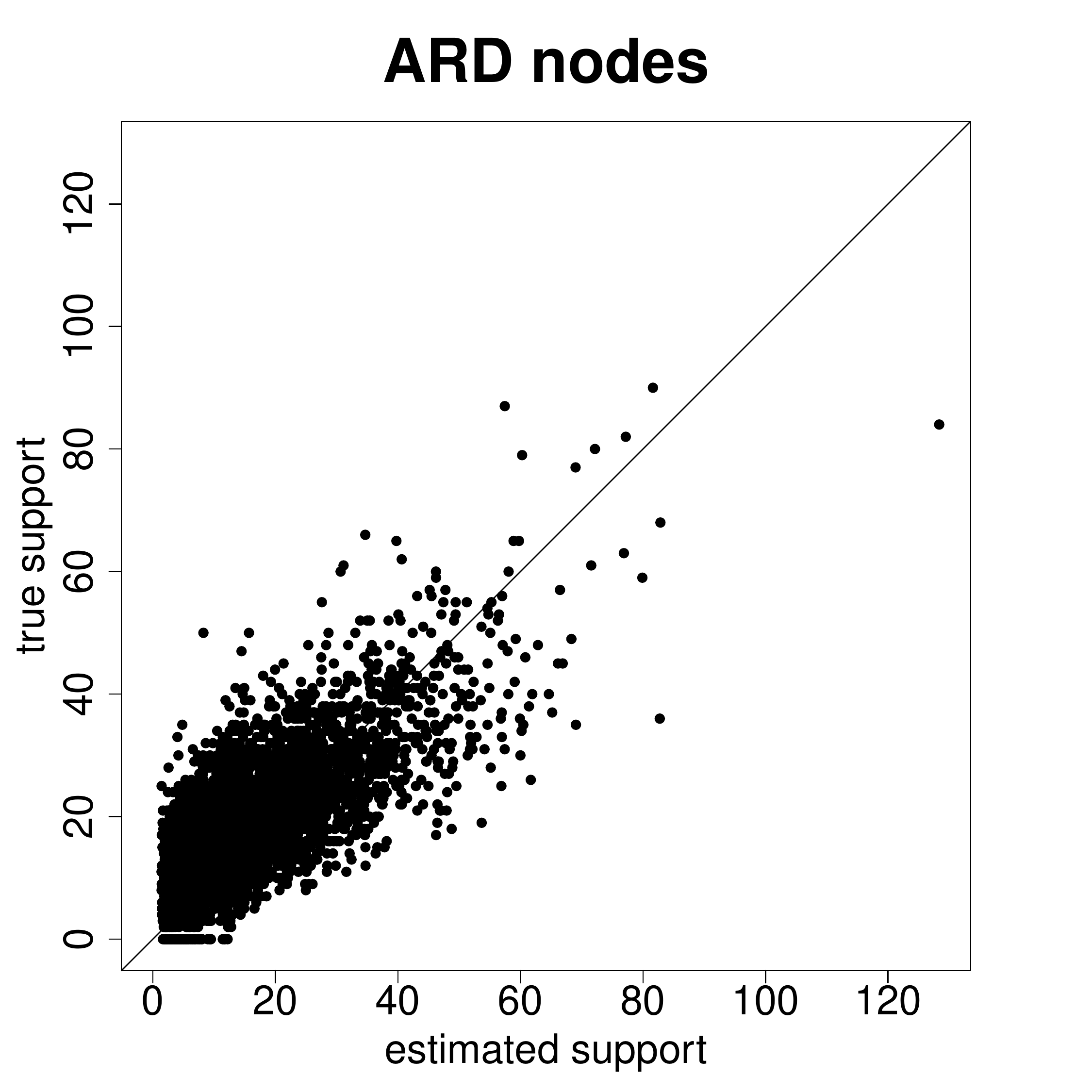}
}

\smallskip
\subfloat[Distance from seed]{
\includegraphics[width=.3\textwidth]{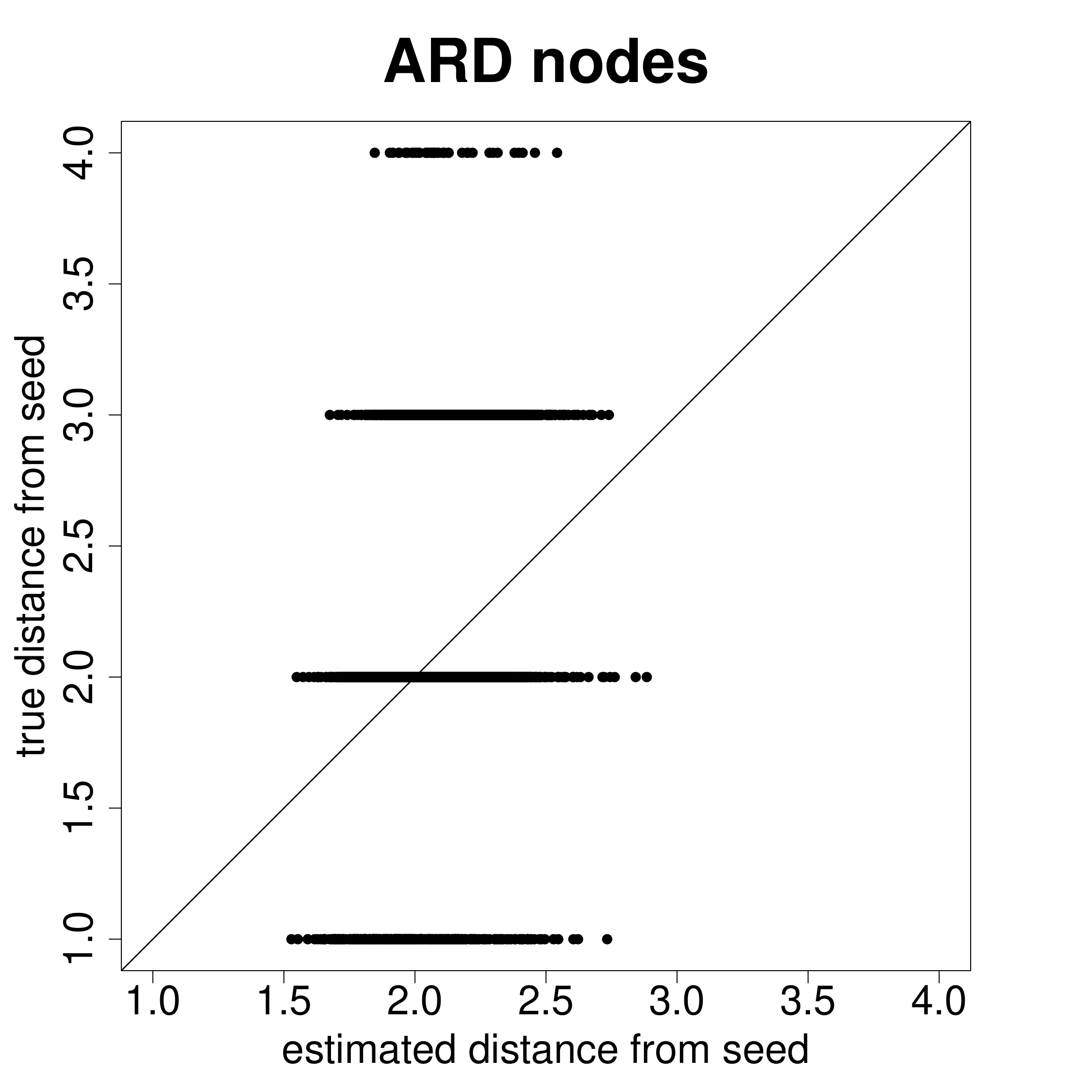}}
\subfloat[Node level clustering]{
\includegraphics[width=.3\textwidth]{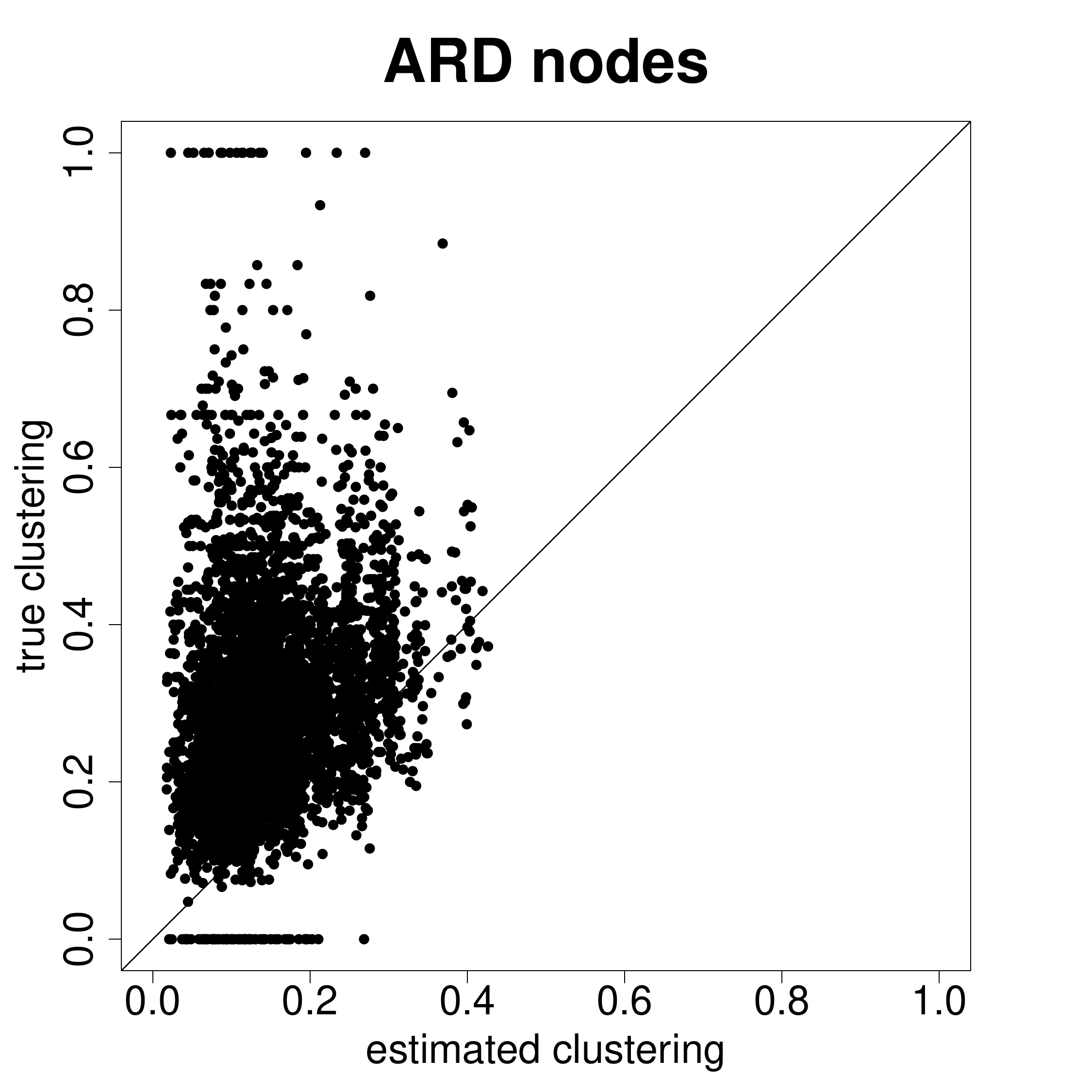}}
\subfloat[Treated neighborhood share]{
\includegraphics[width=.3\textwidth]{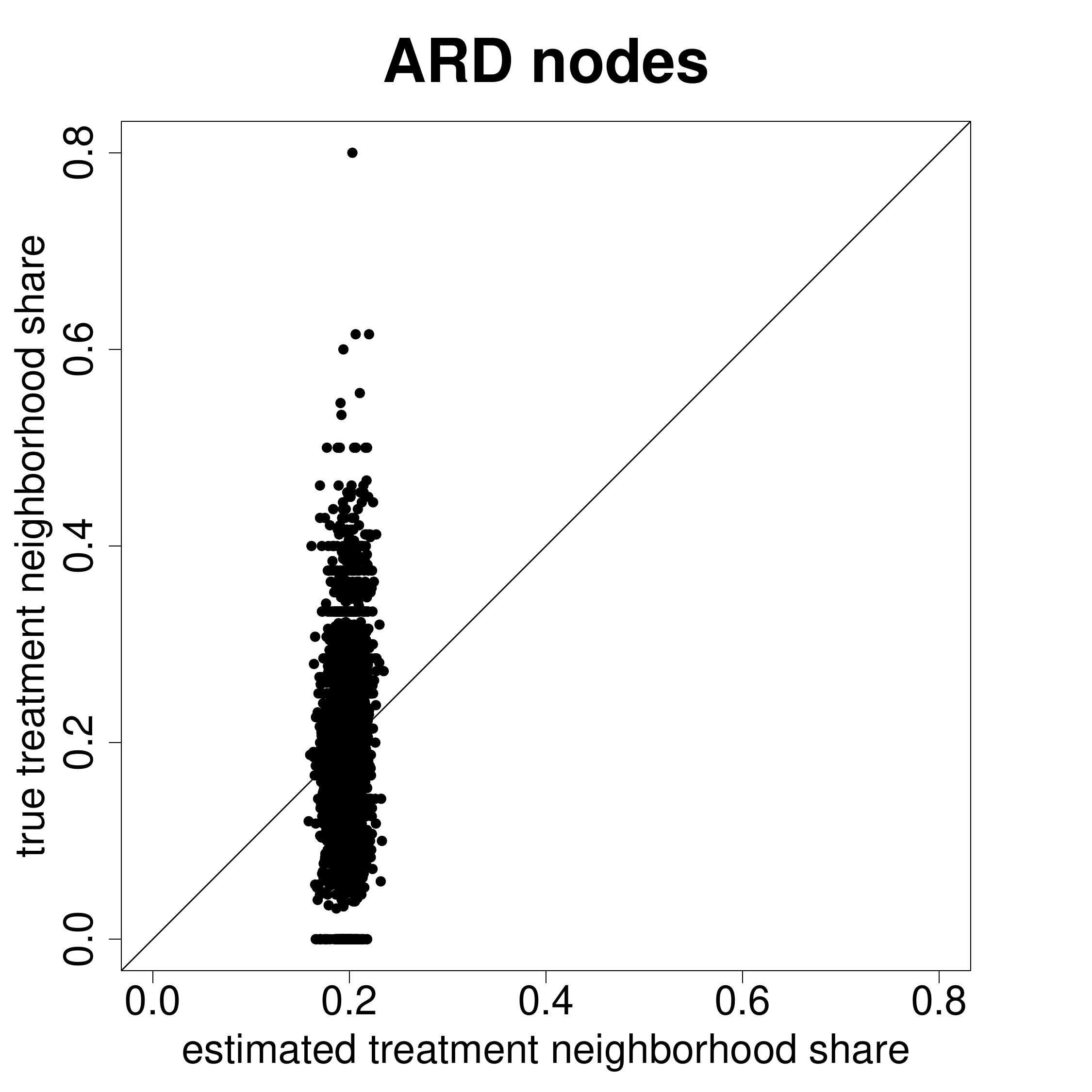}}

\caption{Node level measures estimation for households with ARD response in villages in Karnataka.  These plots show scatterplots across all villages with the estimated node level measure on the x-axis and the measure from the true underlying graph on the y-axis.}
\label{fig:karnataka_ARD_node_appendix}
\end{figure}

\begin{figure}[!h]
\centering
\subfloat[Closeness]{
\includegraphics[width=.3\textwidth]{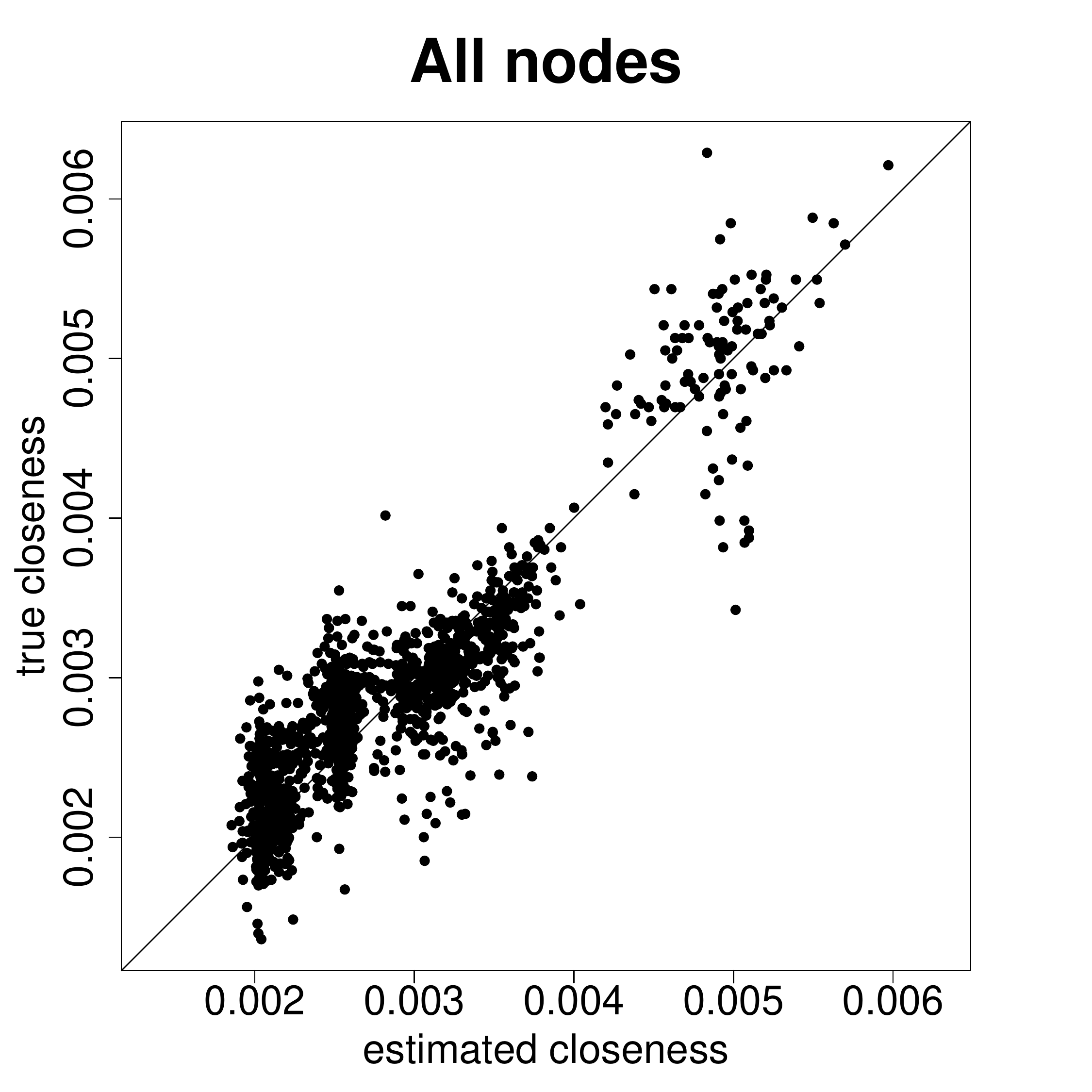}}
\subfloat[Betweenness]{
\includegraphics[width=.3\textwidth]{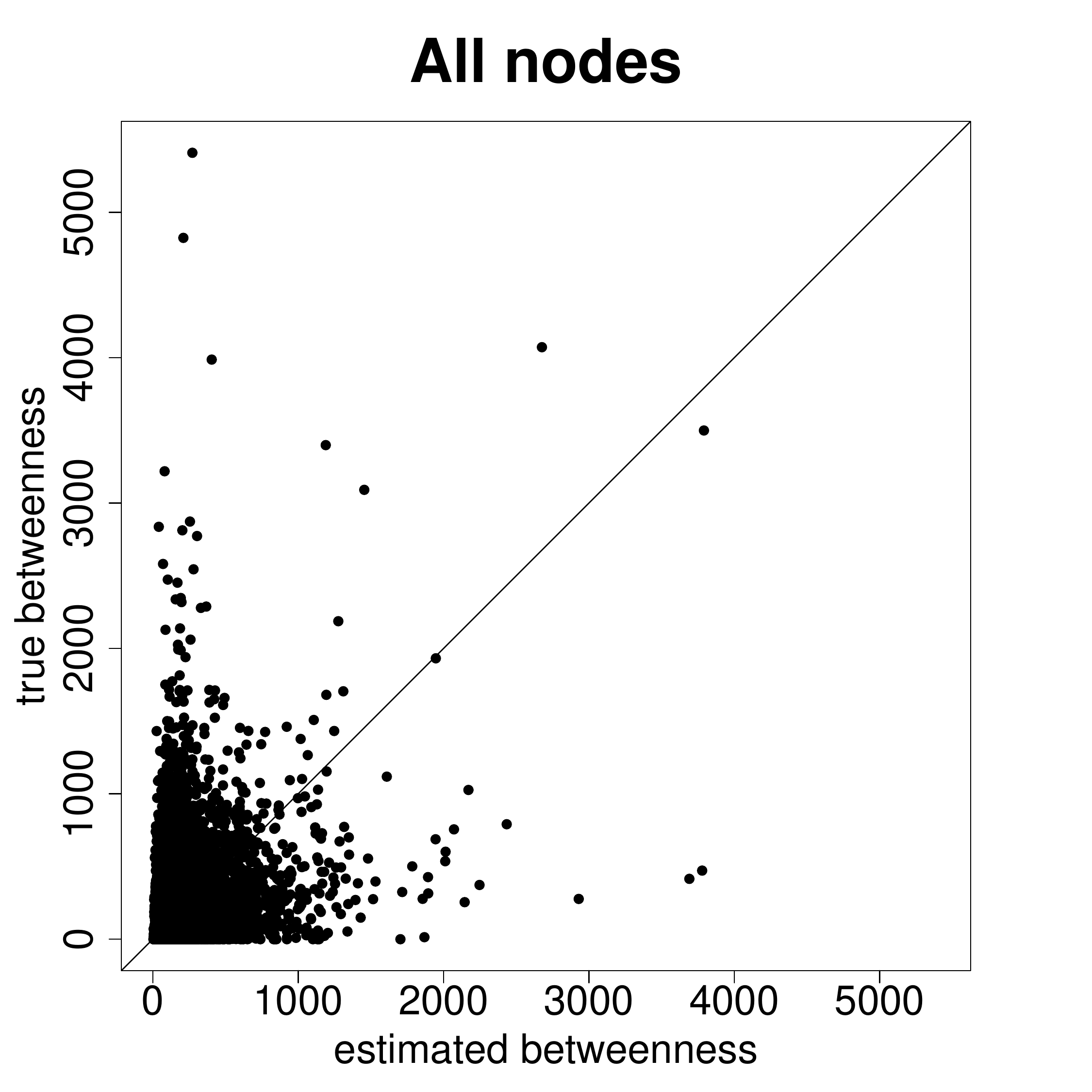}
}
\subfloat[Support]{
\includegraphics[width=.3\textwidth]{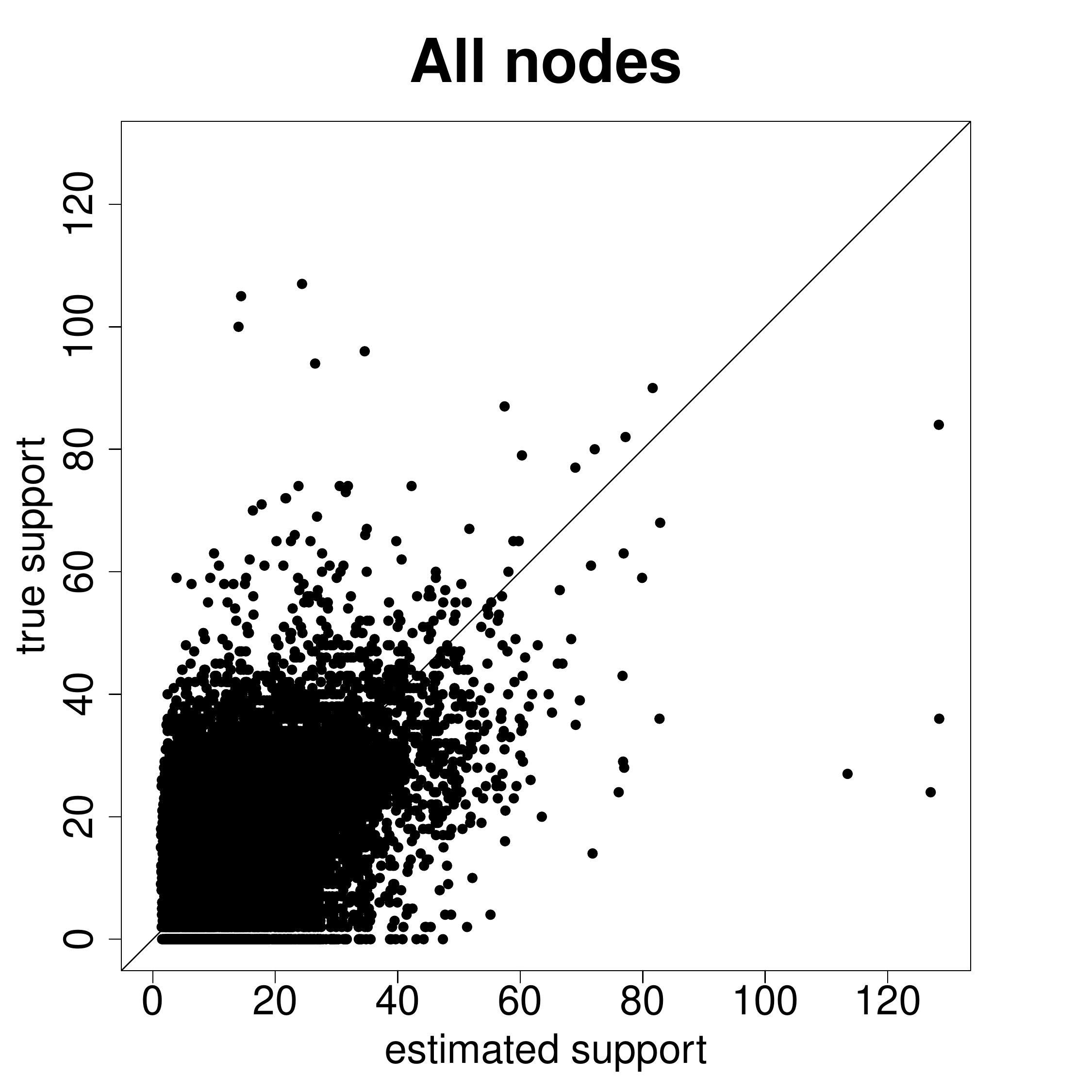}
}

\smallskip
\subfloat[Distance from seed]{
\includegraphics[width=.3\textwidth]{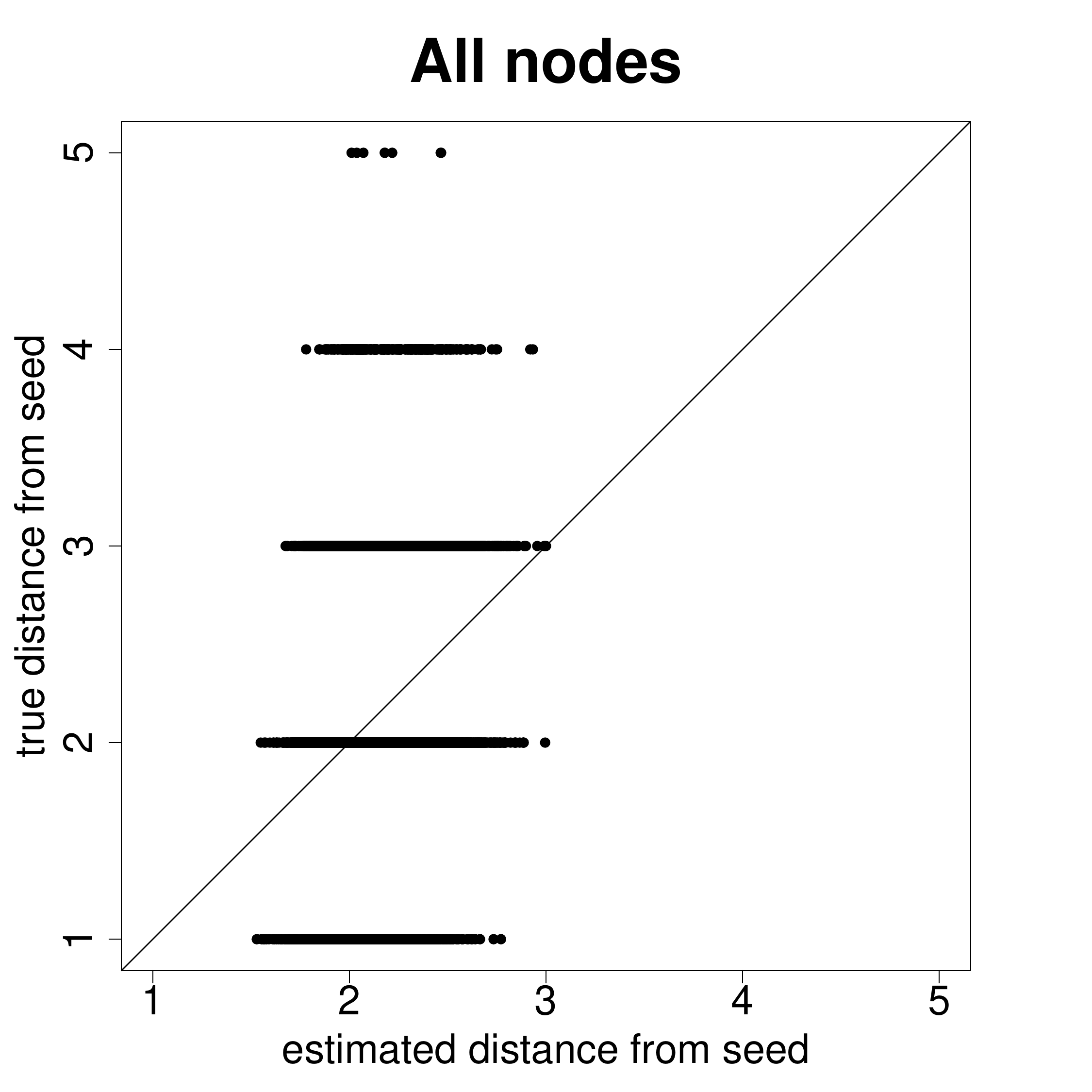}}
\subfloat[Node level clustering]{
\includegraphics[width=.3\textwidth]{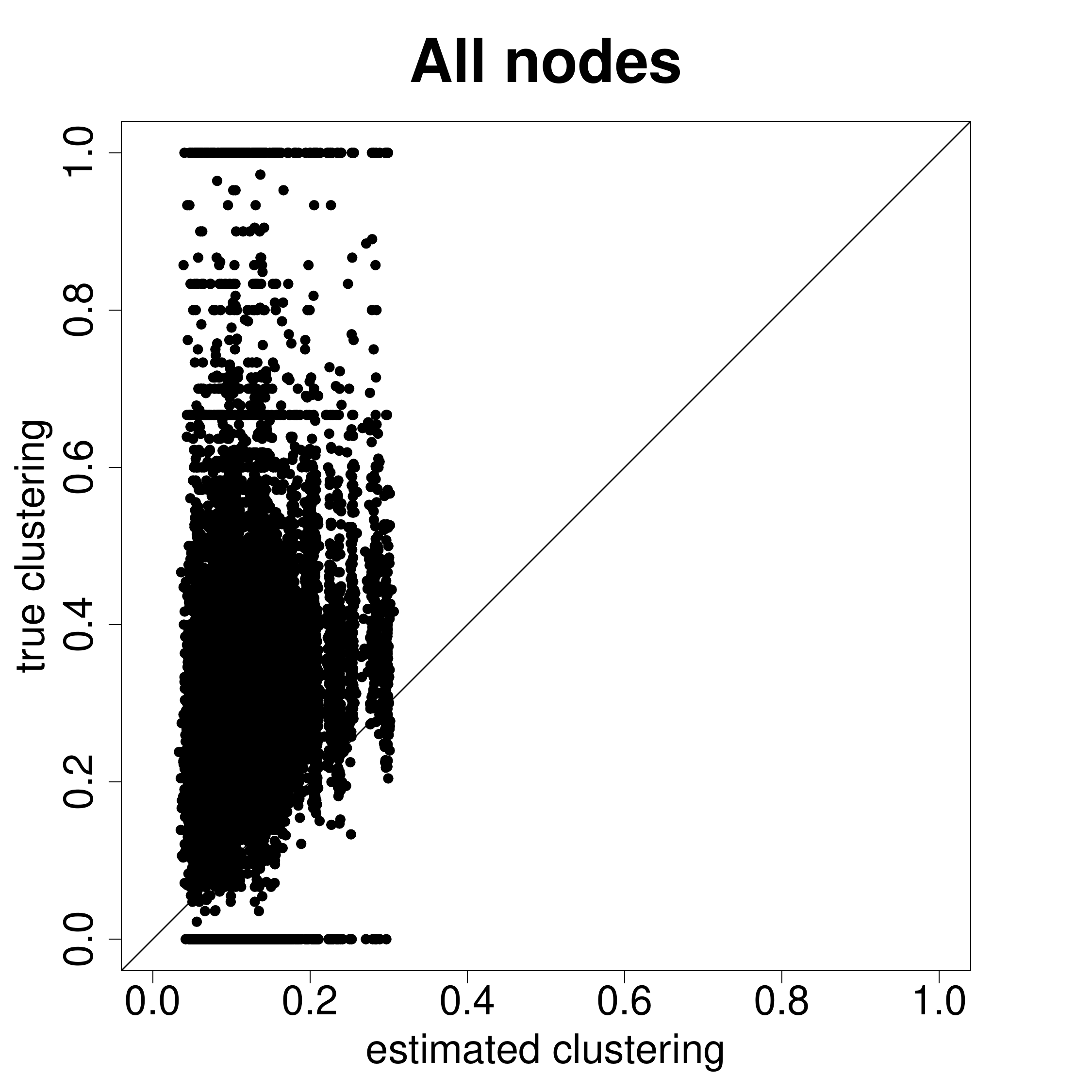}}
\subfloat[Treated neighborhood share]{
\includegraphics[width=.3\textwidth]{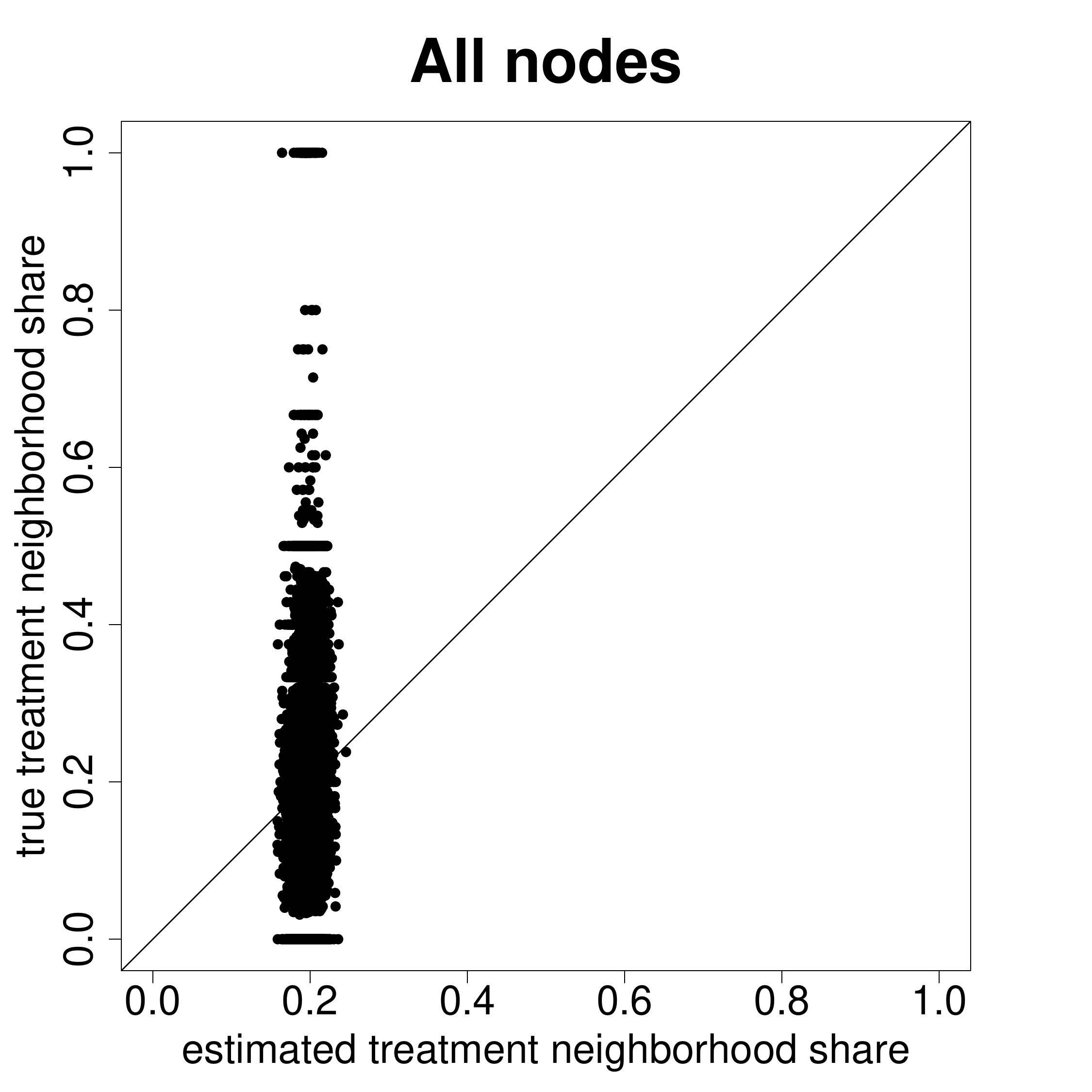}}

\caption{Node level measures estimation for all households in villages in Karnataka.  These plots show scatterplots across all villages with the estimated node level measure on the x-axis and the measure from the true underlying graph on the y-axis.}
\label{fig:karnataka_all_node_appendix}
\end{figure}

\begin{figure}[!h]
\centering
\subfloat[Number of components]{
\includegraphics[width=.3\textwidth]{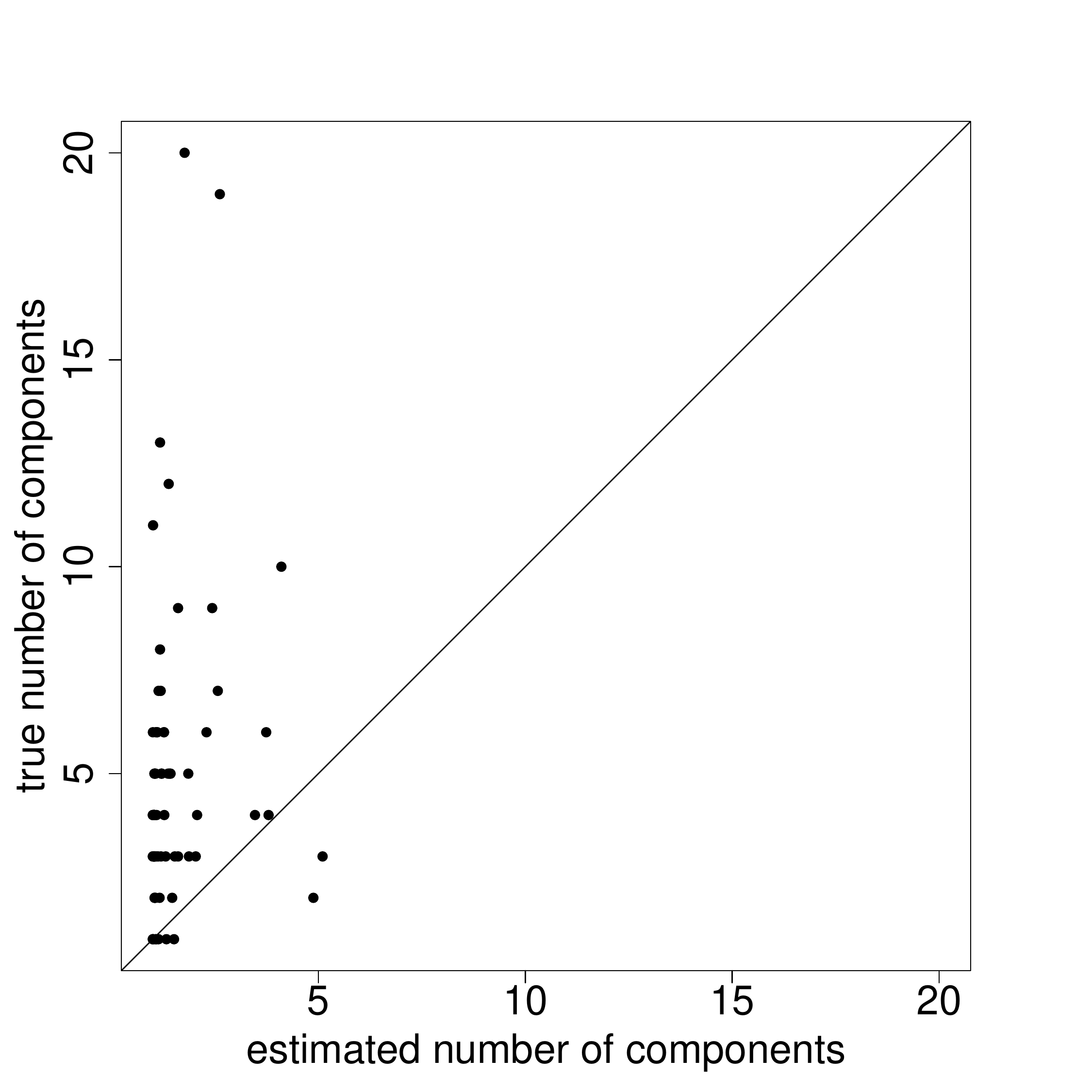}}
\subfloat[Average path length]{
\includegraphics[width=.3\textwidth]{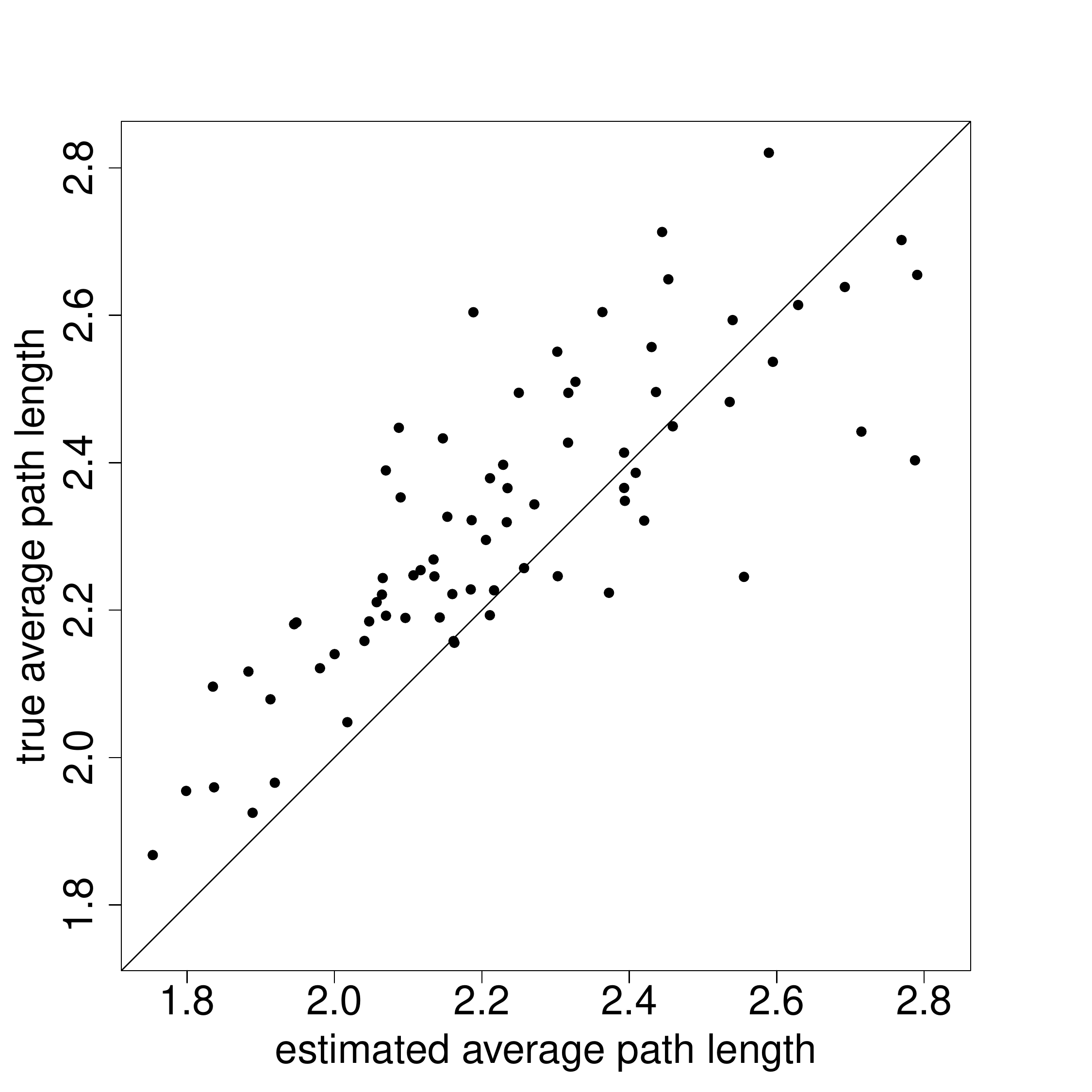}}

\smallskip
\subfloat[Diameter]{
\includegraphics[width=.3\textwidth]{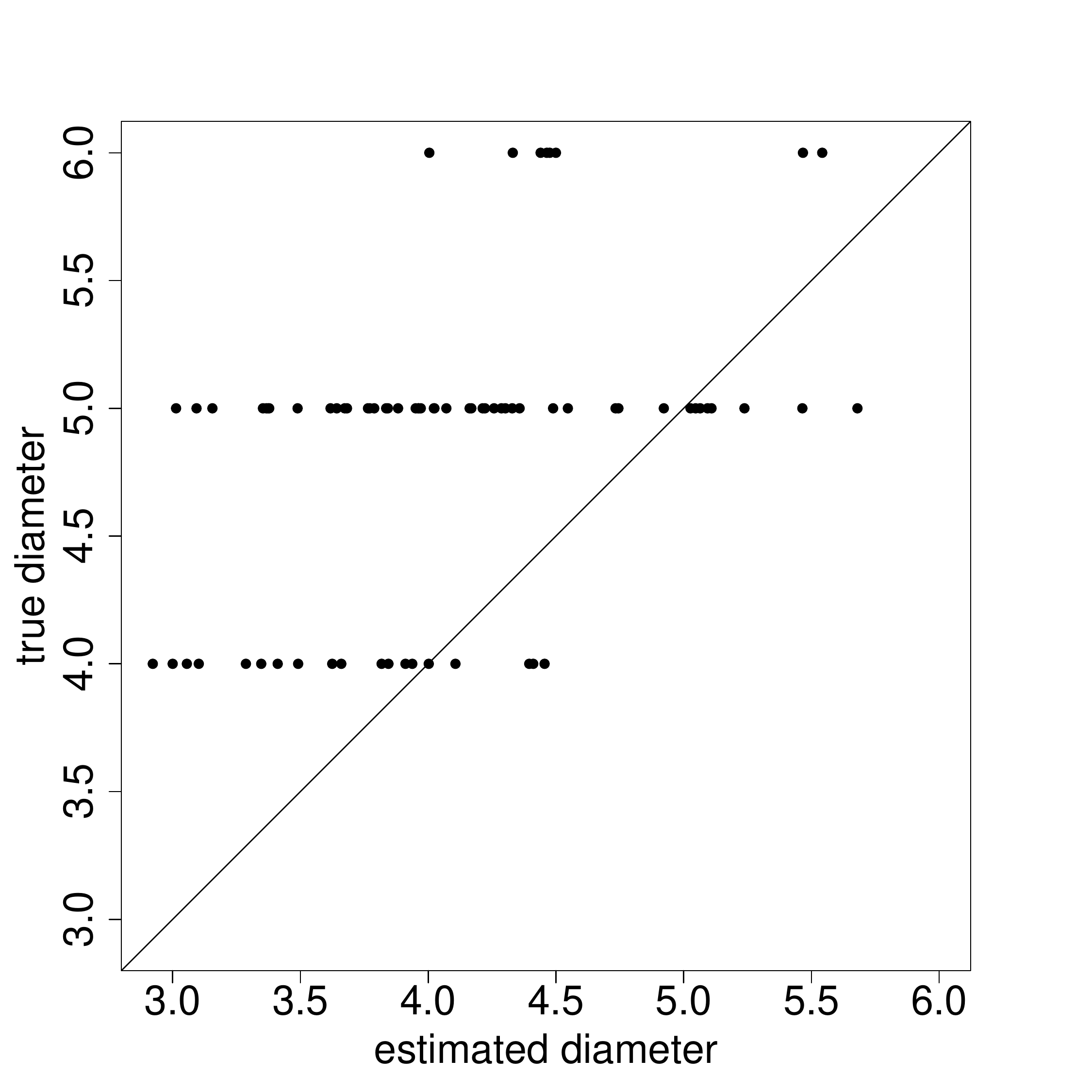}}
\subfloat[Fraction of giant component]{
\includegraphics[width=.3\textwidth]{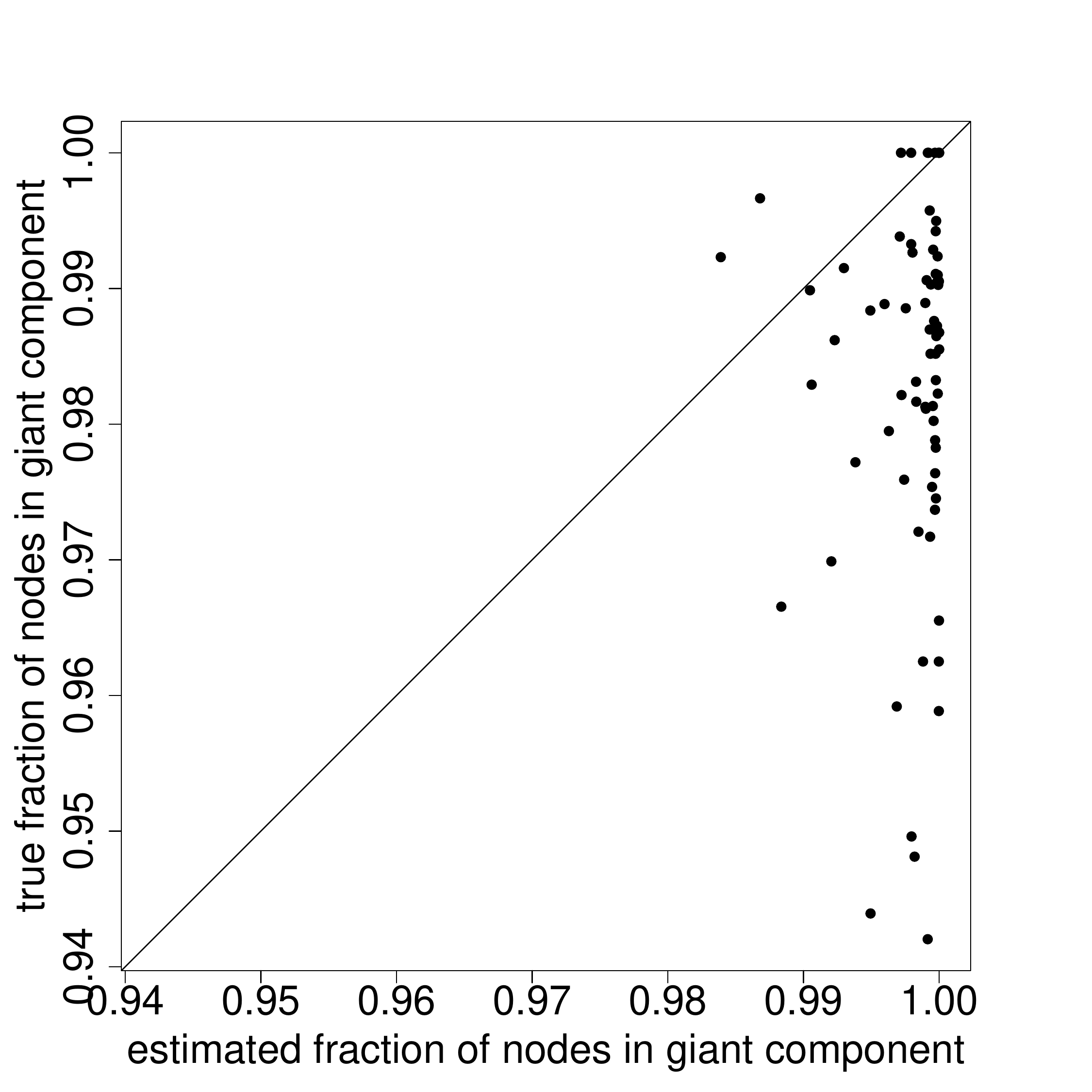}}

\caption{Network level measures estimation for households in villages in Karnataka.  These plots show scatterplots across all villages with the estimated network level measure on the x-axis and the measure from the true underlying graph on the y-axis.}
\label{fig:karnataka_network_appendix}
\end{figure}

\clearpage
\section{Scatterplots for Karnataka villages when households' latent space positions are on the surface of a 4 dimensional hypersphere}
\label{sec:p_3}

\setcounter{figure}{0}
\renewcommand{\thefigure}{J.\arabic{figure}}

\vspace{-4mm}
\begin{figure}[!h]
\centering
\subfloat[Degree]{
\includegraphics[width=.28\textwidth]{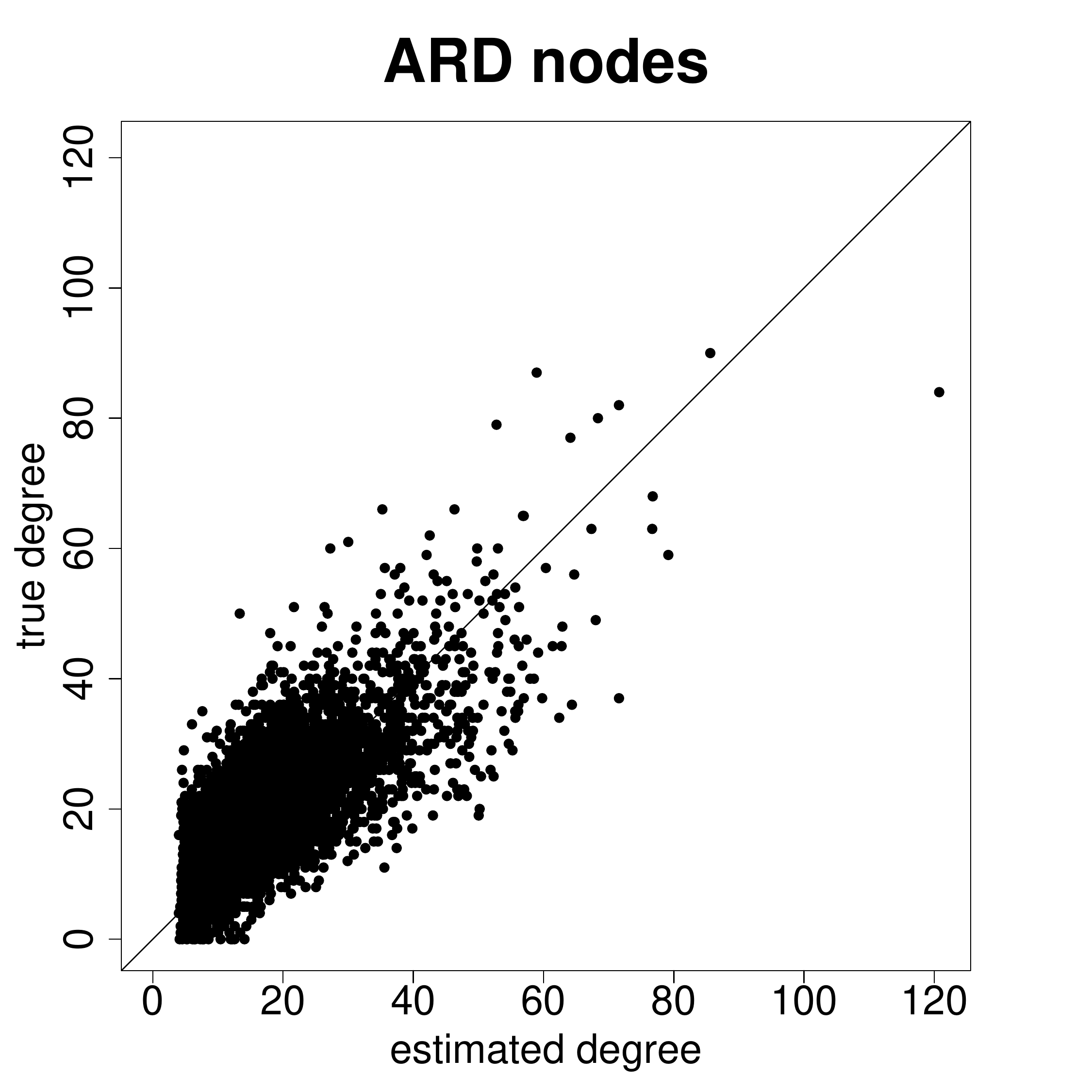}}
\subfloat[Eigenvector Centrality]{
\includegraphics[width=.28\textwidth]{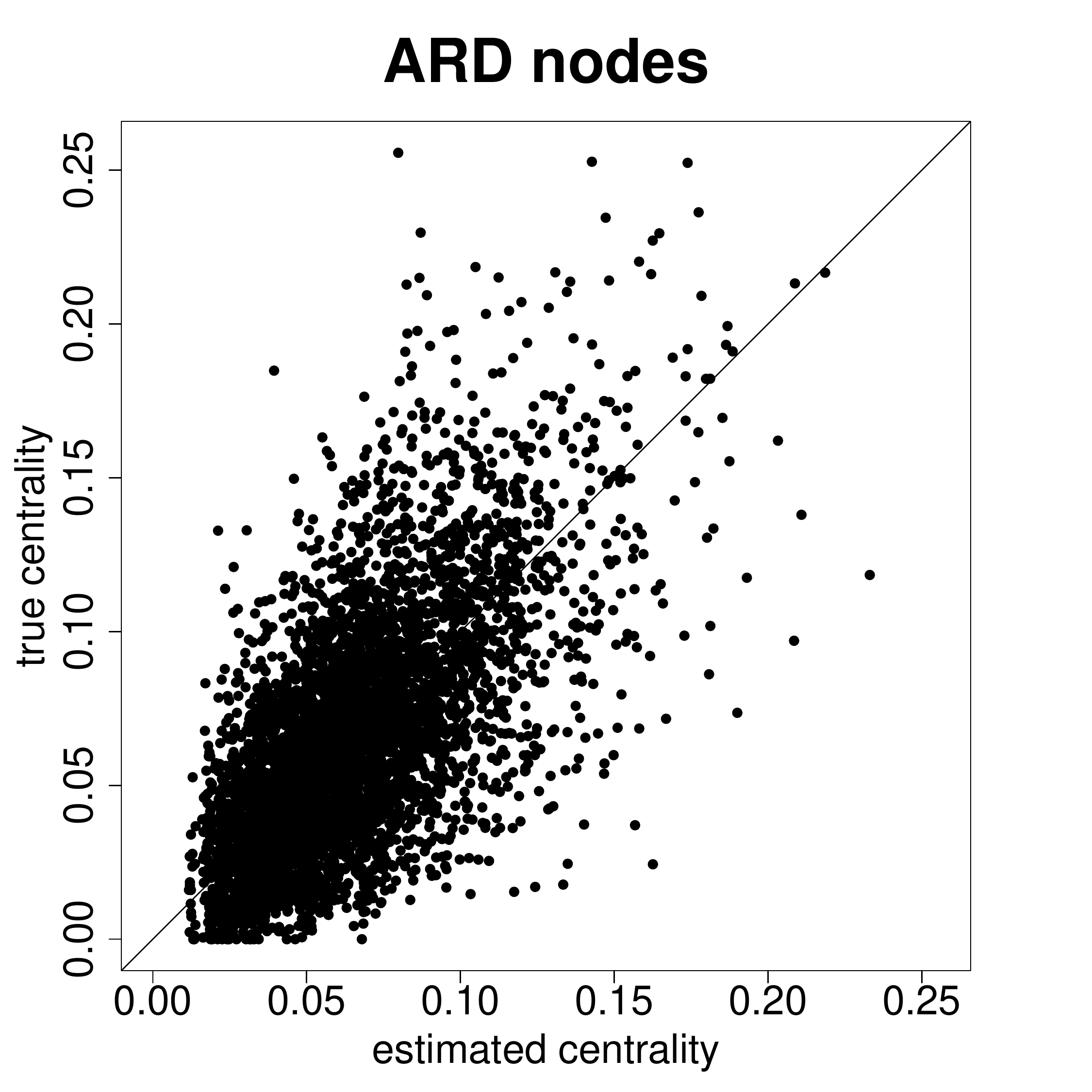}}
\subfloat[Clustering]{
\includegraphics[width=.28\textwidth]{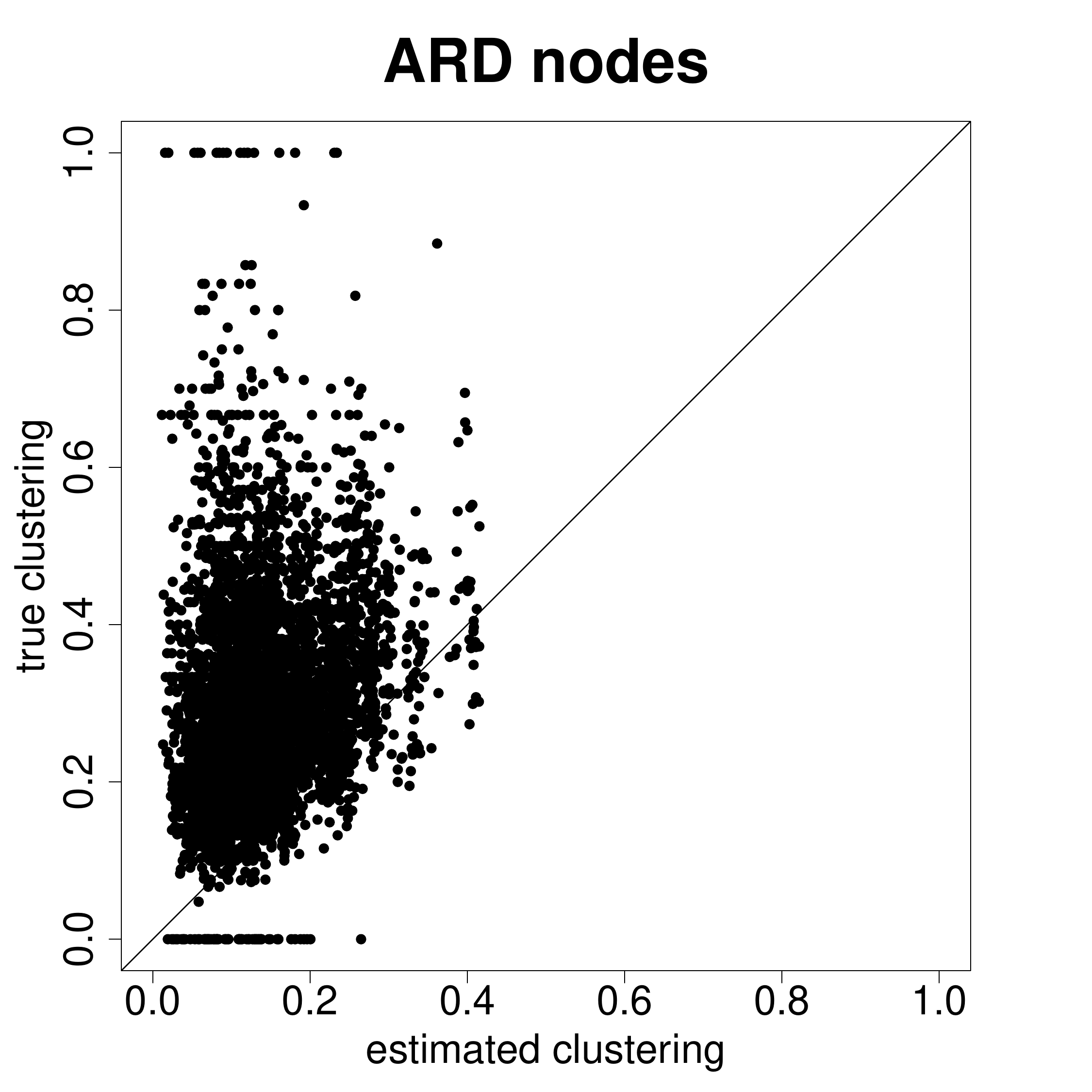}}

\smallskip
\subfloat[Closeness]{
\includegraphics[width=.28\textwidth]{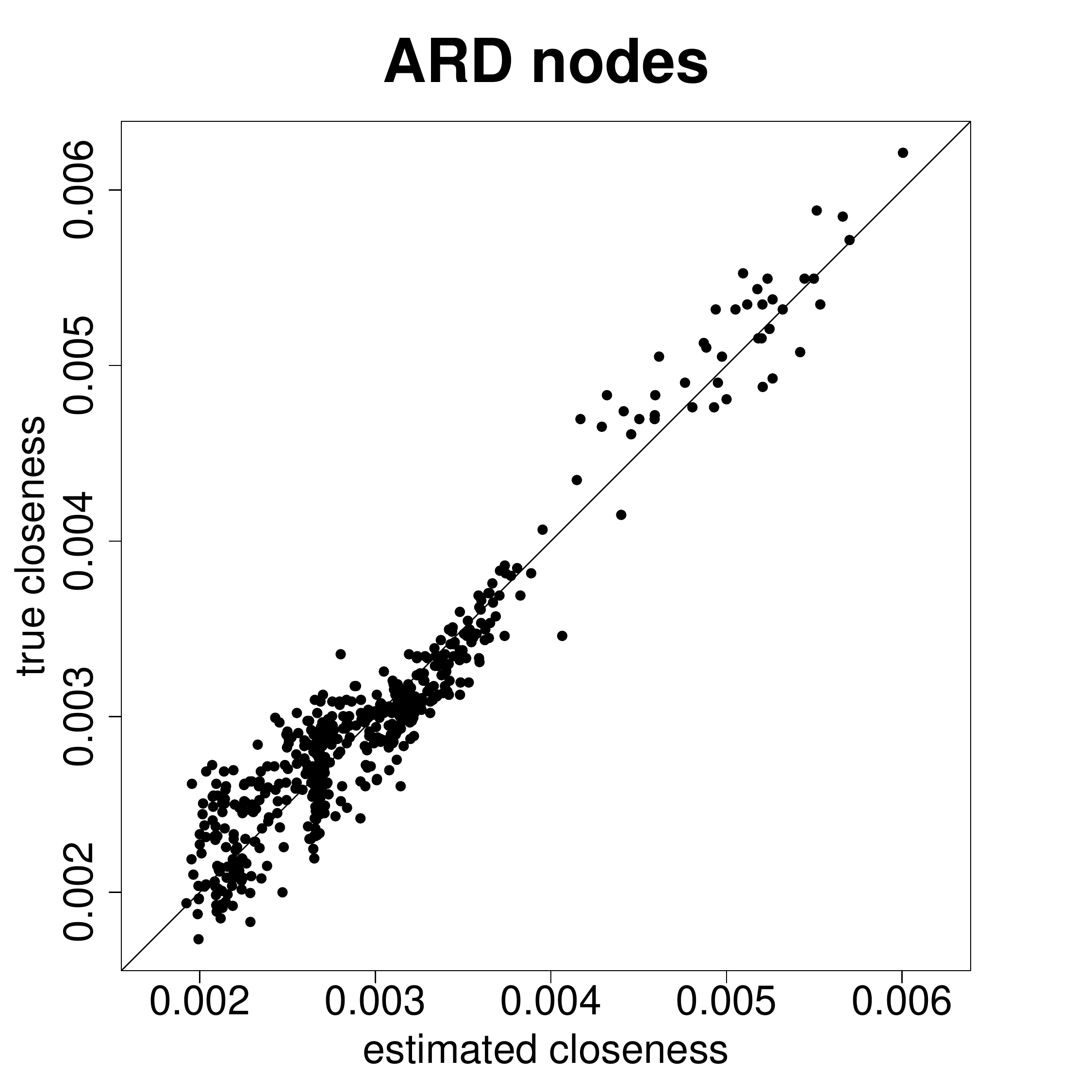}}
\subfloat[Betweenness]{
\includegraphics[width=.28\textwidth]{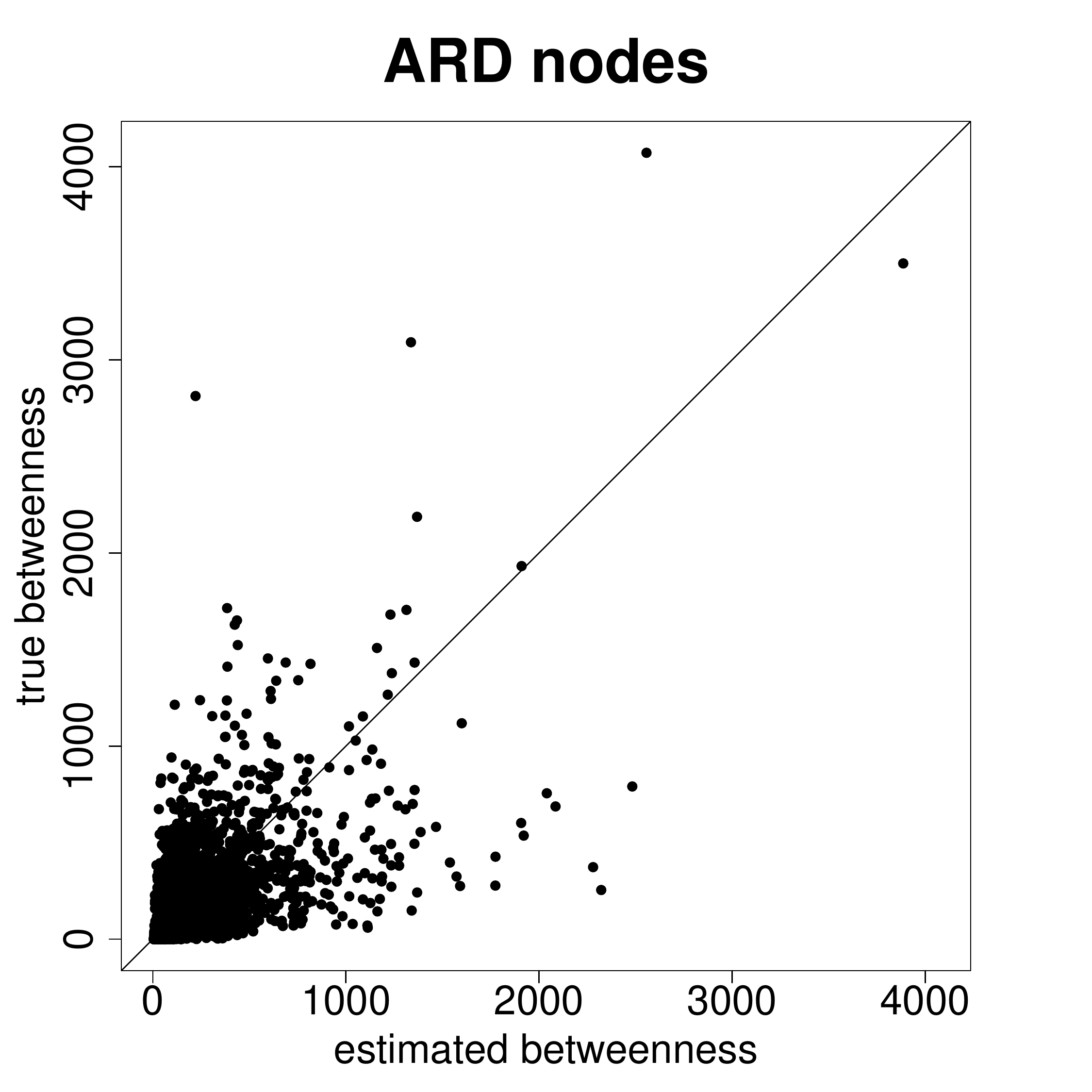}
}
\subfloat[Support]{
\includegraphics[width=.28\textwidth]{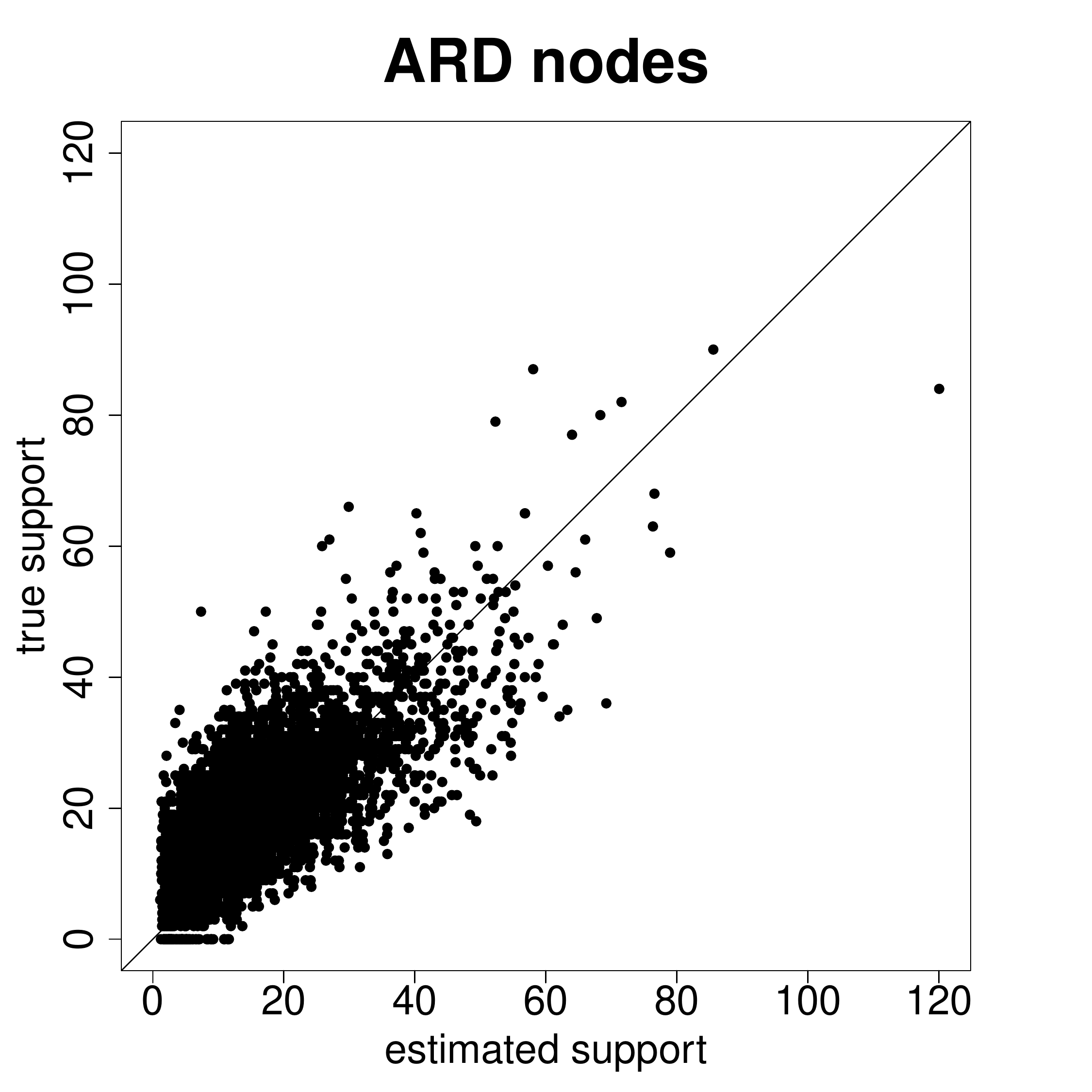}
}

\smallskip
\subfloat[Distance from seed]{
\includegraphics[width=.28\textwidth]{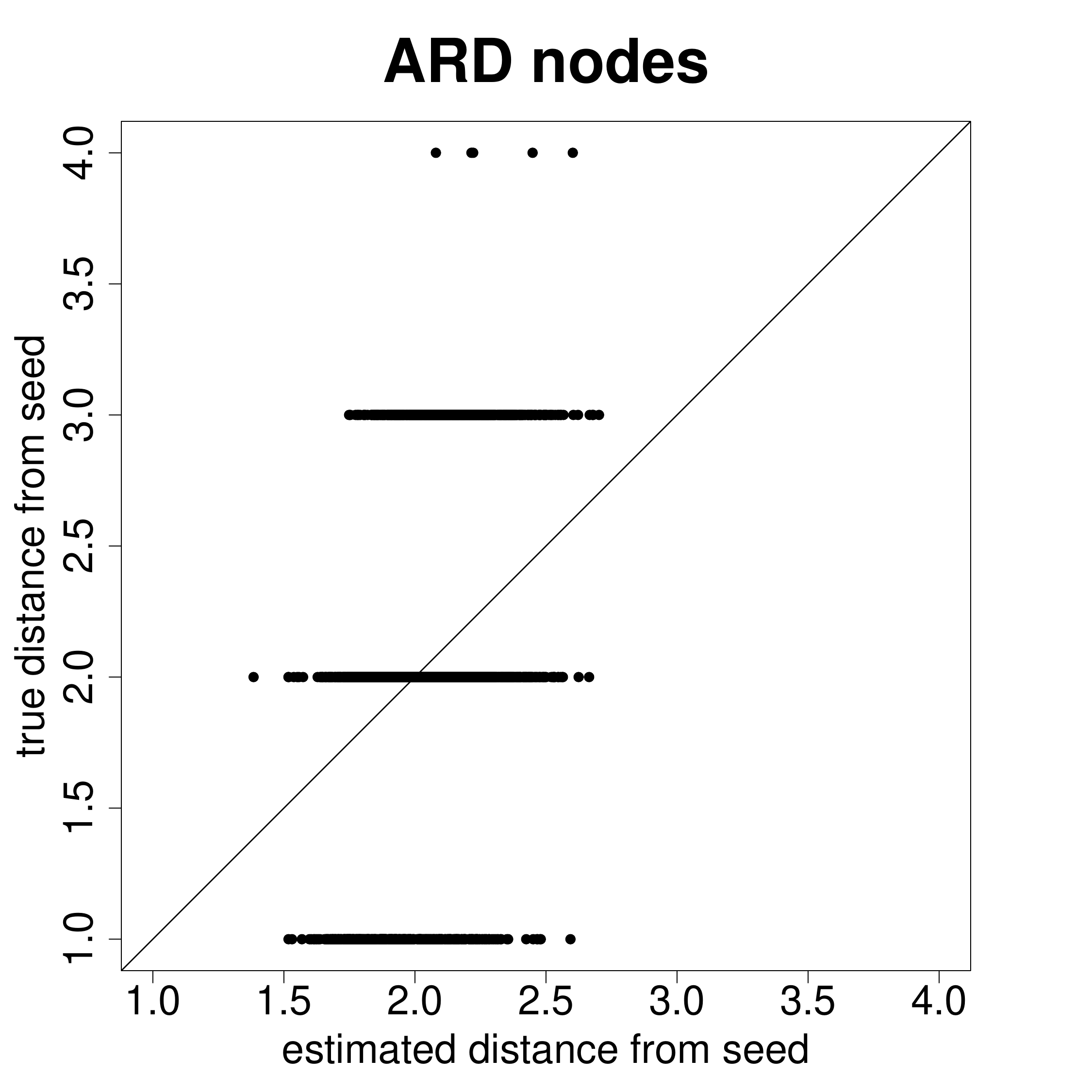}}
\subfloat[Node level clustering]{
\includegraphics[width=.28\textwidth]{plots/p_3/comparison_clustering_ARDnodes.pdf}}
\subfloat[Treated neighborhood share]{
\includegraphics[width=.28\textwidth]{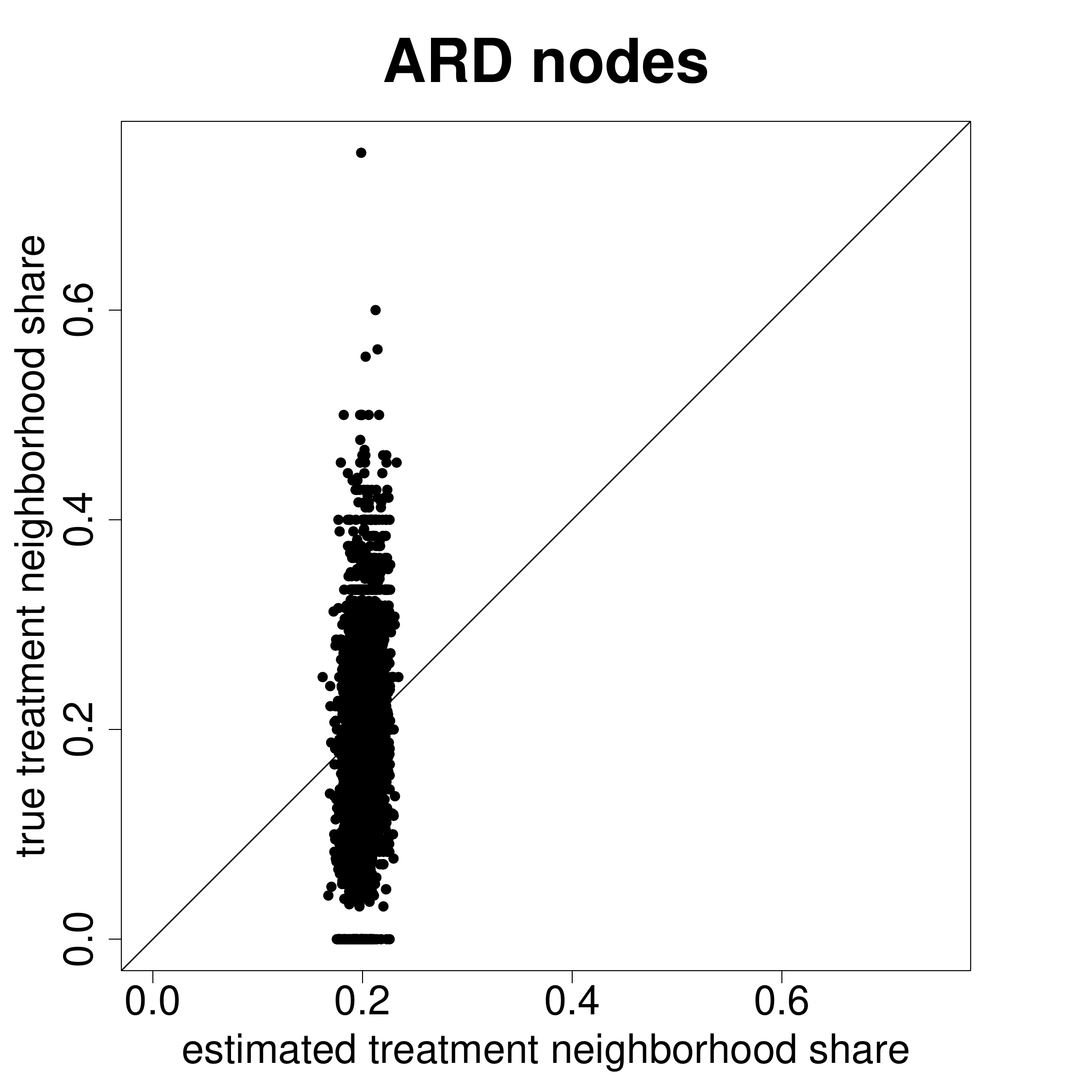}}

\caption{Node level measures estimation for households with ARD response in villages in Karnataka.  These plots show scatterplots across all villages with the estimated node level measure on the x-axis and the measure from the true underlying graph on the y-axis.}
\label{fig:p_3_karnataka_ARD_node_appendix}
\end{figure}

\begin{figure}[!h]
\centering
\subfloat[Degree]{
\includegraphics[width=.28\textwidth]{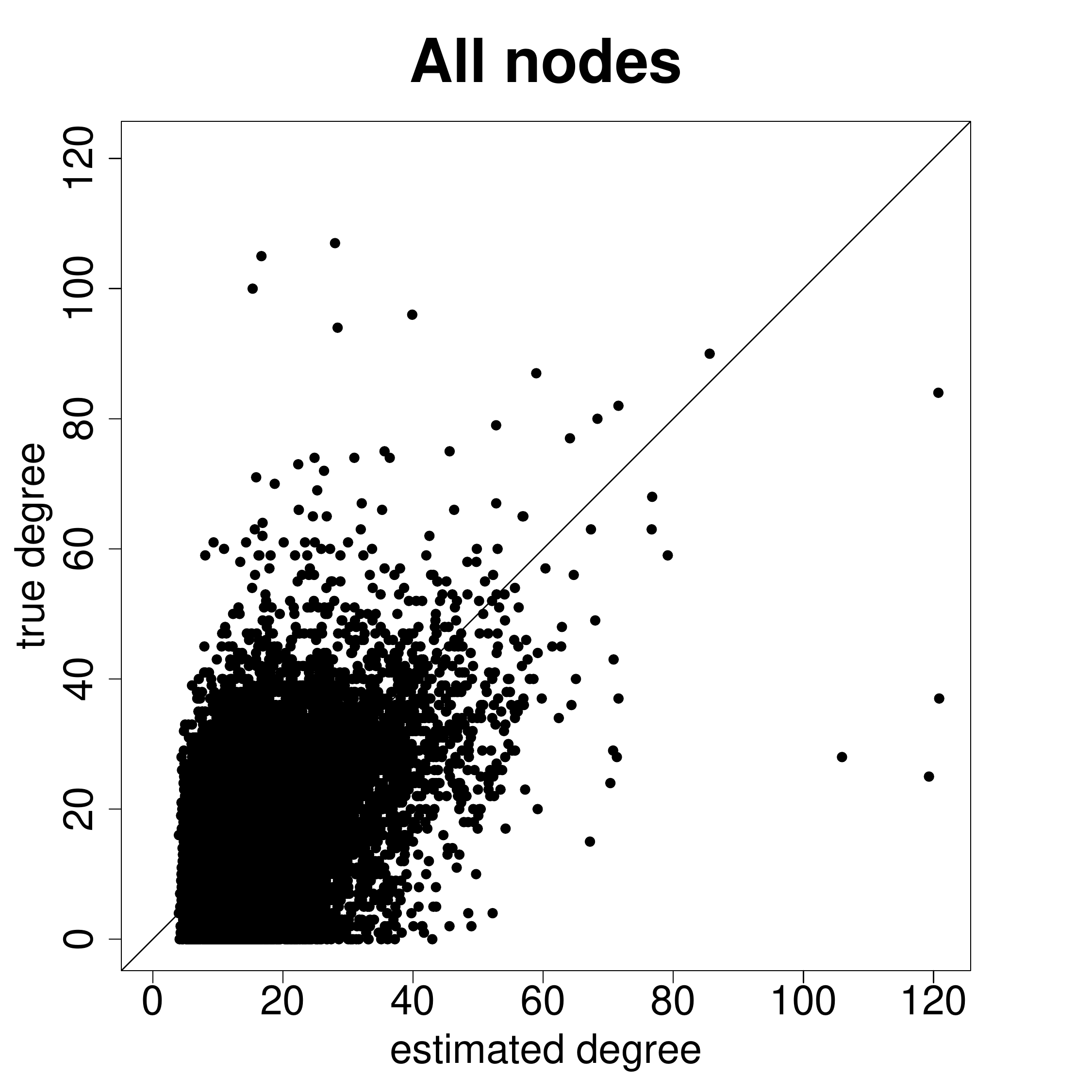}}
\subfloat[Eigenvector Centrality]{
\includegraphics[width=.28\textwidth]{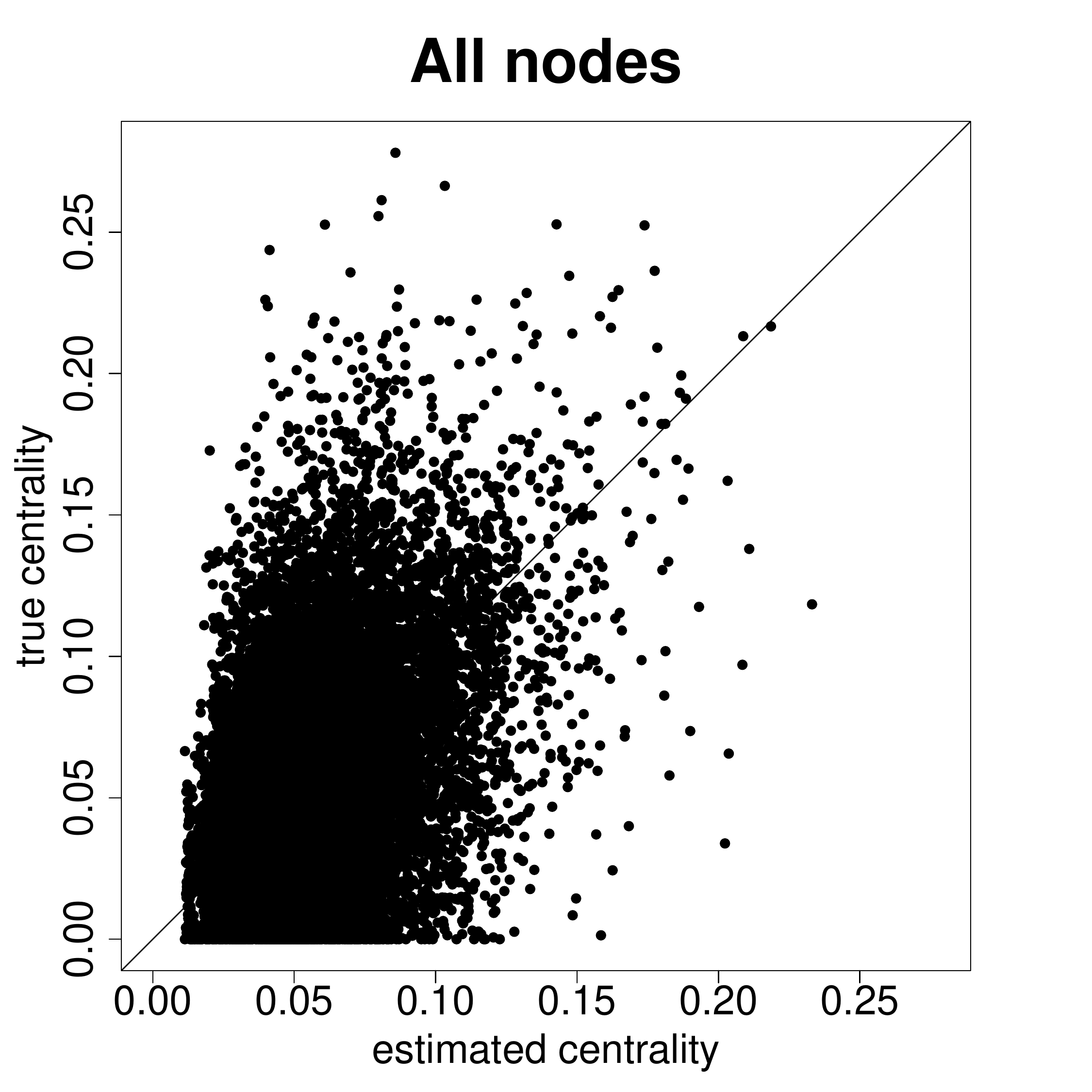}}
\subfloat[Clustering]{
\includegraphics[width=.28\textwidth]{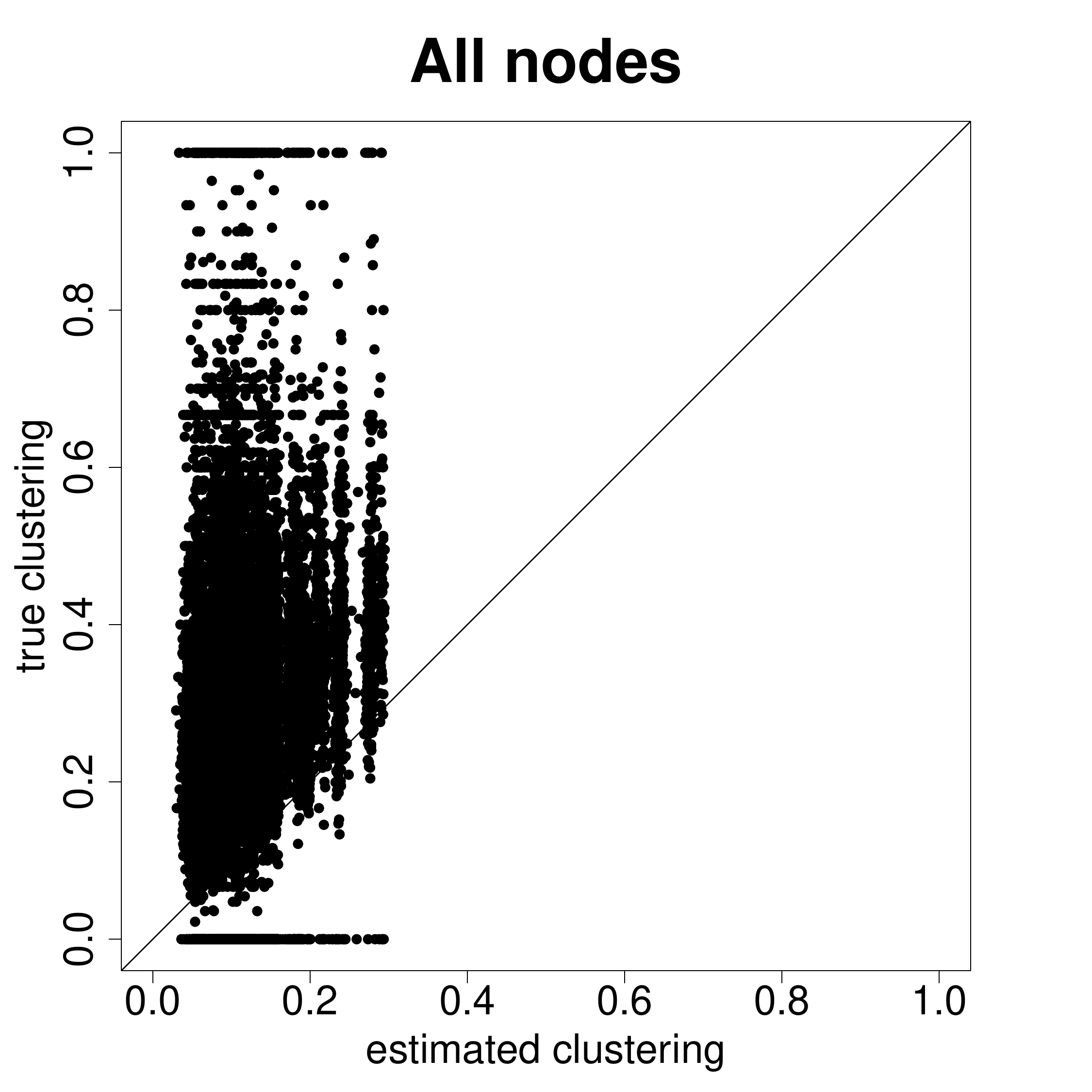}}

\smallskip
\subfloat[Closeness]{
\includegraphics[width=.28\textwidth]{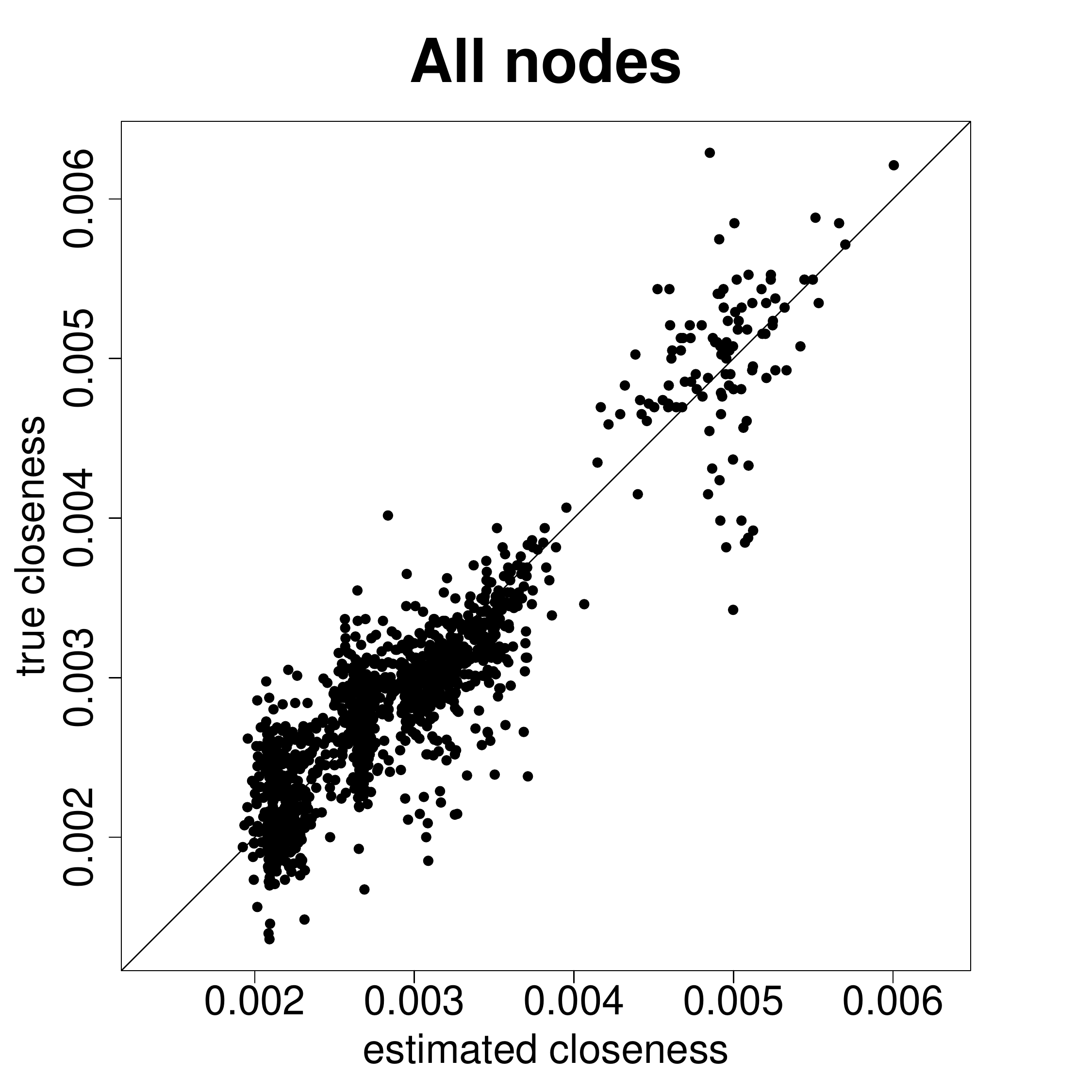}}
\subfloat[Betweenness]{
\includegraphics[width=.28\textwidth]{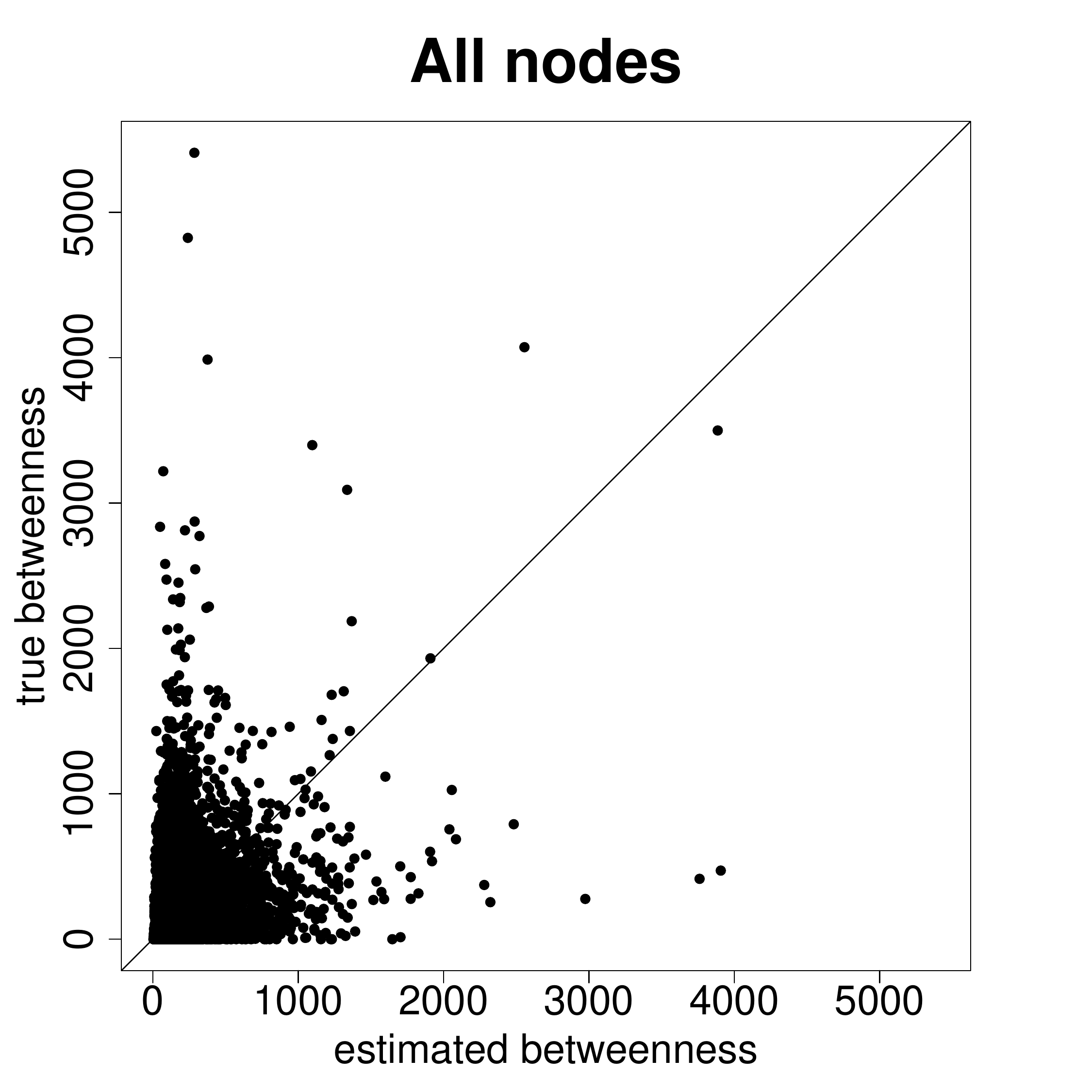}
}
\subfloat[Support]{
\includegraphics[width=.28\textwidth]{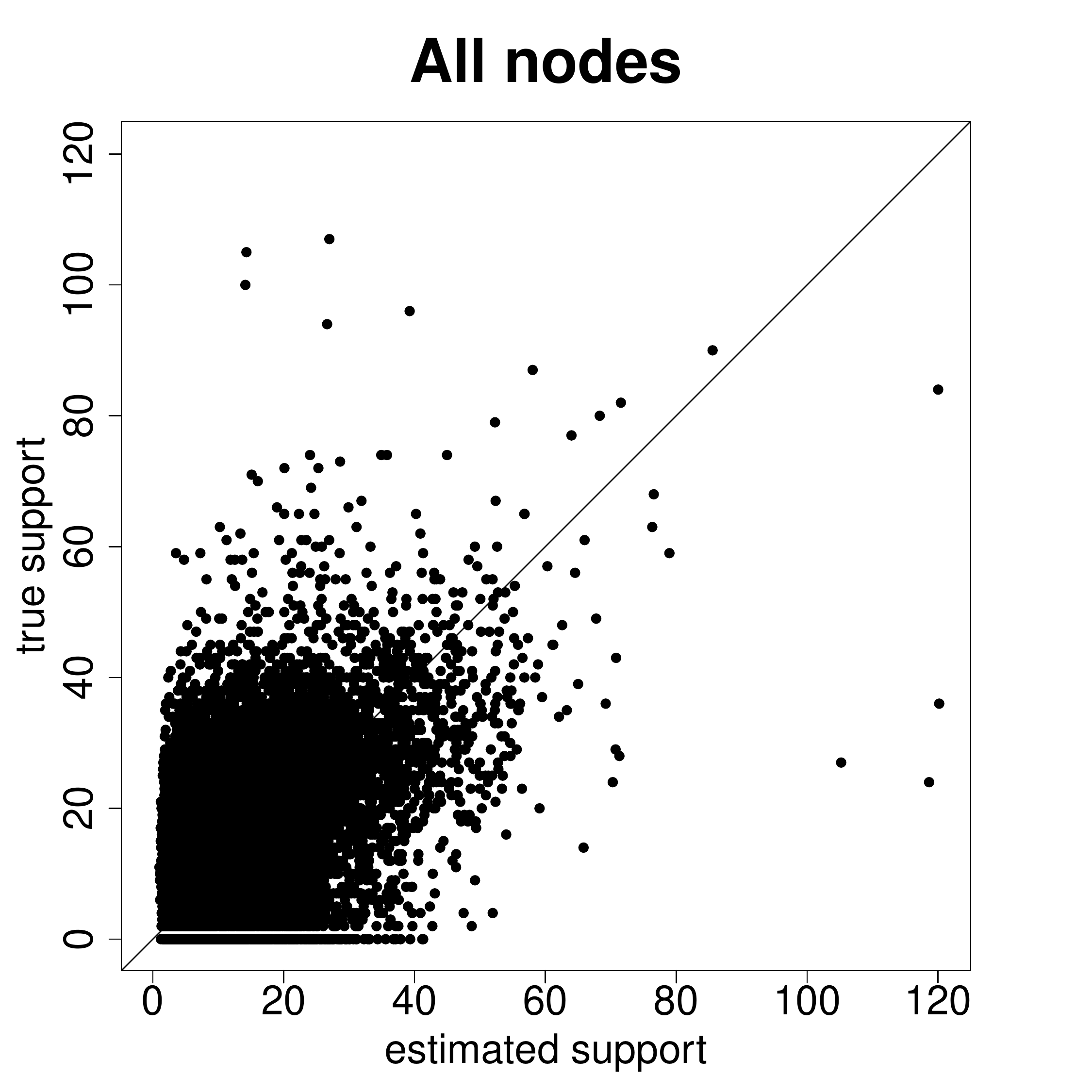}
}

\smallskip
\subfloat[Distance from seed]{
\includegraphics[width=.28\textwidth]{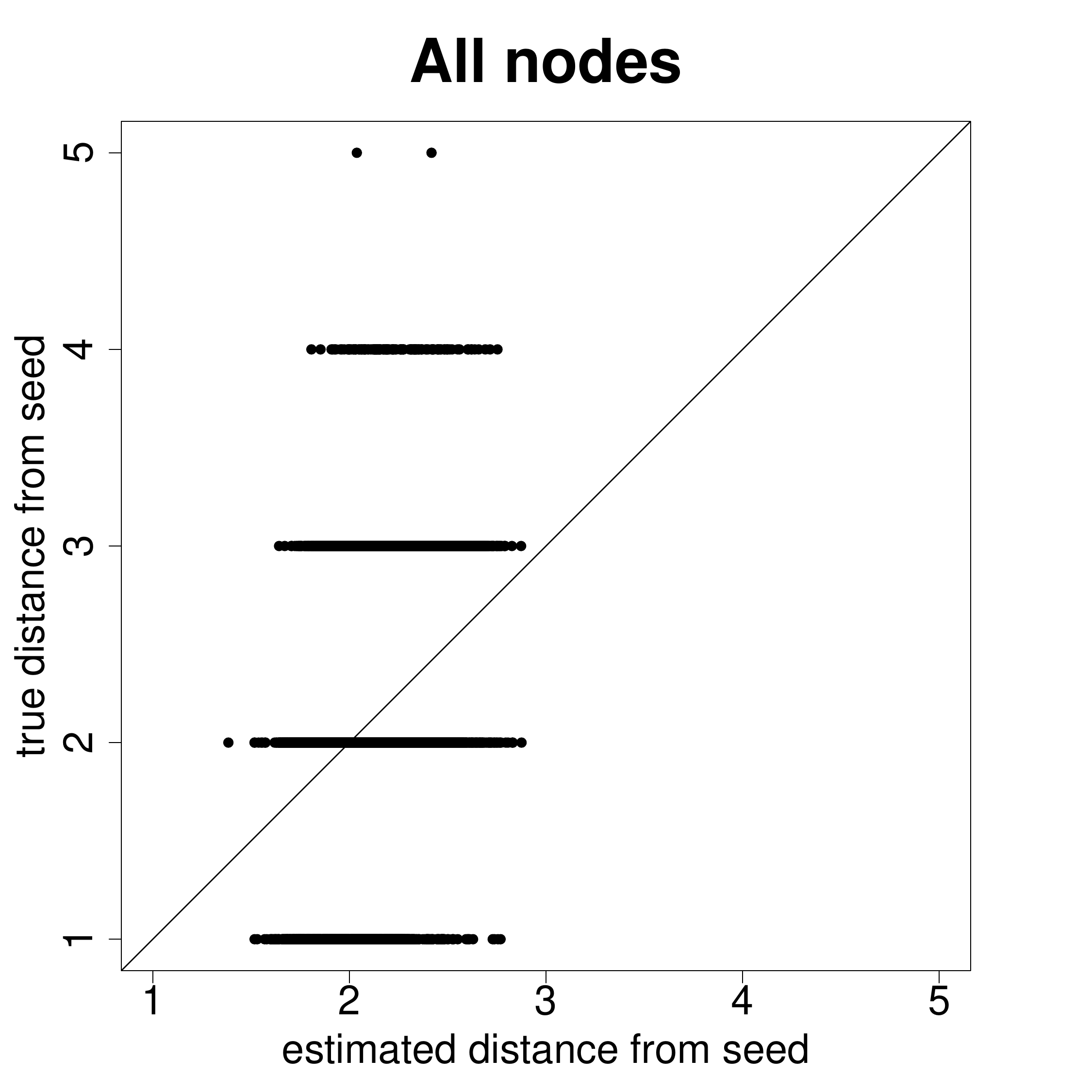}}
\subfloat[Node level clustering]{
\includegraphics[width=.28\textwidth]{plots/p_3/comparison_clustering_allnodes.pdf}}
\subfloat[Treated neighborhood share]{
\includegraphics[width=.28\textwidth]{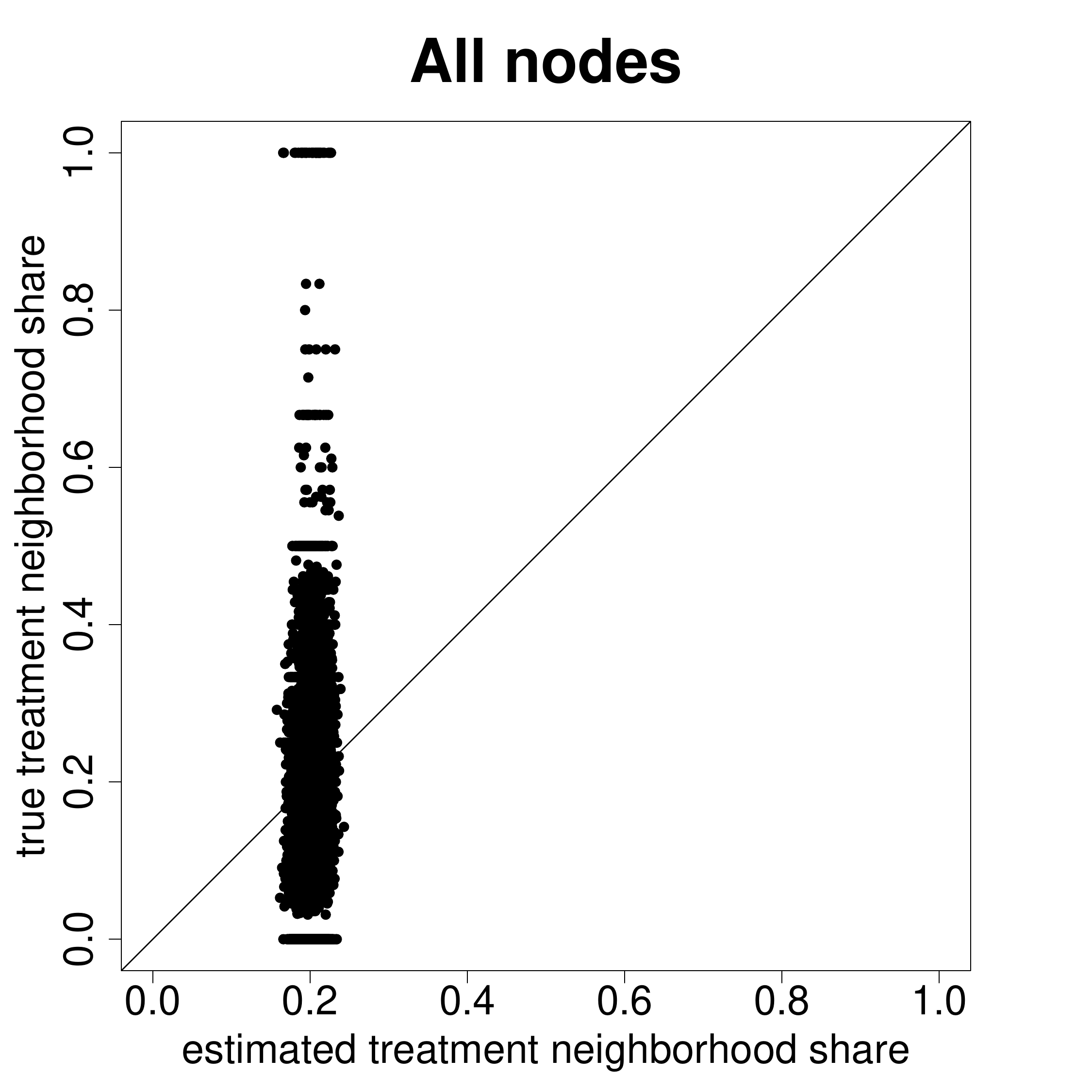}}

\caption{Node level measures estimation for households with all response in villages in Karnataka.  These plots show scatterplots across all villages with the estimated node level measure on the x-axis and the measure from the true underlying graph on the y-axis.}
\label{fig:p_3_karnataka_all_node_appendix}
\end{figure}

\begin{figure}[!h]
\centering
\subfloat[$\lambda_1(g)$]{
\includegraphics[width=.28\textwidth]{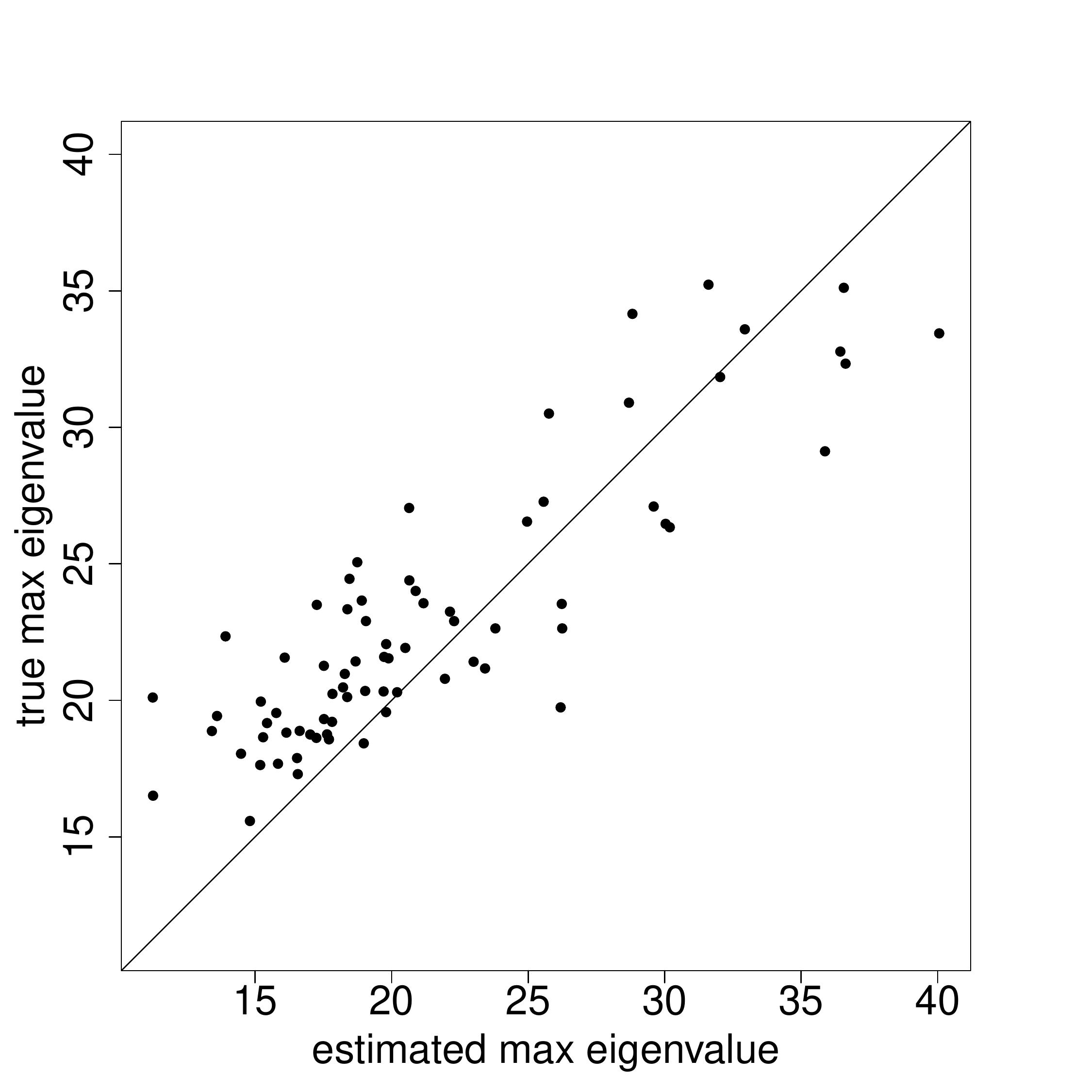}\qquad
}
\subfloat[Social Proximity]{
\includegraphics[width=.28\textwidth]{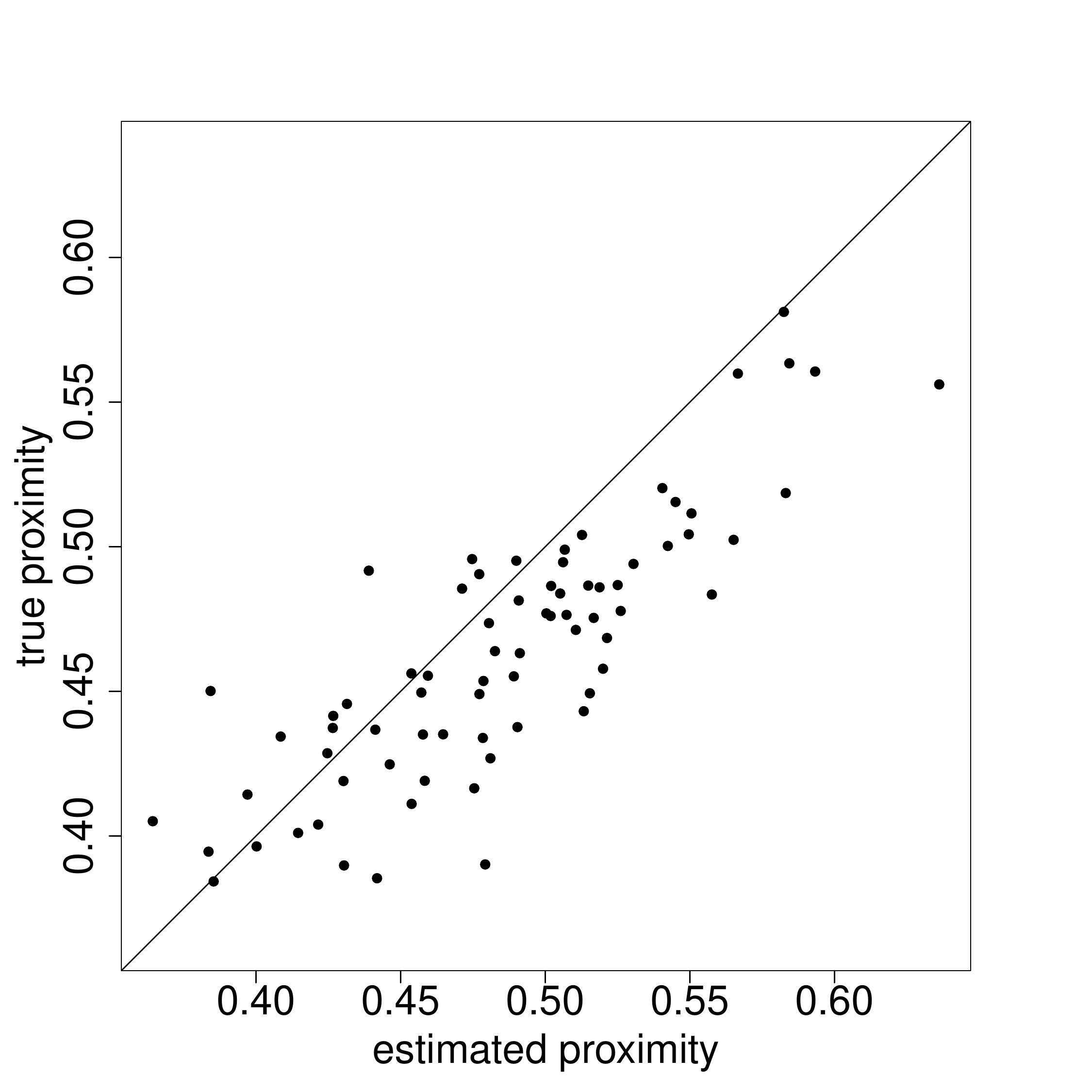}
}
\subfloat[Clustering]{
\includegraphics[width=.28\textwidth]{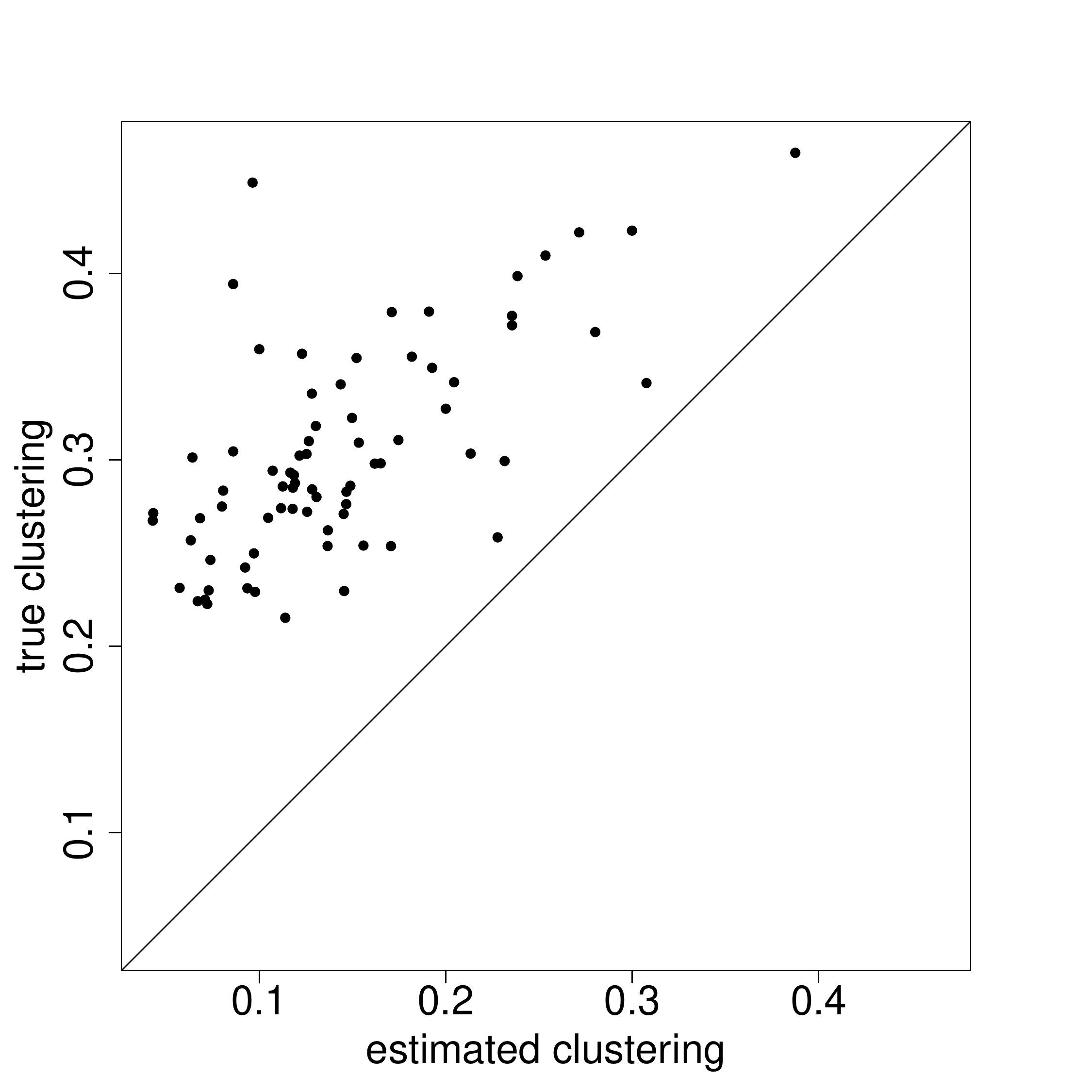}
}

\smallskip
\subfloat[Eigenvector Cut]{
\includegraphics[width=.28\textwidth]{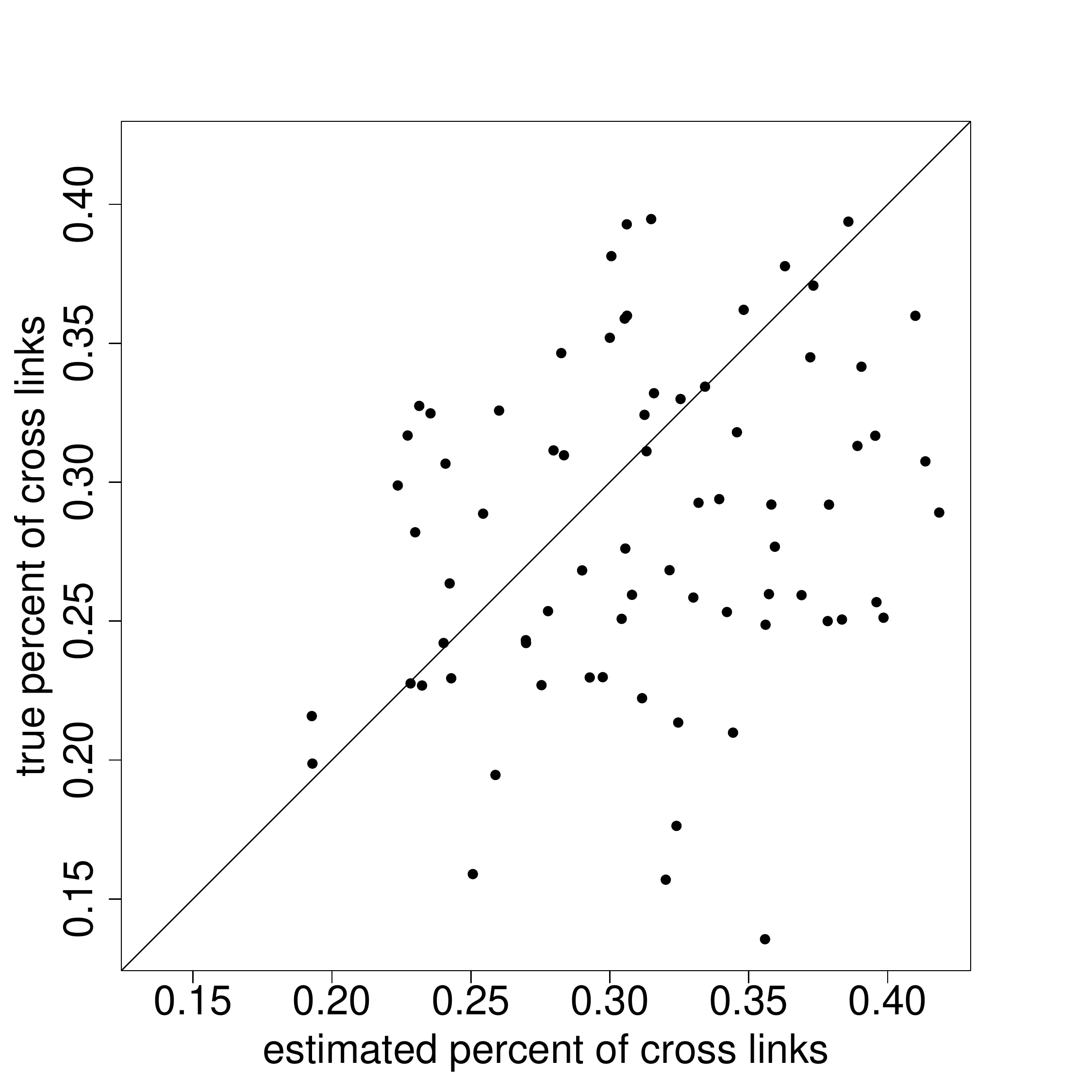}
}
\subfloat[Number of components]{
\includegraphics[width=.28\textwidth]{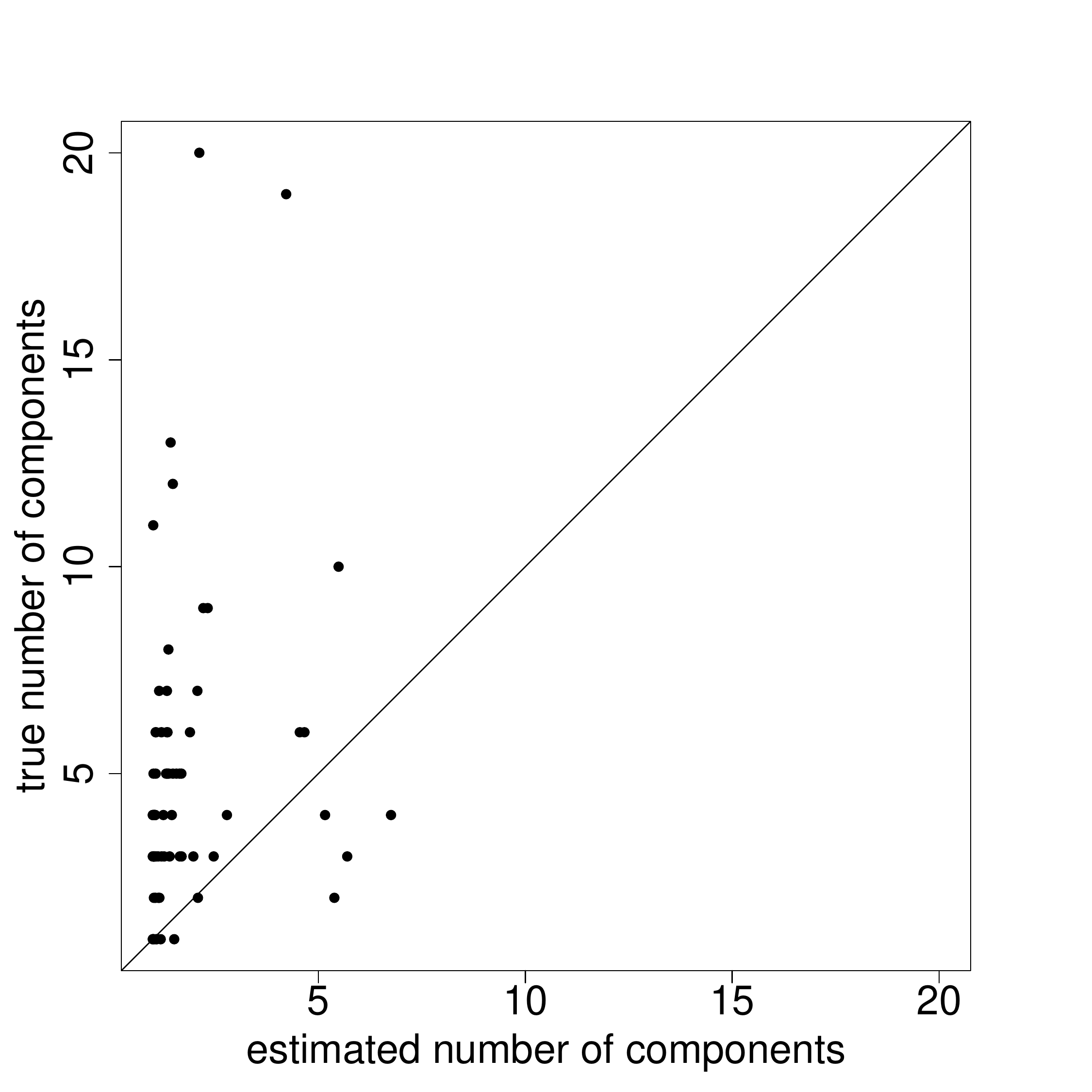}}
\subfloat[Average path length]{
\includegraphics[width=.28\textwidth]{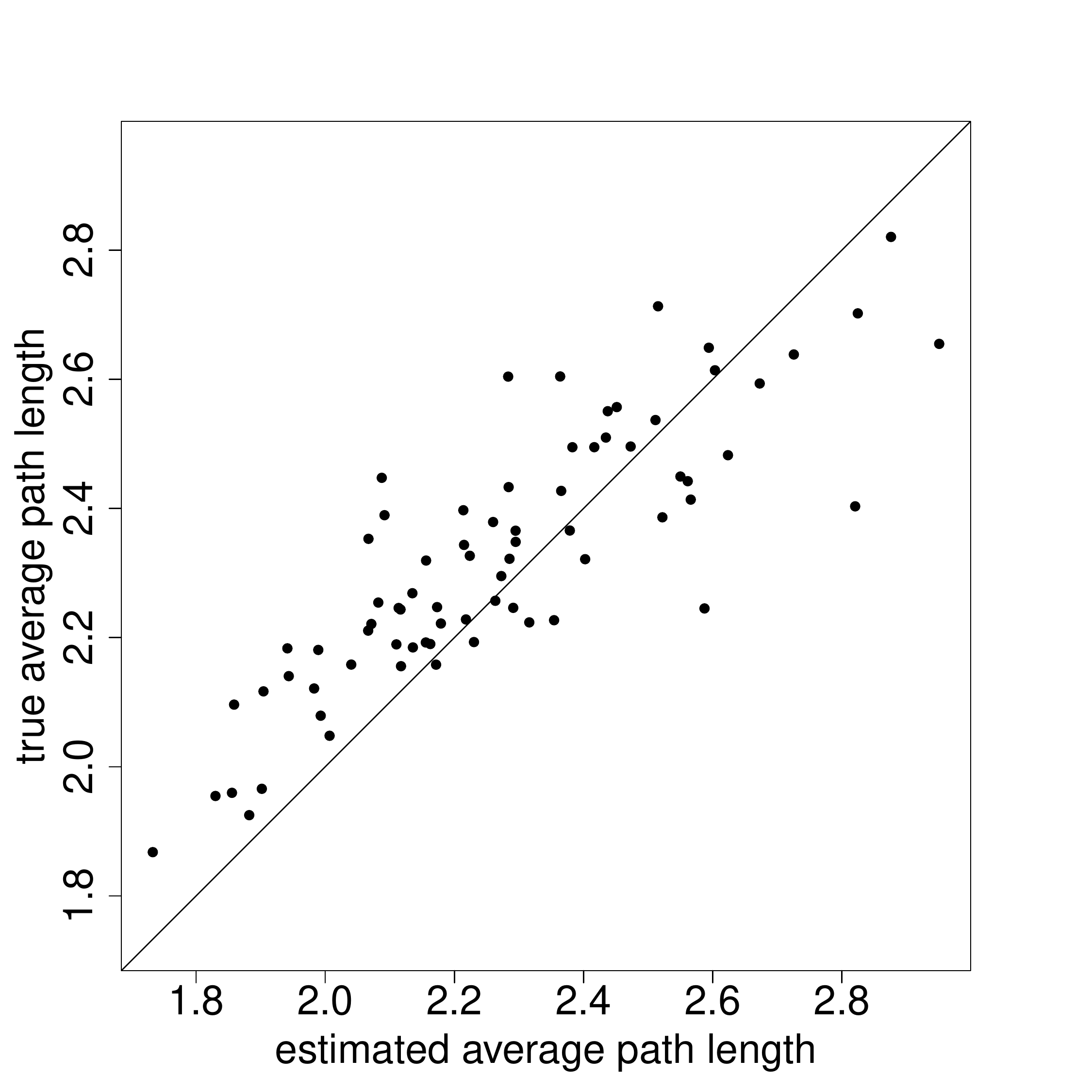}}

\smallskip
\subfloat[Diameter]{
\includegraphics[width=.28\textwidth]{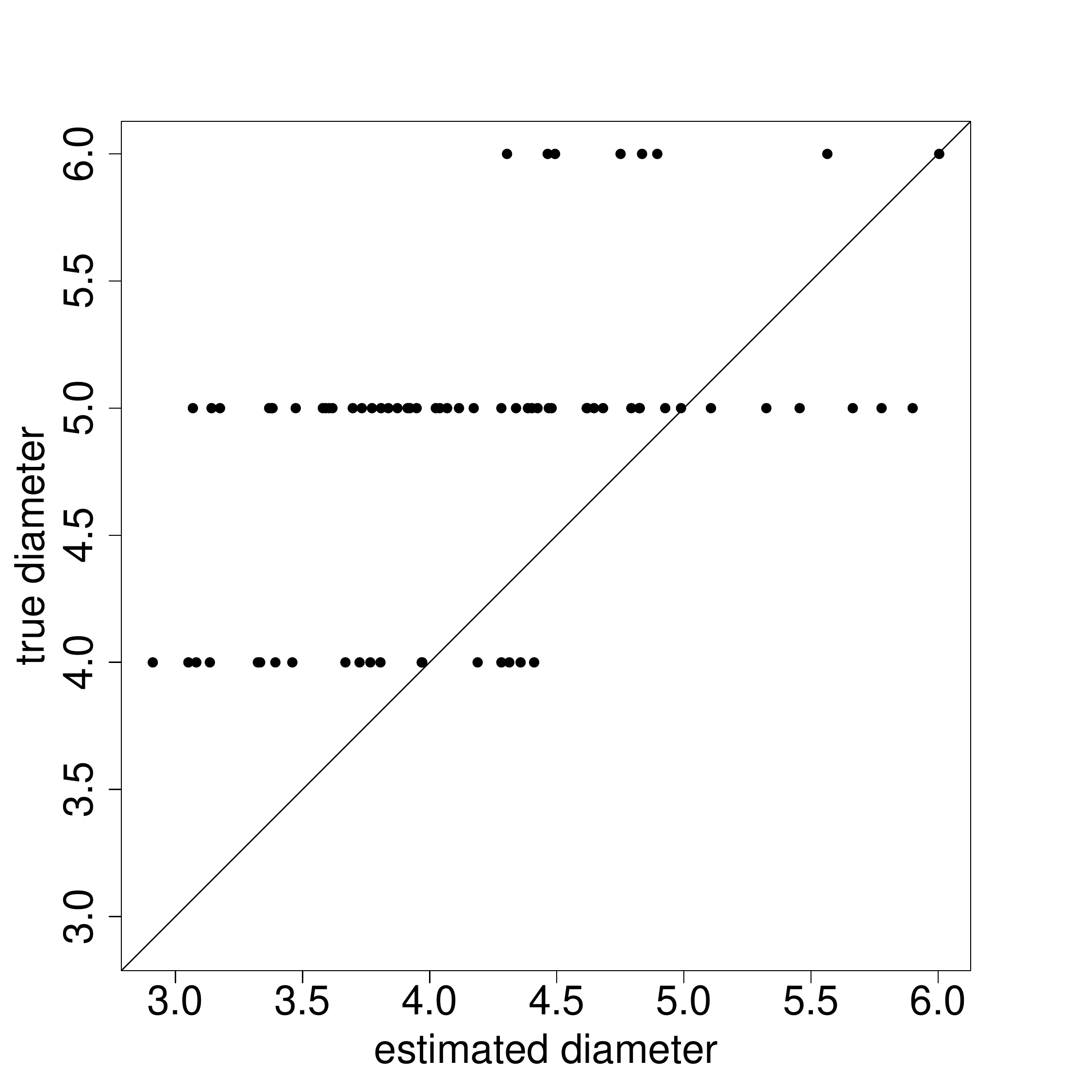}}
\subfloat[Fraction of giant component]{
\includegraphics[width=.28\textwidth]{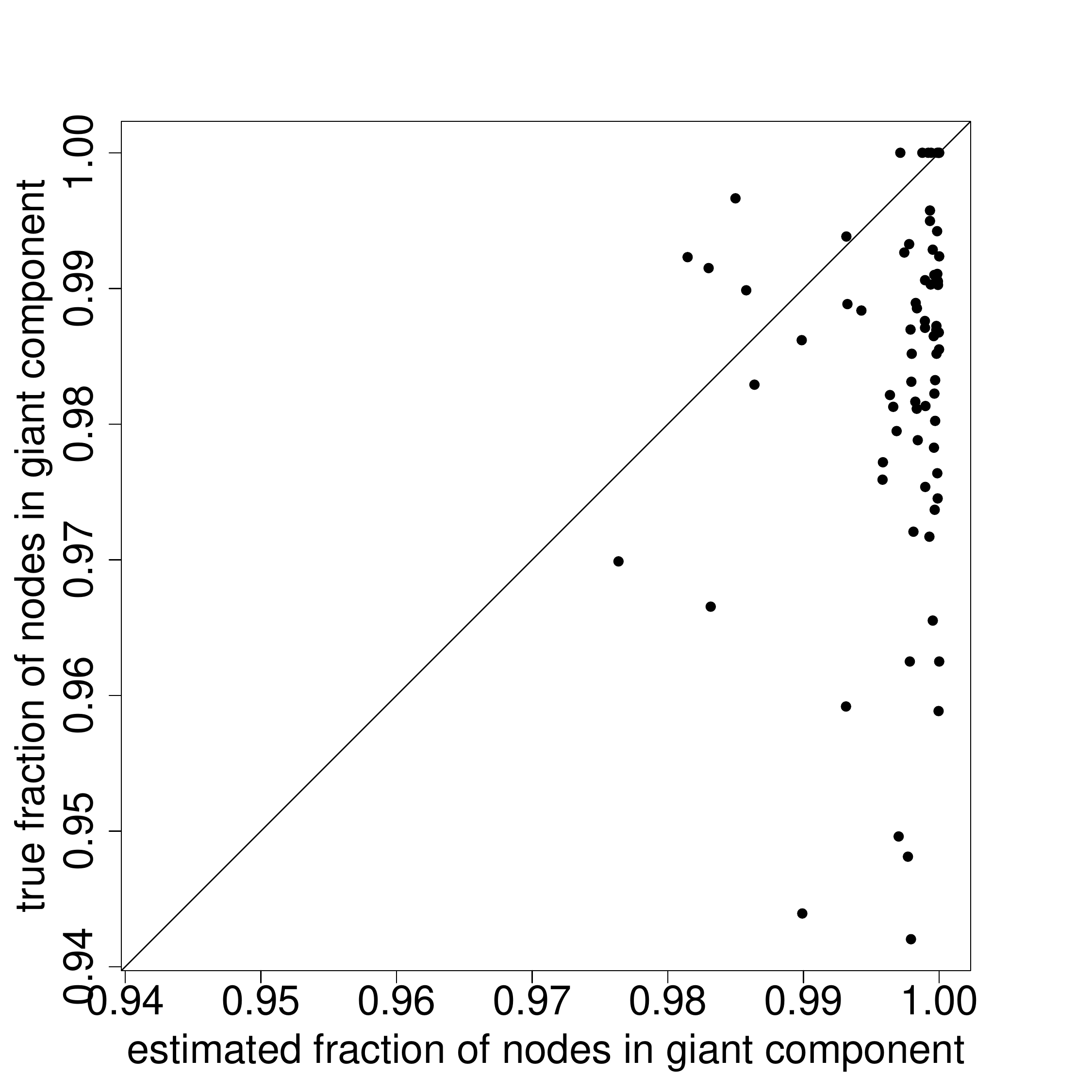}}

\caption{Network level measures estimation for households in villages in Karnataka.  These plots show scatterplots across all villages with the estimated network level measure on the x-axis and the measure from the true underlying graph on the y-axis.}
\label{fig:p_3_karnataka_network_appendix}
\end{figure}

\clearpage
\section{Scatterplots for Karnataka villages when households' latent space positions are on the surface of a 5 dimensional hypersphere}
\label{sec:p_4}

\setcounter{figure}{0}
\renewcommand{\thefigure}{K.\arabic{figure}}

\vspace{-4mm}
\begin{figure}[!h]
\centering
\subfloat[Degree]{
\includegraphics[width=.28\textwidth]{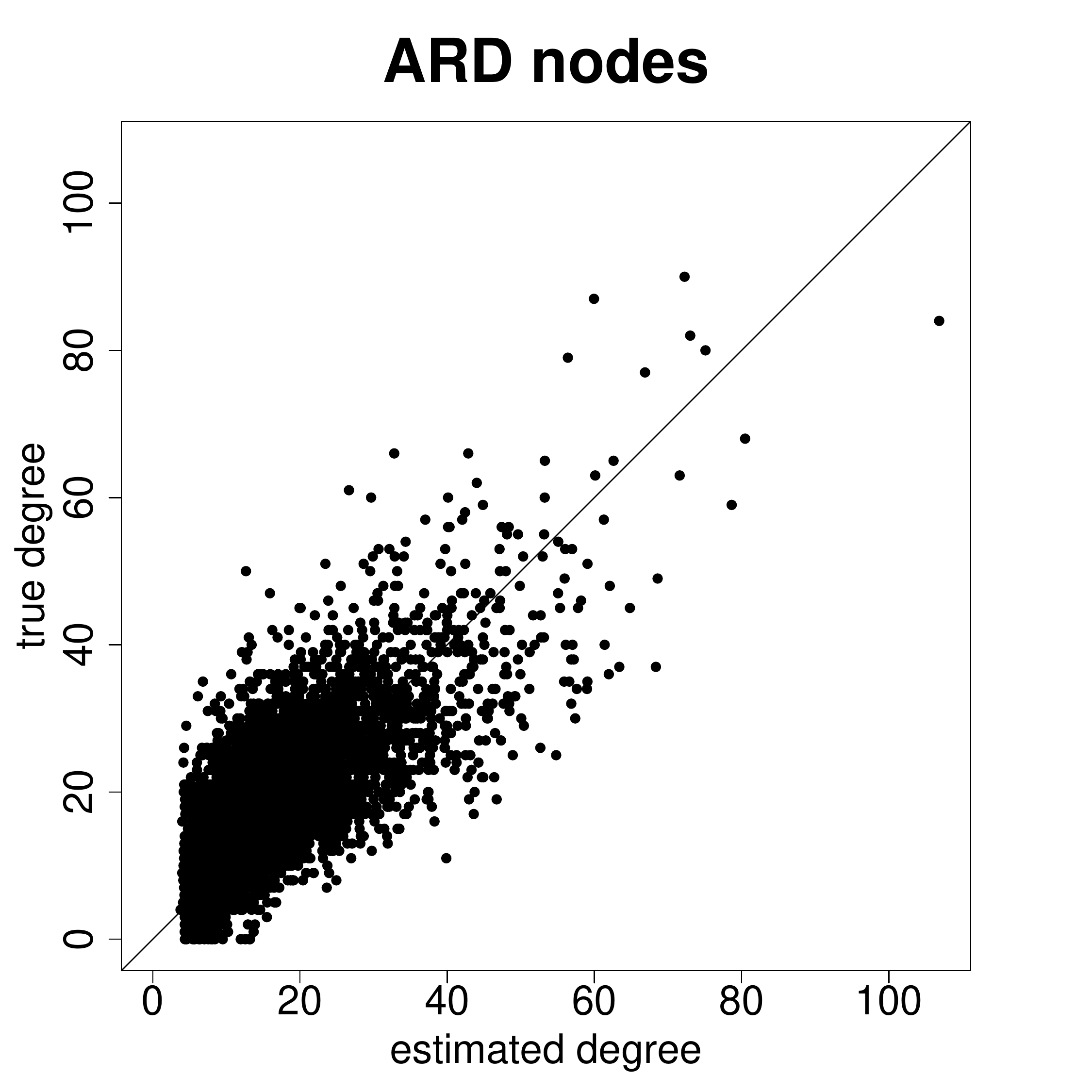}}
\subfloat[Eigenvector Centrality]{
\includegraphics[width=.28\textwidth]{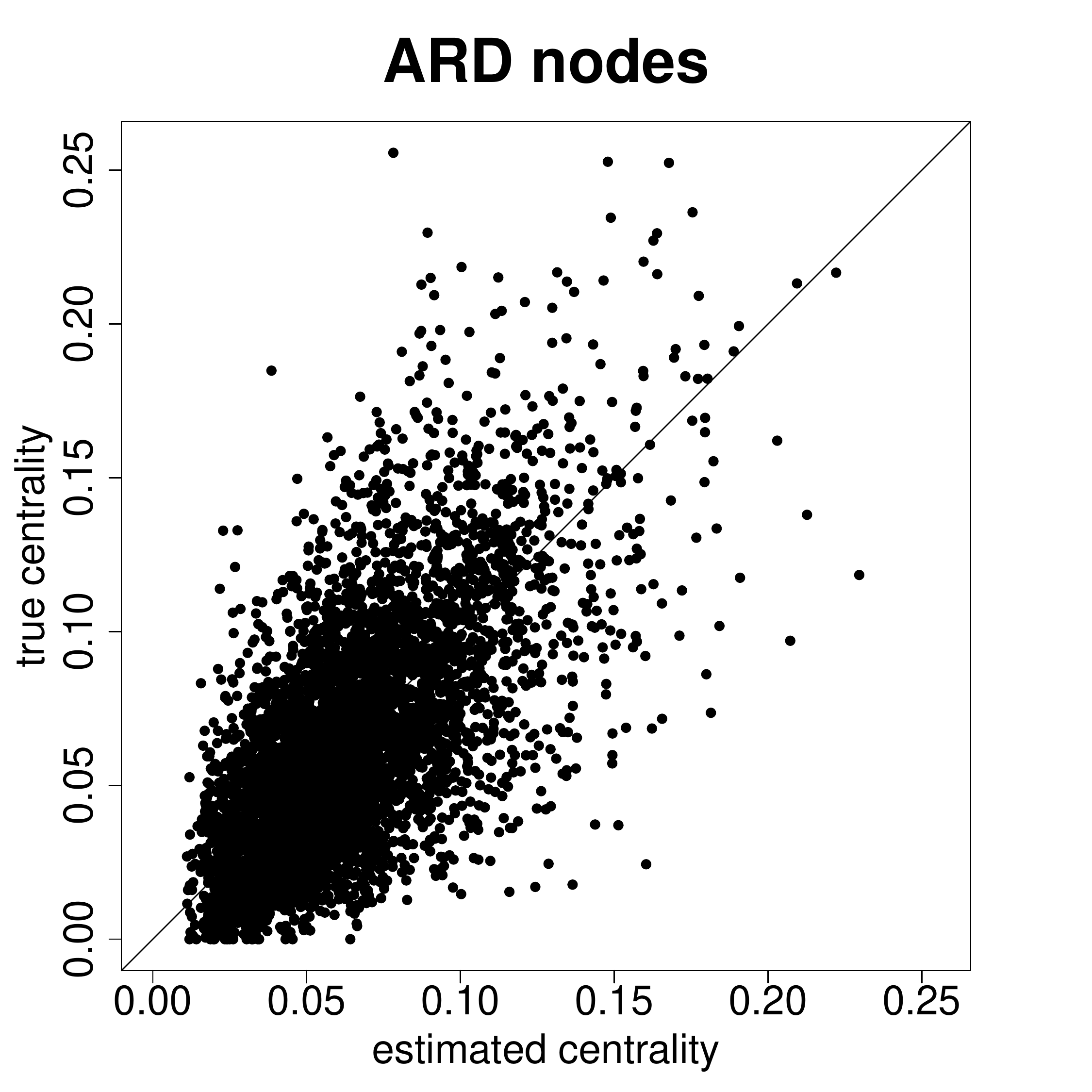}}
\subfloat[Clustering]{
\includegraphics[width=.28\textwidth]{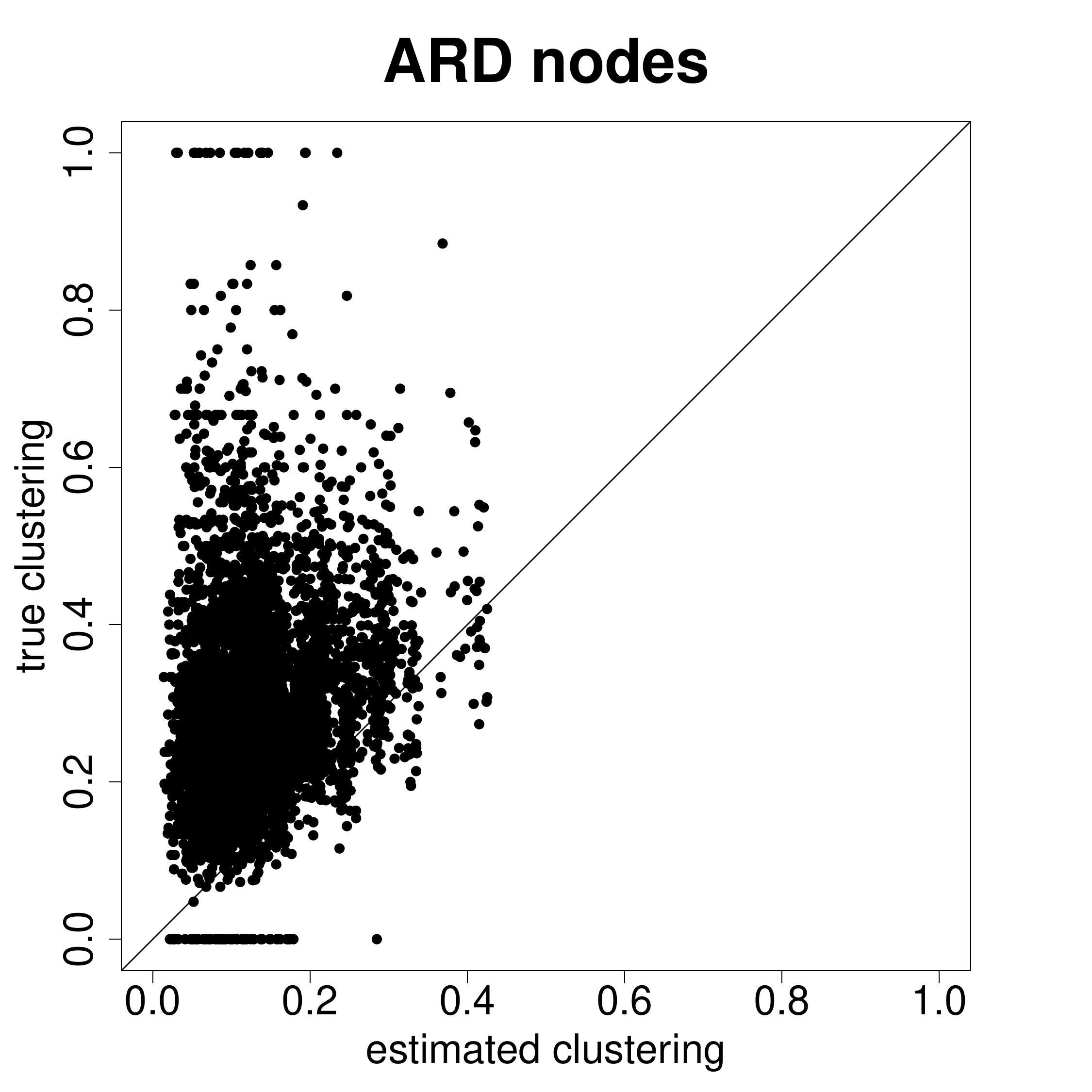}}

\smallskip
\subfloat[Closeness]{
\includegraphics[width=.28\textwidth]{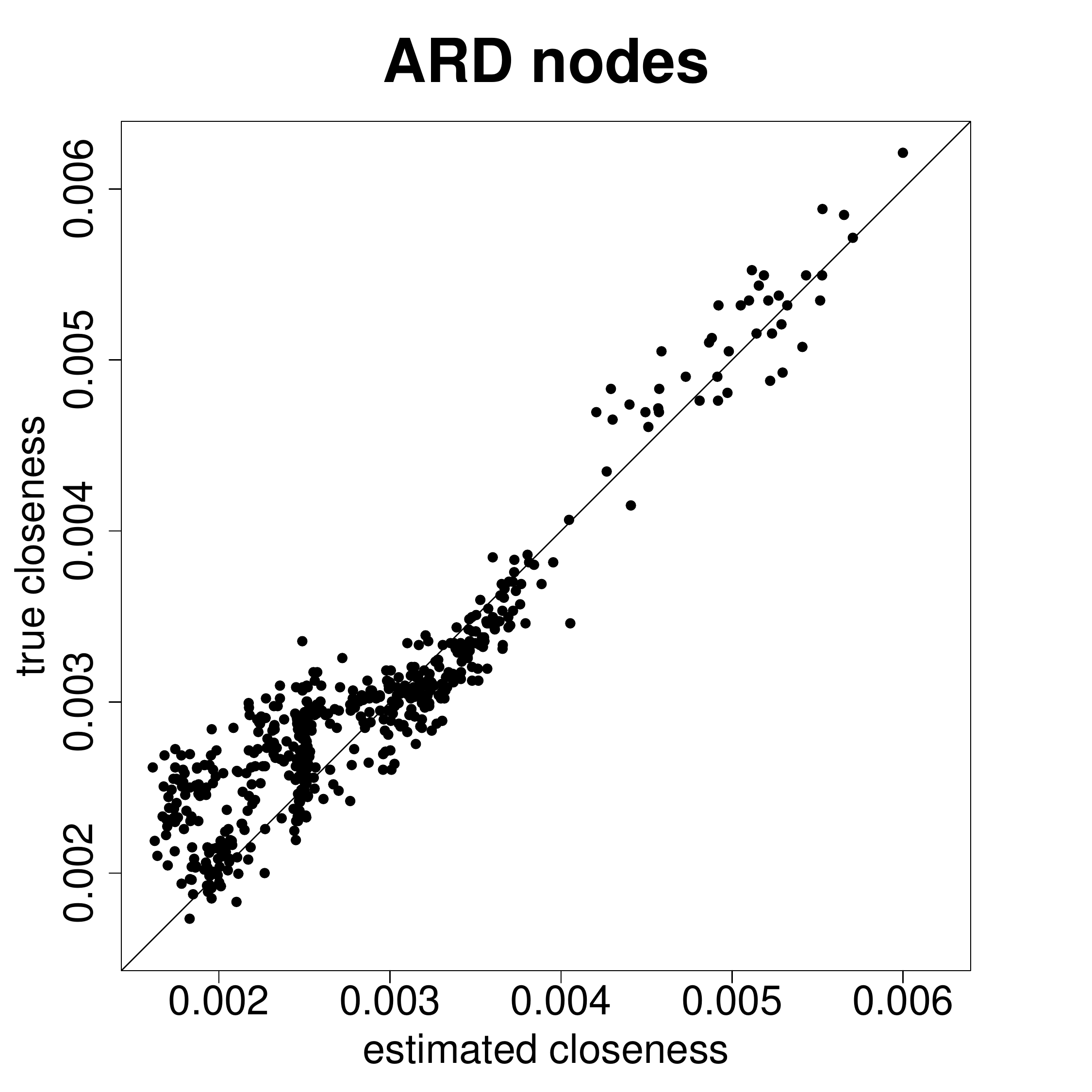}}
\subfloat[Betweenness]{
\includegraphics[width=.28\textwidth]{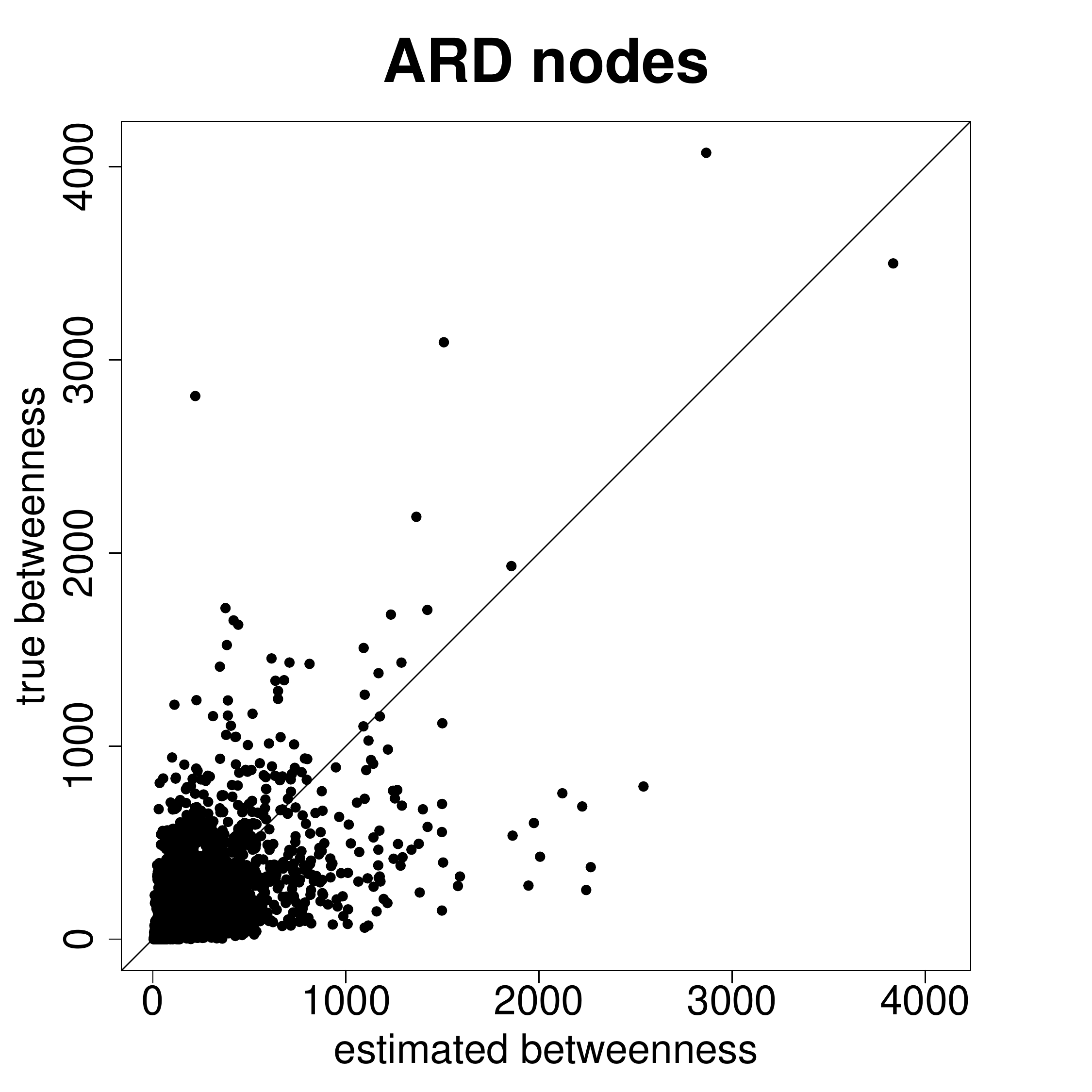}
}
\subfloat[Support]{
\includegraphics[width=.28\textwidth]{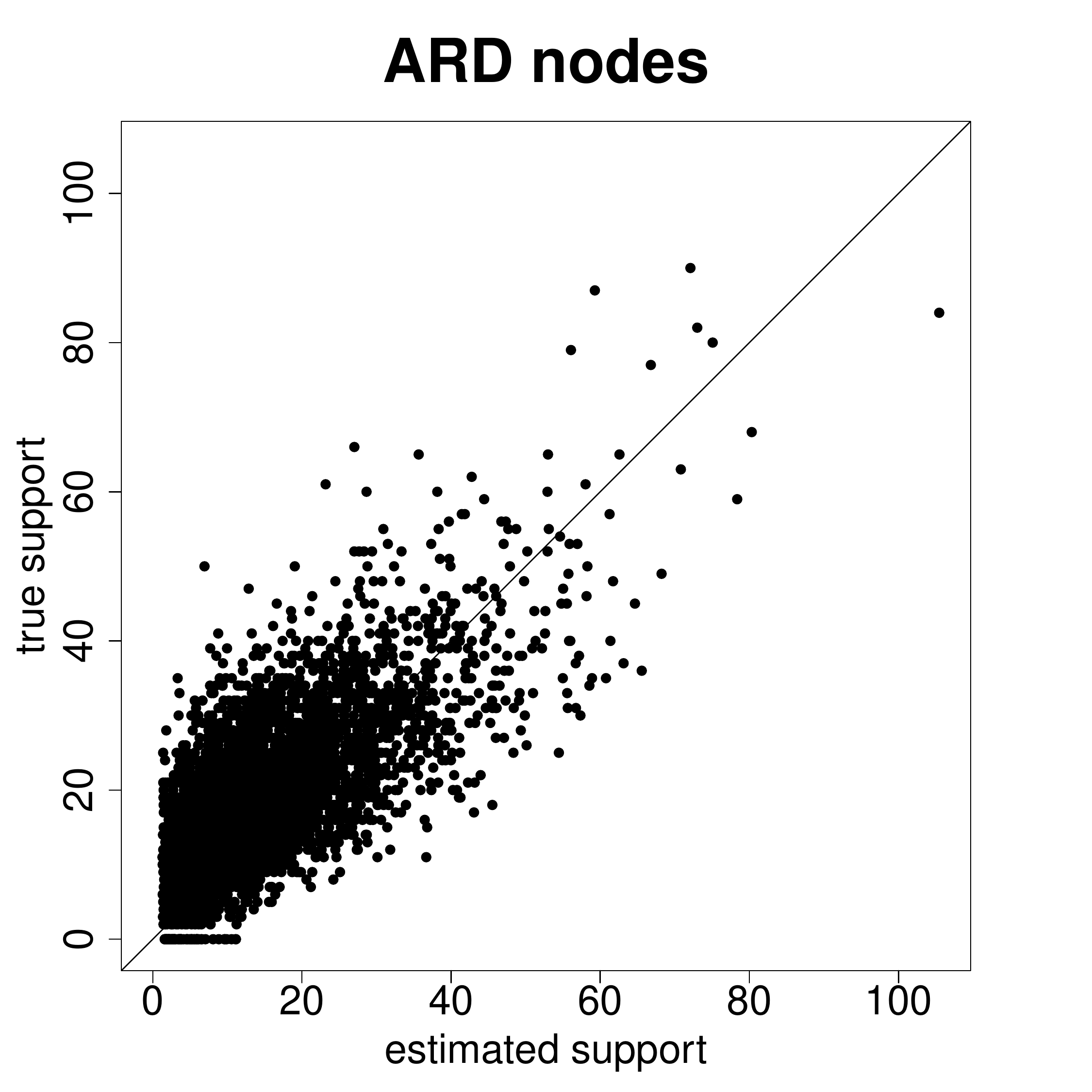}
}

\smallskip
\subfloat[Distance from seed]{
\includegraphics[width=.28\textwidth]{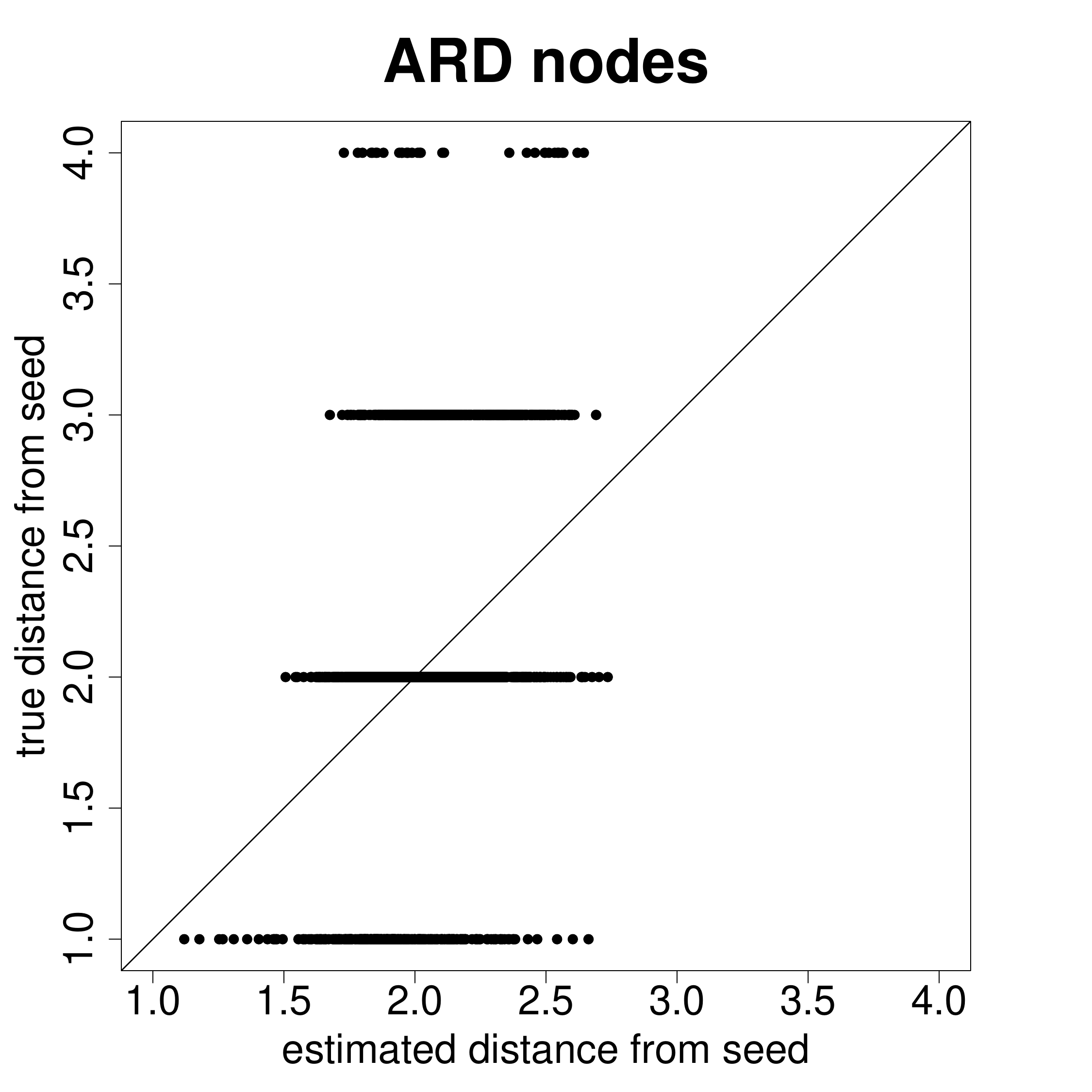}}
\subfloat[Node level clustering]{
\includegraphics[width=.28\textwidth]{plots/p_4/comparison_clustering_ARDnodes.pdf}}
\subfloat[Treated neighborhood share]{
\includegraphics[width=.28\textwidth]{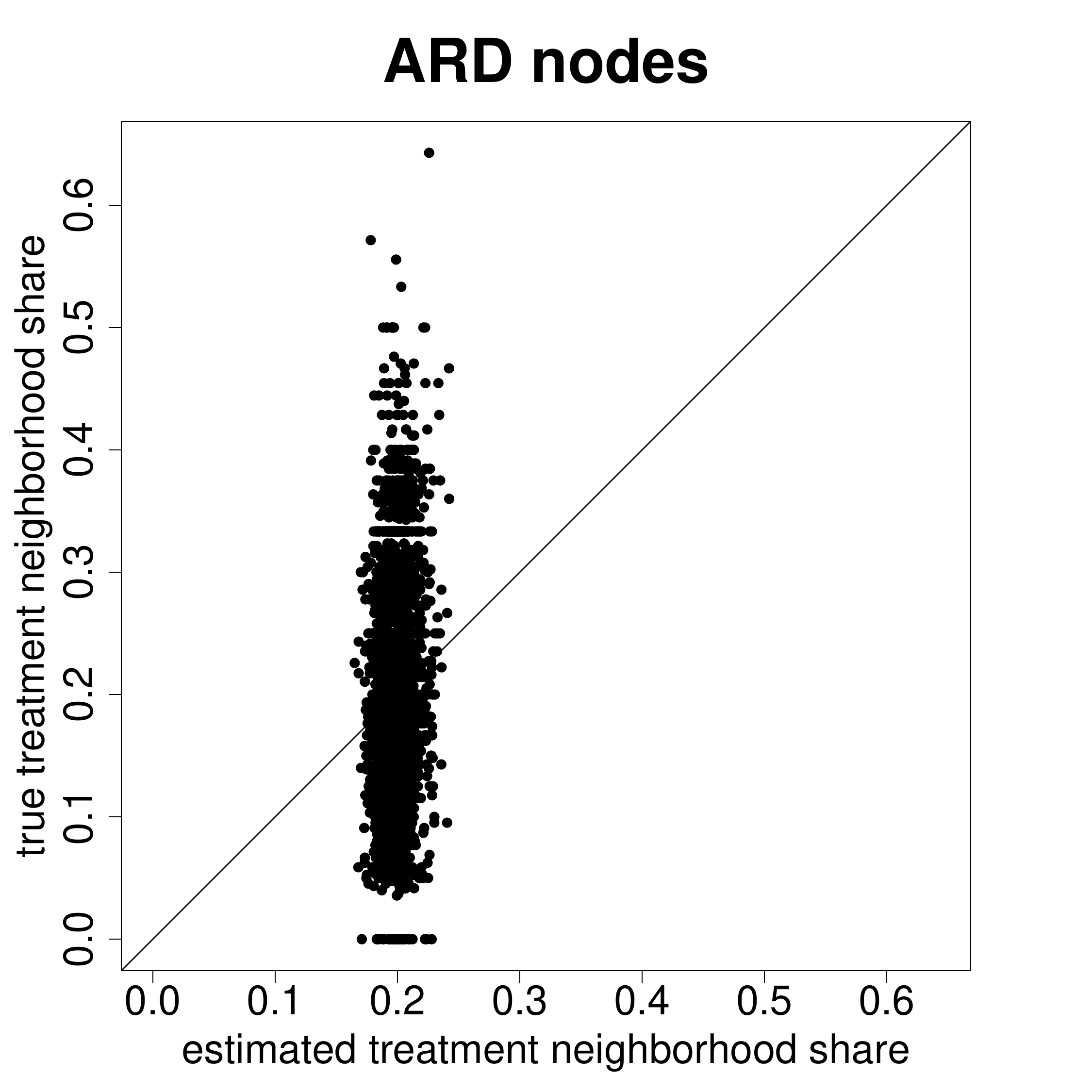}}

\caption{Node level measures estimation for households with ARD response in villages in Karnataka.  These plots show scatterplots across all villages with the estimated node level measure on the x-axis and the measure from the true underlying graph on the y-axis.}
\label{fig:p_4_karnataka_ARD_node_appendix}
\end{figure}

\begin{figure}[!h]
\centering
\subfloat[Degree]{
\includegraphics[width=.28\textwidth]{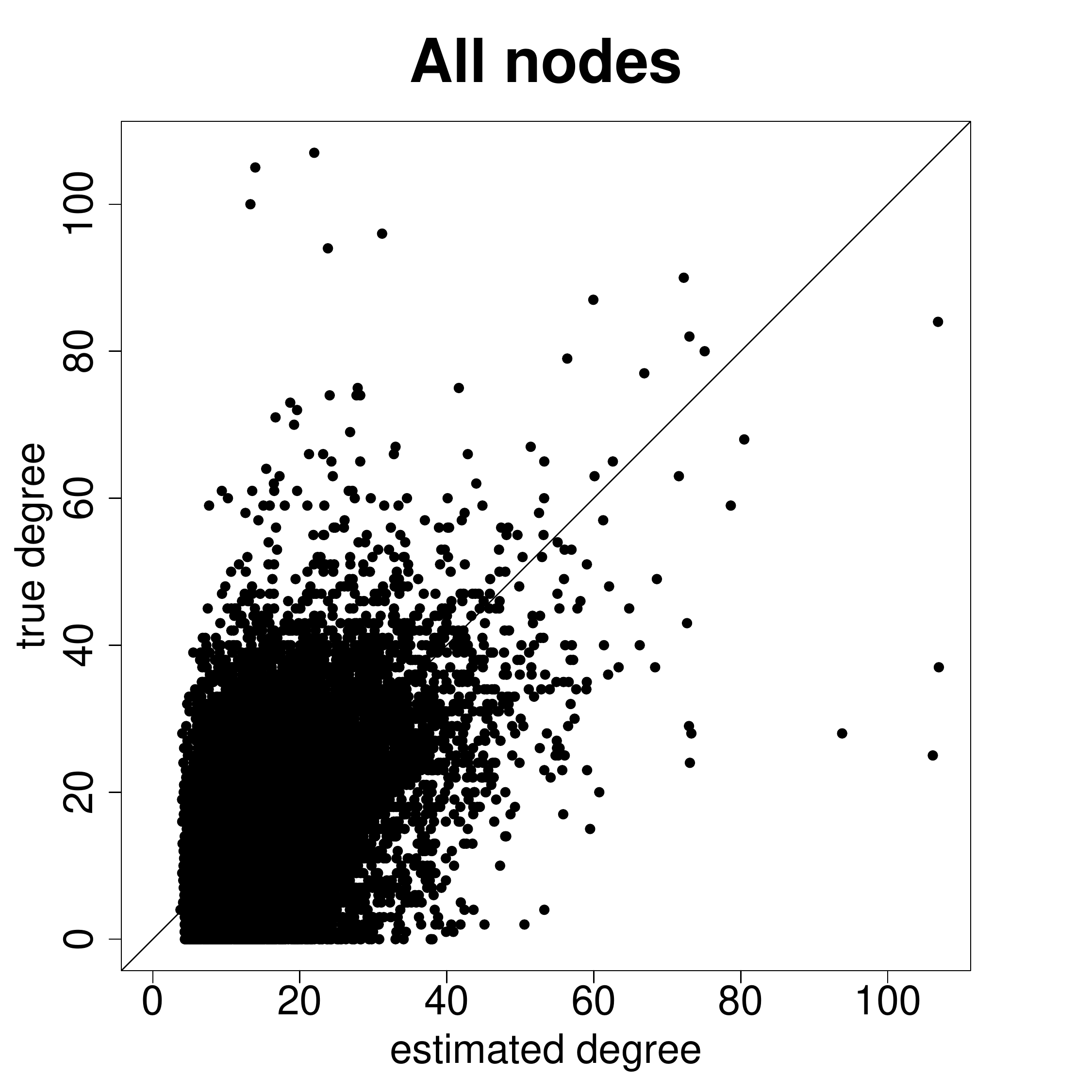}}
\subfloat[Eigenvector Centrality]{
\includegraphics[width=.28\textwidth]{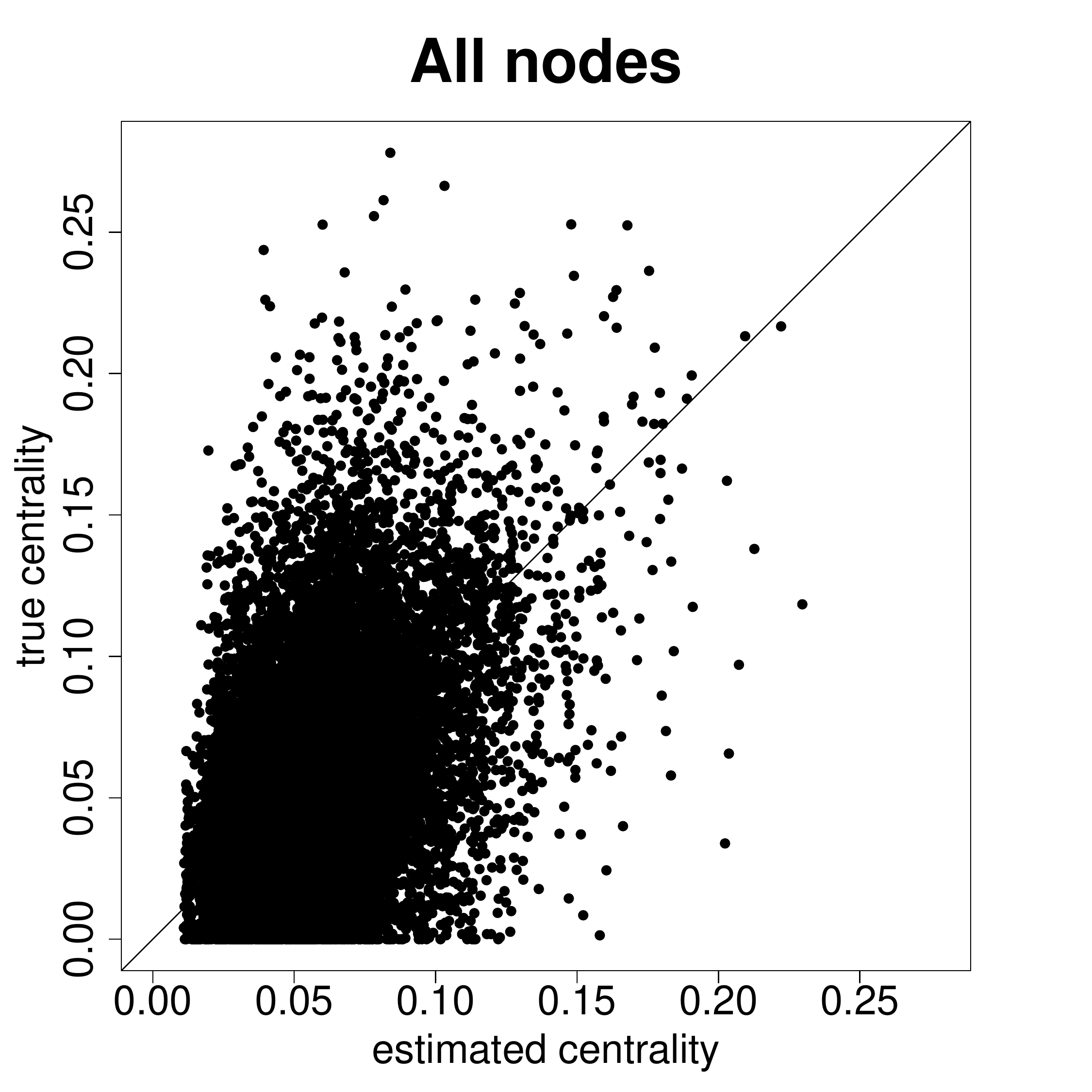}}
\subfloat[Clustering]{
\includegraphics[width=.28\textwidth]{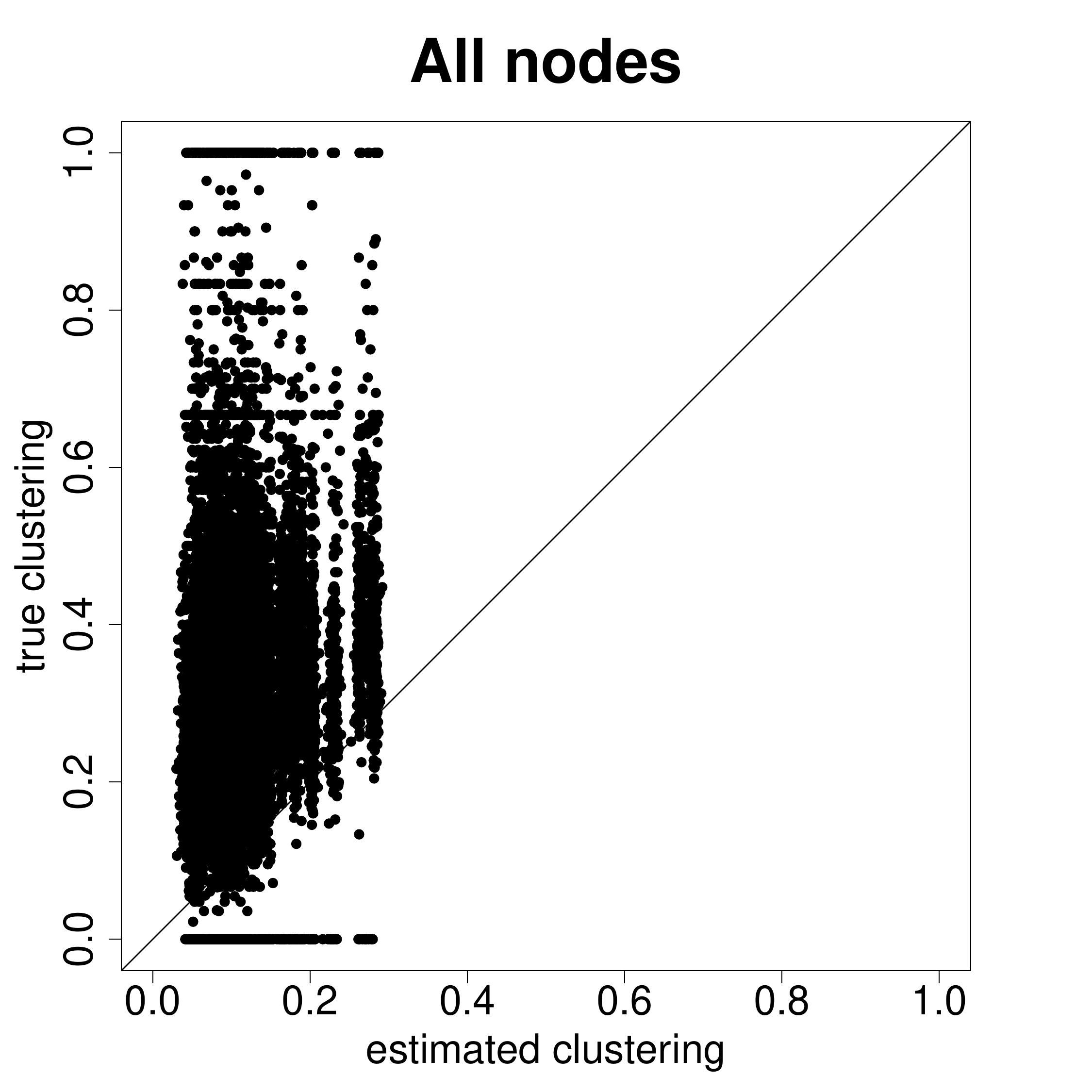}}

\smallskip
\subfloat[Closeness]{
\includegraphics[width=.28\textwidth]{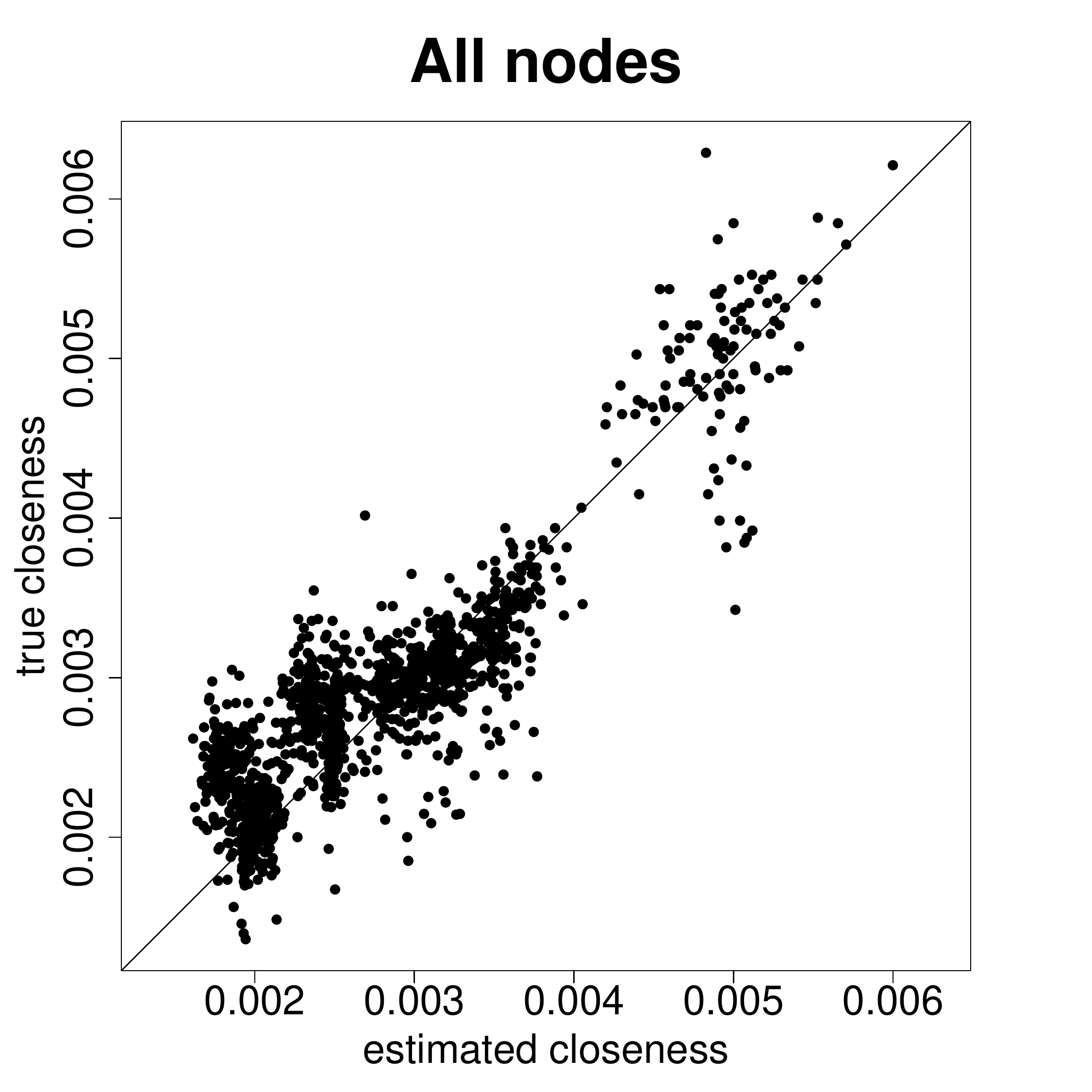}}
\subfloat[Betweenness]{
\includegraphics[width=.28\textwidth]{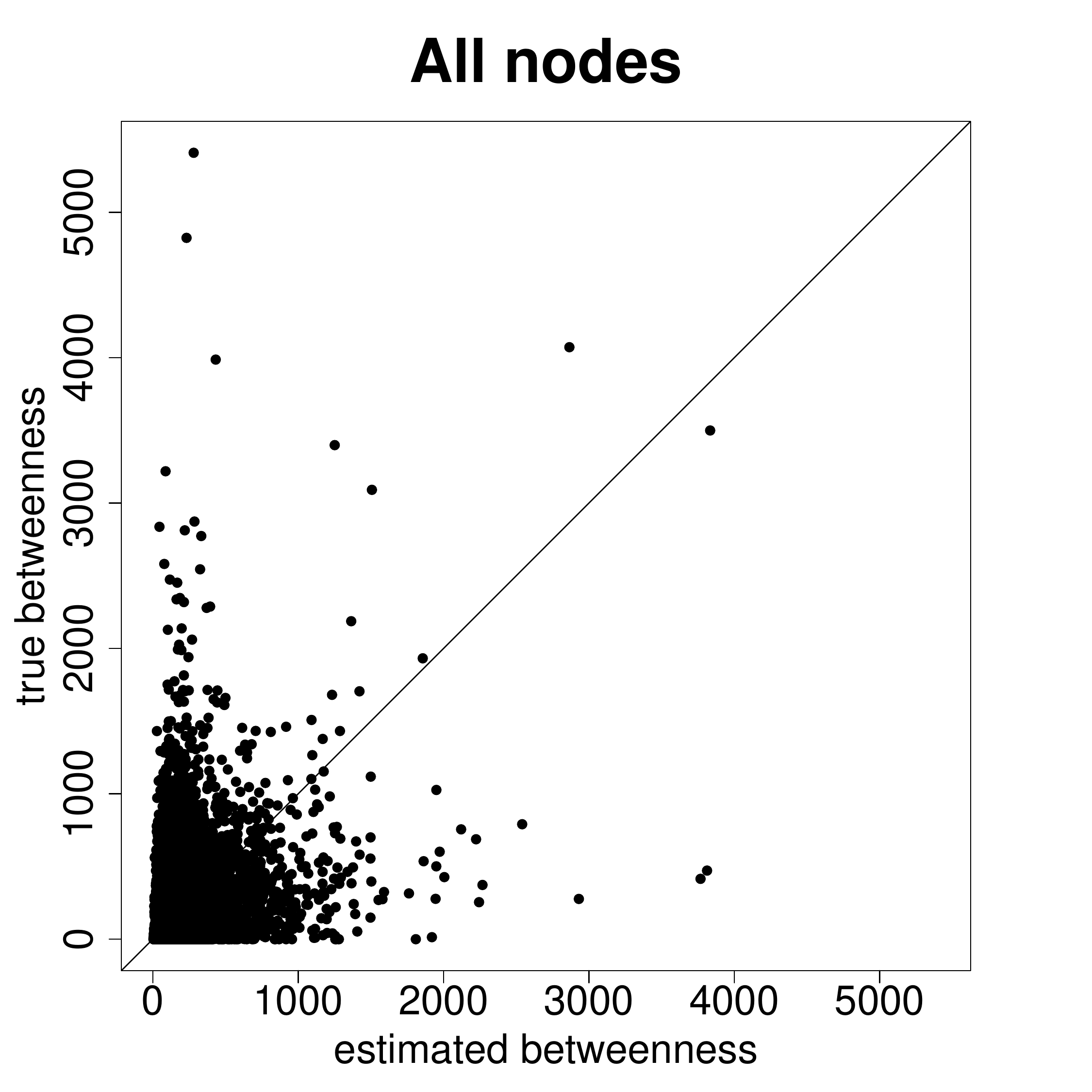}
}
\subfloat[Support]{
\includegraphics[width=.28\textwidth]{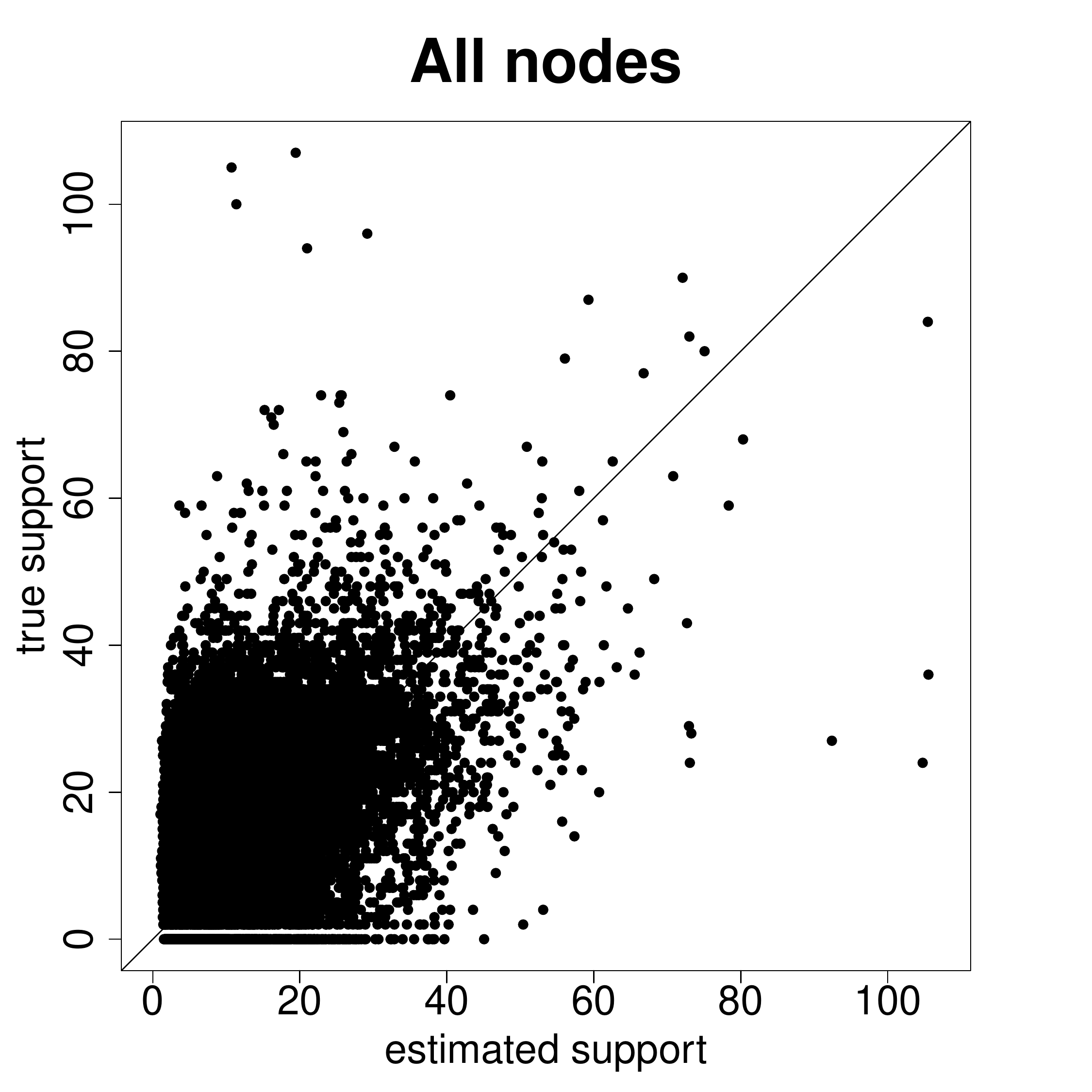}
}

\smallskip
\subfloat[Distance from seed]{
\includegraphics[width=.28\textwidth]{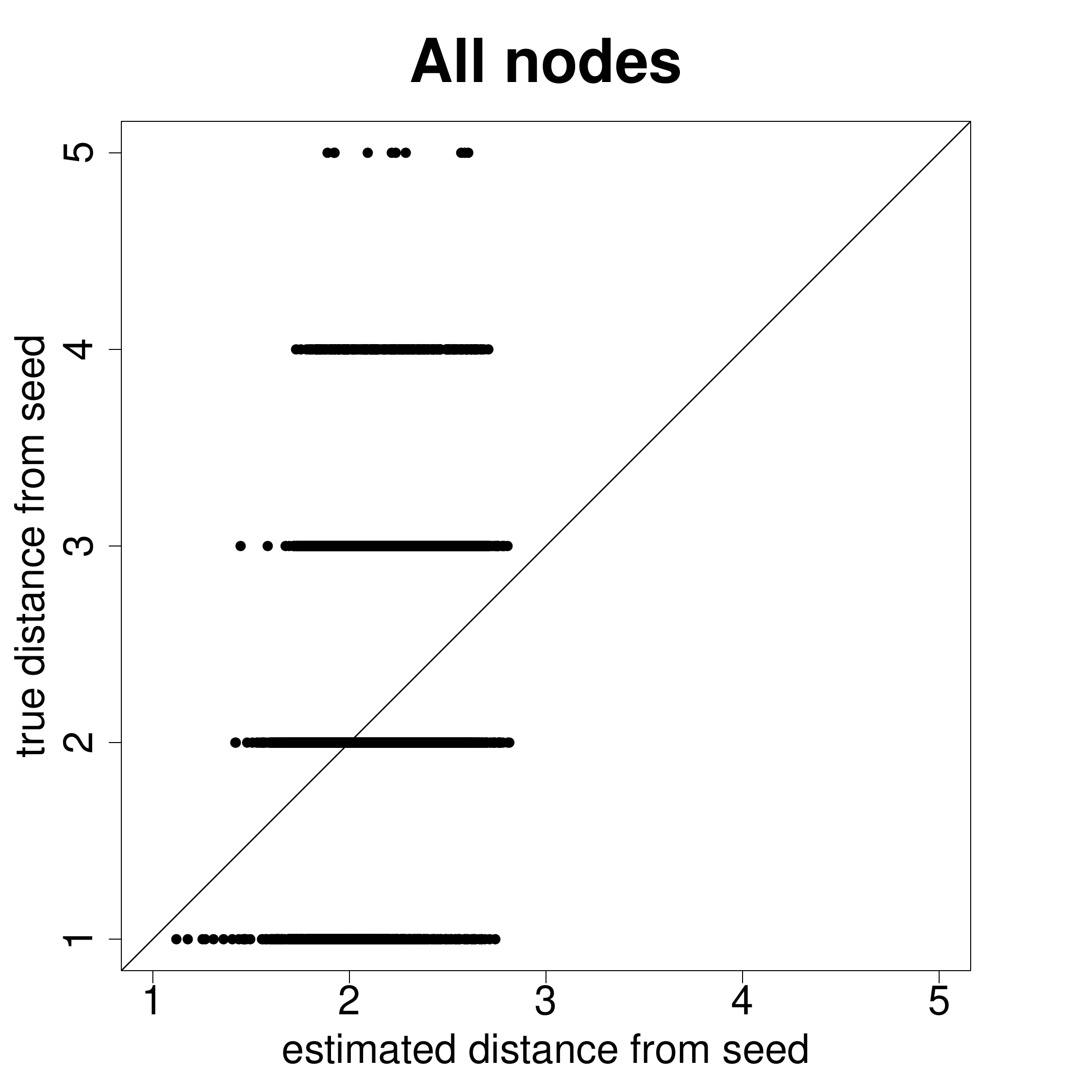}}
\subfloat[Node level clustering]{
\includegraphics[width=.28\textwidth]{plots/p_4/comparison_clustering_allnodes.pdf}}
\subfloat[Treated neighborhood share]{
\includegraphics[width=.28\textwidth]{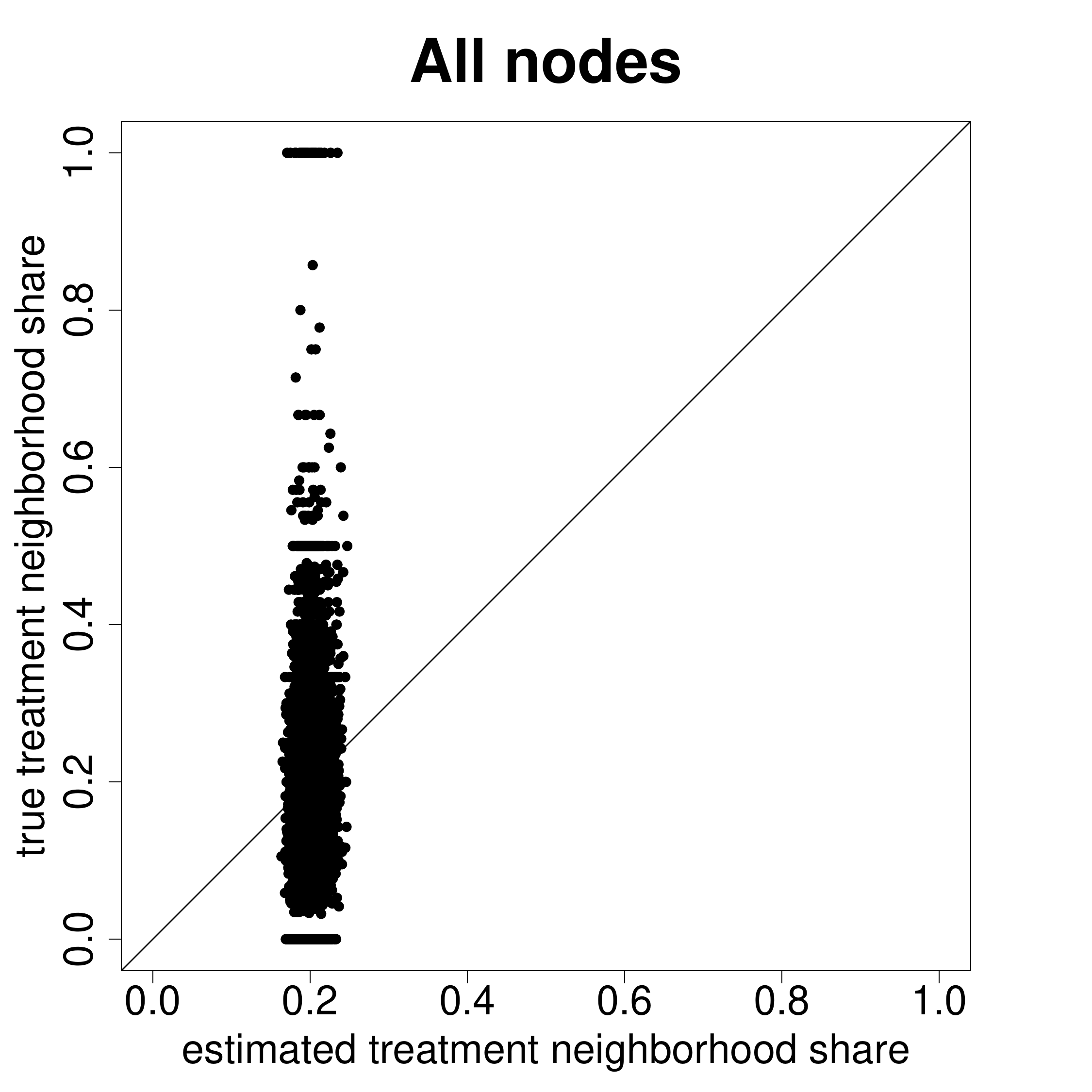}}

\caption{Node level measures estimation for households with all response in villages in Karnataka.  These plots show scatterplots across all villages with the estimated node level measure on the x-axis and the measure from the true underlying graph on the y-axis.}
\label{fig:p_4_karnataka_all_node_appendix}
\end{figure}

\begin{figure}[!h]
\centering
\subfloat[$\lambda_1(g)$]{
\includegraphics[width=.28\textwidth]{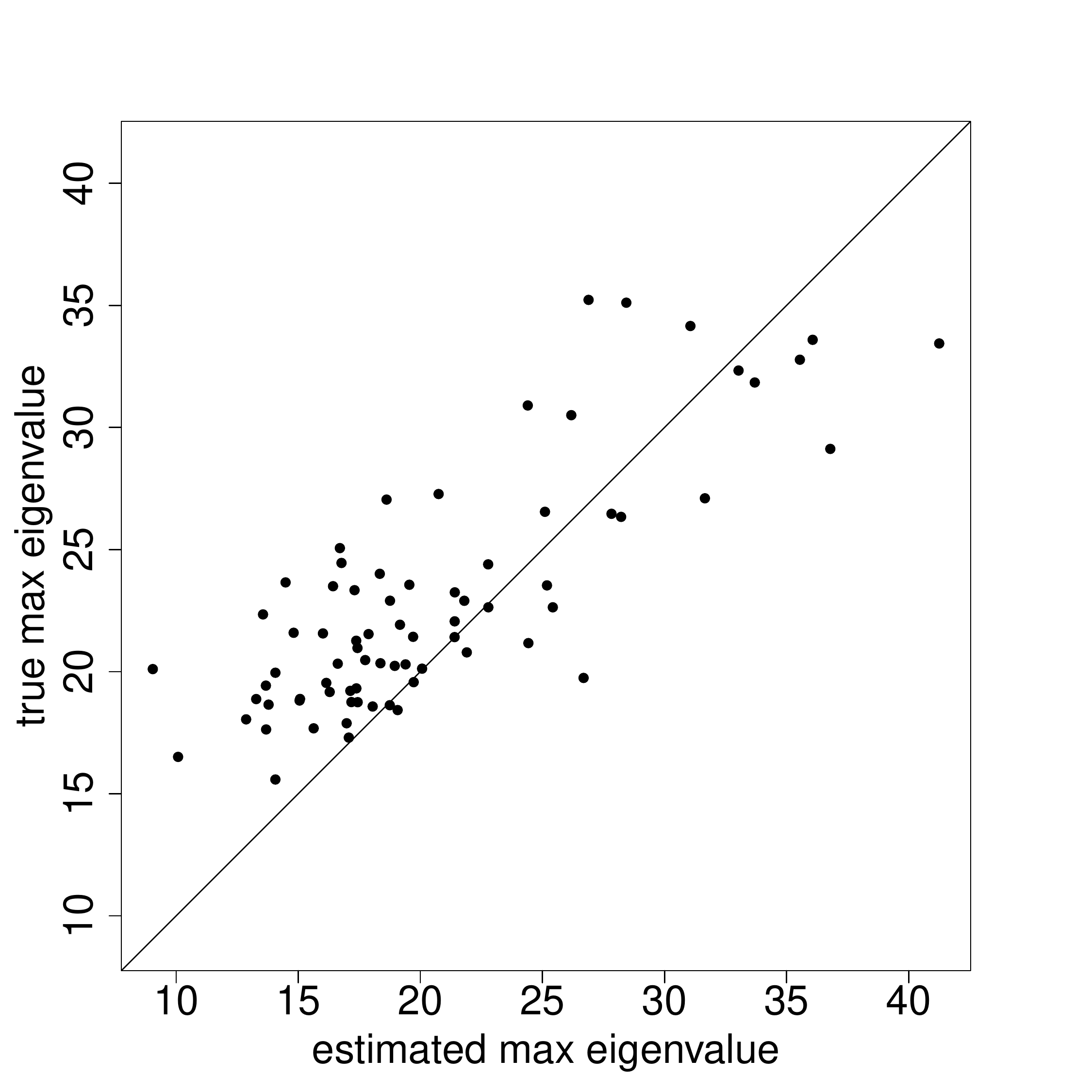}\qquad
}
\subfloat[Social Proximity]{
\includegraphics[width=.28\textwidth]{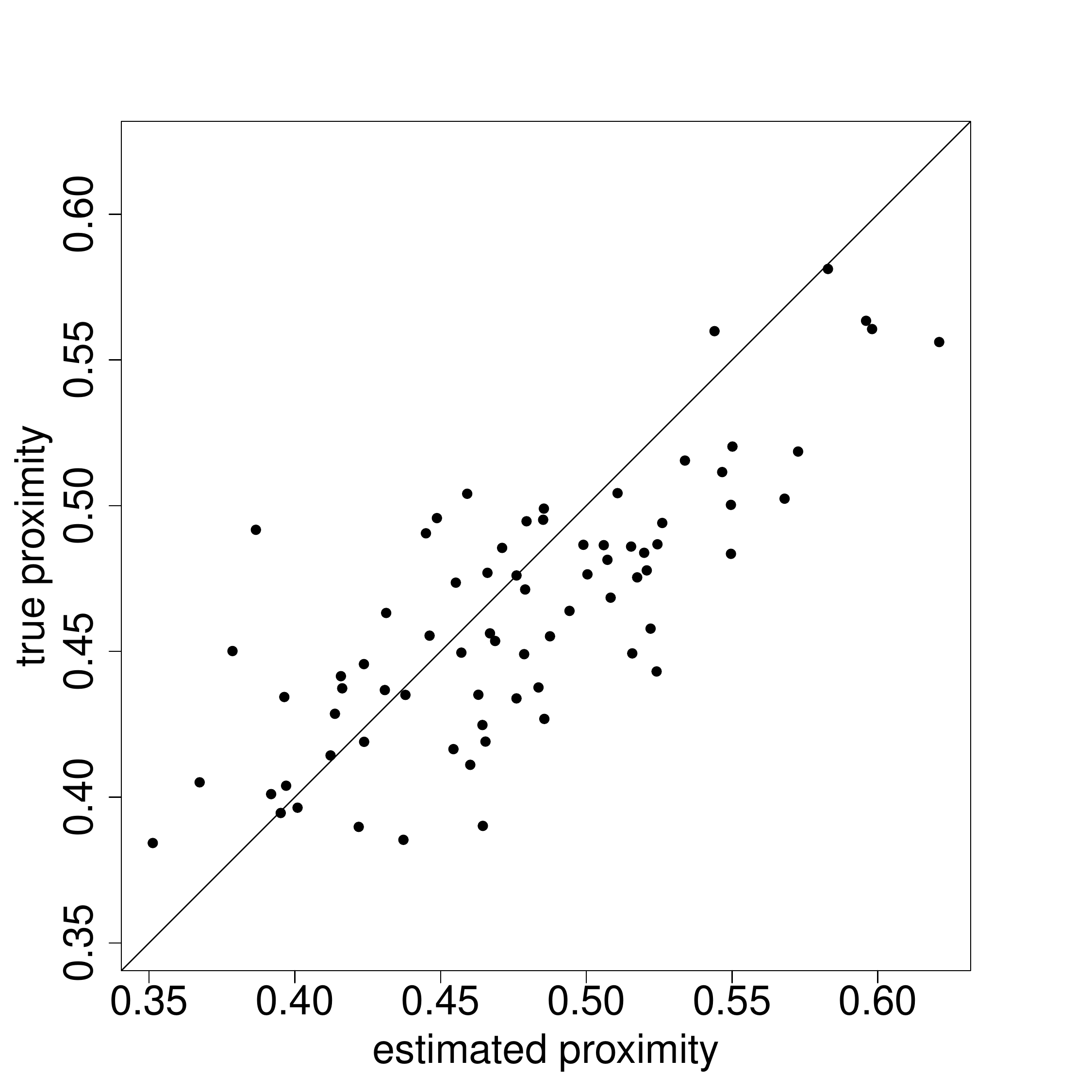}
}
\subfloat[Clustering]{
\includegraphics[width=.28\textwidth]{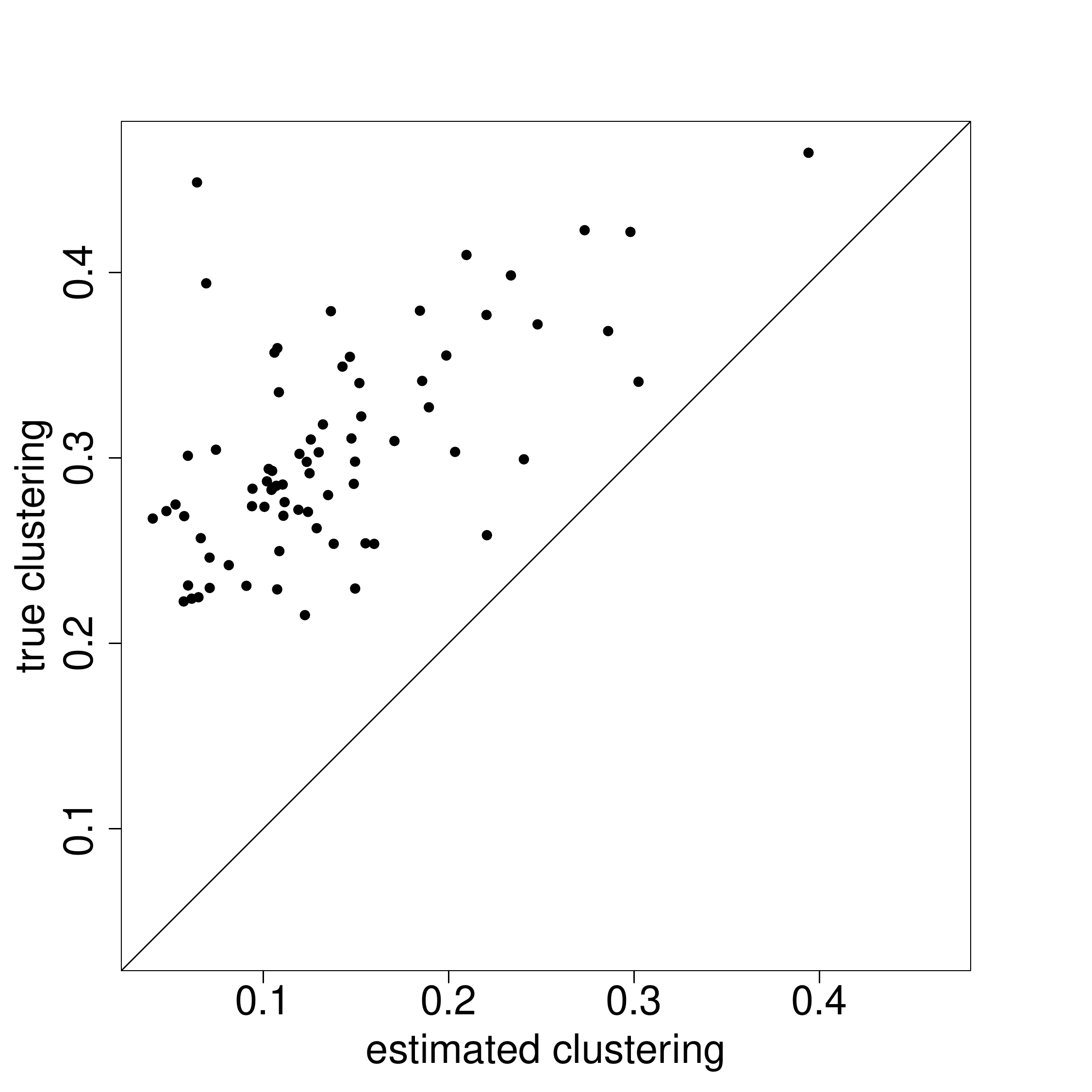}
}

\smallskip
\subfloat[Eigenvector Cut]{
\includegraphics[width=.28\textwidth]{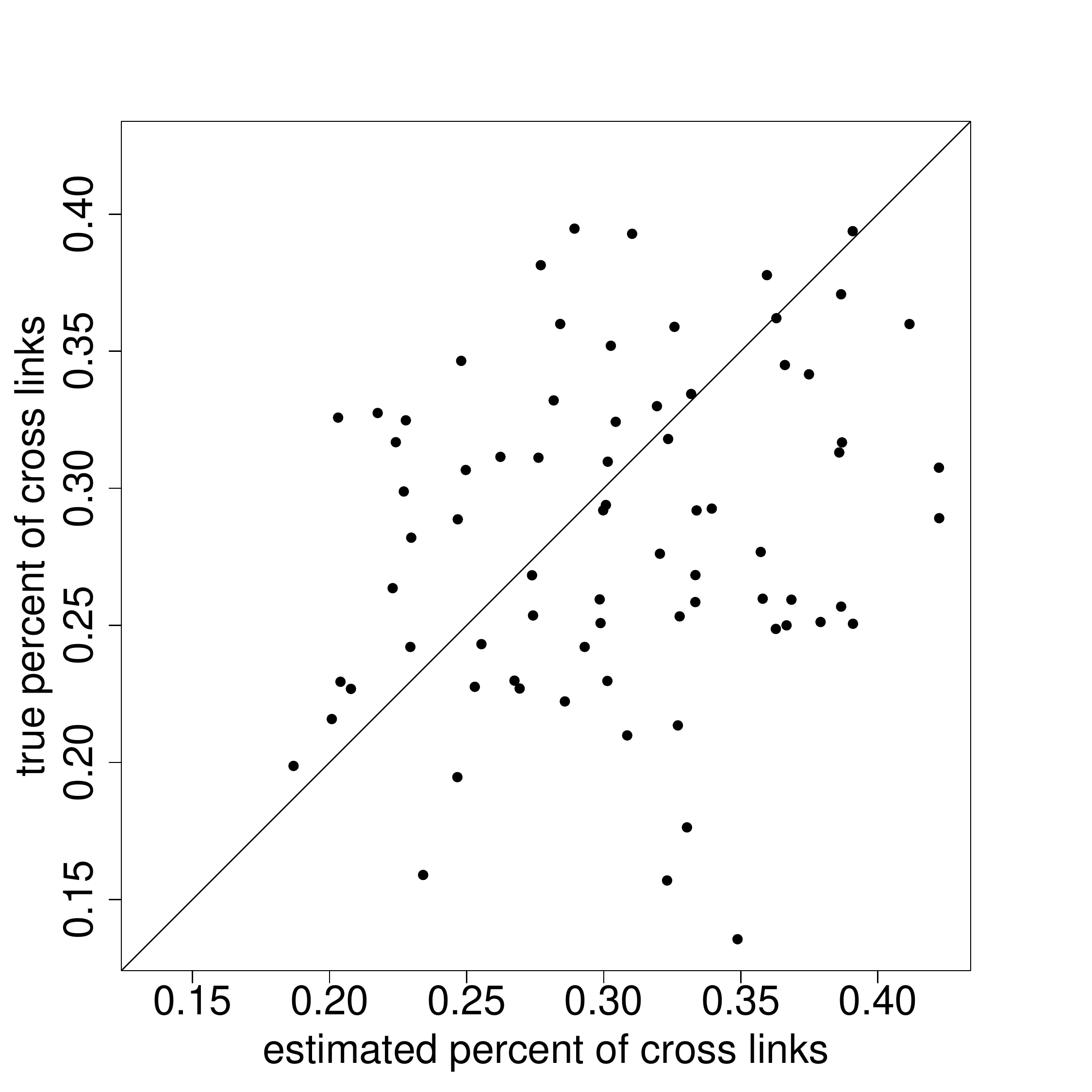}
}
\subfloat[Number of components]{
\includegraphics[width=.28\textwidth]{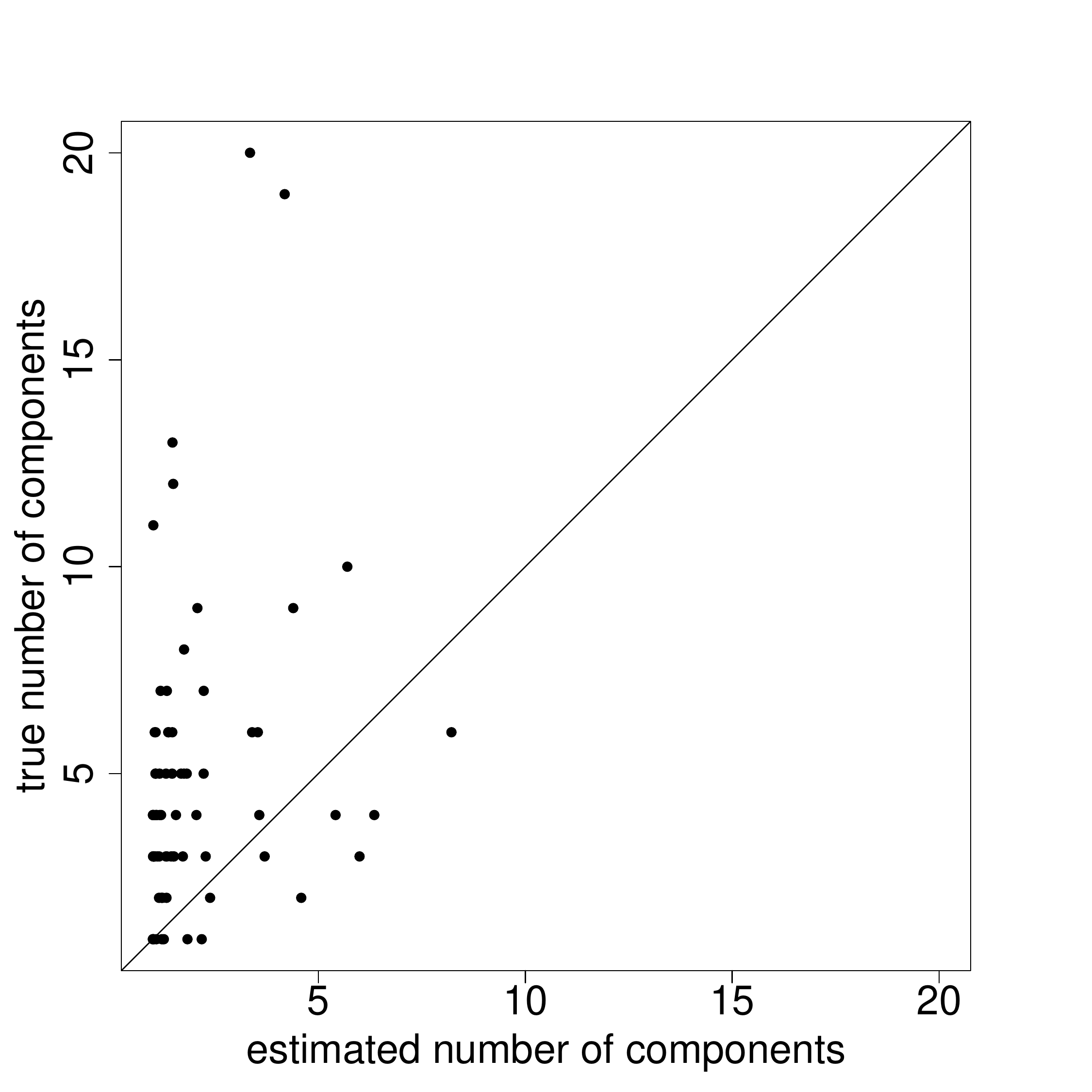}}
\subfloat[Average path length]{
\includegraphics[width=.28\textwidth]{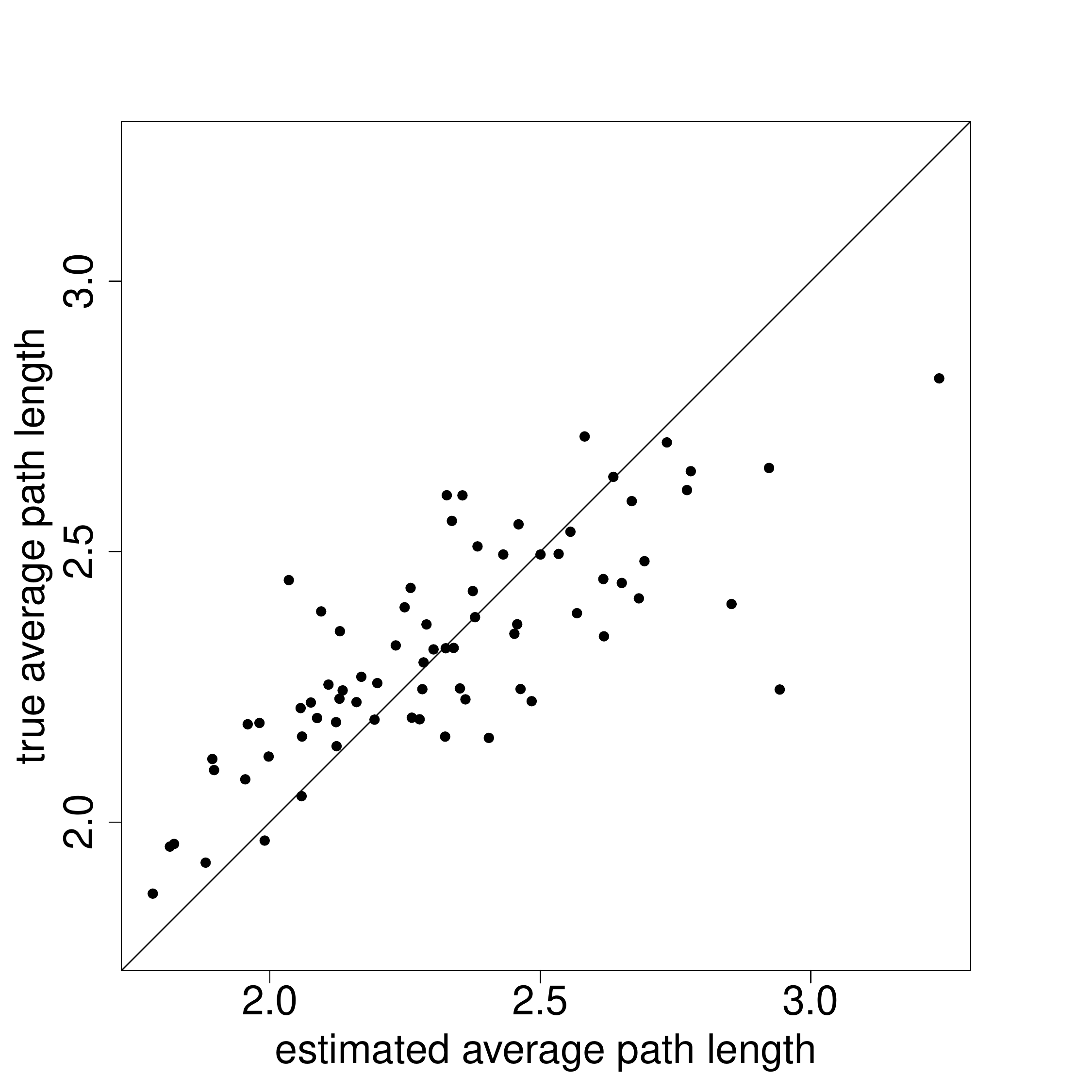}}

\smallskip
\subfloat[Diameter]{
\includegraphics[width=.28\textwidth]{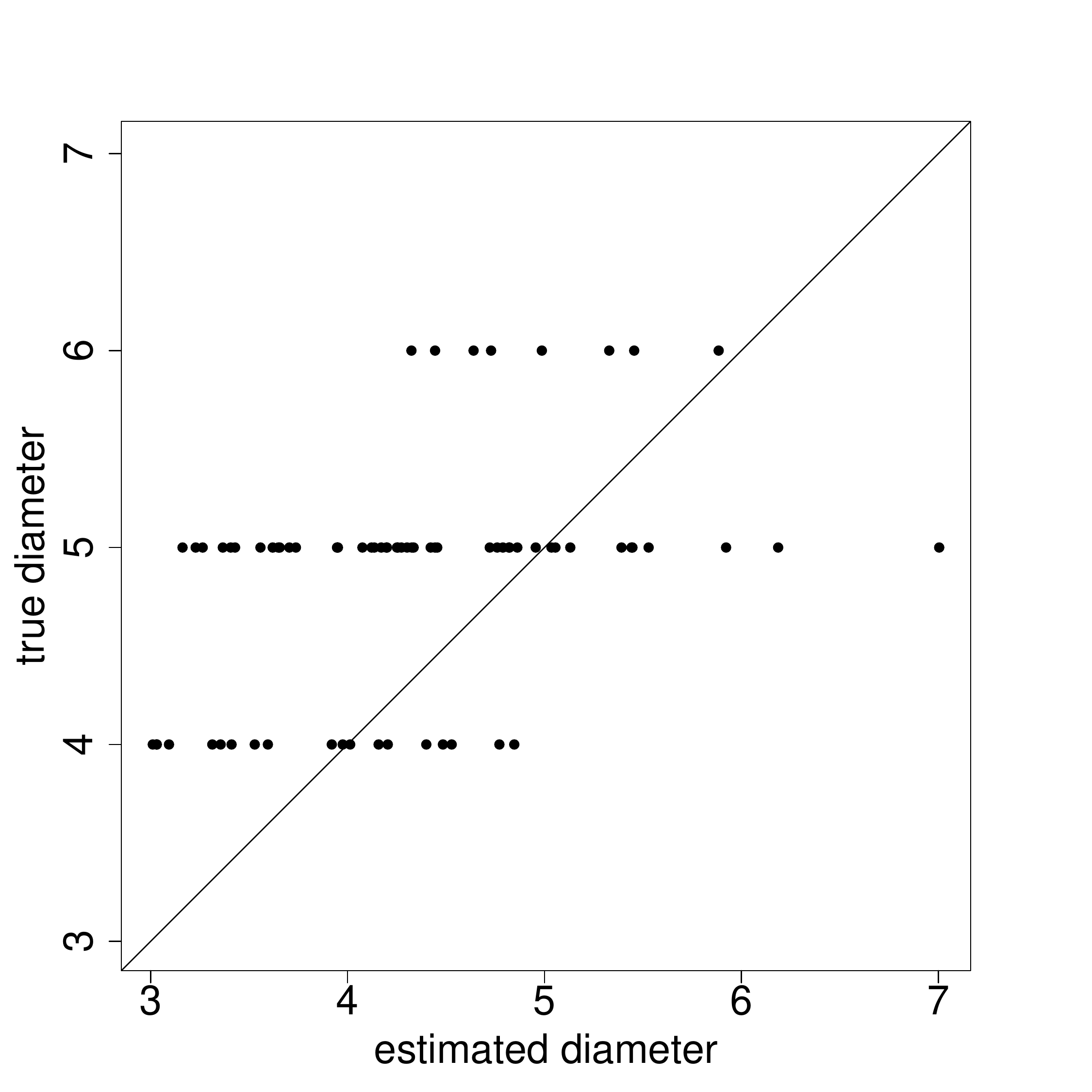}}
\subfloat[Fraction of giant component]{
\includegraphics[width=.28\textwidth]{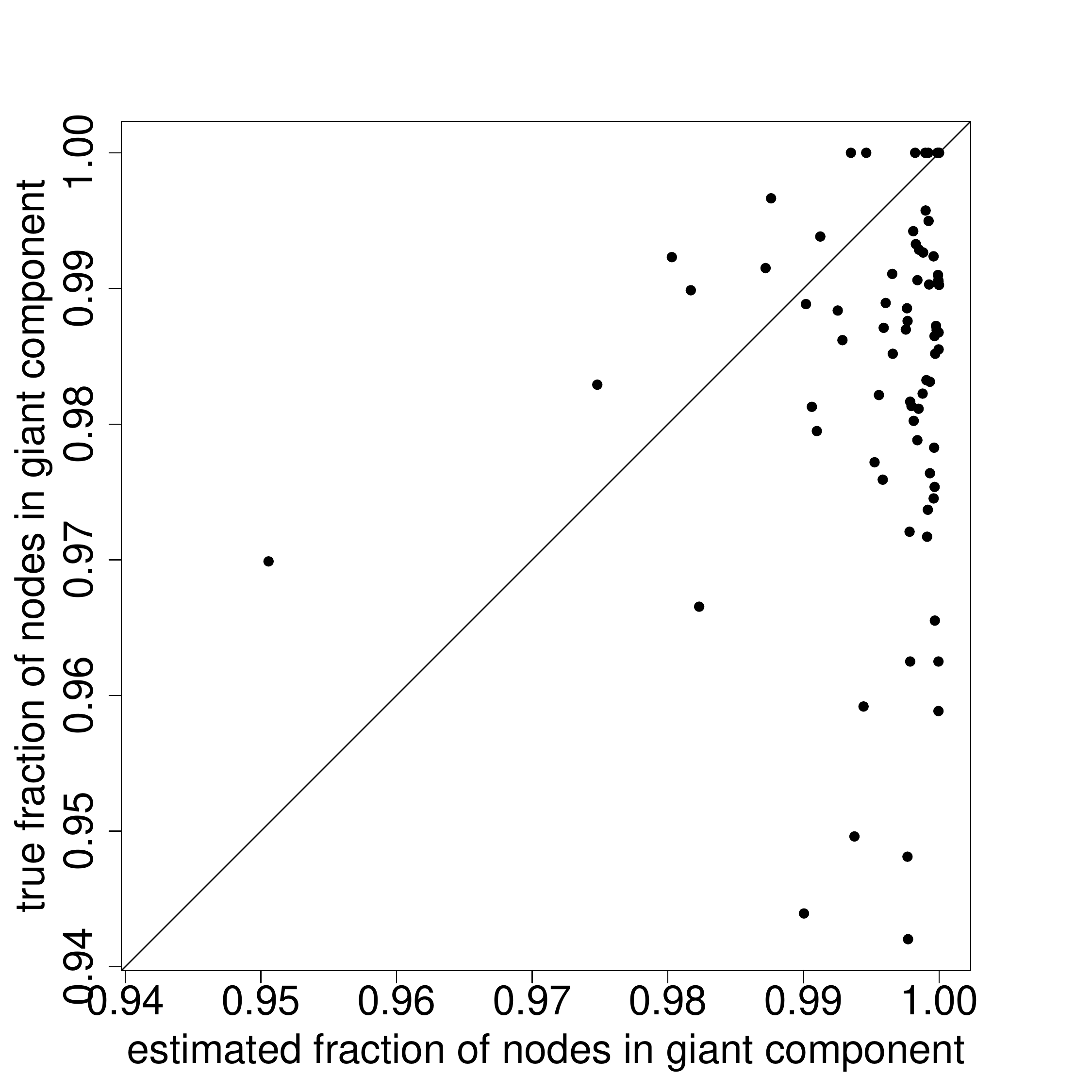}}

\caption{Network level measures estimation for households in villages in Karnataka.  These plots show scatterplots across all villages with the estimated network level measure on the x-axis and the measure from the true underlying graph on the y-axis.}
\label{fig:p_4_karnataka_network_appendix}
\end{figure}

\end{document}